\documentclass[12pt,a4paper, oneside]{book}

\usepackage[utf8]{inputenc}

\usepackage[top=2.7cm,left=2.7cm,right=2.7cm,bottom=2.7cm]{geometry}

\usepackage[numbers]{natbib}

\usepackage[titletoc]{appendix}
\usepackage{latexsym}
\usepackage{amsthm}
\usepackage{geometry} 
\usepackage{mathrsfs}
\usepackage[normalem]{ulem} 

\usepackage{calc}
\usepackage{setspace}

\usepackage{caption,tabularx,booktabs}
\usepackage{multirow}
\usepackage{array}
\newcolumntype{C}[1]{>{\centering\arraybackslash}m{#1}}

\usepackage{amsmath}
\usepackage{amssymb}
\usepackage{graphicx}
\usepackage{multicol}
\usepackage{epigraph} 
\usepackage{makeidx}

\usepackage{caption}
\usepackage{subcaption}

\usepackage{tikz}
\usepackage{tkz-graph}

\usetikzlibrary{shapes.geometric}
\usetikzlibrary{shapes.arrows}
\usepackage{array}

\usepackage[matrix, arrow]{xy}
\usetikzlibrary{matrix,arrows}
\usepackage{amsthm}
\usepackage{longtable}
\usepackage{fancyhdr}
\usetikzlibrary{positioning,calc}
\usetikzlibrary{matrix}
\usetikzlibrary{arrows}
\usepackage{slantsc}
\usepackage{lettrine}

\usepackage{comment}
\usepackage{color}

\definecolor{airforceblue}{rgb}{0.36, 0.54, 0.66}
\definecolor{amaranth}{rgb}{0.9, 0.17, 0.31}
\definecolor{arsenic}{rgb}{0.23, 0.27, 0.29}
\definecolor{bazaar}{rgb}{0.6, 0.47, 0.48}
\definecolor{darkscarlet}{rgb}{0.34, 0.01, 0.1}
\definecolor{prettyorange}{rgb}{0.9372549019607843, 0.5019607843137255, 0.4}

\usepackage[]{hyperref}
\hypersetup{
  linkcolor  = blue,
  citecolor  = magenta,
  urlcolor   = magenta,
  colorlinks = true,
}

\usepackage{marvosym}
\usepackage{setspace} 

\usepackage[titletoc]{appendix}

\makeindex

\newcolumntype{M}[1]{>{\centering\arraybackslash}m{#1}}
\newcolumntype{N}{@{}m{0pt}@{}}

\usepackage{fancyhdr}
\usepackage{lastpage}
 
\newcommand\phantomarrow[2]{%
  \setbox0=\hbox{$\displaystyle #1\to$}%
  \hbox to \wd0{%
    $#2\mapstochar
     \cleaders\hbox{$\mkern-1mu\relbar\mkern-3mu$}\hfill
     \mkern-7mu\rightarrow$}%
 \,}
  
\usepackage{pgfgantt} 
\definecolor{tartorange}{rgb}{0.96, 0.28, 0.25}
\definecolor{lightkaki}{rgb}{0.96,0.94,0.53}

\usepackage{amsthm}
\theoremstyle{definition}
\newtheorem{definition}{Definition}[section]
\newtheorem{observation}{Observation}[section]
\newtheorem{example}{Example}[section]
\newtheorem{remark}{Remark}[section]

\newtheorem{theorem}{Theorem}[section]

\newtheorem{corollary}{Corollary}[section]
\newtheorem{proposition}{Proposition}[section]

\usepackage{mathtools}

\DeclarePairedDelimiter{\diagfences}{(}{)}
\newcommand{\diag}{\operatorname{diag}\diagfences}

\usepackage{tikz-cd}

\usepackage{epigraph}
\DeclareMathOperator{\Ima}{Im}

\usepackage[]{algorithm2e}

\usepackage{csquotes}
\usepackage{epigraph}

\usepackage{etoolbox}
    \makeatletter
    \newlength\epitextskip
    \pretocmd{\@epitext}{\em}{}{}
    \apptocmd{\@epitext}{\em}{}{}
    \patchcmd{\epigraph}{\@epitext{#1}\\}{\@epitext{#1}\\[\epitextskip]}{}{}
    \makeatother

\usepackage{float}
\usepackage{ragged2e}

\makeatletter
\def\bbordermatrix#1{\begingroup \m@th
  \@tempdima 4.75\p@
  \setbox\z@\vbox{%
    \def\cr{\crcr\noalign{\kern2\p@\global\let\cr\endline}}%
    \ialign{$##$\hfil\kern2\p@\kern\@tempdima&\thinspace\hfil$##$\hfil
      &&\quad\hfil$##$\hfil\crcr
      \omit\strut\hfil\crcr\noalign{\kern-\baselineskip}%
      #1\crcr\omit\strut\cr}}%
  \setbox\tw@\vbox{\unvcopy\z@\global\setbox\@ne\lastbox}%
  \setbox\tw@\hbox{\unhbox\@ne\unskip\global\setbox\@ne\lastbox}%
  \setbox\tw@\hbox{$\kern\wd\@ne\kern-\@tempdima\left[\kern-\wd\@ne
    \global\setbox\@ne\vbox{\box\@ne\kern2\p@}%
    \vcenter{\kern-\ht\@ne\unvbox\z@\kern-\baselineskip}\,\right]$}%
  \null\;\vbox{\kern\ht\@ne\box\tw@}\endgroup}
\makeatother

\usepackage[thinlines]{easytable}
\usepackage{lscape}
\usepackage{booktabs}
\usepackage{xcolor,colortbl}
\definecolor{lavender}{rgb}{0.9, 0.9, 0.98}

\usepackage{comment}
\usepackage{bm}

\NewDocumentCommand{\dgal}{sO{}m}{%
  \IfBooleanTF{#1}
    {\dgalext{#3}}
    {\dgalx[#2]{#3}}%
}

\NewDocumentCommand{\dgalext}{m}{%
  \sbox0{%
    \mathsurround=0pt 
    $\left\{\vphantom{#1}\right.\kern-\nulldelimiterspace$%
  }%
  \sbox2{\{}%
  \ifdim\ht0=\ht2
    \{\kern-.625\wd2 \{#1\}\kern-.625\wd2 \}%
  \else
    \left\{\kern-.7\wd0\left\{#1\right\}\kern-.7\wd0\right\}%
  \fi
}

\NewDocumentCommand{\dgalx}{om}{%
  \sbox0{\mathsurround=0pt$#1\{$}%
  \sbox2{\{}%
  \ifdim\ht0=\ht2
    \{\kern-.625\wd2 \{#2\}\kern-.625\wd2 \}%
  \else
    \mathopen{#1\{\kern-.7\wd0 #1\{}
    #2
    \mathclose{#1\}\kern-.7\wd0 #1\}}
  \fi
}

\begin{document}

\newpage

\clearpage
\thispagestyle{empty}

\begin{center}

\begin{figure}[h!]
\centering
\includegraphics[scale=0.08]{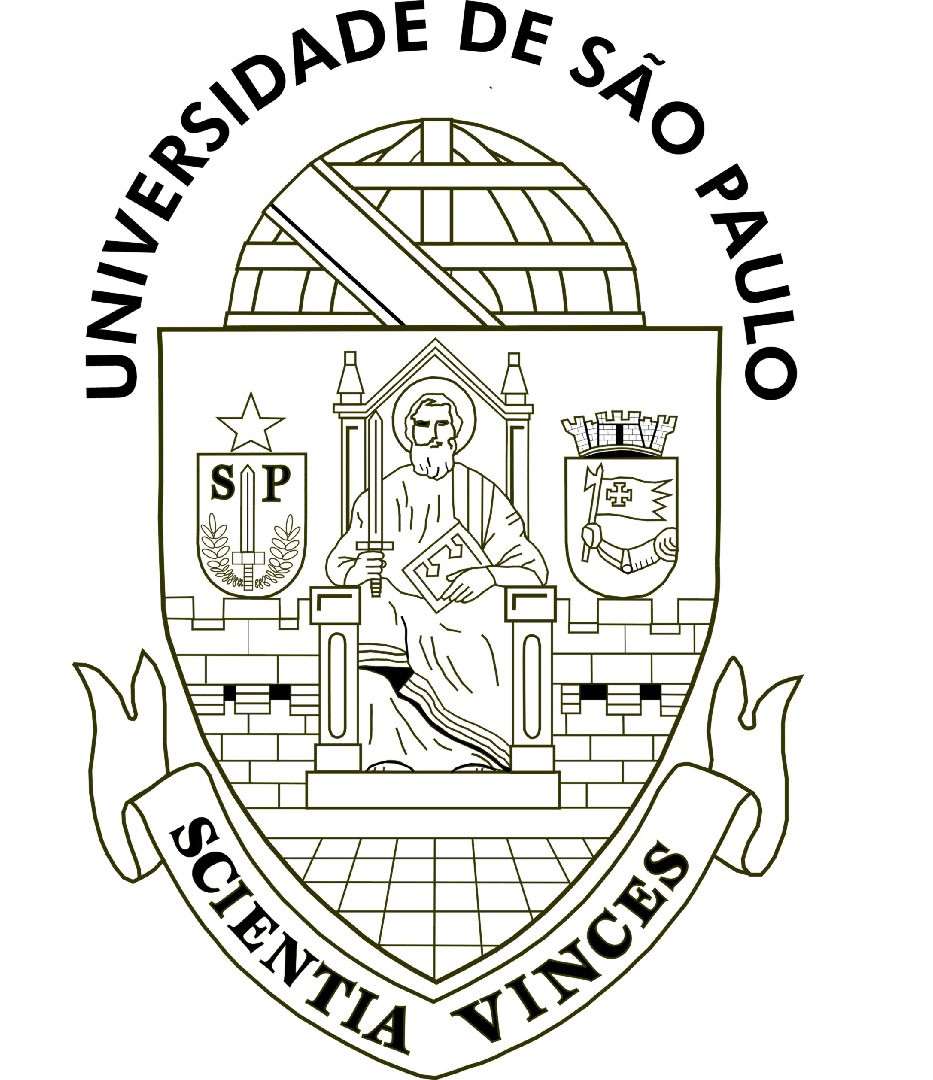}
\end{figure}

\textsc{\large University of São Paulo}

\smallskip



\textsc{\large Interunit Bioinformatics Graduate Program}

\end{center}


\bigskip
\bigskip
\bigskip
\bigskip
\bigskip
\bigskip

\newenvironment{pnc}{\fontfamily{pnc}\selectfont}{\par}

\begin{pnc}

\thispagestyle{plain}

\begin{center}
\textsc{\large Heitor Baldo}
\end{center}

\bigskip
\bigskip
\bigskip
\bigskip
\bigskip
\bigskip

\begin{center}
 \Large \textsc{Towards a Quantitative Theory of Digraph-Based Complexes and its Applications in Brain Network Analysis}
\end{center}


\end{pnc}

\bigskip
\bigskip
\bigskip
\bigskip
\bigskip

\begin{center}
\textsc{PhD Thesis}
\end{center}


\bigskip
\bigskip
\bigskip
\bigskip
\bigskip
\bigskip
\bigskip
\bigskip
\bigskip

\begin{center}
\textsc{This study was financed by the Coordenação de Aperfeiçoamento de Pessoal de
Nível Superior (CAPES) - Finance Code 88887.464712/2019-00}
\end{center}

\bigskip

\begin{center}
\textsc{São Paulo}

\textsc{2024}
\end{center}

\newpage

\clearpage
\thispagestyle{empty}

~\\*[0.05cm]

\begin{center}
\large Heitor Baldo
\end{center}

\bigskip
\bigskip
\bigskip
\bigskip

\begin{center}
 \textbf{ \large Towards a Quantitative Theory of Digraph-Based Complexes and its Applications in Brain Network Analysis}
\end{center}

\bigskip
\bigskip

\begin{center}
 \textbf{ \large Rumo a uma Teoria Quantitativa de Complexos Baseados em Dígrafos e suas Aplicações na Análise de Redes Cerebrais}
\end{center}

\bigskip
\bigskip
\bigskip
\bigskip
\bigskip
\bigskip

\begin{center}
\textbf{\large Final Version}
\end{center}

\bigskip
\bigskip
\bigskip
\bigskip
\bigskip
\bigskip


\begin{flushleft}
\hspace{80mm}Ph.D. thesis \hspace{0.02mm} presented to \hspace{0.01mm} the \hspace{0.002mm} Interunit\\
\hspace{80mm}Bioinformatics Graduate Program  at the\\
\hspace{80mm}University \hspace{0.01mm} of \hspace{0.01mm} São Paulo \hspace{0.01mm} to \hspace{0.01mm} obtain \hspace{0.01mm} the\\
\hspace{80mm}degree of Doctor of Science.\\
\bigskip
\hspace{80mm}Concentration area: Bioinformatics\\
\bigskip
\hspace{80mm}Supervisor:  Prof. Dr. Koichi Sameshima\\
\hspace{80mm}{\small Faculdade de Medicina-USP}\\
\hspace{80mm}Co-supervisor: Prof. Dr. André Fujita\\
\hspace{80mm}{\small Instituto de Matemática e Estatística-USP}
\end{flushleft}

\bigskip
\bigskip
\bigskip
\bigskip
\bigskip
\bigskip
\bigskip
\bigskip
\bigskip
\bigskip
\bigskip
\bigskip
\bigskip
\bigskip
\bigskip

\begin{center}
São Paulo

2024
\end{center}

\newpage

\clearpage
\thispagestyle{empty}

~\\*[0.2cm]

{\footnotesize

\begingroup\centering\singlespacing\small
Ficha catalográfica elaborada com dados inseridos pelo(a) autor(a)\\
Biblioteca Carlos Benjamin de Lyra\\
Instituto de Matemática e Estatística\\
Universidade de São Paulo\par
\vspace{2\baselineskip}\hrule\vspace{.8\baselineskip}
\RaggedRightRightskip 0pt plus 30pt minus 0pt\relax
\RaggedRightParfillskip 20pt plus 40pt minus 10pt\relax
\ttfamily\hspace{2em}\begin{minipage}[t]{125mm}
\RaggedRight\sloppy\setlength{\parindent}{\widthof{123}}

\noindent Baldo, Heitor

Towards a Quantitative Theory of Digraph-Based Complexes and its Applications in Brain Network Analysis / Heitor Baldo; orientador, Koichi Sameshima; coorientador, André Fujita. - São Paulo, 2024.

213 p.: il.

\vspace{1\baselineskip}

\vspace{1\baselineskip}

Tese (Doutorado) - Programa Interunidades de Pós-Graduação em Bioinformática / Instituto de Matemática e Estatística / Universidade de São Paulo.

Bibliografia

Versão corrigida

\vspace{1\baselineskip}

\vspace{1\baselineskip}

1. Graph Theory. 2. Digraph-Based Complexes. 3. Partial Directed Coherence. 4. Brain Connectivity. 5. Epilepsy. I. Sameshima, Koichi. II. Título.

\end{minipage}\par
\vspace{1\baselineskip}\hrule\vspace{.5\baselineskip}\rmfamily
Bibliotecárias do Serviço de Informação e Biblioteca\\
Carlos Benjamin de Lyra do IME-USP, responsáveis pela\\
estrutura de catalogação da publicação de acordo com a AACR2:\\
Maria Lúcia Ribeiro CRB-8/2766; Stela do Nascimento Madruga CRB 8/7534.
\par\endgroup

}

\let\cleardoublepage\clearpage








  



\newpage

\clearpage
\thispagestyle{empty}

~\\*[19.5cm]

\begin{flushright}
\textit{Lovingly dedicated to my mother,}\\ 
\textit{Maria Aparecida Marchi Baldo}\\
\smallskip
\textit{(in memoriam)}
\end{flushright}

\newpage

\clearpage
\thispagestyle{empty}

~\\*[1mm]

\begin{center}
\Large \textbf{Acknowledgments}
\end{center}

\bigskip

I’d first like to acknowledge my thesis advisor, Prof. Koichi Sameshima, for allowing me to do research in this unique field, for his endless patience, and for all the valuable discussions. Secondly, I’d like to acknowledge my co-advisor, Prof. André Fujita, for all the support. I'd also like to express my sincere gratitude to Prof. Luiz Baccalá for all the valuable insights, tips, and discussions, and to all my friends and colleagues for the friendly conversations and emotional support during this period.

Also, I'd like to thank my parents, my mother, Maria Aparecida Marchi Baldo (\textit{in memoriam}), and my father, João Edgar Baldo, who always supported my education and always encouraged me in difficult times throughout my life.

Finally, I'd like to thank the Interunit Bioinformatics Graduate Program at the University of São Paulo, firstly for accepting me as a graduate student and, secondly, for all the support, and CAPES for the financial support.











\newpage

\clearpage
\thispagestyle{empty}

~\\*[1mm]

\begin{center}
\Large \textbf{Abstract}
\end{center}

\bigskip

The development of mathematical methods for studying the structure, organization, and functioning of the brain has become increasingly important, especially in the context of brain connectivity networks studies, highlighting, in recent decades, methods associated with graph theory, network science, and computational (algebraic) topology. In particular, these methods have been used to study neurological disorders associated with abnormal structural and functional properties of brain connectivity, such as epilepsy, Alzheimer's disease, Parkinson's disease, and multiple sclerosis.

In this work, we developed new mathematical methods for analyzing network topology and we applied these methods to the analysis of brain networks. More specifically, we rigorously developed quantitative methods based on complexes constructed from digraphs (digraph-based complexes), such as path complexes and directed clique complexes (alternatively, we refer to these complexes as ``higher-order structures," or ``higher-order topologies," or ``simplicial structures"), and, in the case of directed clique complexes, also methods based on the interrelations between the directed cliques, what we called ``directed higher-order connectivities." This new quantitative theory for digraph-based complexes can be seen as a step towards the formalization of a “quantitative simplicial theory.”

Subsequently, we used these new methods, such as characterization measures and similarity measures for digraph-based complexes, to analyze the topology of digraphs derived from brain connectivity estimators, specifically the estimator known as information partial directed coherence (iPDC), which is a multivariate estimator that can be considered a representation of Granger causality in the frequency-domain, particularly estimated from electroencephalography (EEG) data from patients diagnosed with left temporal lobe epilepsy, in the delta, theta and alpha frequency bands, to try to find new biomarkers based on the higher-order structures and connectivities of these digraphs. In particular, we attempted to answer the following questions: How does the higher-order topology of the brain network change from the pre-ictal to the ictal phase, from the ictal to the post-ictal phase, at each frequency band and in each cerebral hemisphere? Does the analysis of higher-order structures provide new and better biomarkers for seizure dynamics and also for the laterality of the seizure focus than the usual graph theoretical analyses?

We found that all simplicial characterization measures considered in the study showed statistically significant increases in their magnitudes from the pre-ictal phase to the ictal phase, for several higher orders, for both cerebral hemispheres, particularly in the delta and theta bands but no statistically significant changes were observed from the ictal to the post-ictal phase, which may suggest that several topological and functional aspects of brain networks change from the pre-ictal to the ictal phase, at various levels of the higher-order topological organization. Regarding the laterality of the seizure focus, the analysis based on simplicial similarities found no statistically significant difference between the clique topology of the left and right hemispheres in the ictal phase. We conclude from this study that, despite a number of limitations, there may be evidence supporting the viability and reliability of using higher-order structures associated with digraphs to identify biomarkers associated with epileptic networks. However, further research is needed, and the applicability of the newly introduced methods to other disorders of brain connectivity networks will also depend on future studies.

\bigskip

\medskip


\newpage

\clearpage
\thispagestyle{empty}

~\\*[1mm]

\begin{center}
\Large \textbf{Resumo}
\end{center}

\bigskip

O desenvolvimento de métodos matemáticos para o estudo da estrutura, organização e funcionamento do cérebro tem se tornado cada vez mais importante, sobretudo no âmbito de estudos das redes de conectividade cerebral, destacando-se, nas últimas décadas, os métodos associados à teoria dos grafos, ciência de redes e topologia (algébrica) computacional. Em particular, esses métodos têm sido usados para estudar desordens neurológicas associadas com propriedades estruturais e funcionais anormais da conectividade cerebral, tais como epilepsia, doença de Alzheimer, doença de Parkinson e esclerose múltipla.

Neste trabalho, desenvolvemos novos métodos matemáticos para análise da topologia de redes e a aplicação destes métodos à análise de redes cerebrais. Mais especificamente, desenvolvemos rigorosamente métodos quantitativos baseados em complexos construídos a partir de dígrafos (complexos baseados em dígrafos), como complexos de caminhos e complexos de cliques direcionados (alternativamente, referimo-nos à esses complexos por ``estruturas de ordem superior", ou ``topologias de ordem superior", ou ``estruturas simpliciais"), e, no caso de complexos de cliques direcionados, também métodos baseados nas interrelações entre os cliques direcionados, o que chamamos de ``conectividades de ordem superior direcionadas". Essa nova teoria quantitativa para complexos baseados em dígrafos pode ser vista como um passo em direção à formali-zação de uma “teoria quantitativa simplicial”.

Subsequentemente, usamos esses novos métodos, como medidas de caracterização e de similaridade para complexos baseados em dígrafos, para analisar a topologia de dígrafos derivados de estimadores de conectividade cerebral, especificamente do estimador conhecido como coerência parcial direcionada informacional (iPDC), que é um estimador multivariado que pode ser considerado uma representação da causalidade de Granger no domínio da frequência, particularmente estimados a partir de dados de eletroencefalografia (EEG) de paciente diagnosticados com epilepsia do lobo temporal esquerdo, nas bandas de frequência delta, teta e alfa, para tentar encontrar novos biomarcadores baseados nas estruturas e conectividades de ordem superior direcionadas desses dígrafos. Em particular, tentamos responder as seguintes questões: Como a topologia de ordem superior da rede cerebral muda da fase pre-ictal para a ictal, da fase ictal para a pós-ictal, em cada banda de frequência e em cada hemisfério cerebral? A análise de estruturas de ordem superior fornece biomarcadores novos e melhores para a dinâmica das crises e também para a lateralidade do foco da crise do que as análises de grafos usuais?

Encontramos que todas as medidas de caracterização simplicial consideradas no estudo apresentaram aumentos estatisticamente significativos em suas magnitudes da fase pré-ictal para a fase ictal, para várias ordem superiores, para ambos os hemisférios cerebrais, particularmente nas bandas delta e teta, mas nenhuma mudança estatisticamente significativa foi observada da fase ictal para a fase pós-ictal, o que pode sugerir que vários aspectos topológicos e funcionais das redes cerebrais mudam da fase pré-ictal para a fase ictal, em vários níveis de organização topológica de ordem superior. Quanto à lateralidade do foco da crise, a análise baseada em similaridades simpliciais não encontrou diferença estatisticamente significativa entre a topologia de cliques dos hemisférios esquerdo e direito na fase ictal. Concluímos deste estudo que, apesar de uma série de limitações, pode haver evidências que apoiam a viabilidade e confiabilidade do uso de estruturas de ordem superior associadas à dígrafos para identificar biomarcadores associados à redes epilépticas. Entretanto, são necessários mais estudos nessa direção, e a aplicabilidade dos métodos recentemente introduzidos à outros distúrbios de redes de conectividade cerebral também dependerá de estudos futuros.

\bigskip

\medskip

\textbf{Palavras-chave:} Teoria dos Grafos; Complexos de Cliques Direcionados; Complexos de Caminhos; Conectividade de Ordem Superior Direcionada; Coerência Parcial Direcionada; Conectividade Cerebral; Eletroencefalografia; Epilepsia.

\newpage


\clearpage
\thispagestyle{empty}

~\\*[1mm]

\begin{center}
\Large \textbf{List of Abbreviations}
\end{center}

\bigskip
\bigskip
\bigskip
\bigskip

\begin{tabular}{cp{1.1\textwidth}}

  ADSC &\textit{Abstract Directed Simplicial Complex} \\
  AIC & \textit{Akaike Information Criterion} \\
  ASC & \textit{Abstract Simplicial Complex} \\
  DAG & \textit{Directed Acyclic Graph}\\
  DQC & \textit{Directed Quasi-Clique}\\
  DRE & \textit{Drug-Resistent Epilepsy}\\
  DTF & \textit{Directed Transfer Function} \\
  EEG & \textit{Electroencephalography} \\
  EZ & \textit{Epileptogenic Zone} \\
  fMRI & \textit{Functional Magnetic Resonance Imaging} \\
  GC & \textit{Granger Causality}\\
  GED & \textit{Graph Edit Distance} \\
  gPDC & \textit{Generalized Partial Directed Coherence} \\
  GTA & \textit{Graph Theoretical Analysis}\\
  ICA & \textit{Independent Component Analysis} \\
  ILAE & \textit{International League Against Epilepsy} \\
  iPDC & \textit{Informational Partial Directed Coherence} \\
  MRI & \textit{Magnetic Resonance Imaging}\\
  MTLE & \textit{Mesial Temporal Lobe Epilepsy}\\
  MVAR & \textit{Multivariate Autoregressive Model} \\
  PDC & \textit{Partial Directed Coherence} \\
  TDA & \textit{Topological Data Analysis}\\
  TLE & \textit{Temporal Lobe Epilepsy}\\
  VAR & \textit{Vector Autoregressive Model} \\
\end{tabular}\\


{
  \hypersetup{linkcolor=black}
  \tableofcontents
}

\clearpage
\thispagestyle{empty}

\thispagestyle{plain}

\let\cleardoublepage\clearpage


\chapter[Introduction]{Introduction}
\label{chap:chap1}

\pagenumbering{arabic}
\setcounter{page}{7}

\epigraph{Because we do not understand the brain very well we are constantly tempted to use the latest technology as a model for trying to understand it. In my childhood we were always assured that the brain was a telephone switchboard. (‘What else could it be?’) I was amused to see that Sherrington, the great British neuroscientist, thought that the brain worked like a
telegraph system. Freud often compared the brain to hydraulic and electro-magnetic systems. Leibniz compared it to a mill, and I am told some of the ancient Greeks thought the brain functions like a catapult. At present, obviously, the metaphor is the digital computer.}{--- John Searle \citep{Searle}}


\bigskip

The human brain is the most complex biological system known by science, it is believed to contain between 80 and 100 billion neurons and trillions of synaptic connections \citep{Fornito-2016, Herculano}. Essentially, the brain is a network of nervous cells with intricate connections, and the architecture of this network is the structural substrate for its functioning \citep{Sporns-2012}. Understanding how the brain works is one of the greatest scientific challenges of our time, and a crucial component of this challenge is understanding brain networks.

The structural and functional properties of brain networks are the substrate of many brain processes, and the disruption of these networks may be the cause of many neurological disorders, such as Alzheimer's disease and epilepsy. Brain network studies analyze the structural and functional properties of these networks in different contexts, such as in healthy individuals performing motor, sensory, or cognitive tasks, and in individuals diagnosed with some neurological disease. Recently, these studies have used mathematical methods from graph theory, network science, and computational topology to assist in brain network analysis, and this is the scenario that we will deal with in this text.

\section{Brain Networks}

Brain activity depends on several different brain areas rather than being restricted to one or more particular brain regions, and the interaction among remote-located neuroanatomical structures produces structural and functional brain networks. These networks are characterized by the different types of structures involved in the connections and the different types of connectivity, the latter being classified into three main types \citep{Lang, Sporns-Scholar}: \textit{structural connectivity} (or \textit{anatomical connectivity}), which refers to a set of anatomical connections that link neural elements; \textit{functional connectivity}, which refers to statistical relationships between different and remote-located populations of neurons; and \textit{effective connectivity}, which refers to causal relationships between activated brain regions \citep{Friston-1994, Friston-2011, Horwitz, Sporns-2010}.

The inference of the different types of brain networks is carried out through the application of mathematical and statistical methods, called \textit{connectivity estimators}, to neurophysiological signal data that may be obtained through various different technologies, such as electroencephalography (EEG) \citep{Schomer}, functional magnetic resonance imaging (fMRI) \citep{Huettel}, magnetoencephalography (MEG) \citep{Papanicolaou}, and diffusion tensor imaging (DTI) \citep{Mukherjee}. In the literature, there are several types of connectivity estimators, for instance, there are directed and non-directed estimators, bivariate and multivariate estimators, time-domain and frequency-domain estimators \citep{Chiarion}. Typically, effective and functional connectivity networks are estimated through methods that involve some correlation, covariance, or coherence property between the different time series measured from different cortical areas \citep{Fallani3, Fallani}. For the estimation of functional connectivity, correlation, and coherence are the two most commonly used estimators. Other commonly used estimators are mutual information, transfer entropy, and phase synchronization \citep{Pereda}. For the estimation of effective connectivity, an important class of methods are those based on the concept of \textit{Granger causality} (GC) \citep{Granger}, which is a cause-effect relation idea, where the past values of one time series can predict the current values of another. Possibly the most popular among GC-based connectivity estimators are transfer entropy, Granger causality index (GCI) \citep{Brovelli, Geweke}, directed coherence (DC) \citep{Saito}, partial directed coherence (PDC) \citep{Baccala-2001}, and directed transfer function (DTF) \citep{Kaminski1}.

Among the estimators mentioned above, of special interest in this work is the PDC. PDC is a directed, model-based, multivariate technique that is used to simultaneously determine the directional influences and spectral properties of the interaction between any pairs of brain signals (e.g., pairs of channels in EEG) given in a multivariate ensemble (e.g., a multivariate autoregressive model (MVAR)) \citep{Fallani, Sameshima1}. The PDC is an estimator that is based on the coefficients of an MVAR model (or other multivariate models, such as the vector moving average (VMA) model and the vector autoregressive moving average (VARMA) model \citep{Baccala2022}), transformed in the frequency-domain, and it can be considered a representation of the concept of GC in the frequency-domain \citep{Sameshima1}. Additionally, the PDC is able to distinguish between direct and indirect causal flows in the estimated connectivity pattern \citep{Fallani, Sameshima1}. For instance, PDC tends not to add ``erroneous" causality flows between signals recorded from one structure to another. This property makes PDC suitable to be applied to brain signals. Furthermore, a generalized form of PDC (the \textit{generalized} PDC (gPDC)) was introduced in \citep{Baccala8}. In an attempt to understand precisely how the PDC relates to the information flow, Takahashi et al. \citep{Takahashi2, Takahashi} introduced the \textit{information} PDC (iPDC), a modification of the PDC expression that formalizes the relationship between the PDC and the information flow based on the concept of mutual information rate between two time series. 



Mathematically, networks can be represented by \textit{graphs}, which are formed by a set of nodes and a set of links connecting these nodes. Commonly, the terms ``networks" and ``graphs" are used interchangeably. In the case of brain networks, nodes may represent brain areas or neural elements \citep{Zalesky-2010a}, and the links between them can represent, for example, anatomical connections, statistical dependencies, or causal influences. Therefore, the anatomical and functional organization of the brain can be studied from its mathematical representations as graphs, and this fact is what has led in recent years to the widespread adoption of methods from graph theory and network science in network neuroscience \citep{Bullmore1, Bullmore2, Fallani3, Papo-2014, Rubinov-2010, Sporns-2010, Sporns-2014}.

\section{Graph Theoretical Analysis of Brain Networks}

As mentioned in the previous section, there has been a significant increase in research on brain networks using methods from graph theory and network science, and, commonly, the analyses that are based on methods from these fields are called \textit{graph theoretical analysis} (GTA). Numerous GTA studies have demonstrated that several non-trivial topological and organizational properties \citep{Zalesky-2012}, such as hierarchical organization \citep{Bassett-2008}, clustering, modularity \citep{Rubinov-2010, Sporns-2010},  presence of structural and functional network motifs \citep{Sporns-2010, Sporns-2004c}, and small-world organization \citep{Achard, Bassett-2017, Papo-2016, Sporns-2004a} are exhibited by brain networks, depending on the context. Furthermore, recent works on brain network analysis have used concepts from spectral graph theory, such as eigenvectors \citep{Abdelnour-2015b, Abdelnour2-2018} and the spectrum \citep{deLange} of the Laplacian matrix, the latter being able to reveal an integrative community structure of the networks.


Additionally, GTA has proven to be an important tool in the search for network-based biomarkers for many aspects of brain functioning, for instance, in identifying structural and functional anomalies in network connection patterns, such as those indicating neurological illnesses or certain cognitive processes \citep{Ward}. In particular, small-world network topology has been regularly detected in various brain network studies, mainly in data from healthy patients, which has been used as a parameter to differentiate healthy and neuropathological network patterns \citep{Ahmadlou, Bassett-2017, Chiang}. The comprehensive evaluation of GTA investigations of brain networks using fMRI data conducted by Farahani et al. \citep{Farahani} noted that applications of graph theory in human cognition include the identification of biomarkers (fMRI-based biomarkers) for human intelligence \citep{vandenHeuvel-2009}), working memory \citep{Stanley}, aging brain \citep{Farahani}, and behavioral performance in natural environments \citep{Qian}, and applications in brain diseases lie in the discovery of biomarkers for conditions like epilepsy \citep{Vlooswijk}, Alzheimer's disease (AD) \citep{deHaan}, multiple sclerosis (MS) \citep{Liu2016}, autism spectrum disorders (ASD) \citep{Keown}, and attention-deficit/hyperactivity disorder (ADHD) \citep{Wang-2009}.  

Of particular interest in the study of neuropathologies are the EEG-based biomarkers, due to the advantages of EEG over fMRI, such as portability and low cost. There are several studies that use GTA on EEG data obtained from patients diagnosed with some type of neurological disorder in the search for biomarkers \citep{Liu, Stam-2014}. Among these disorders, we can mention Parkinson's disease (PD) \citep{Utianski}, epilepsy \citep{Horstmann} and schizophrenia (SCZ) \citep{Yin2017}.

\section{Topological Data Analysis of Brain Networks}

Together with GTA, concepts from computational topology, such as simplicial complexes, homotopy, homology, Betti numbers, and persistent homology, which, eventually, in the data analysis scenario, have been put together under the umbrella term \textit{topological data analysis} (TDA) \citep{Chazal}, have been used in recent investigations of brain networks.

Methods from TDA are suitable for evaluating topological characteristics of topological spaces, in particular discrete topological spaces such as simplicial complexes. Briefly, an \textit{abstract simplicial complex} (which, for the sake of simplicity, we refer to as simplicial complex) is a finite collection of finite sets (called simplices) that is closed under subset inclusion \citep{Edelsbrunner}, and they can be seen as a generalization of graphs. Simplicial complexes can be constructed from graphs, or directed graphs (digraphs), in several ways, e.g., by considering the cliques of a graph as the simplices of the complex (called \textit{clique complex} or \textit{flag complex}) \citep{Aharoni}, or by considering the directed cliques of a digraph as directed simplices of a directed simplicial complex (where directed simplices are considered to be ordered sets)  (called \textit{directed clique complex} or \textit{directed flag complex}) \citep{Lutgehetmann, Masulli, Reimann}. Other complexes that can be associated with digraphs are the \textit{path complexes} \citep{Grigoryan-2020}, which can contribute with additional insights into the substructures of digraphs. 

One of the reasons for considering these complexes rather than just the network nodes is that network characterization measures cannot always provide us with relevant insights into the network topology, for example, two nodes may have the same clustering coefficient, but the topology of their neighborhoods can differ significantly \citep{Kartun}. Another point that we can highlight is that methods such as persistent homology can analyze the topological features of a space at different scales and their persistence (or their ``lifetime persistence").

Several examples of applications of TDA in network neuroscience are the following: characterizing functional brain networks of patients with ADHD and ASD \citep{Lee1-AD, Lee2-AD}; utilizing clique topology and persistent homology of clique complexes constructed out of brain networks to assess neural functions and structures \citep{Giusti2016, Giusti2015, Petri, Reimann, Sizemore}; and utilizing persistent homology to detect epileptic seizures \citep{Fernandez, Piangerelli, Sun2023, Wang}.

\section{Epilepsy Studies through Brain Connectivity Networks}


Although a number of brain disorders were mentioned in the previous sections, our focus in this study is epilepsy. Epilepsy is one of the most common disorders of the central nervous system (CNS), characterized by recurrent and non-induced seizures \citep{Alarcon, Wasade}. It is also a brain connectivity network disorder, typified by a clear relation between pathological symptoms and aberrant network dynamics \citep{Frohlich}.

Over the years, studies have consistently shown that, compared to healthy individuals, patients diagnosed with epilepsy present changes in the topology of brain connectivity networks \citep{Farahani, Liu, Stam-2014}, and many of these discoveries come from GTA and TDA performed on these networks. Furthermore, these techniques have aided in understanding how brain network architecture changes during the ictal phase (i.e., during the seizure) and how epileptic networks can be described in terms of their topologies.

Some findings that show how epilepsy and alterations in brain network topology are related, based on different data acquisition techniques and connectivity estimation methods, are: Bernhardt et al. \citep{Bernhardt} found that patients with temporal lobe epilepsy (TLE) showed changes in the distribution of hubs (important brain regions), increased path lengths, and increased clustering coefficients compared with healthy controls; Liao et al. \citep{Liao} found that patients with mesial temporal lobe epilepsy (MTLE) showed significantly increased local connectivity and decreased global connectivity compared with healthy individuals; Bonilha et al. \citep{Bonilha} found that patients with MTLE showed increased degree, local efficiency, and clustering coefficient, in certain areas compared with healthy controls; Horstmann et al. \citep{Horstmann} found that patients with drug-resistant epilepsy (a pharmacoresistant form of epilepsy) showed abnormally regular functional networks compared to healthy individuals.

\section{About this Thesis}

First and foremost, before proceeding further in the text, I would like to warn the reader that some relevant research may not have been discussed or cited. I apologize in advance for the research works that were overlooked.

\subsection{Objectives and Scientific Relevance}

\subsubsection{Objectives}

The main objectives of this work, in simple terms, are the development of new mathematical methods for analyzing the topology of networks and the application of these methods to the analysis of brain networks. More specifically, we are interested in building new ways of looking at the topology of digraphs; for this aim, we chose to develop quantitative methods based on complexes built out of digraphs (or \textit{digraph-based complexes}) such as path complexes and directed cliques complexes, what we call their ``higher-order structures" (or ``higher-order topologies," or ``simplicial structures") and, in the case of directed clique complexes, also methods based on the interrelationships between the directed cliques, what we call their ``(directed) higher-order connectivities," and, ultimately, use these new methods to analyze the topology of digraphs derived from brain connectivity estimators (especially the iPDC estimator), particularly estimated from epileptic individuals, to try to find new network-based biomarkers. We can put this more clearly into two main objectives:

\begin{enumerate}

\item To develop rigorously a new quantitative theory for digraph-based complexes (or, as we can consider it, a step towards the formalization of a ``quantitative simplicial theory"), with special emphasis on directed higher-order connectivity between directed cliques;

\item To apply the methods of the new theory to epileptic brain networks obtained through iPDC to quantitatively investigate their higher-order topologies and search for new biomarkers based on their directed higher-order structures and connectivities, thus pointing out potential applications of the theory in network neuroscience.
  
\end{enumerate}

\subsubsection{Clinical and Scientific Relevance}
 
The new methods introduced in this work may be helpful to the academic community in several areas involving the study of networks, such as biology, social sciences, computer science, and, in particular, network neuroscience, especially in the study of neurological diseases associated with disorders of brain connectivity, such as epilepsy, Alzheimer's disease, and Parkinson's disease. In the case of epilepsy, for example, they may be helpful to the community in answering fundamental questions such as: How do the brain connectivity networks change from one seizure phase to another? How do these networks change during a seizure? Is it possible to associate reliable network-based biomarkers with epileptic brain networks and with the laterality of the seizure focus? Furthermore, these methods may be useful in clinical practice, for instance, in assisting epilepsy surgeries that depend on the precise location of the epileptogenic zone.

\subsection{Outline of the Thesis}


This thesis is divided into two parts: the first part deals with the development and formalization of a quantitative theory of digraph-based complexes (objective 1), including the exposition of the necessary basic tools, such as the fundamentals of graph theory; the second part deals with the presentation of the theory of brain connectivity networks, together with its utility in the study of epilepsy, and the application of the methods developed in the first part in the analysis of epileptic brain networks (objective 2). In the following, we present a brief description of the chapters that compose each of these two parts.

\bigskip
\noindent \textbf{Part I - Towards a Quantitative Theory of Digraph-Based Complexes}
\medskip

\begin{itemize}

\item \textbf{Chapter 2 - Fundamentals of Graph Theory:} In this chapter, we introduce the fundamental concepts of graph theory, starting with some formal definitions involving relations and orders, such as equivalence relations and ordered sets, then move on to introduce concepts related to graphs and digraphs, algebraic and spectral graph theory, graph measures, graph similarities, and finally a brief discussion on random graphs.

\item \textbf{Chapter 3 - Digraph-Based Complexes and Directed Higher-Order Connectivity:} In this chapter, we present the theory of (abstract) simplicial complexes and directed clique complexes associated with digraphs, along with the case of weighted digraphs, passing through simplicial homology, persistent homology, and combinatorial Laplacians associated with these complexes. Next, we present the concept of path complexes, their homologies, and a brief discussion about combinatorial Laplacians associated with them. Finally, we introduce a new theory related to directed higher-order connectivity between directed cliques. This theory leads to the conception of new concepts such as \textit{directed higher-order adjacencies} (upper and lower adjacencies) and \textit{maximal/lower $q$-digraphs} (defined as digraphs whose nodes are maximal directed cliques), as well as provides new concepts and formalisms for directed Q-Analysis.

\item \textbf{Chapter 4 - Quantitative Approaches to  Digraph-Based Complexes:} In this chapter, based on graph measures, we introduce new quantifiers for characterizing $q$-digraphs, and, based on graph similarity comparison methods, we introduce new similarity comparison methods for directed clique complexes and path complexes. The set of all these quantifiers and methods can be seen as the formalization of a simplicial analogue of the quantitative graph  theory, that is, a formalization of a ``quantitative simplicial theory." Finally, we present some examples with random digraphs.

\end{itemize}

\bigskip
\noindent \textbf{Part II - Brain Connectivity Networks and a Quantitative Graph/Simplicial Analysis of Epileptic Networks}
\medskip

\begin{itemize}

\item \textbf{Chapter 5 - Brain Connectivity Networks:} In this chapter, we study the theory behind brain connectivity networks, briefly going over the biophysical principles of brain signals and the methods for obtaining them, paying particular attention to EEG. Next, we discuss, also briefly, bivariate and multivariate connectivity estimators, focusing on PDC and its variants, particularly gPDC and iPDC. Finally, we cover the various kinds of brain connectivity networks, particularly structural, functional, and effective networks, as well as the modern uses of graph theory and computational (algebraic) topology in the analysis of these networks to explore the dynamics of brain activity in various contexts, especially in the study of neurological disorders.

\item \textbf{Chapter 6 - Epilepsy as a Disorder of Brain Connectivity:} In this chapter, we look more closely at the neuropathology of epilepsy, covering its main characteristics, etiologies, epidemics, diagnoses, and treatments. We also discuss a number of studies that showed differences between the brain networks of epileptic patients and the brain networks of healthy people in terms of network properties, such as clustering coefficient, characteristic path length, and degree distribution, as well as differences between the ictal and non-ictal periods.

\item \textbf{Chapter 7 - Quantitative Graph/Simplicial Analysis of Epileptic Networks:} In this chapter, we perform an analysis of epileptic brain networks, estimated through iPDC from EEG data from patients with left temporal lobe epilepsy, using the new quantitative methods developed in previous chapters for directed higher-order networks ($q$-digraphs), to explore how certain properties of these networks change according to the seizure phases, as well as according to the cerebral hemispheres, in different frequency bands.

\item \textbf{Chapter 8 - Final Considerations:} Finally, in this last chapter, we present a summary of the objectives and developments of this thesis, some relevant considerations, and the gaps left in this work that we intend to complete in future studies.

\end{itemize}

\let\cleardoublepage\clearpage

\part{Towards a Quantitative Theory of Digraph-Based Complexes}
\label{partI}

\chapter[Fundamentals of  Graph Theory]{Fundamentals of  Graph Theory}
\label{chap:chap2}

\epigraph{Graph theory serves as a mathematical model for any system involving a binary relation.}{--- Frank Harary \citep{Harary}}


\bigskip

Bearing in mind the multidisciplinary nature of this thesis and to make our discussion self-contained, we have chosen to offer enough background information so that readers with an undergraduate mathematics background may follow along. Accordingly, in this second chapter, we present the basic concepts, terminology, and notations of graph theory that will be necessary throughout the text.



\section[Fundamental Concepts]{Fundamental Concepts}
\label{sec:graph-theory}

Graphs are the central objects of this thesis because, as we made explicit in the introduction, we will deal with how different areas of the brain interact with each other, and these interactions can be represented abstractly through \textit{graphs}, which are mathematical representations of the relationships between units of a system. Besides neuroscience, graph theory has numerous applications in a variety of other scientific fields, such as social science, biology, computer science, linguistics, and transportation planning \citep{Chung2006, Newman}. 


Some authors use the term ``networks" to designate real-world networks and ``complex networks" to designate large real-world networks, whereas the term ``graphs" is used to designate their mathematical representations. However, here we will not make these distinctions; thus, henceforth, we will use these terms interchangeably.

In this section, we present the mathematical formalism of the areas of graph theory necessary for the development of the text. Some of the propositions and theorems presented here do not include their respective proofs, but instead, we provide appropriate references to their proofs.

Henceforth, we will use the Bourbaki notation for the sets of natural, integer, and real numbers, $\mathbb{N}$ ($0 \notin \mathbb{N}$), $\mathbb{Z}$, and $\mathbb{R}$, respectively. For these sets, some additional notations are adopted: $ \mathbb{N}_{0} = \mathbb{N} \cup \{ 0 \}$, $\mathbb{R}_{\ge 0} = \{ x \in \mathbb{R}: x \ge 0\}$, and $\mathbb{R}_{+} = \{ x \in \mathbb{R}: x > 0\}$. Also, we may use the Dirac notation to denote vectors in the Euclidean space, e.g. $|v \rangle \in \mathbb{R}^{n}$, and $\langle v| = |v \rangle^{T}$.


\subsection{Relations and Orders}
\label{sec:relations-orders}

Let us begin with a basic yet crucial distinction between unordered and ordered pairs of elements. The primary reference for this section is \citep{Jech}.

The axiom of extensionality of Zermelo-Fraenkel axiomatic set theory states that if two sets have the same elements, then they are equal. From this axiom, we have $\{x,y\} = \{y,x\}$, for any set of two elements (or \textit{unordered pair}). We define an \textit{ordered pair} $(x,y)$ so as to satisfy the condition: $(x, y) = (z, w)$ if and only if $x = z$ and $y = w$. In general, we define an \textit{ordered $n$-tuple} $(x_{1},..., x_{n})$ so as to satisfy the condition: $(x_{1},..., x_{n}) = (y_{1},...,y_{n})$ if and only if $x_{i} = y_{i}$, for all $i=1,...,n$.

As a usual convention, braces $\{\}$ are used to denote sets and parentheses $()$ to denote ordered $n$-tuples; however, in some cases, we may commit an abuse of notation and use $()$ to denote unordered pairs, as we will see in the definition of undirected graphs.

\smallskip

\begin{definition}\label{def:relation}
Let $S_{1},..., S_{n}$ be $n$ sets, not necessarily distinct. A \textit{relation} (or \textit{$n$-ary relation}) on these sets is a set $R \subseteq S_{1} \times... \times S_{n}$ whose elements are $n$-tuples  $(x_{1},..., x_{n})$ such that $x_{i} \in S_{i}$, $\forall i = 1,..,n$. We call $R$ a \textit{binary relation} if $n=2$, and, in this case, we denote a pair $(x_{1}, x_{2}) \in R$ as $x_{1}Rx_{2}$.
\end{definition}

\begin{definition}\label{def:equivalence-relation}
Let $\sim \subseteq S \times S$ be a binary relation on a set $S$. We say that $\sim$ is an \textit{equivalence relation} on $S$ if it satisfies the following conditions, for all $x,y,z \in S$:
\begin{enumerate}
\item $x \sim x$ (reflexivity);
\item $x \sim y \iff y \sim x$  (symmetry);
\item if $x \sim y$ and $y \sim z$ $\Rightarrow$ $x \sim z$ (transitivity).
\end{enumerate}
\end{definition}

\smallskip

An equivalence relation $\sim$ on $S$ defines a \textit{class of equivalence} for each element $x \in S$: $[x] = \{ y \in S : x \sim y \}$. The classes of equivalence of $S$ define a set called \textit{quotient set}, and it is denoted as $S/\sim$, i.e. $S/\sim = \{[x] : x \in S\}$.

\smallskip


\begin{definition}\label{def:order}
Given a set $S$, an \textit{order} (or \textit{partial order}) defined in $S$ is a binary relation, denoted by $\le$, which satisfies the following conditions, for all $x,y,z \in S$:  

\begin{enumerate}
\item $x \le x$ (reflexivity);
\item if $x \le y$ and $y \le x$ $\Rightarrow$ $x=y$ (antisymmetry); 
\item if $x \le y$ and $y \le z$ $\Rightarrow$ $x \le z$ (transitivity).
\end{enumerate}

Moreover, if either $x \le y$ or $y \le x$, $\forall x, y \in S$, we said that $\le$ is a \textit{total order} (or \textit{linear order}).
\end{definition}

\smallskip

A set equipped with an order (partial order) is called a \textit{ordered set} (\textit{partially ordered set} or \textit{poset}), and a set equipped with a total order is called a \textit{totally ordered set} (or \textit{linearly ordered set}).

It's important to note that if two partially/totally ordered sets have the same elements, they are identical as sets, but they are not identical as partially/totally ordered sets if the respective orders of their elements are different.

As a last consideration, since a set does not allow repetitions of its elements, i.e. $\{x,x\} = \{x\}$, and since sometimes it is necessary to take these repetitions into account, below we present the formal definition of \textit{multiset}, which is a generalization of the concept of set \citep{Blizard}.

\smallskip

\begin{definition}\label{def:multiset}
A \textit{multiset} is a collection of elements in which repetition of elements is allowed. The \textit{multiplicity} of an element is the number of times it appears in the multiset. The cardinality of a multiset is equal to the total number of its elements, counting their multiplicities.
\end{definition}

\smallskip

To avoid confusion, we will adopt the notation $\dgal{ }$ to denote multisets.

\smallskip

\begin{example}
Consider the multiset $M = \dgal{x,x,y,y,y}$. The multiplicity of $x$ is equal to $2$, the multiplicity of $y$ is equal to $3$, and the cardinality of $M$ is equal to $5$. 
\end{example}


\subsection{General Concepts in Graph Theory}
\label{sec:graph-basics}

In this subsection, we present a brief introduction to the main concepts of graph theory that will be necessary for the development of subsequent concepts. For this first part, the basic references are \citep{Bollobas, Diestel, Harary, West}, however, additional references will be presented throughout the text.

\subsubsection{Graphs: Basic Definitions}

\begin{definition}\label{def:und-graph}
A \textit{graph} is a pair $G = (V, E)$ of disjoint sets such that $E \subseteq V \times V$ is a binary relation, i.e. the elements of $E$ are unordered pairs of elements of $V$. The elements of $V$ are the \textit{vertices} (or \textit{nodes}) and the elements of $E$ are the \textit{edges} (or \textit{links}) of the graph $G$. The edges are denoted as $(v,u)$ $(= (u,v))$, where $v,u \in V$. The notations $V = V(G)$ and $E = E(G)$ are also commonly used. The cardinalities of $V$ and $E$ are denoted by $|V|$ and $|E|$, respectively, and $|V|$ is said to be the \textit{order} of $G$. If $|V| < \infty$ and $|E| < \infty$, the graph is said to be \textit{finite}.
\end{definition}

\smallskip

Strictly speaking, in the previous definition, we defined an \textit{undirected graph}, whose edges are actually sets $(v,u) = \{v,u\}$, where we are committing an abuse of notation by using parentheses. As a convention, we will use parentheses whenever we deal with edges of graphs.

It is usual to represent graphs through diagrams of points (vertices) connected by lines (edges) (see Figure \ref{fig:cliques-und}).


\smallskip

\begin{definition}\label{def:incidence}
Given a graph $G = (V, E)$, a vertex $x \in V$ is said to be \textit{incident with an edge} $e = (v,u) \in E$ if $x \in e$, i.e. if $x = v$ or $x = u$ (we can also say that the edge $e$ is \textit{incident to} $x$). Furthermore, two edges are said to be \textit{adjacent} if they share a common vertex, and two vertices, $v, u \in V$, are said to be \textit{adjacent} (or \textit{neighbors}) if the edge $(v,u)$ exists.
\end{definition}

\smallskip

Some authors define the \textit{empty graph} as the graph having at least one vertex but no edges, and the \textit{null graph} as the graph having no vertices and no edges. Here we will adopt this convention.

\smallskip

\begin{definition}\label{def:loop}
A \textit{loop} (or \textit{self-loop}) in a graph $G = (V, E)$ is an edge $(v,v) \in E$, i.e. it is an edge that links a vertex $v \in V$ to itself. If there is more than one edge between two vertices, we say that $G$ contains  \textit{multiple edges}. If a graph contains no loops or multiple edges, it is said to be a \textit{simple graph}. A graph that allows multiple edges and loops is called \textit{pseudograph}\footnote{Strictly speaking, a pseudograph is not a graph since, by definition, a graph has no loops, but for practical reasons we use this terminology.} and a pseudograph without loops is called \textit{multigraph}.
\end{definition}

\smallskip

We can make the vertices of a graph distinguishable from one another by associating names or \textit{labels} with each of them. More formally, a non-null graph is said to be \textit{labeled} if it is equipped with a bijection between the finite set of vertices and a finite set of labels. For instance,  for a graph with $k+1$ vertices, we can assign labels such as $v_{0},...,v_{k}$, or  non-negative integers $0,1,...,k$, to its vertices.

Henceforth, all graphs will be considered non-empty, non-null, finite, simple, and labeled, unless said otherwise.

\smallskip

\begin{definition}\label{def:subgraph}
The graph $G' = (V', E')$ is a \textit{subgraph} of $G = (V, E)$ if $V' \subseteq V$ and $E' \subseteq E$, and in this case we write $G' \subseteq G$. If $G' \subseteq G$ and $G' \neq G$, then $G'$ is called a \textit{proper subgraph} of $G$. If $G'$ contains all edges $(v,u) \in E$, with $v,u \in V'$, then $G'$ is said to be an \textit{induced subgraph} of $G$, and we say that $V'$ \textit{induces} (or \textit{generates}) $G'$ into $G$.
\end{definition}

\smallskip

Typically, in network science literature, subgraphs that appear at a significantly higher frequency in a given graph than in equivalent random graphs are called \textit{motifs} \citep{Milo}.

\smallskip


\begin{definition}\label{def:path}
A \textit{path} is a graph $P = (V, E)$, $V = \{v_{0},...,v_{k}\}$, such that $E = \{e_{1} = (v_{0}, v_{1}),..., e_{k} = (v_{k-1}, v_{k})\}$, with $v_{i} \neq v_{j}$ if $i \neq j$, for all $i,j = 0,..., k$, and $e_{m} \neq e_{n}$ if $m \neq n$, for all $m,n = 1...,k$. Also, we can denote a path simply by the sequence of its vertices, i.e. $P = v_{0}v_{1}...v_{k}$.
\end{definition}

\begin{definition}\label{def:cycle}
A \textit{cycle} is a graph $C = (V, E)$, $V = \{v_{0},...,v_{k}\}$, with $k \ge 3$, such that $E = \{e_{1} = (v_{0}, v_{1}),..., e_{k} = (v_{k-1}, v_{k}), e_{k+1} = (v_{k}, v_{0})\}$, with $v_{i} \neq v_{j}$ if $i \neq j$, for all $i,j = 0,..., k$, and $e_{m} \neq e_{n}$ if $n \neq m$, for all $m,n = 1...,k+1$. 
\end{definition}


\begin{definition}\label{def:path-walk-trail}
A \textit{walk} between two vertices $v_{0}$ and $v_{k}$ (or \textit{$(v_{0}, v_{k})$-walk}) is a sequence of vertices and edges (not necessarily distinct) $v_{0}e_{1}v_{1}...v_{k-1}e_{k}v_{k}$ such that $e_{i} = (v_{i-1}, v_{i})$, for all $1 \le i \le k$.  If $v_{0} = v_{k}$, the walk is called \textit{closed}, and is called \textit{open} otherwise. If the edges of a walk are all distinct, it is said to be a \textit{trail}.  If the vertices of a walk are all distinct (and consequently all of the edges), then it is said to be a \textit{path}. A closed walk with all distinct vertices and with $k \ge 3$ is a \textit{cycle}. The \textit{length} of a walk is equal to its number of edges. Also, a walk of length $k$ is said to be a \textit{$k$-walk}. 
\end{definition}

\smallskip

It's important to highlight that, as sequences, paths, and cycles are special cases of trails, and trails are special cases of walks, but paths and cycles form actually simple graphs. In a simple graph, it's common to denote a walk by the sequence of its vertices, i.e. $W = v_{0}v_{1}...v_{k}$, $v_{i} \in V$, $i = 0,1,...,k$. 

\smallskip

\begin{definition}\label{def:connected}
Given a graph $G = (V, E)$, two vertices $v,u \in V$ are said to be \textit{connected} if either $v = u$ or there exists a path connecting $v$ to $u$. Also, if every vertex in $G$ is connected to every other, $G$ is said to be \textit{connected}. If $G$ is not connected, it is said to be \textit{disconnected}.
\end{definition}

\begin{proposition}\label{prop:connected}
\textit{The property of two vertices of a graph being connected is an equivalence relation on its vertex set.}
\end{proposition}
\begin{proof}
Let $G = (V, E)$ be a graph and let $\sim$ denote the relation ``is connected to" on the vertex set $V$. For arbitrary vertices $x,y,z \in V$, the following properties are satisfied:

\begin{enumerate}
\item  $x \sim x$. Indeed, by Definition \ref{def:connected}, every vertex is connected to itself.

\item $x \sim y$ $\Rightarrow$ $y \sim x$. Indeed, if $x \sim y$, there is a path connecting $y$ to $x$ as well.

\item $x \sim y$ and $y \sim z$ $\Rightarrow$ $x \sim z$. Indeed, let $P_{xy} = xv_{0}...v_{k}y$ be the path connecting $x$ to $y$ and let $P_{yz} = yv_{k+1}...v_{k+n}z$ be the path connecting $y$ to $z$, $v_{i} \in V$. The path  $P_{xz} = xv_{0}...v_{k}yv_{k+1}...v_{k+n}z$ is a path connecting $x$ to $z$, thus  $x \sim z$.

\end{enumerate}
Therefore, the relation $\sim$ is an equivalence relation since it is reflexive, symmetric, and transitive. 
\end{proof}

\begin{definition}\label{def:connected-component}
A \textit{maximal connected subgraph} of a graph $G$ is a connected subgraph which is not a proper subgraph of any other connected subgraph of $G$.  A \textit{connected component} of $G$ is a maximal connected subgraph of $G$. The largest connected component of $G$ is called its \textit{giant component}.
\end{definition}

\begin{definition}\label{def:density}
Given a graph $G= (V, E)$, with $|V|=n$ and $|E|=m$, its \textit{density} is defined as $\mbox{den}(G)  = 2m/n(n-1)$, since the maximum number of edges in the graph is equal to ${n \choose 2} = n(n-1)/2$. As a convention, a graph is said to be \textit{dense} if $\mbox{den}(G) > 0.5$, and it is said to be \textit{sparse} otherwise.
\end{definition}

\begin{definition}\label{def:neighborhood}
Given a graph $G = (V, E)$, the \textit{neighborhood} of 
a vertice $v$ in $G$, denoted by $\mathcal{N}_{G}(v)$, is the set of all vertices that are adjacent to $v$, i.e.
\begin{equation}\label{eq:neighborhood}
\mathcal{N}_{G}(v) = \{u \in V : (v, u) \in E\}.
\end{equation}
\end{definition}

\smallskip

Moreover, we say that $\mathcal{N}_{G}(v)$ is an \textit{open neighborhood} if $v \not\in \mathcal{N}_{G}(v)$, and that it is a \textit{closed neighborhood} otherwise. In the last case, we use the notation $\mathcal{N}_{G}[v]$.


\smallskip

\begin{definition}\label{def:degree}
Given a graph $G = (V, E)$, the \textit{degree} of a vertex $v \in V$, denoted by $\deg_{G}(v) = \deg(v)$, is equal to the number of edges incident to $v$. If $\deg(v) = 0$, the vertex is said to be \textit{isolated}. 
\end{definition}

\smallskip

In general, in network science literature, if the degree of a vertex far exceeds the average degree of the other nodes in a graph, it is called a \textit{hub}\footnote{Other definitions of hub were proposed either based on other node-wise measures, such as betweenness centrality and clustering coefficient \citep{Bullmore1, Sporns-2007}, or based on the connection with the concept of \textit{authorities} \citep{Kleinberg}.}.

Now that we have introduced the definition of vertex degree, let's introduce one of the most fundamental properties of a graph: its \textit{degree distribution} \citep{Newman}. Let $G=(V,E)$ be a graph with $|V|=n$. Let $\delta(k)$ be the number of vertices having degree $k$. The probability that a vertex chosen uniformly at random has a degree equal to $k$ is given by
\begin{equation}\label{eq:prob-degree}
p(k) = \frac{\delta(k)}{n}.
\end{equation}

The fractions $p(k)$ represent the \textit{degree distribution}  of the graph, and they describe how frequently a vertex with a certain degree appears in the graph. Furthermore, we can represent the degree distribution of $G$ graphically as the plot of $p(k)$ versus $k$.

\smallskip



\begin{definition}\label{def:complete-graph}
Given a graph $G$, if all its vertices are adjacent to each other, $G$ is said to be \textit{complete}. A complete graph with $n$ vertices is commonly denoted by $K_{n}$.
\end{definition}

\smallskip

The complete graph $K_{n}$ has $n(n -1)/2$ edges. The graph $K_{3}$ is called \textit{triangle}. Figure \ref{fig:cliques-und} illustrates the complete graphs $K_{n}$, for $n=1,2,3,4,5$.

\smallskip

\begin{definition}\label{def:clique}
Let $G=(V, E)$ be a graph. A \textit{$(k+1)$-clique} in $G$ is a complete induced subgraph with $k+1$ vertices, $0 \le k \le |V|-1$. A clique is said to be \textit{maximal} if it is not a proper subgraph of any other clique in $G$. Also, the  \textit{clique number} of $G$, denoted by $\omega(G)$, is the number of vertices contained in the largest clique of $G$.
\end{definition}

\begin{example}
Figure \ref{fig:cliques-und} shows examples of $(k +1)$-cliques, for $k=0,1,2,3,4$. From left to right: $1$-clique (vertex), $2$-clique (edge), $3$-clique (triangle), $4$-clique, $5$-clique. The $(k+1)$-clique is the complete graph $K_{(k+1)}$.
\begin{figure}[h!]
   \centering
\includegraphics[scale=1.3]{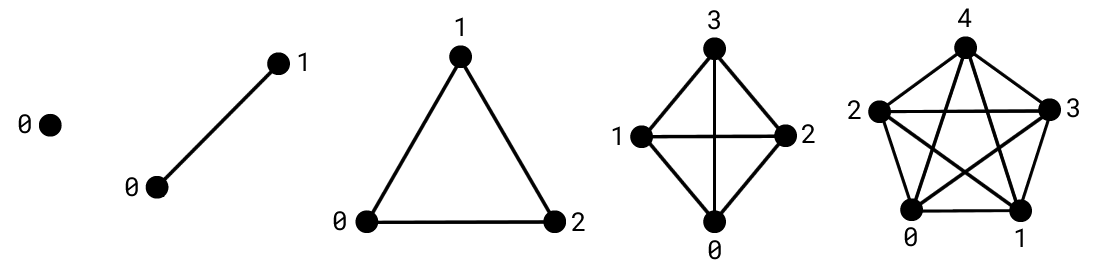}
\caption{Examples of $(k+1)$-cliques, for $k=0,1,2,3, 4$.}
    \label{fig:cliques-und}
\end{figure}
\end{example}

The term ``clique" originated in the study of social networks to denote the formation of a group of two or more people if the condition of being mutual friends is satisfied \citep{Luce}, and later was adopted to denote a complete induced subgraph. The process of clique enumeration (finding and listing all cliques in a graph) is very useful and widely used in network analysis since cliques might represent functional units within real-world networks, but the problem of finding cliques is NP-complete \citep{Cormen}. 

A natural generalization of cliques is the \textit{quasi-cliques}, which are \textit{almost complete} subgraphs. In the literature, there are slightly different ways to define them \citep{Brunato, Pattillo}, and here we will use the definition of \textit{degree-based quasi-clique} \citep{Sanei-Mehri}.

\smallskip

\begin{definition}\label{def:quasi-clique}
Given a graph $G=(V, E)$, a subgraph $H = (V', E') \subseteq G$, with $|V'| = m$, is called a \textit{$\gamma$-quasi-clique}, for a parameter $0 < \gamma \le 1$, if  $\deg_{H}(v) \ge \gamma (m-1)$, for all $v \in V'$.
\end{definition}

\smallskip

Note that if $\gamma = 1$, the $\gamma$-quasi-clique is actually a clique.

\medskip
\subsubsection{Directed and Weighted Graphs}
\label{sec:directed-graphs}

Many of the definitions and results presented here can be found in \citep{Bang, Harary}.

\smallskip

\begin{definition}\label{def:digraph}
A \textit{directed graph} (or \textit{digraph}) is a pair $G = (V, E)$ of disjoint sets such that the elements of $E$ are ordered pairs of elements of $V$ (set of vertices). The elements of $E$ are called \textit{directed edges}, (or \textit{arcs}, or \textit{arrows}), and are denoted by $(v,u)$ $(\neq (u,v))$, $v,u \in V$.  The first vertex $v$ of an arc $(v,u)$ is called \textit{tail} and the second vertex $u$ is called \textit{head}, and we say that the \textit{direction} of $(v,u)$ is from $v$ to $u$ (or from its tail to its head). Also, we say that a vertex $u$ \textit{arrives} at a vertex $v$ if the arc $(u,v)$ exists (equivalently, $(u,v)$ \textit{arrives} at $v$), and that a vertex $w$ \textit{leaves} $v$ if the arc $(v,w)$ exists (equivalently,  $(v,w)$ \textit{leaves} $v$).
\end{definition}

\smallskip

Unlike undirected edges, if we change the order of the vertices in an arc, we obtain an arc in the opposite direction. Moreover, in a digraph, we can have two arcs between the same two vertices, but with opposite directions, and when this occurs, we say that the digraph has a \textit{bidirectional edge} (or \textit{double edge}).

Most of the definitions made for undirected graphs in the previous section are straightforwardly extended to digraphs. For instance, we say that a digraph is \textit{simple} if it has no directed loops (arcs in which their tails coincide with their heads) and no multiple arcs (more than one arc with the same tail and head). The Definition \ref{def:subgraph} of subgraphs in a graph is the same as for a digraph, with the difference that now we are dealing with \textit{subdigraphs} in a digraph. 
The definitions of walk, trail, path, and cycle are easily extended for the directed case, as we will see in the next definition. 

Henceforth, all digraphs will be considered non-empty, non-null, finite, simple, and labeled, unless said otherwise.

\smallskip

\begin{definition}\label{def:dir-walk}
A  \textit{directed walk from $v_{0}$ to $v_{k}$} (or \textit{directed $(v_{0}, v_{k})$-walk}) is a sequence of vertices and arcs (not necessarily distinct) $v_{0}e_{1}v_{1}...v_{k-1}e_{k}v_{k}$ such that $e_{i} = (v_{i-1}, v_{i})$, i.e., $v_{i-1}$ is the tail and $v_{i}$ is the head of the arc $e_{i}$,  for all $1 \le i \le k$.  If $v_{0} = v_{k}$, the directed walk is called \textit{closed}, and is called \textit{open} otherwise. If the arcs of a directed walk are all distinct, it is said to be a \textit{directed trail}.  If the vertices of a directed walk are all distinct (and consequently all of the arcs), then it is said to be a \textit{directed path}. A closed directed walk with all distinct vertices and with $k \ge 2$ is a \textit{directed cycle}. Also, the \textit{length} of a directed walk is equal to its number of arcs. 
\end{definition}

\smallskip

Note that a directed cycle of length $2$ is actually a double edge.
		
\smallskip

\begin{definition}\label{def:underlying}
Given a digraph $G = (V, E)$, the  \textit{underlying undirected graph} of $G$ is the undirected graph, with the same set of vertices $V$, formed by replacing all directed edges in $E$ with undirected edges.
\end{definition}

\begin{definition}\label{def:dag}
A \textit{directed acyclic graph} (DAG), or \textit{acyclic digraph}, is a digraph that has no directed cycles.
\end{definition}

\smallskip

A notable property of (finite) DAGs is that they have at least one source and at least one sink (see \citep{Bang}, p. 32,  for proof).

\smallskip

\begin{definition}\label{def:isomorphism}
Two graphs $G_{1} = (V_{1}, E_{1})$ and $G_{2} = (V_{2}, E_{2})$ are said to be \textit{isomorphic} if there is a bijection $f: V_ {1} \rightarrow V_{2}$ such that if the vertices $v, u \in V_{1}$ are adjacent, then the vertices $f(v)$ and $ f(u)$ are adjacent in $V_{2}$ and vice versa, i.e. if and only if the bijection $f$ preserves adjacencies. Likewise, if $G_{1}$ and $G_{2}$ are digraphs, $f$ is an isomorphism if and only if $(u,v)$ is an arc in $V_{1}$ then $(f(u), f(v))$ is an arc in $V_{2}$.
\end{definition}

\smallskip

Graph properties that are invariant under graph/digraph isomorphism are called \textit{graph invariants}. For example, the order and the number of edges/arcs of a graph/digraph are graph invariants \citep{Diestel, Harary}.

\smallskip

\begin{definition}\label{def:reachable}
Given a digraph $G = (V, E)$, a vertex $v \in V$ is said to be \textit{reachable} from a vertex $u \in V$ if either $v = u$ or there exists a directed path from $u$ to $v$ in $G$.
\end{definition}

\smallskip

A difference between undirected and directed graphs is that, for digraphs, we have two different concepts of connectivity: \textit{weak connectivity} and \textit{strong connectivity}.

\smallskip

\begin{definition}\label{def:wcc}
Let $G = (V, E)$ be a digraph. A vertex $v \in V$ is said to be \textit{weakly connected} to another vertex $u \in V$ if there is an undirected path between $v$ and $u$ in the underlying undirected graph of $G$. We say that $G$ is \textit{weakly connected} if every vertex in $G$ is weakly connected to every other. A \textit{weakly connected component} of $G$ is a maximal subdigraph that is weakly connected. The largest connected component of G is called its giant component. Analogously to the undirected case, the largest weakly connected component of $G$ is called its \textit{giant component}.
\end{definition}

\begin{definition}\label{def:scc}
Let $G = (V, E)$ be a digraph. A vertex $v \in V$ is said to be \textit{strongly connected} to another vertex $u \in V$ if $u$ is reachable from $v$ and $v$ is reachable from $u$ in $G$.
We say that $G$ is \textit{strongly connected} if every vertex in $G$ is strongly connected to every other. A \textit{strongly connected component} of $G$ is a maximal subdigraph that is strongly connected.
\end{definition}


\begin{proposition}\label{prop:wcc-scc}
\textit{In a digraph, weak connectivity and strong connectivity are both equivalence relations on its vertex set.}
\end{proposition}
\begin{proof}
Let $G = (V, E)$ be a digraph. Let $\sim_{w}$ denote the relation ``is weakly connected to" and let $\sim$ denote the relation ``is connected to" in the underlying undirected graph of $G$, on the vertex set $V$. For arbitrary vertices $v,u \in V$, by definition, $v \sim_{w} u$ $\Leftrightarrow$ $v \sim u$. Since $\sim$ is an equivalence relation, $\sim_{w}$ is also an equivalence relation on $V$.

Now, denote by $\sim_{s}$ the relation ``is strongly connected to" on $V$. For arbitrary vertices $v,u,w \in V$, by Definition \ref{def:reachable}, every vertex is reachable from itself (reflexivity), if $v \sim_{s} u$, then $v$ is reachable from $u$ and $u$ is reachable from $v$ $\Rightarrow$ $u \sim_{s} v$ (symmetry), and if $v \sim_{s} u$ and $u \sim_{s} w$, similarly to the argument used in the proof of Proposition \ref{prop:connected}, there exists a directed path from $v$ to $w$ and a directed path from $w$ to $v$, thus $v \sim_{s} w$ (transitivity).  
\end{proof}

\smallskip

Another difference between undirected and directed graphs is that the density of a digraph $G = (V, E)$, with $|V|=n$ and $|E|=m$, is $\mbox{den}(G)  =  m/n(n-1)$, since the maximum number of arcs in the digraph is equal to $n(n-1)$ (see Definition \ref{def:density}).

\smallskip

\begin{definition}\label{def:dir-neighborhood}
Given a digraph $G = (V, E)$, the \textit{in-neighborhood} of 
a vertex $v$, denoted by $\mathcal{N}_{G}^{in}(v)$ or $\mathcal{N}_{G}^{-}(v)$, is the set of all vertices that arrive at $v$, and the \textit{out-neighborhood} of $v$, denoted by  $\mathcal{N}_{G}^{out}(v)$ or $\mathcal{N}_{G}^{+}(v)$, is the set of all vertices that leave $v$, i.e.
\begin{equation}
\mathcal{N}_{G}^{in}(v) = \mathcal{N}_{G}^{-}(v) = \{u \in V : (u, v) \in E\},
\end{equation}
\begin{equation}
\mathcal{N}_{G}^{out}(v) = \mathcal{N}_{G}^{+}(v) = \{u \in V : (v, u) \in E\}.
\end{equation}
\end{definition}

\smallskip

Note that if a digraph $G$ have no double edges, the neighborhood of a vertex $v$ in the underlying undirected graph is $\mathcal{N}_{G}(v) = \mathcal{N}_{G}^{in}(v) \cup \mathcal{N}_{G}^{out}(v)$. 

\smallskip

\begin{definition}\label{def:dir-degree}
Let $G = (V, E)$ be a digraph. The number of arcs arriving at a vertex $v \in V$ is its \textit{in-degree}, denoted by $\deg^{in}(v) = \deg^{-}(v)$, and the number of arcs leaving $v$ is its \textit{out-degree}, denoted by $\deg^{out}(v) = \deg^{+}(v)$. The \textit{total degree} of $v$, denoted by $\deg^{tot}(v)$, is the sum of its in-degree and its out-degree.
\end{definition}

\begin{example}
Figure \ref{fig:dir-graph2} exemplifies the digraph  $G = (V = \{ 0,1,2,3 \}, E = \{ (1,0), (2,1),$ $(0,2), (3,2), (1,3) \})$ in which $\deg^{in}(1) = 1$, $\deg^{out}(1) = 2$, and $\deg^{tot}(1) = 1 + 2 = 3$.
\end{example}

Furthermore, since a digraph may have double edges, let's denote by $\deg^{\pm}(v)$ the number of double edges incident to $v$.
It is clear that the degree of $v$ in the underlying undirected graph is given by $\deg(v) = \deg^{-}(v) + \deg^{+}(v) -  \deg^{\pm}(v)$.

Analogously to the undirected case, for a directed graph $G=(V,E)$ with $|V|=n$, if $\delta^{in}(k)$ is the number of vertices having in-degree $k$, the probability that a vertex chosen uniformly at random has in-degree equal to $k$ is given by
\begin{equation}\label{eq:in-degree-dist}
p^{in}(k) = \frac{\delta^{in}(k)}{n}.
\end{equation}

Similarly, the probability that a vertex chosen uniformly at random has out-degree equal to $k$ is given by
\begin{equation}\label{eq:out-degree-dist}
p^{out}(k) = \frac{\delta^{out}(k)}{n}.
\end{equation}

\smallskip

The fractions $p^{in}(k)$ and $p^{out}(k)$ represent the \textit{in-degree distribution} and the \textit{out-degree distribution}, respectively, of the digraph.

Before introducing the concept of weight, we need to introduce the concepts of \textit{metric}, \textit{quasi-metric}, and \textit{pre-metric}, which will be useful in the next sections as well.

\smallskip
 
\begin{definition}\label{def:metric}
Let $X$ be a set. A \textit{metric} (or \textit{distance}) on $X$ is a function $d: X \times X \rightarrow \mathbb{R}_{\ge 0}$ satisfying the following conditions, for all $x, y, z \in X$:

\begin{enumerate}
    \item $d(x,y) = 0 \iff x=y$ (identity);
    \item $d(x,y) = d(y,x)$ (symmetry);
    \item $d(x,y) \le d(x,z) + d(z,y)$ (triangular inequality).
\end{enumerate}

We say that $d$ is a \textit{quasi-metric} (or \textit{quasi-distance}) if $d$ does not necessarily satisfy the symmetry property; $d$ is called a \textit{semi-metric} (or \textit{semi-distance}) if $d$ does not necessarily satisfy the triangular inequality; if $d$ does not necessarily satisfy both conditions, symmetry and triangular inequality, then $d$ is called a \textit{pre-metric} (or \textit{pre-distance}). Given a (quasi/semi/pre-) metric $d$, the pair $(X, d)$ is called a (quasi/semi/pre-) \textit{metric space}. 
\end{definition}

In the following, weights in the set $\mathbb{R}_{\ge 0}$ are taken into consideration for defining weighted graphs (digraphs); alternatively, a different set of numbers may be used.

\begin{definition}\label{def:weig-function}
A \textit{weighted graph} (\textit{digraph}) is a triple $G^{\omega} = (V, E, \omega)$, where $G = (V, E)$ is a graph (digraph), and $\omega: V \times V \rightarrow \mathbb{R}_{\ge 0}$ is a pre-metric. The real number $\omega(v, u) = \omega_{vu}$ is the \textit{weight} of the edge $(v, u) \in E$. 
\end{definition}

\begin{example}
Figure \ref{fig:weig-graph2} exemplifies a weighted undirected graph with four vertices in which the edge thicknesses represent the weights that satisfy the relation $\omega_{02} < \omega_{23} < \omega_{13}   < \omega_{01} < \omega_{12}$.
\end{example}

\begin{figure}[h!]
\centering
\begin{subfigure}{.3\textwidth}
  \centering
  \includegraphics[scale=0.8]{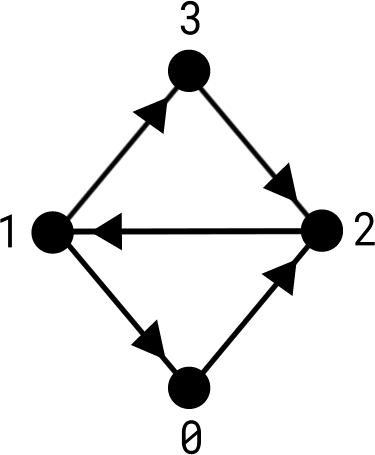}
  \caption{A directed graph.}
  \label{fig:dir-graph2}
\end{subfigure}
\begin{subfigure}{.3\textwidth}
  \centering
  \includegraphics[scale=0.8]{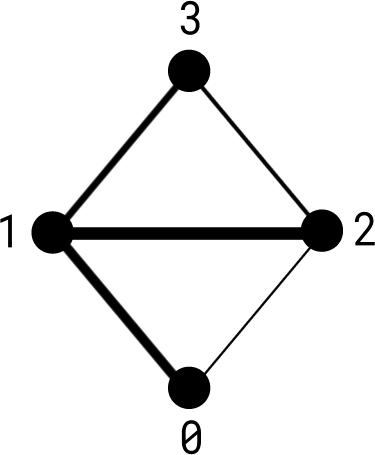}
  \caption{A weighted graph.}
  \label{fig:weig-graph2}
\end{subfigure}
\caption{Examples of directed and weighted graphs (the weights $\omega_{ij}$ are visually represented by the thicknesses of the edges).}
\label{fig:graphs-dir-weig}
\end{figure}

\smallskip

Defining the weight function $\omega$ as a pre-metric is suitable for the general case, since if $G$ is a digraph, we may have $\omega_{vu} \neq \omega_{uv}$, for some vertices $v, u \in G$.

\smallskip

\begin{definition}\label{def:weig-degree}
Let $G^{\omega}$ be a weighted graph. The \textit{weighted degree} of a vertex $v$ is the sum of all weights associated with the edges to which $v$ is incident, i.e.  $\deg_{\omega}(v)= \sum_{u \in \mathcal{N}(v)} \omega(v,u)$. Analogously, if $G^{\omega}$ is a weighted digraph, the \textit{weighted in-degree} of a vertex $v$ is $\deg_{\omega}^{-}(v) = \sum_{u \in \mathcal{N}^{-}(v)} \omega(u, v)$, and its \textit{weighted out-degree} is $\deg_{\omega}^{+}(v) = \sum_{u \in \mathcal{N}^{+}(v)} \omega(v, u)$.
\end{definition}

\smallskip

To define a distance in a weighted graph, we need a \textit{weight-to-distance} conversion function. However, for digraphs, the symmetry property is not necessarily satisfied; therefore, it would be more suitable to define a pre-distance.


\smallskip
 
\begin{definition}\label{def:weig-dist}
Given a weighted graph (or digraph) $G^{\omega} = (V, E, \omega)$ with a (normalized) weight function $\omega: V \times V \rightarrow [0,1]$, let $D^{\omega}: \mbox{Im}(\omega) \subseteq [0,1] \rightarrow \mathbb{R}_{\ge 0} \cup \{+\infty\}$ be a function defined by
\begin{equation}\label{eq:weig-distance}
D^{\omega}(\omega_{vu}) = \begin{cases}
\omega_{vu}^{-1} - 1, \mbox{ if } v \neq u;\\
+\infty, \mbox{ if } \omega_{vu} = 0;\\
0, \mbox{ if } v=u.
\end{cases}
\end{equation}

The function $D^{\omega}$ is a pre-distance. 
\end{definition}

\smallskip

It is clear that $D^{\omega}$ produces small values when applied to large weights and vice versa, thus we can say that two vertices are ``closer" when the weight of the connection between them is greater. Although $D^{\omega}$ is actually a pre-distance, from now on we will adopt an abuse of notation and call it a distance.

\begin{observation}\label{obs:weig-function}
When the weights are not normalized and take values $\ge 1$, we can define a weight-to-distance function by modifying the formula of $D^{\omega}$ by replacing the expression $(\omega_{vu}^{-1} - 1)$ with the expression $\omega_{vu}^{-1}$.
\end{observation}

\section{Algebraic and Spectral Graph Theory}
\label{sec:spectral-graph}


Algebraic graph theory is concerned with translating properties of graphs into algebraic properties, such as matrices and groups, in such a way that concepts from abstract and linear algebra can be applied. An important branch of algebraic graph theory is the \textit{spectral graph theory}, which studies the properties of graphs through eigenvalues, eigenvectors, and other spectra-related concepts of graph matrices.

The fundamental bibliography for this section is \citep{Beineke-alg, Biggs, Chung}. Undergraduate knowledge of matrix theory is assumed.


\subsection{Algebraic Graph Theory}
\label{sec:AGT}

Let's start with the definition of the most common matrix representation of a graph: the \textit{adjacency matrix}.

\smallskip

\begin{definition}\label{def:adjacency-matrix}
Let $G = (V, E)$ be a graph with $|V| = n$. The \textit{adjacency matrix} of $G$ is the $n \times n$ matrix $A = A(G) = (a_{ij})$ whose entries are given by
\begin{equation}
a_{ij} =
\begin{cases}
1 , &\text{if } (v_{i}, v_{j}) \in E, \text{ for } v_{i}, v_{j} \in V, i \neq j; \\
0, &\text{otherwise.}\\
\end{cases}
\end{equation}
\end{definition}

\medskip

Note that, by definition, $A$ is a binary matrix (formed by $0$'s and $1$'s), symmetric, and with trace $\mbox{Tr}(A) = \sum a_{ii} = 0$. In the case where $G$ is a weighted graph, the entries of $A$ will be equal to the edge weights, i.e. $a_{ij} = \omega_{ij}$. On the other hand, if $G$ is a digraph, then the matrix $A$ might be asymmetric, since we might have $a_{ij} \neq a_{ji}$, since $a_{ij} = 1$ if and only if $(v_{i}, v_{j}) \in E$, and if there is no connection in the opposite direction, i.e. if $(v_{j}, v_{i}) \not\in E$, then $a_{ji} = 0$.

\smallskip

\begin{example}
Consider the complete graph $K_{4}$, the digraph $G$ represented in Figure \ref{fig:dir-graph2}, and the weighted graph $G^{\omega}$ represented in Figure \ref{fig:weig-graph2}. The respective adjacency matrices of these graphs are:
$$
A(K_{4}) = \begin{bmatrix}0 & 1 & 1 & 1\\ 1 & 0 & 1 & 1\\ 1 & 1 & 0 & 1\\ 1 & 1 & 1 & 0\\ \end{bmatrix}, \hspace{1.4mm}
A(G) = \begin{bmatrix}0 & 0 & 1 & 0\\ 1 & 0 & 0 & 1\\ 0 & 1 & 0 & 0\\ 0 & 0 & 1 & 0\\ \end{bmatrix},
\hspace{1.4mm} \mbox{and} \hspace{1.4mm}
A(G^{\omega}) = \begin{bmatrix}0 & \omega_{01} & \omega_{02} & 0\\ \omega_{10} & 0 & \omega_{12} & \omega_{13}\\ \omega_{20} & \omega_{21} & 0 & \omega_{23}\\ 0 & \omega_{31} & \omega_{32} & 0\\ \end{bmatrix}.
$$
\end{example}

\begin{definition}\label{def:incident-matrix}
Let $G = (V, E)$ be a graph with $|V| = \{v_{0},...,v_{n-1}\}$ and $|E| = \{e_{0},...,e_{m-1}\}$. The (vertex-edge) \textit{unoriented incidence matrix} of $G$ is the $n \times m$ matrix $B = B(G) = (b_{ij})$ whose entries are given by
\begin{equation}\label{eq:UIM}
b_{ij} =
\begin{cases}
1 , &\text{if } v_{i} \text{ is incident with } e_{j};\\
0, &\text{otherwise.}\\
\end{cases}
\end{equation}

If $G$ is a directed graph, the entries $b_{ij}$ of its (vertex-arc) \textit{oriented incidence matrix} are given by
\begin{equation}\label{eq:OIM}
b_{ij} =
\begin{cases}
1, &\text{if } v_{i}  \text{ is the tail of the arc } e_{j};\\
-1 , &\text{if } v_{i} \text{ is the head of the arc }e_{j};\\
0, &\text{otherwise.}\\
\end{cases}
\end{equation}
\end{definition}

\smallskip

We can associate an oriented incidence matrix (\ref{eq:OIM}) with an undirected graph by considering arbitrary directions on its edges. The following proposition shows how to obtain the number of walks of a certain size between two vertices in an undirected or directed graph from the powers of its adjacency matrix.

\smallskip

\begin{proposition}\label{prop:power-adj}
\textit{Let $G = (V,E)$ be a graph. For a non-negative integer $k$, the $(i,j)$-entries of the $k$-th power of its adjacency matrix, $A^{k}$, is equal to the number of $(v_{i}, v_{j})$-walks of length $k$ in $G$, $v_{i}, v_{j} \in V$. Analogously, if $G$ is a digraph, the $(i,j)$-entries of $A^{k}$ is the number of directed $(v_{i}, v_{j})$-walks of length $k$ in the digraph.}
\end{proposition}

\smallskip

A simple proof (by induction in $k$) of the previous proposition for the undirected case (but straightforwardly extended to the directed case) can be found in \citep{Biggs}, p. 9.

\smallskip

\begin{corollary}\label{coro:trace-cycles}
\textit{Let $G$ be a graph with adjacency matrix $A$. For some non-negative integer $k$, the trace of the $k$-th power of $A$, $\mathrm{Tr}(A^{k})$, counts the number of closed $k$-walks in $G$. Analogously, if $G$ is a digraph, $\mathrm{Tr}(A^{k})$ counts the number of closed directed $k$-walks in $G$}.
\end{corollary}

\begin{definition}\label{def:graph-laplacian}
Given an undirected or directed graph $G$, let $B$ be its oriented incidence matrix. The \textit{Laplacian matrix} of $G$ is defined by
\begin{equation}\label{eq:graph-laplacian}
L = BB^{T}.
\end{equation}
\end{definition}

\smallskip


Note that whether $G$ is an undirected or directed graph, its Laplacian matrix $L(G)$ is symmetric ($BB^{T} = (BB^{T})^{T}$). Be aware that there are other alternative definitions for symmetric and non-symmetric versions of the Laplacian of a digraph since the above definition is \textit{independent} of the directions of the edges, but in this work, we will adopt this definition.

\smallskip

\begin{definition}\label{def:PSD}
Let $M \in \mathbb{R}^{n \times n}$ be a symmetric matrix. $M$ is said to be \textit{positive semi-definite} if $v^{T}M v \ge 0$, for all $v \in \mathbb{R}^{n}$.
\end{definition}

\begin{proposition}\label{prop:PSD}
\textit{The Laplacian matrix of a given undirected or directed graph is positive semi-definite.}
\end{proposition}
\begin{proof}
Let $G=(V,E)$ be an undirected or directed graph with $|V|=n$, and let $L$ be its Laplacian matrix. For any real-valued vector $v \in \mathbb{R}^{n}$, we have
\begin{equation}
v^{T}Lv = v^{T}BB^{T}v = (B^{T}v)^{T}(B^{T}v) = ||B^{T}v||^{2} \ge 0,
\end{equation}
\noindent where $||.||$ is the Euclidean norm.
\end{proof}




\subsection{Spectral Graph Theory}
\label{sec:SGT}

As previously stated, the spectral graph theory deals with the study of graphs through the spectra of their matrices; therefore, let's start by defining the components of these spectra: their \textit{eigenvalues}.

\smallskip

\begin{definition}\label{def:eigenvalues}
Given an undirected or directed graph $G$, the \textit{characteristic polynomial} of its adjacency matrix $A$ is given by $p_{A}(\lambda) = \det(\lambda I - A)$. The zeros of $p_{A}(\lambda)$ are the \textit{eigenvalues} of $A$. Analogously, the characteristic polynomial of its Laplacian matrix $L$ is given by $p_{L}(\mu) = \det(\mu I - L)$, and its zeros are the \textit{eigenvalues} of $L$ (or \textit{Laplacian eigenvalues}). The \textit{eigenvectors} associated with the eigenvalues of $A$ (respectively $L$) are the vectors $\mathrm{v}$ such that $Av = \lambda v$ (respectively $Lv = \mu v$).
\end{definition}

\smallskip

The \textit{spectrum} (respectively \textit{Laplacian spectrum}) of a graph $G$ is the set of all eigenvalues of $A$ (respectively $L$), together with their multiplicities. Commonly, the eigenvalues of $A$ are ordered in decreasing order: $\lambda_{1} \ge \hdots \ge \lambda_{min}$. On the other hand, the eigenvalues of the Laplacian matrix are ordered in increasing order: $\mu_{1} \le \hdots \le \mu_{max}$. 




A notable property of the Laplacian matrix is that it is positive semi-definite (Proposition \ref{prop:PSD}), and thus all of its eigenvalues are non-negative, as proved below.

\smallskip

\begin{proposition}\label{prop:eigenvalues}
\textit{All eigenvalues of the Laplacian matrix of a given undirected or directed graph are non-negative.}
\end{proposition}
\begin{proof}
Let $G=(V,E)$ be an undirected or directed graph with $|V|=n$ and with Laplacian matrix $L$. Since $L$ is positive semi-definite (Proposition \ref{prop:PSD}), for any eigenvector $v \in \mathbb{R}^{n}$ of $L$, with eigenvalue $\mu$, we have
\begin{equation}
v^{T}Lv = v^{T} \mu v = \mu v^{T} v = \mu||v||^{2} \ge 0,
\end{equation}

\noindent but $\mu||v||^{2} \ge 0$ if and only if $\mu \ge 0$, where $||.||$ is the Euclidean norm.
\end{proof}

\smallskip

From the previous proposition, we conclude that zero is always the smallest eigenvalue of the Laplacian matrix, i.e. $\mu_{1} = 0$.

\smallskip

\begin{definition}\label{def:singular-values}
Let $M \in \mathbb{R}$ be a matrix. The square roots of the eigenvalues of the positive semi-definite matrix $M^{T}M$ are called the \textit{singular values} of $M$. Typically, the notation $\{ \sigma_{i} \}_{i=1}^{n}$ is used to denote the set of singular values (with their multiplicities). 
\end{definition}

\smallskip

If $A$ is the adjacency matrix, since $A$ is real and symmetric,  all singular values of $A$ are equal to the absolute value of its eigenvalues.

An important theorem associated with non-negative matrices is the \textit{Perron-Frobenius theorem} \citep{Beineke-alg}.

\smallskip

\begin{theorem}\label{theo:perron-frobenius}
(\textbf{Perron-Frobenius}) Let $M$ be a square non-negative matrix. Then $M$ has an eigenvalue $\hat{\lambda} \ge 0$ such that $|\lambda| \le \hat{\lambda}$, for all eigenvalue $\lambda$ of $M$, and the eigenvector associated with $\hat{\lambda}$ is a non-negative vector.
\end{theorem}

\smallskip

By definition, the adjacency matrix of an undirected or directed graph is always square and non-negative, thus the Perron-Frobenius theorem guarantees that a graph or digraph has a non-negative eigenvector. 






\section{Graph Measures}
\label{sec:graph-measures}


Graphs in general present a wide variety of structural features; for example, they can contain certain types of subgraphs, their nodes can have specific degree distributions, they can have specific average path lengths, their nodes can be organized into clusters or communities, etc. Accordingly, it is natural to try to define specific measures to quantify each of these structural characteristics. Once defined, they can be used to study the functionality, topology, and dynamics of real-world networks \citep{Chung2006, Newman}. 

Over the years, a large number of graph measures have been proposed, each with the purpose of quantifying some information about a specific characteristic of the network \citep{Costa, Rubinov-2010}. For instance, the \textit{global efficiency} tries to quantify how efficiently information is propagated through the graph; the \textit{degree centrality} tries to quantify how influential or central a node is in the graph; the \textit{closeness centrality} tries to quantify the ability of a node to transmit information \citep{Dehmer-2014, Dehmer-2011, Estrada-2011, Estrada-2015}. Recently, all these quantitative approaches to dealing with graphs have been brought together into a new branch within graph theory called \textit{quantitative graph theory} (QGT) \citep{Dehmer-2017}, which we will discuss in more detail in Chapter \ref{chap:chap5}.

One can classify graph measures into two major categories:

\begin{itemize}
\item  \textbf{Local measures:} refer to the measures that try to extract the properties of the nodes in a graph.

\item \textbf{Global measures:} refer to the measures that try to extract the global properties of a graph by taking into account the graph as a whole.
\end{itemize}




In the next subsections, we present several well-known graph measures (most of them can be found in the aforementioned bibliography), and all of them represent \textit{graph invariants}. Also, all graphs (or digraphs), $G = (V, E)$, will be considered with $|V| = n$ and $|E| = m$, and its vertices and edges (or arcs) will be denoted by $i, j \in V$ and $(i,j) \in E$, respectively, and $A = (a_{ij})$ will denote their adjacency matrix. Unless otherwise specified, the variants related to weighted graphs are directly extended to weighted digraphs.

\subsection{Distance-Related Measures}
\label{subsec:Distance-Related-Measures}

Perhaps the most common and intuitive measures related to graphs are the node degree and node degree distribution. The next most common approaches to characterize graphs are the \textit{distance-based measures}. These measures allow us to quantify integration at the global level of the network, that is, how each node interacts with all other nodes, and thus we can classify them as measures of \textit{global integration} \citep{Sporns-2010}. In what follows, we present some of the most relevant measures associated with distances.

\bigskip
\noindent \textbf{Shortest Path and Distance} (\textit{global}):
The \textit{shortest path} (or \textit{geodesic path}) between two given vertices in an undirected graph is the path with the minimum number of edges between them. The \textit{distance} (or \textit{geodesic distance}) between two vertices $i$ and $j$, denoted as $d(i,j) = d_{ij}$, is the length of the shortest path between $i$ and $j$, and it can be written in terms of the adjacency matrix entries as
\begin{equation}\label{eq:shortest-distance}
   d_{ij} = \sum_{x,y \in g_{i\leftrightarrow j}} a_{xy},
\end{equation}

\noindent where $g_{i\leftrightarrow j}$ is the geodesic path between $i$ and $j$. Since the adjacency matrix is symmetric, we have $d_{ij} = d_{ji}$, for all $i, j$. Also, we define $d_{ij} = +\infty$ for every disconnected pair $i,j$. 

Analogously, for directed graphs, the \textit{directed distance} from a vertex $i$ to a vertex $j$ is the length of the shortest directed path from $i$ to $j$, i.e.
\begin{equation}\label{eq:dir-shortest-distance}
   \vec{d}_{ij} = \sum_{x,y \in \vec{g}_{i\rightarrow j}} a_{xy},
\end{equation}

\noindent where $\vec{g}_{i\rightarrow j}$ is the geodesic directed path from $i$ to $j$.  Notice that $\vec{d}_{ij}$ is an asymmetric function, since we might have $\vec{d}_{ij} \neq \vec{d}_{ji}$ for some $i$ and $j$, therefore, strictly speaking, $\vec{d}$ is a \textit{quasi-distance} (Definition \ref{def:metric}). If there is no directed path from $i$ to $j$, then we put $\vec{d}_{ij} = +\infty$. Also, it is worth noting that the length of the shortest path is equal to the length of the shortest walk in both undirected and directed cases.

Lastly, for weighted graphs, let $f$ be a weight-to-distance function. The \textit{weighted distance} between two vertices $i$ and $j$ is the length of the shortest path in relation to the function $f$, i.e.
\begin{equation}\label{eq:weig-shortest-distance}
d_{ij}^{\omega} = \sum_{x,y \in g_{i\leftrightarrow j}(f)} f(\omega_{xy}),
\end{equation}

\noindent where $g_{i\leftrightarrow j}(f)$ is the geodesic path between $i$ and $j$ in relation to $f$ (or the \textit{weighted geodesic path}). Here, since we want a higher weight to be associated with a shorter path, if the weights are normalized, then we identify $f$ as the weight-to-distance function defined by Equation (\ref{eq:weig-distance}), i.e. $f = D^{\omega}$, otherwise, we can use the modified version of $D^{\omega}$ as explained in Observation \ref{obs:weig-function}.

\bigskip
\noindent \textbf{Characteristic Path Length} (\textit{global}):
The \textit{characteristic path length} (or \textit{average path length} or \textit{average shortest path length}) \citep{Watts} is a measure that represents the path length that is most likely to occur in the graph, i.e. it is the average distance between all possible pairs of vertices. Essentially, it is a measure of how efficiently information travels through the graph. Mathematically, it is defined by
\begin{equation}\label{eq:CPL}
L(G)  = \frac{1}{n} \sum_{i \in V}\frac{\sum_{j \in V, j\neq i} d_{ij}}{n-1} = \sum_{\substack{i,j\in V \\ i\neq j}} \frac{d_{ij}}{n(n-1)}.
\end{equation}

For directed graphs, \textit{directed characteristic path length}, $\vec{L}$, is obtained by replacing the distance $d_{ij}$ with the directed distance $\vec{d}_{ij}$, and for the weighted case, the \textit{weighted characteristic path length}, $L^{\omega}$, is obtained by replacing $d_{ij}$ with $d^{\omega}_{ij}$. Some algorithms consider $d_{ij} = 0$ (respectively $\vec{d}_{ij} = 0$) when $(i,j) \not\in E$, in this case $n$ is the order of $G$, otherwise $n$ must be the order of its giant component.

\bigskip
\noindent \textbf{Eccentricity} (\textit{global}):
For a connected graph, the \textit{eccentricity} of a vertex $i$ is the maximum distance between $i$ and any other vertex $j$ in the graph, i.e.
\begin{equation}\label{eq:eccentricity}
\mathrm{ecc}(i) = \max_{j \in V}  d(i, j) .
\end{equation}

Analogously, for strongly connected digraphs, the \textit{directed eccentricity}, denoted as $\vec{\mbox{ecc}}(i)$, is the maximum directed distance from $i$ to any other vertex $j$ in the digraph.

\bigskip
\noindent \textbf{Diameter and Radius} (\textit{global}):
The \textit{diameter} of a connected graph $G$ is the maximum eccentricity among all of its vertices, i.e.
\begin{equation}\label{eq:diameter}
\mathrm{diam}(G) = \max_{i \in V} (\mathrm{ecc}(i)).
\end{equation}

If $G$ is a strongly connected digraph, the \textit{directed diameter} is the maximum directed eccentricity among all of its vertices. A related measure is the \textit{radius} of a connected graph $G$, which is the minimum eccentricity among all of its vertices:
\begin{equation}\label{eq:radius}
\mathrm{rad}(G) = \min_{i \in V}  (\mathrm{ecc}(i)).
\end{equation}

Again, for strongly connected digraphs, the \textit{directed radius} is obtained by replacing $\mbox{ecc}(i)$ with the directed eccentricity $\vec{\mathrm{ecc}}(i)$ in the previous formula.

\bigskip
\noindent \textbf{Global Efficiency} (\textit{global}):
Let's define by $\epsilon_{ij} = d_{ij}^{-1}$ the \textit{efficiency} in the communication between two vertices $i$ and $j$ of a graph $G$. Latora and Marchiori \citep{Latora} defined the \textit{global efficiency} (or \textit{average efficiency}) of $G$ as 
\begin{equation}\label{eq:global-efficiency}
 E_{glob}(G) = \frac{1}{n} \sum_{i \in V}\frac{\sum_{j \in V, j\neq i} \epsilon_{ij}}{n-1}.
\end{equation}

This measure is closely related to the characteristic path length, and it is another way of trying to capture how efficiently information is propagated through the graph. The \textit{directed global efficiency}, $\vec{E}_{glob}$, is obtained by replacing $d_{ij}$ with the directed distance $\vec{d}_{ij}$, and the \textit{weighted global efficiency}, $E_{glob}^{\omega}$, is obtained by replacing $d_{ij}$ with $d^{\omega}_{ij}$.


\bigskip
\noindent \textbf{Communicability} (\textit{global}):
The \textit{communicability} between two nodes \citep{Estrada-2009a, Estrada-2015} is a measure based on the number of walks that exist between them, and that assigns different importance to these walks based on their lengths. It can be expressed in terms of the powers of the adjacency matrix as
\begin{equation}\label{eq:communicability-1}
CM(i,j) = \sum^{\infty}_{k=0} c_{k}(A^{k})_{ij},
\end{equation}

\noindent since $(A^{k})_{ij}$ is the number of walks of length $k$ between $i$ and $j$ (Proposition \ref{prop:power-adj}), and the coefficients $c_{k}$ must be defined in such a way that: guarantees the convergence of the series; the values of $c_{k}$ increase as we decrease $k$ and vice versa (i.e. assign greater importance to shorter walks); and produce positive values for all pairs $i,j$, $i \neq j$. 

A convenient choice for the coefficients $c_{k}$ is $c_{k} = 1/k!$, since in this case all of the above conditions are satisfied and we have: 
\begin{equation}\label{eq:communicability-2}
CM(i,j) = \sum^{\infty}_{k=0} \frac{(A^{k})_{ij}}{k!} = (\exp(A))_{ij}.
\end{equation}

Note that the communicability between a pair of nodes increases when the number of walks between them increases. 


\bigskip
\noindent \textbf{Returnability} (\textit{global}):
Returnability \citep{Estrada-2009b, Estrada-2015} is a measure that tries to quantify the amount of information that flows through a digraph and returns to its original sources by considering the relative contribution of all closed directed walks presented in the digraph. It may be considered as a measure of ``returnability of information." Formally, since $\mbox{Tr}(A^{k})$ is equal to the number of closed directed $k$-walks in a digraph (Corollary \ref{coro:trace-cycles}), we define its \textit{returnability} as
\begin{equation}\label{eq:returnability1}
K_{r}(G) = \sum^{\infty}_{k=2} c_{k} \mbox{Tr}(A^{k}),
\end{equation}

\noindent where the coefficients $c_{k}$ must satisfy the same three conditions exposed previously for the communicability measure. By choosing $c_{k} = 1/k!$ as before, we have:
\begin{equation}\label{eq:returnability2}
K_{r}(G) = \sum^{\infty}_{k=2} \frac{\mbox{Tr}(A^{k})}{k!} = \mbox{Tr}(\exp(A)) - n.
\end{equation}

The term $n$ in the right-hand side of Equation (\ref{eq:returnability2}) comes from our assumption that the digraph does not contain any self-loops or directed cycles of length $1$, so the two first terms, $\mathrm{Tr}(A^{0})= n$ and $\mathrm{Tr}(A^{1}) = 0$, are removed. Also, we can define the \textit{relative returnability} as
\begin{equation}\label{eq:relative-returnability}
K_{r}'(G) = \frac{\mathrm{Tr}(\exp(A)) - n}{\mathrm{Tr}(\exp(A')) - n},
\end{equation}

\noindent where $A'$ is the adjacency matrix of the underlying undirected graph.



\subsection{Measures of Centrality}
\label{subsec:centrality-measures}

In this part, we present the most relevant centrality measures related to a node found in the literature, such as degree centrality, closeness centrality, and betweenness centrality. Each one of these measures tries to quantify some specific property or role of a node in the graph. In general, node centrality measures try to quantify the ``importance," ``influence," or ``centrality" of a node within a graph, in the sense of capturing the capacity that a node has to ``spread" and/or ``receive" information for/from other nodes, which may be characterized by the direct contact with other nodes, its closeness to a large number of other nodes, and the number of pairs of nodes that require this node as an intermediary in their interactions.

Let's start with the simplest and most straightforward measure of node centrality: the \textit{degree centrality}.

\bigskip
\noindent \textbf{Degree Centrality} (\textit{local}):
The basic idea behind the degree centrality of a node \citep{Freeman1978} is to quantify the importance or centrality of this node based on the number of edges which are incident to it in the graph, i.e. if a node's degree is higher than another's, then that node is more influential or central than the other. As mentioned above, by ``influential" or ``central" we mean the capacity of spreading and/or receiving information for/from other nodes. Formally, the \textit{degree centrality} of a node $i$ is simply its degree divided by $(n-1)$, which is the maximum number of nodes that $i$ can be adjacent to, i.e. it is the proportion of nodes that are adjacent to $i$:
\begin{equation}\label{eq:degree-centrality}
   C_{dg}(i) = \frac{\deg(i)}{n-1}.
\end{equation}

For directed graphs, we have two different centralities, namely, the in-degree centrality and the out-degree centrality. The \textit{in-degree centrality} is the proportion of nodes which arrives at $i$:
\begin{equation}\label{eq:in-degree-centrality}
   C^{-}_{dg}(i) = \frac{\deg^{-}(i)}{n-1}.
\end{equation}

Similarly, the \textit{out-degree centrality} is the proportion of nodes which leaves $i$:
\begin{equation}\label{eq:out-degree-centrality}
   C^{+}_{dg}(i) = \frac{\deg^{+}(i)}{n-1}.
\end{equation}

Moreover, for a digraph without double edges, it's clear that $C_{dg}(i) = C^{-}_{dg}(i) + C^{+}_{dg}(i)$. On the other hand, if double edges exist in the digraph, then $C_{dg}(i) = C^{-}_{dg}(i) + C^{+}_{dg}(i) - \deg^{\pm}(i)/(n-1)$, where $\deg^{\pm}(i)$ is the number of double edges incident to $i$. The weighted versions of each of the previous degree centralities are obtained by replacing the degrees with their respective weighted formulas (Definition \ref{def:weig-degree}).


\bigskip
\noindent \textbf{Closeness Centrality} (\textit{local}):
Closeness centrality, originally proposed by Bavelas \citep{Bavelas}, is a measure that quantifies how relatively close a node is to all other nodes in the graph; that is, it identifies the nodes that can spread/receive information in an efficient way. Mathematically, the \textit{closeness centrality} of a node $i$ can be defined as the inverse of the sum of the shortest paths between $i$ and every other node in the graph, i.e.
\begin{equation}\label{eq:closeness-centrality}
Cl(i) = \frac{1}{\sum_{\substack{j \in V \\ j\neq i}} d_{ij}}.
\end{equation}

It's important to note that closeness centrality is defined for connected graphs. Thus, if $N$ is the order of the giant component, we can define the \textit{normalized closeness centrality} by multiplying the formula (\ref{eq:closeness-centrality}) by $(N-1)$:
\begin{equation}\label{eq:normal-closeness-centrality}
Cl(i) = \frac{N-1}{\sum_{\substack{j \in V \\ j\neq i}} d_{ij}}.
\end{equation}

For the directed case, we simply replace the distance $d_{ij}$ with the directed distance $\vec{d}_{ij}$, and for the weighted case, we replace $d_{ij}$ with $d^{\omega}_{ij}$.

\bigskip
\noindent \textbf{Harmonic Centrality} (\textit{local}):
As mentioned previously, the closeness centrality is not defined for disconnected graphs. To overcome this problem, the \textit{harmonic centrality} \citep{Boldi} of a node $i$ was introduced as the sum of the inverse of all distances between $i$ and all other nodes, i.e.
\begin{equation}\label{eq:harmonic-centrality}
HC(i) = \sum_{\substack{j \in V \\ j \neq i}} \frac{1}{d(i,j)},
\end{equation}

\noindent where the convention $1/\infty = 0$ is adopted. For the directed case, we simply replace $d_{ij}$ with $\vec{d}_{ij}$, and for the weighted case, we replace $d_{ij}$ with $d^{\omega}_{ij}$.


\bigskip
\noindent \textbf{Betweenness Centrality} (\textit{local}): 
Betweenness centrality \citep{Freeman1978} is a measure of how much control or influence a node has over the information flow in the graph. It takes into account the proportion between all geodesic paths that pass through a specific node and all other geodesic paths between all other nodes in the graph. Formally, the \textit{betweenness centrality} of a node $i$ is defined as 
\begin{equation}\label{eq:betweenness-centrality}
B(i) =  \sum_{\substack{h, j \in V\\h\neq j, h\neq i, j\neq i }} \frac{\rho_{hj}(i)}{\rho_{hj}},
\end{equation}

\noindent where $\rho_{hj}(i)$ is the number of geodesic paths from $h$ to $j$ that pass through $i$, and $\rho_{hj}$ is the total number of geodesic paths from $h$ to $j$.

Also, note that betweenness centrality is defined for connected graphs, and we can define the \textit{normalized betweenness centrality} by multiplying Equation (\ref{eq:betweenness-centrality}) by the normalization term $2/(N-1)(N-2)$, where $N$ is the number of nodes in the giant component:
\begin{equation}\label{eq:normal-betweenness-centrality}
B(i) = \frac{2}{(N-1)(N-2)} \sum_{\substack{h, j \in V\\h\neq j, h\neq i, j\neq i }} \frac{\rho_{hj}(i)}{\rho_{hj}}.
\end{equation}

For directed graphs, we consider $\vec{\rho}_{hj}(i)$ as the number of directed geodesic paths from $h$ to $j$ that pass through $i$ and $\vec{\rho}_{hj}$ the total number of directed geodesic paths from $h$ to $j$. Moreover, the normalization term for the directed case is $1/(N-1)(N-2)$. Similarly, for weighted graphs, we define $\rho^{\omega}_{hj}(i)$ and $\rho^{\omega}_{hj}$ in an analogous way as before, but considering the weighted geodesic paths.

Over the years, several variants of betweenness centrality have been proposed, such as the flow betweenness centrality, the random walk betweenness centrality, and the communicability betweenness centrality (see \citep{Estrada-2011}, chap. 7).


\bigskip
\noindent \textbf{Reaching Centrality} (\textit{local/global}):
Before introducing the measure known as ``reaching centrality," let's introduce the concept of hierarchy in networks \citep{Mones}. In the literature, there are three main types of hierarchies, namely: \textit{order hierarchy}, \textit{nested hierarchy}, and \textit{flow hierarchy}. Order hierarchy is described as the presence of an order in the elements of a set, that is, it is equivalent to an ordered set (e.g., an order in the vertex set of a network); nested hierarchy is described as the presence of higher and lower level components in the network, such that higher level components comprise of and include lower level components (e.g., cliques formed by smaller cliques); finally, flow hierarchy comprises the influence that a node has on other nodes, thus these other nodes are considered to be at a lower level than the node that influences them.
 
The concept of reaching centrality of a node, as proposed in \citep{Mones}, quantifies the concept of flow hierarchy in a digraph, i.e. how a node influences the flow of information through the digraph. Let $G$ be a digraph and let $r_{G}(i)$ be the number of nodes in $G$ that are reachable from the node $i$. The \textit{local reaching centrality} of $i$ is defined as the proportion of the nodes which are reachable from $i$, i.e.
\begin{equation}\label{eq:local-reaching-centrality}
C_{R}(i) = \frac{r_{G}(i)}{n-1}.
\end{equation}

Let $C_{R}^{max} = \max_{i \in V} \{ C_{R}(i) \}$ be the maximum local reaching centrality obtained in $G$. We define the \textit{global reaching centrality} as the average of the difference between $C_{R}^{max}$ and $C_{R}(i)$, over all nodes in the digraph, i.e.

\begin{equation}\label{eq:global-reaching-centrality}
GRC = \frac{\sum_{i \in V} [C_{R}^{max} - C_{R}(i)]}{n-1}.
\end{equation}

For a weighted digraph, there are variants of the local and global reaching centralities introduced in \citep{Mones}. However, here we will remain with the same formulas as we are solely considering the count of the reachable nodes.

\subsection{Measures of Segregation}
\label{subsec:segration-measures}

By ``measures of segregation," we mean to specify measures that attempt to quantify the tendency of nodes to segregate into clusters, communities, or modules, which are different ways to refer to densely connected neighborhoods. In the literature, there are a variety of such measures, but here we will focus on three of them, namely: \textit{clustering coefficient}, \textit{rich-club coefficient}, and \textit{local efficiency}.

\bigskip
\noindent \textbf{Clustering Coefficient} (\textit{local/global}): 
The clustering coefficient, as introduced in \citep{Watts}, is a measure that tries to quantify the tendency of nodes to form clusters in the network. Formally, the \textit{clustering coefficient} (or \textit{agglomeration coefficient}) (local) of a node $i$ is defined as the ratio between the number of triangles containing $i$, denoted by $t(i)$, and the maximum number of edges between its neighbors (equal to $\deg(i)(\deg(i) - 1)/2$), i.e.
\begin{equation}\label{eq:clustering-coef}
 C(i) =  \frac{2t(i)}{\deg(i)(\deg(i) - 1)}.
\end{equation}

The quantity $t(i)$ can be expressed in terms of the adjacency matrix entries as
\begin{equation}\label{eq:number-triangles}
t(i) = \frac{1}{2} \sum_{j,h\in V} a_{ij}a_{ih}a_{jh}.
\end{equation}

The formula (\ref{eq:clustering-coef}) is also known as \textit{local clustering coefficient}. Furthermore, we can define the \textit{average clustering coefficient} as the sum of $C(i)$ over all nodes $i \in V$ normalized by the order of the graph, i.e.
\begin{equation}\label{eq:avg-clustering-coef}
\bar{C} = \frac{1}{n} \sum_{i \in V} C(i) = \frac{1}{n} \sum_{i \in V} \frac{2t(i)}{\deg(i)(\deg(i) - 1)}.
\end{equation}

For directed graphs, Fagiolo \citep{Fagiolo} proposed a variant of the formula (\ref{eq:clustering-coef}) by considering directed triangles and the total number of arcs between the neighbors of $i$ (excluding the double edges), i.e.
\begin{equation}\label{eq:dir-clustering-coef}
\vec{C}(i) =  \frac{\vec{t}(i)}{\deg^{tot}(i)(\deg^{tot}(i) - 1) - 2\deg^{\pm}(i)},
\end{equation}

\noindent where $\deg^{tot}(i) = \deg^{-}(i) + \deg^{+}(i)$, $\deg^{\pm}(i)$ is the number of double edges incident to $i$, and $\vec{t}(i)$ is the number of directed triangles containing $i$ as one of their nodes, that is,
\begin{equation}\label{eq:dir-number-triangles}
\vec{t(i)}  = \frac{1}{2} \sum_{j,h\in V} (a_{ij} + a_{ji})(a_{ih} + a_{hi})(a_{jh} + a_{hj}).
\end{equation}

Analogously to the formula (\ref{eq:avg-clustering-coef}), we define the \textit{average directed clustering coefficient} as the sum of $\vec{C}(i)$ over all nodes $i \in V$ normalized by the order of the graph.

A weighted version of the clustering coefficient was proposed by  Onnela et al. \citep{Onnela} by defining the quantity $t^{\omega}(i)$ as the sum of the geometric mean of the scaled weights of the triangles (triangle intensities):
\begin{equation}\label{eq:weig-number-triangles}
t^{\omega}(i) = \frac{1}{2} \sum_{j,h\in V} ( \hat{\omega}_{ij}\hat{\omega}_{ih}\hat{\omega}_{jh})^{1/3}.
\end{equation}

\noindent where $\hat{\omega}_{ij} = \omega_{ij}/\max(\omega_{ij})$, and then replacing $t(i)$ with $t^{\omega}(i)$ in the formula (\ref{eq:clustering-coef}).

\bigskip
\noindent \textbf{Rich-Club Coefficient} (\textit{local}):
The rich-club coefficient, first proposed by Zhou and Mondragon \citep{Zhou}, is a measure that attempts to quantify the tendency of nodes with high degrees (called the ``rich nodes") to be densely connected to each other. Formally, let $N_{>k}$ be the number of nodes with degree $> k$, and let $E_{>k}$ be the number of edges among the nodes with degree $> k$, the \textit{rich-club coefficient} is defined as the ratio:
\begin{equation}\label{eq:rich-club-coef}
\phi(k) = \frac{2E_{>k}}{N_{>k}(N_{>k} - 1)}.
\end{equation}

For directed graphs, Smilkov and Kocarev \citep{Smilkov} defined the \textit{in-degree rich-club coefficient}, which considers the in-degree instead of the undirected degree, i.e.
\begin{equation}\label{eq:in-rich-club-coef}
\phi^{in}(k) = \frac{E^{in}_{>k}}{N^{in}_{>k}(N^{in}_{>k} - 1)},
\end{equation}

\noindent where $N^{in}_{>k}$ is the number of nodes $i$ having $\deg^{-}(i) > k$, and $E^{in}_{>k}$ is the number of directed edges connecting those $N^{in}_{>k}$ nodes. 

The formula (\ref{eq:rich-club-coef}) can also be defined using the out-degree instead of the in-degree, and in this case, we have the \textit{out-degree rich-club coefficient}:
\begin{equation}\label{eq:out-rich-club-coef}
\phi^{out}(k) = \frac{E^{out}_{>k}}{N^{out}_{>k}(N^{out}_{>k} - 1)},
\end{equation}

\noindent where $N^{out}_{>k}$ is the number of nodes $i$ having $\deg^{+}(i) > k$, and $E^{out}_{>k}$ is the number of directed edges connecting those $N^{out}_{>k}$ nodes.

\bigskip
\noindent \textbf{Local Efficiency} (\textit{global}):
Let $G(i)$ denote the subgraph of $G$ formed by the open neighborhood of a vertex $i$. Latora and Marchiori \citep{Latora} defined the \textit{local efficiency} of $G$ as the average efficiency of the local subgraphs $G(i)$, i.e.
\begin{equation}\label{eq:local-efficiency}
 E_{loc}(G) = \frac{1}{n} \sum_{i \in V} E_{glob}(G(i)).
\end{equation}

The \textit{directed} and \textit{weighted local efficiency} are is obtained by replacing $E_{glob}(G(i))$ with $\vec{E}_{glob}(G(i))$ and with $E_{glob}^{\omega}(G(i))$, respectively.

\subsection{Entropy Measures}
\label{subsec:entropy-measures}

The concept of \textit{entropy} first appeared in the study of thermodynamic systems when, in an 1865 study, Rudolf Clausius coined the term to mean ``transformation content" in the sense of availability/unavailability of energy in a system, and it was further studied in other contexts \citep{Gleick}. In his 1948 paper \citep{Shannon}, Claude Shannon, studying transmission of signals in communication systems, proposed an idea of ``information entropy" (which came to be known as \textit{Shannon entropy}) as a measure of \textit{uncertainty} in the sense of how much ``randomness" a signal carries or ``the `amount of surprise' a message source has for a receiver" \citep{Mitchell2009}.

Here, we are interested in the concept of entropy associated with networks. The \textit{graph entropy} or the \textit{topological information content} of a graph can be interpreted as a type of \textit{structural entropy}, and it was first proposed by Rashevsky \citep{Rashevsky} as the Shannon entropy of some probability distributions obtained from the symmetric structure of the vertices of a graph, which can also be computed in terms of the orbits of its automorphism group \citep{Mowshowitz}. Nonetheless, over the years, several new approaches to defining graph entropy have been proposed \citep{Dehmer2011}, for example, entropies based on the degree of the nodes \citep{Cao, Wang2006}, entropies based on the eigenvalues of the adjacency matrix \citep{Sivakumar, Sun}, and entropies based on the Laplacian eigenvalues \citep{Passerini, Ye}.

In what follows, we discuss the \textit{entropy of the degree distribution} for undirected and directed graphs.

\bigskip
\noindent \textbf{Entropy of the Degree Distribution} (\textit{global}):
The \textit{entropy of the degree distribution} (or \textit{degree distribution entropy}) \citep{Wang2006} is a measure that tries to capture the heterogeneity of the edge distribution in a given graph $G$. It is defined as the Shannon entropy of the node degree distributions (\ref{eq:prob-degree}), i.e.
\begin{equation}\label{eq:distribution-entropy}
H(G) = - \sum_{k=1}^{n-1}  p(k) \log_{2} p(k).
\end{equation}

In the case where $G$ is a directed graph, we have two possible definitions: the \textit{in-degree distribution entropy}  and the \textit{out-degree distribution entropy}. These entropies are respectively defined by
\begin{equation}\label{eq:in-distribution-entropy}
H^{in}(G) = - \sum_{k=1}^{n-1}  p^{in}(k) \log_{2} p^{in}(k),
\end{equation}

\begin{equation}\label{eq:out-distribution-entropy}
H^{out}(G) = - \sum_{k=1}^{n-1}  p^{out}(k) \log_{2} p^{out}(k),
\end{equation}

\noindent where $p^{in}(k)$ are the node in-degree distributions (\ref{eq:in-degree-dist}) and $p^{out}(k)$ are the node out-degree distributions (\ref{eq:out-degree-dist}).

Notice that if all nodes of a graph or digraph have the same degree, or same in-degree, or same out-degree, the respective entropies reach their minimum, i.e. $H(G)=H^{in}(G)=H^{out}(G)=0$ (with the convention $0\log_{2} 0 = 0$), and reach their maximum if $p(k)=p^{in}(k)=p^{out}(k)=1/(n-1)$, for all $k=1,2,...,n-1$. Thus, we can roughly say that these entropy measures quantify the ``degree of disorder" or ``degree of randomness" (\textit{in relation to} the inner or outer flux in the case of $H^{in}$ or $H^{out}$, respectively) of a network since they produce higher values for networks that are closer to a random model and lower values for those that are closer to a regular model (see Section \ref{sec:random-graphs} for an overview of random graph models).

\subsection{Spectrum-Related Measures}
\label{subsec:spectrum-measures}

In this last part, we present some measures that are related to the spectra of the adjacency and Laplacian matrices of an undirected or directed graph, such as \textit{graph energy}, \textit{Katz centrality}, \textit{eigenvector centrality}, and \textit{spectral entropy}.

\bigskip
\noindent \textbf{Graph Energy} (\textit{global}):
Gutman \citep{Gutman} originally defined the graph energy of an undirected graph as the sum of the absolute values of the eigenvalues (with multiplicities) of its adjacency matrix. As observed by Nikiforov \citep{Nikiforov}, the \textit{trace norm} of a matrix, denoted by $|| \cdot ||_{*}$,  is the sum of its singular values, and as we observed earlier, for a real symmetric matrix, its singular values are equal to the absolute value of its eigenvalues, thus the energy of a graph can be defined as the trace norm of its adjacency matrix.

In effect, the trace norm can be used to define energy for digraphs as well. Arizmendi and Arizmendi \citep{Arizmendi} defined the energy of a directed graph as the trace of the matrix $|A|^{+} = (AA^{T})^{1/2}$ (equivalently for $|A|^{-} = (A^{T}A)^{1/2}$), which is equal to the sum of the singular values of $A$ (with their respective multiplicities), which in turn is equal to the trace norm as discussed previously. Therefore, for both types of graphs, undirected and directed, we define the \textit{graph energy} as
\begin{equation}\label{eq:graph-energy}
\varepsilon(G) = ||A||_{*} = \mbox{Tr}(|A|^{+}) = \mbox{Tr}(|A|^{-}) = \sum^{n}_{i=1} \sigma_{i}.
\end{equation}

In particular, if $G$ is an undirected graph, we have $\varepsilon(G) = \sum^{n}_{i=1} |\lambda_{i}|$. Graph energy may be seen as a measure of graph connectivity \citep{Shatto}.

\bigskip
\noindent \textbf{Katz Centrality} (\textit{local}):
Katz centrality, first proposed by L. Katz \citep{Katz}, is based on the idea that the importance of a node is not only influenced by its immediate neighbors but also by nodes that are farther away in the graph. Unlike the degree centrality of a node $i$, for example, which takes into account solely the influence of its nearest-neighbors (walk of length $1$), the Katz centrality of $i$ also takes into account all the other nodes that are connected to $i$ by a walk of length $> 1$. Formally, we can take these nodes into account by considering the series $A^{0} + A^{1} + ... + A^{k} + ...$, since the $(i,j)$-entries of $A^{k}$ are the number of $(i,j)$-walks of length $k$ (Proposition \ref{prop:power-adj}). However, this series might diverge, then we need to introduce an \textit{attenuation factor} $\alpha \in \mathbb{R}_{+}$ in such a way that the series converges. This leads us to the definition of the \textit{Katz centrality} of a node $i$:

\begin{equation}\label{eq:katz-centrality}
K(i) = \Bigg[ \Bigg( \sum_{k=0}^{\infty} \alpha^{k} A^{k} \Bigg) |1 \rangle \Bigg]_{i},
\end{equation}

\noindent where $|1 \rangle = (1,...,1)^{T}$ and the subscript $i$ in the brackets represents the $i$-th position of the vector inside the brackets. In order to guarantee the convergence of (\ref{eq:katz-centrality}), the attenuation factor must be $\alpha \neq 1/\lambda_{1}$, where $\lambda_{1}$ is the largest eigenvalue of $A$. In this case, Equation (\ref{eq:katz-centrality}) can be written as 
\begin{equation}\label{eq:katz-centrality-mat}
K(i) = \big[ (I_{n} - \alpha A)^{-1} |1 \rangle \big]_{i},
\end{equation}

\noindent where $I_{n}$ is the $n \times n$ identity matrix. Typically, the attenuation factor is chosen to be $\alpha < 1/\lambda_{1}$. For directed graphs, we consider either the arcs arriving at $i$ ($K^{in}(i)$) or the arcs leaving $i$ ($K^{out}(i)$), i.e.
\begin{equation}\label{eq:dir-katz-centrality-in}
K^{in}(i) = \big[ \langle 1| (I_{n} - \alpha A)^{-1}  \big]_{i},
\end{equation}
\begin{equation}\label{eq:dir-katz-centrality-out}
K^{out}(i) = \big[ (I_{n} - \alpha A)^{-1} |1 \rangle  \big]_{i}.
\end{equation}

\bigskip
\noindent \textbf{Eigenvector Centrality} (\textit{local}):
The \textit{eigenvalue centrality} \citep{Bonacich} of a node $i \in V$ is defined as the $i$-th entry of the eigenvector associated with the largest eigenvalue ($\lambda_{1}$) of the adjacency matrix $A$ of $G$, i.e.
\begin{equation}\label{eq:eigenvector-centrality}
C_{e}(i) = \big( v_{1} \big)_{i} = \Bigg( \frac{1}{\lambda_{1}} A v_{1} \Bigg)_{i}.
\end{equation}

If $G$ is a digraph, the adjacency matrix $A$ might be non-symmetric ($A^{T} \neq A$), then we consider the right eigenvector ($Av = \lambda_{1}v$) or the left eigenvector ($A^{T}v = \lambda_{1}v$) associated with $\lambda_{1}$ in formula (\ref{eq:eigenvector-centrality}), and then we can have right and left eigenvector centralities associated with the node $i$. 

Moreover, the Perron-Frobenius theorem (Theorem \ref{theo:perron-frobenius}) guarantees that the eigenvector associated with $\lambda_{1}$ is non-negative, thus the eigenvalue centrality of every node is non-negative. 

The eigenvalue centrality can be seen as a modification of the Katz centrality, as demonstrated in \citep{Estrada-2015}, and it can be interpreted as a measure that tries to quantify the importance of a node according to the importance of its neighbors.



\bigskip
\noindent \textbf{Spectral Entropy} (\textit{global}):
In Subsection \ref{subsec:entropy-measures}, we have already discussed the concept of entropy and presented graph entropy based on the degree distribution. Now we define the \textit{spectral entropy} of a graph or digraph $G$ based on its Laplacian eigenvalues as follows. Let $L(G)$ be the Laplacian matrix of $G$ with eigenvalues $\{\mu_{i}\}_{i}$ (with multiplicities), and let
\begin{equation}\label{eq:spectral-prob}
p(\mu_{i}) = \frac{\mu_{i}}{\sum_{i} \mu_{i}}
\end{equation}

\noindent be the ``eigenvalue probabilities," i.e. the contribution of $\mu_{i}$ in the Laplacian spectrum (Proposition \ref{prop:eigenvalues} guarantees that $\mu_{i} \ge 0$, $\forall i$). Assuming the conventions $0/0=0$ and $0 \log_{2} 0 = 0$, we define the \textit{spectral entropy} of $G$ as the Shannon entropy of the eigenvalue probabilities, i.e.
\begin{equation}\label{eq:spectral-graph-entropy}
S(G) = - \sum_{i} p(\mu_{i}) \log_{2} p(\mu_{i}).
\end{equation}

\smallskip

Different definitions of spectral entropy associated with digraphs were presented by Sun et al. \citep{Sun} and Ye et al. \citep{Ye}.


\section{Graph Similarity}
\label{sec:graph-similarity}

In this section, we discuss methods and algorithms that are used to quantify how similar (or dissimilar) two graphs are, i.e. quantitative approaches to the \textit{graph similarity comparison problem}. In the previous section, we presented several measures capable of characterizing topological aspects of graphs; now, we discuss how to use them to compare graphs and, in addition, we present distance-based algorithms that try to quantify the differences between two graphs through a ``similarity" score. In what follows, all discussion applies to both directed and undirected cases.

The graph similarity comparison problem can be summarized in the following question: given two graphs, how similar are they? Despite the efforts employed in the development of methods to evaluate the similarity of graphs, there is still no one capable of satisfactorily answering this question \citep{Roy}. Also, the meaning of ``similarity"  may vary depending on the application. For example, in different contexts, we may be interested in answering one of the following questions: Are the graphs copies of each other? What changes should we apply in the graphs to transform one into the other? If we allow the graph to change over time, can we assess whether two graphs were generated from a common ancestor graph (by the successive application of evolutionary rules)? To tackle these and other questions, several methods were proposed \citep{Borgwardt, Wills, Zager}.

Mheich et al. \citep{Mheich} described two major classes of methods of graph comparison, namely:

\begin{itemize}
\item \textbf{Statistical comparison methods:} These methods consist of applying (local and global) measures of topological characterization (e.g., measures of centrality, segregation, integration, spectral measures, etc.)  in different groups of graphs in order to compare them by using statistical tests. Furthermore, for real-world networks, this comparison can be done in two ways: either by comparing them with equivalent random networks (null models) or by comparing them with other correlated groups of real-world networks. Other methods include graph correlation \citep{Fujita2016} and analysis of graph variability \citep{Fujita2017}.

\item \textbf{Distance-based comparison algorithms:} These algorithms produce ``similarity" scores as a result of comparing two graphs. Typically, these scores are normalized (i.e., they output $1$ if the two graphs are totally different from each other and output $0$ if they are totally similar). Among these algorithms are graph/subgraph isomorphisms, edit distances, graph kernels,  and structure distances (distances based on the presence/absence of edges, cliques, quasi-cliques, or other subgraphs).
\end{itemize}

Moreover, we note that when a statistical comparison analysis is performed using local/global measures, we may use the nomenclatures \textit{node-wise analysis}/\textit{global-level analysis}, respectively.

In what follows, we present some distance-based comparison algorithms that are of interest in this text.

\subsection{Distance-Based Comparison Algorithms}
\label{sec:Distance-Based-Comparison}

\bigskip
\noindent \textbf{Graph Edit Distance:}
The graph edit distance (GED) is one of the most common approaches in determining the similarity between two graphs. As observed in \citep{Gao}, a graph can be transformed into another by a finite sequence of graph editing operations (node/edge insertion, node/edge deletion), and the GED is defined as the minimum cost of these editing operations, for some suitable cost function. Formally, we can express the GED between two graphs $G_{1}$ and $G_{2}$ as 
\begin{equation}\label{eq:GED}
d_{GED}(G_{1}, G_{2}) = \min \sum_{k=1}^{N_{op}} c(e_{k}),
\end{equation}

\noindent where $c(e_{k})$ is the cost of the $k$-th editing operation $e_{k}$, and $N_{op}$ is the total number of editing operations.

\bigskip
\noindent \textbf{Graph Kernels:}
The fundamental idea behind graph kernels \citep{Borgwardt} is to build feature vectors from graphs by mapping these features into a Hilbert space\footnote{Hilbert space is a vector space provided with an inner product such that it is complete with respect to the induced norm.} $\mathcal{H}$ (feature space), and then define the kernel function as the inner product of this space. 

Let $G_{1}$ and $G_{2}$ be two graphs and let $\phi(G_{1}), \phi(G_{2}) \in \mathcal{H}$ be their feature vectors. Let $ \langle \cdot, \cdot \rangle$ denote the inner product of $ \mathcal{H}$. The graph kernel is 
\begin{equation}\label{eq:graph-kernel}
k(G_{1}, G_{2}) = \langle \phi(G_{1}), \phi(G_{2}) \rangle.
\end{equation}

The value $k(G_{1}, G_{2})$ is the similarity score used to compare the graphs. There are several ways to extract feature vectors from graphs and then define a kernel, for example, by counting the matching random walks between two graphs (random walk graph kernels) or by comparing all shortest path lengths in two graphs (shortest path kernel) \citep{Borgwardt, Vishwanathan}.

\smallskip



\bigskip
\noindent \textbf{Structure Distance:}
Similarly to the case of graph kernels, the basic idea behind defining a structure distance is to extract structure vectors (feature vectors) from graphs, based, for example, on the counting of meaningful substructures (subgraphs/ motifs), and embed them into a metric space. Since we are considering all graphs finite, here we consider the $n$-dimensional real metric space $(\mathbb{R}^{n}, d_{p})$ whose metric $d_{p}$ is induced by the $p$-norm $|| \cdot ||_{p}$, $p \in [1, \infty)$.

Let  $v^{1} = (v^{1}_{1},...,v^{1}_{n}), v^{2} = (v^{2}_{1},...,v^{2}_{n})\in \mathbb{R}^{n}$ be structure vectors associated with $G_{1}$ and $G_{2}$, respectively. We can define structure distances based on the $p$-norm as
\begin{equation}\label{eq:structure-dist}
d_{str}^{p}(G_{1}, G_{2}) = ||v^{1} - v^{2}||_{p} = \Bigg( \sum_{i=1}^{n} |v^{1}_{i} - v^{2}_{i}|^{p} \Bigg)^{1/p}.
\end{equation}

When $p=2$, $d_{str}^{p}$ represents the Euclidean distance. Also, if all entries of the structure vectors are considered non-negative, and assuming that at least one of the vectors is different from the null vector, we can normalize Equation (\ref{eq:structure-dist}) by dividing it by $||v^{1}||_{p} + ||v^{2}||_{p}$.


\section{Random Graphs}
\label{sec:random-graphs}

This section is based on the books \citep{Newman, Watts-small}, however, the fundamental concepts of random graphs presented here can be found in the classic book by Bollobás \citep{Bollobas-rand}. Moreover, an algorithmic approach to random graphs can be found in \citep{Joyner}.

\subsection{Erdős-Rényi Model}
\label{sec:er}

Roughly speaking, a random graph of order $n$ is nothing more than a set of $n$ vertices and a set of edges connecting pairs of these vertices in some random way. Almost all of random graph theory is concerned either with the analysis of the \textit{Erdős-Rényi $G(n, M)$ random models} (in homage to the contributions of these authors in the study of this model \citep{Erdos}), or with the analysis of the \textit{Erdős-Rényi $G(n, p)$ random models} (also called \textit{binomial random models} or \textit{Gilbert models}, since they were first proposed by Edgar Gilbert \citep{Gilbert}), and the relationship between these two models. Below, their formal definitions are presented.

\bigskip
\noindent \textbf{The G(n, M) model:} The graph $G(n, M)$ is generated as follows: given a set of $n$ (fixed) vertices, we choose a (fixed) number $M$ of distinct pairs of vertices uniformly at random, among the ${ n \choose 2}$ possible pairs, and we connect each pair with an edge. That is, there are ${{n \choose 2} \choose M}$ ways to place the $M$ edges, and we simply choose any of them with equal probability.

Strictly speaking, a random graph model is defined not just as a single randomly generated graph but as a set of graphs, that is, a probability distribution over the possible graphs. Therefore, the model $G(n, M)$ is correctly defined as a probability distribution $P(G)$ over all simple graphs $G$, with $n$ vertices and $M$ edges, such that
\begin{equation}
P(G) = \frac{1}{{{n \choose 2} \choose M}}.
\end{equation}

\begin{observation}
A directed version of the $G(n, M)$ model can be generated through a straightforward adaptation of the previously described method, i.e., given $n$ fixed vertices and a positive integer $N$, we add an arc from a vertex to another vertex randomly until we get $M$ arcs.
\end{observation}

\bigskip
\noindent \textbf{The G(n, p) model:} $G(n, p)$ is the graph with a set of $n$ (fixed) vertices, in which all possible ${n \choose 2}$ edges exist with probability $0 \le p \le 1$. That is, the (fixed) probability of placing an edge between each distinct pair of vertices is $p$. In this graph, the number $M$ of edges is not fixed. Formally, the model $G(n, p)$ is the set of simple graphs with $n$ vertices in which each graph $G$ appears with probability
\begin{equation}
P(G) = p^M (1-p)^{{n \choose 2}-M},
\end{equation}
\noindent where $M$ is the number of edges in the graph.

\begin{observation}
Similarly to the $G(n, M)$ model, a directed version of the $G(n, p)$ model can be obtained through a straightforward adaptation of the previous method: given $n$ fixed vertices, we add an arc from a vertex to another randomly with a given probability $p$.
\end{observation}

\subsection{$k$-Regular Model}
\label{sec:kr}

There are several models for generating $k$-regular random graphs uniformly at random, that is, random graphs whose nodes have a given fixed degree $k$. Perhaps the most common approach is the \textit{pairing model}. In what follows, we present the algorithm behind this method \citep{Steger}.

Given a set of $n$ nodes, create $n$ sets with $kn$ elements. Choose a pair of elements at random from these sets. Then create an edge $(i,j)$ in the graph if there is a pair composed of elements of the $i$'th and $j$'th sets. We disregard equal pairs and pairs of type the $(i,i)$ (self-loops). The resulting random graph is $k$-regular.

Moreover, we can generate a $k$-regular random digraph by considering $k = \deg^{in}(v)$ $+ \deg^{out}(v)$, for all nodes $v$ in the digraph, and considering $(i,j)$ as arcs in the paring model. Also, we can use the Havel-Hakimi algorithm \citep{Kleitman} to generate a digraph from a given sequence of in-degree and a given sequence of out-degree, and by choosing a fixed in- and out-degree $k$, the resulting digraph is $k$-regular (in this case, we have $k = \deg^{in}(v) = \deg^{out}(v)$, for all nodes $v$).

\subsection{Watts-Strogatz Model}
\label{sec:ws}


In this part, we present the \textit{Watts-Strogatz model} (WS model), sometimes called \textit{small-world network model}, proposed by Watts and Strogatz \citep{Watts}. Although it is not the only small-world network model, it is the most frequently used one. The reference for this part is \citep{Newman}.

The WS model is a model originally created to illustrate how two features of social networks (\textit{high clustering coefficient} (high $C$ ) and \textit{low path lengths} (low $L$)) can coexist in the same network. In retrospect, the main contribution of this model is being able to show why the \textit{small-world effect} (the existence of short paths between most vertices) is prevalent in networks of all types, especially in real-world networks.

We can define this model as follows. Let's start with a regular network (or regular grid) of some type. For example, let us arrange $n$ vertices in a circle and connect each of them to the nearest $k$ vertices ($k$ an even number). Then, we randomize this network, traversing each of the edges around the circle in succession, and with probability $p$, we remove an edge and replace it with another that connects two randomly chosen vertices in a uniform way (i.e., we reconnect some vertices). The probability $p$ is therefore said to be the \textit{reconnection probability} or \textit{rewiring probability}. For directed graphs, the process of rewiring is analogous, even though there are other processes to produce weighted and directed small-world networks \citep{Ramezanpour, Xu}.

In a small-world network model, the rewiring probability $p$ takes intermediate values between the value $p$ of the initial regular circular network and the value $p$ of a random network. That is, when $p = 0$, no edges are reconnected, keeping the configuration of the initial regular network (high $C$ and high $L$), when $p=1$, all edges reconnect, producing a random network (low $C$ and low $L$), and when $p$ takes intermediate values, e.g. $p=0.3$, only some edges are reconnected, and the generated network (small-world network) takes on properties of both networks (high $C$ and low $L$), as depicted in Figure \ref{fig:random-nets}.

\begin{figure}[h!]
\centering
\begin{subfigure}{.3\textwidth}
  \centering
  \includegraphics[scale=1.1]{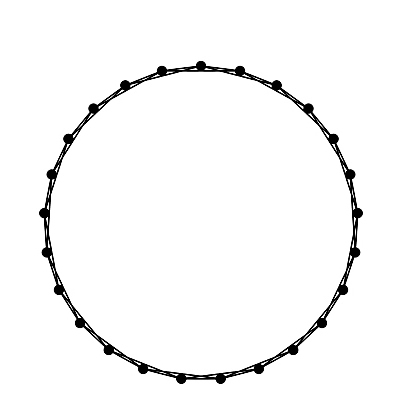}
  \caption{Regular network.}
  \label{fig:sub1}
\end{subfigure}%
\begin{subfigure}{.3\textwidth}
  \centering
  \includegraphics[scale=1.1]{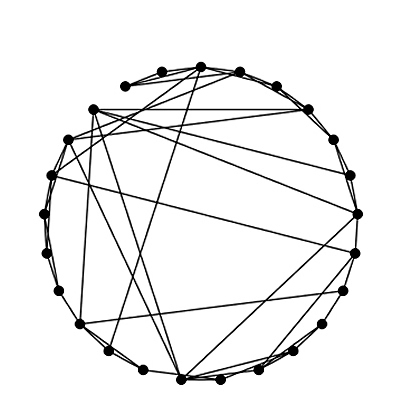}
  \caption{Small-world network.}
  \label{fig:sub2}
\end{subfigure}
\begin{subfigure}{.3\textwidth}
  \centering
  \includegraphics[scale=1.1]{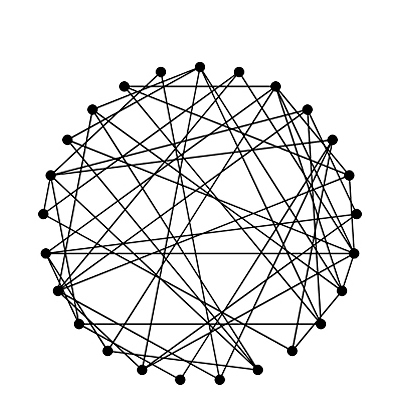}
  \caption{Random network.}
  \label{fig:sub3}
\end{subfigure}
\caption{Networks obtained through the NetLogo software \citep{Wilensky}, using the Watts and Strogatz model. All networks have $25$ nodes. (a) Regular network: high $C$ and high $L$ ($p = 0, C = 0.5, L = 3.5$). (b) Small-world network: high $C$ and low $L$ ($p = 0.3, C = 0.357, L = 2.51$). (c) Random network: low $C$ and low $L$ ($p = 1, C = 0.086, L = 2.317$).}
\label{fig:random-nets}
\end{figure}

A quantitative way of defining whether a given network $G$ has a “small-world” structure is by computing its \textit{small-worldness} \citep{Humphries}, which is defined as the ratio
\begin{equation}\label{eq:small-world-coef}
   S_{SW}(G) = \frac{C/C_{rand}}{L/L_{rand}},
\end{equation}

\noindent where $C$ and $L$ are defined as before, and $C_{rand}$ and $L_{rand}$ are these same quantities but computed for an equivalent Erdős-Rényi random network. We say that $G$ is a small-world network if $S_{SW}(G) >> 1$. For directed graphs, we use the directed versions $\vec{C}$ and $\vec{L}$.






\subsection{Barabási-Albert Model}
\label{sec:ba}

The last random graph model that we are going to discuss is the so-called \textit{Barabási-Albert model} (BA model) \citep{Barabasi-1999}. This model is based on two central rules: \textit{growth} and \textit{preferential attachment}\footnote{Preferential attachment is informally known as the “rich-get-richer" effect.}. 

Briefly, a network is generated through growth and preferential attachment as follows: given a set of $n_{0}$ initial nodes, a new node is added, at each time step, to the network and linked to $r$ other existing nodes (where $r$ is a parameter of the model) such that the connection probability is proportional to the degree of each node. Formally, the probability of a new node connecting to a node $i$ is $p(i) = \deg(i)/\sum_{j} \deg(j)$, where the sum is performed over all other nodes $j$ existing in the network \citep{Joyner}. 

The previous two rules produce networks whose degree distributions follow a power-law, i.e. let $p(k)$ be the probability of a vertex at chosen at random having degree $k$, then this probability is proportional to the power $k^{-\eta}$, $\eta \in \mathbb{R}_{+}$, i.e. 
\begin{equation}\label{eq:power-law}
p(k) \sim k^{-\eta}.
\end{equation}

Probability distributions of the form (\ref{eq:power-law}) are called \textit{power-laws} \citep{Newman3}. These networks are called ``scale-free networks," since their degree distributions have the same shape at different scales.


Furthermore, Bollobás et al. \citep{Bollobas2003} proposed a model which generates scale-free digraphs through the preferential attachment algorithm depending on both in-degrees and out-degrees of the nodes. The model is based on the parameters $\delta_{in}$, $\delta_{out}, \alpha, \beta, \gamma \in \mathbb{R}_{\ge 0}$, where  

\begin{itemize}
\item $\delta_{in}$ and $\delta_{out}$ are biases for choosing nodes based on the in-degree/out-degree distribution, respectively;

\item $\alpha$ is the probability of add a new node $v$ and an arc $(v,u)$ to an existing node $u$, where $u$ is chosen with a probability proportional to $\deg^{in}(u) + \delta_{in}$;

\item $\beta$ is the probability of adding a new arc from an existing node $v$ to another existing node $u$, where $v$ is chosen with a probability proportional to $\deg^{out}(v) + \delta_{out}$, and $u$ is chosen with a probability proportional to $\deg^{in}(u) + \delta_{in}$;

\item $\gamma$ is the probability of add a new node $v$ and an arc $(u,v)$ to an existing node $u$, where $u$ is chosen with a probability proportional to $\deg^{out}(u) + \delta_{out}$.
\end{itemize} 

These parameters must be chosen so that the equality $\alpha + \beta + \gamma = 1$ is satisfied. The scale-free digraph grows as we perform, at each time step, one of the actions associated with one of the previous probabilities. 

\chapter{Digraph-Based Complexes and Directed Higher-Order Connectivity}
\label{chap:chap4}

\epigraph{(...) the heart of the scientific method, and of all rational study of any human activity, lies in the process of identifying sets and of understanding the structural properties of relations between sets.}{--- R. H. Atkin \citep{Atkin1974a}}

\bigskip

In Chapter \ref{chap:chap2}, we discussed the foundations of graph theory and the elements of quantitative graph theory. In this chapter, we will extend our discussion to the field of \textit{simplicial complexes}, which can be seen as a generalization of graphs.

Graphs can be considered hierarchically structured: ``higher-level" (larger) subgraphs contain other ``lower-level" (smaller) subgraphs. This is the idea behind the \textit{clique organization} or \textit{clique topology} of a graph, since larger cliques contain smaller cliques. This hierarchical organization allows us to build simplicial complexes from the cliques of the graphs, the so-called \textit{clique complexes}. When we are dealing with digraphs, this same type of organization is present, but in this scenario, we may take the directionality of the edges into account and construct special types of complexes out of the digraphs, the so-called \textit{directed clique complexes} (or \textit{directed flag complexes}). 

To be more specific, in this chapter, we will deal mainly with directed flag complexes constructed from digraphs without double edges, which constitute special types of simplicial complexes. Furthermore, we will briefly discuss more general types of complexes obtained from digraphs, the so-called \textit{path complexes}. We use the umbrella term \textit{digraph-based complexes} to represent these types of complexes, that is, complexes constructed from digraphs. Here we will focus not only on the development of the theory of directed flag complexes, but also on the development of the theory of homology and persistent homology for these complexes, and mainly, on the development of a \textit{directed Q-Analysis} as a directed analogue of classical Q-Analysis that takes the directionality of the higher-order connectivity between directed simplices into account, introducing, therefore, the concept of \textit{directed higher-order adjacencies} (lower and upper directed adjacencies). In addition, we will discuss all these constructions for the weighted case, i.e. when these complexes are obtained from weighted digraphs.

\section{Directed Flag Complexes of Digraphs}
\label{sec:DFC}

In this section, we present the classical theory of simplicial complexes, with considerations about directed simplicial complexes and semi-simplicial sets; subsequently, we introduce the directed flag complexes and all the mathematical formalism behind these structures, including the weighted case, i.e., the case when they are obtained from weighted digraphs. Furthermore, we present several concepts from algebraic topology and computational algebraic topology for simplicial complexes constructed out of digraphs, such as simplicial homology, persistent homology, and the combinatorial Hodge Laplacians.

The basic bibliography for this section is \citep{Edelsbrunner, Hatcher, Kozlov, Munkres}. Also, all graphs/digraphs are considered to be non-empty, non-null, finite, simple, and labeled, and all sets are considered to be finite unless said otherwise. The notation $I_{n}^{*} = \{0,1,...,n \}$ is adopted for the index set.

\subsection{Simplicial Complexes and Semi-Simplicial Sets}
\label{sec:sc-semi-sc}

In this subsection, we present the formal definitions of abstract simplicial complexes and abstract directed simplicial complexes, the latter being the fundamental abstract structures behind directed flag complexes. Furthermore, we present the concept of semi-simplicial set, which can be seen as a generalization of the concept of an abstract directed simplicial complex.

\subsubsection{Abstract Simplicial Complexes}

\begin{definition}\label{def:ASC}
An \textit{abstract simplicial complex} (ASC) is a finite collection $\mathcal{X}$ of finite sets, such that if $\sigma \in \mathcal{X}$, then for all $\tau \subseteq \sigma$ we have $\tau \in \mathcal{X}$ (closed under subset inclusion).
\end{definition}

Let $\mathcal{X}$ be an ASC. Each set $\sigma \in \mathcal{X}$ is called a \textit{simplex} (or \textit{abstract simplex}), or an $n$-\textit{simplex}, if $|\sigma| = n +1$ is it cardinality; in this case we define the \textit{dimension} of $\sigma$ as $\dim \sigma = n$. The dimension of $\mathcal{X}$ is the maximum dimension of $\sigma$, $\forall \sigma \in \mathcal{X}$, i.e. 
$$
\dim \mathcal{X} = \max_{\sigma \in \mathcal{X}} ( \dim \sigma ).$$

Any element $v_{i}$ of an $n$-simplex $\sigma = \{ v_{0},..., v_{n}\}$ is called a \textit{vertex of $\sigma$}.  A simplex is uniquely determined by its vertices. Sometimes we may denote an $n$-simplex by $\sigma^{(n)}$, where the superscript $n$ denotes its dimension. The empty set is considered to be a subset of every simplex, therefore $\emptyset \in \mathcal{X}$. Sometimes the empty set is represented as a simplex with no vertices, i.e. a $(-1)$-dimensional simplex $\sigma^{(-1)} = \emptyset$ (also called \textit{null simplex}). A \textit{$k$-face} of an $n$-simplex $\sigma$, $0 \le k \le n$, is a $k$-simplex $\tau$ such that $\tau \subseteq \sigma$ (we'll use the notation $\tau \subseteq \sigma$ to denote that $\tau$ is a face of $\sigma$); in contrast, $\sigma$ is said to be a \textit{coface} of $\tau$. If $\tau \subseteq \sigma$ is a face of $\sigma$ such that $\tau \neq \sigma$, then $\tau$ is said to be a \textit{proper face}, and in this case we denote $\tau \subset \sigma$. The $(n-1)$-faces of an $n$-simplex are said to be its \textit{boundary}. A simplex is said to be \textit{maximal} in $\mathcal{X}$ if it is not a face of any other simplex in $\mathcal{X}$. The set of all $k$-simplices in $\mathcal{X}$ is denoted by $X_{k}$; in particular, $X_{0}$ is the set of all unit sets of $\mathcal{X}$. The \textit{vertex set} of $\mathcal{X}$  is the set of all vertices of its simplices, i.e.
$$
V_{\mathcal{X}} = \bigcup_{\sigma \in \mathcal{X}} \sigma.
$$


It's common to say that $\mathcal{X}$ is an ASC on the vertex set $V_{\mathcal{X}}$. A \textit{subcomplex} of $\mathcal{X}$ is a subcollection $\mathcal{S}_{\mathcal{X}} \subseteq \mathcal{X}$ such that $\mathcal{S}_{\mathcal{X}}$ is closed under subset inclusion. If the simplices of a subcomplex $\mathcal{S}_{\mathcal{X}} \subseteq \mathcal{X}$ have dimensions at most $k$, $0 \le k \le \dim \mathcal{X}$, then $\mathcal{S}_{\mathcal{X}}$ is called the \textit{$k$-skeleton} of $\mathcal{X}$, and denoted by $\mathcal{S}_{\mathcal{X}}^{(k)}$, i.e. the $k$-skeleton of $\mathcal{X}$ is equal to the union 
$$
\mathcal{S}_{\mathcal{X}}^{(k)} = \bigcup_{i=0}^{k} X_{i}.
$$

\begin{example}
A simple undirected graph $G = (V, E)$ is equivalent to a $1$-dimensional abstract simplicial complex on $V$, where the $0$-simplices are the vertices of $G$ and the $1$-simplices are the edges of $G$. We also can say that an undirected graph is the $1$-skeleton of an abstract simplicial complex.
\end{example}

Moreover, from the previous definitions, $\mathcal{X} \subseteq \mathcal{P}(V_{\mathcal{X}})$, where $\mathcal{P}(V_{\mathcal{X}})$ is the power-set\footnote{The \textit{power-set} $\mathcal{P}(S)$ of a given set $S$, also denoted by $2^{S}$, is the set formed by all possible subsets of $S$, including the empty set.} of $V_{\mathcal{X}}$ and if $\sigma, \sigma' \in \mathcal{X}$, the intersection $\sigma \cap \sigma'$ is either a face of both simplices or the empty set.

\begin{remark}\label{rem:topological-space}
The power-set of the vertex set of every ASC defines a topology (called \textit{discrete topology}) on the vertex set. Thus, every ASC defines a topological space.
\end{remark}

A less restrictive definition is the definition of simplicial family: a \textit{simplicial family} is a set formed by arbitrary simplices (the faces of the simplices do not need to belong to the set). Notice that every simplicial family determines an abstract simplicial complex.

As a last remark, we point out that some authors simply use the expression \textit{simplicial complex} to refer to an abstract simplicial complex, and here we may switch between the two nomenclatures deliberately.

\subsubsection{Geometric Simplicial Complexes}

One can associate with any abstract simplex $\sigma$ a \textit{geometric simplex}, which is the convex hull\footnote{The convex hull of a subset $S \subseteq \mathbb{R}^{n}$ is the intersection of all convex sets containing $S$.} of the vertices of $\sigma$ in the Euclidean space. More formally:

\begin{definition}\label{def-geo-simp}
Let  $v_{0}, ..., v_{k}$ be $k+1$ distinct points in $\mathbb{R}^{n}$. We say that $\{ v_{0}, ..., v_{k}\}$ is an \textit{affinely independent set} if for real numbers $a_{0},...,a_{k}$ such that $\sum_{i=0}^{k} a_{i}v_{i} = 0$ and $\sum_{i=0}^{k} a_{i} = 0$, then $a_{0} =... = a_{k} = 0$. If $\{ v_{0}, ..., v_{k}\}$ is an affinely independent set of points in $\mathbb{R}^{n}$, the \textit{geometric simplex} spanned by it is the set

\begin{equation}
\sigma = \Bigg\{ \sum_{i=0}^{k} \lambda_{i} v_{i} : \lambda_{i} \ge 0 \mbox{ and }  \sum_{i=0}^{k} \lambda_{i} = 1 \Bigg\}.
\end{equation}

\end{definition}

Analogously to Definition \ref{def:ASC}, a \textit{geometric simplicial complex} (or \textit{Euclidean simplicial complex}) $\mathcal{K}$ is defined as a finite collection of geometric simplices such that if $\sigma \in \mathcal{K}$ and $\tau \subseteq \sigma$, then $\tau \in \mathcal{K}$, and if $\sigma, \sigma' \in \mathcal{K}$, then either $\sigma \cap \sigma' \in \mathcal{K}$ or $\sigma \cap \sigma' = \emptyset$.

The \textit{boundary} of a geometric $n$-simplex is constituted of all its $(n-1)$-faces, and its \textit{interior} is constituted of all points which do not belong to its boundary. Every other definition made previously for abstract simplices is defined analogously to geometric simplices.

\begin{example}
Figure \ref{geometric_sc} presents examples of geometric $k$-simplices, for $k=0,1,2,3$. From left to right, we have: a 0-simplex (vertex), a 1-simplex (edge), a 2-simplex (triangle with interior), and a 3-simplex (tetrahedron with interior). Note that the edge has two vertices as its faces, the triangle has three edges as its faces, and the tetrahedron has four triangles as its faces.

\begin{figure}[h!]
\centering
  \includegraphics[scale=1.35]{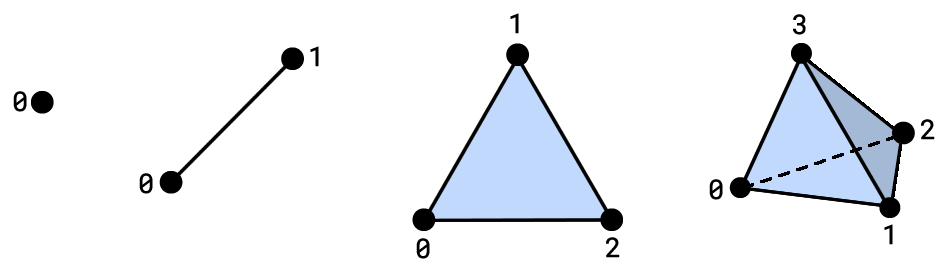}
  \caption{Examples of geometric simplices.}
  \label{geometric_sc}
\end{figure}
\end{example}

\begin{definition}
The \textit{vertex scheme} of a geometric simplicial complex is the ASC built out of the sets of vertices of its geometric simplices. On the other hand, the \textit{geometric realization} of an ASC $\Delta$ is the geometric simplicial complex whose vertex scheme is isomorphic to $\Delta$.
\end{definition}

Every finite ASC has a geometric realization on an Euclidean space, as stated by the geometric realization theorem: ``every $d$-dimensional abstract simplicial complex has a geometric realization in $\mathbb{R}^{2d+1}$" (see \citep{Edelsbrunner}, p. 53, for a proof).


\smallskip

\subsubsection{Abstract Directed Simplicial Complexes}

We can define an \textit{orientation} on an $n$-simplex by defining a total ordering on its vertices. For instance, given an $n$-simplex $\{ v_{0},..., v_{n} \}$, an orientation is obtained by ordering $v_{i} < v_{j}$, whenever $i < j$.

 In the literature, it is usual to denote an \textit{ordered} (or \textit{oriented}) \textit{$n$-simplex} by $[v_{0},..., v_{n}]$, which represents a totally ordered set. Just as we defined for abstract simplices, we may use the notation $\sigma^{(n)} = [v_{0},..., v_{n}]$, where the superscript represents its dimension. The edges of an ordered simplex inherit its ordering, i.e. $[v_{i}, v_{j}]$, $v_{i} < v_{j}$, if $i < j$. Here, a caveat is necessary because we are committing an abuse of notation: $i$ and $j$ are actually indicating the position of the vertices in the simplex, even though we are using the same indices on the vertex labels. With that said, we'll adopt the notation $v_{i} = i$ for the labels, and write $[0...n] = [0,...,n] = [v_{0},..., v_{n}]$, but keep in mind that the labels do not necessarily represent the position of the vertex in the ordered set.

Two orientations are equivalent if they differ by an even permutation, e.g., the ordered $2$-simplex $[0,1,2]$ in Figure \ref{fig:simp-ori-a} has the same clockwise orientation as those produced by the permutations $[1,2,0]$ and $[2,0,1]$,  while the permutation $[1,0,2]$ produces a counterclockwise orientation, as depicted in Figure \ref{fig:simp-ori-b}. It's important to note that the set $\{ 0,1,2 \}$ corresponds to a single simplex, regardless of its orientation, and the ordered simplices $[1,2,0]$ and $[2,0,1]$ are not equal as ordered sets, despite the same orientation.

\begin{figure}[h!]
\centering
\begin{subfigure}{.39\textwidth}
  \centering
  \includegraphics[scale=1.35]{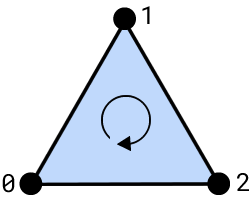}
  \caption{Clockwise orientation.}
  \label{fig:simp-ori-a}
\end{subfigure}%
\begin{subfigure}{.39\textwidth}
  \centering
  \includegraphics[scale=1.35]{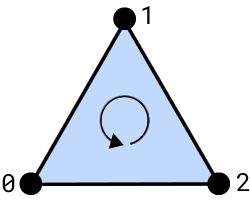}
  \caption{Counterclockwise orientation.}
  \label{fig:simp-ori-b}
\end{subfigure}
\caption{Different orientations of a $2$-simplex.}
\label{fig:orientations}
\end{figure}

In summary, an ordered $n$-simplex is merely an $n$-simplex with a total ordering in its vertices. Therefore, we can formulate an analogous definition of the definition of ASC where our simplices are ordered simplices, i.e. the definition of \textit{abstract ordered simplicial complex} or \textit{abstract directed simplicial complex} \citep{Reimann}.

\begin{definition} An \textit{abstract directed simplicial complex} (ADSC) is a finite collection $\mathcal{X}$ of totally ordered finite sets, such that if $\sigma \in \mathcal{X}$, then for all $\tau \subseteq \sigma$ (with an ordering inherited from $\sigma$) we have $\tau \in \mathcal{X}$ (closed under totally ordered subset inclusion).
\end{definition}

Henceforth, we'll adopt the nomenclature \textit{directed simplices} for the elements of an ADSC instead of ordered simplices, and we'll keep using the same notations introduced for ordered simplices. Also, all definitions and nomenclatures made previously for abstract simplicial complexes apply to abstract directed simplicial complexes as well.

It's important to point out that, as mentioned in Subsection \ref{sec:relations-orders}, two totally ordered sets with the same elements are identical if and only if they have the same ordering; therefore,  two different directed simplices may have the same set of vertices. This fact implies that, despite the name, ADSCs are not examples of ASCs in general, except in the cases where the ADSCs are built out of ASCs, since, in these cases, every ordered set in the ADSC corresponds to a single simplex.

A convenient and practical way to access the faces of a directed simplex is through \textit{face maps}, which are defined as follows.

\begin{definition}\label{facemap}
Given an abstract directed simplicial complex $\mathcal{X} = \bigcup^{\dim \mathcal{X}}_{k=0} X_{k}$, where $X_{k}$ is the set of all directed $k$-simplices, we define maps $\hat{d}_{i}: X_{k+1} \rightarrow X_{k}$, for each $0 \le k \le \dim \mathcal{X}-1$ and each $0 \le i \le k+1$, such that $\hat{d}_{i}([v_{0},...,v_{k}]) = [v_{0},...,\hat{v_{i}},...,v_{k}]$, where $\hat{v_{i}}$ denotes the deletion of the $i$-th vertex $v_{i}$ from $[v_{0},...,v_{k}]$. The map $\hat{d}_{i}$ is called \textit{i-th face map}.
\end{definition}

One can easily see that indeed $\hat{d}_{i}$ takes a simplex on its $i$-th face. An important property of the face maps is that if $i < j$, then

\begin{equation}\label{eqfacemap}
\hat{d}_{i}\circ \hat{d}_{j} = \hat{d}_{j-1} \circ \hat{d}_{i}.
\end{equation}

In fact, $\hat{d}_{i}(\hat{d}_{j}([v_{0},...,v_{k}])) = [v_{0},...,\hat{v_{i}},...,\hat{v_{j}},..., v_{k}] = \hat{d}_{j-1}(\hat{d}_{i}([v_{0},...,v_{k}]))$.

\smallskip

\subsubsection{Semi-Simplicial Sets}

A generalization of the concept of abstract simplicial complexes is the concept of \textit{semi-simplicial complexes}, which was originally proposed by S. Eilenberg and J. A. Zilber  \citep{Eilenberg}. The contemporary theory adopted the nomenclature \textit{semi-simplicial sets}, and the following definition is based on Friedman's work \citep{Friedman}.

\begin{definition}
A \textit{semi-simplicial set} (or \textit{Delta set}, or \textit{$\Delta$-set}) is a collection of sets $\mathcal{X} = \{X_{0}, X_{1},X_{2},... \}$  together with maps $\hat{d}_{i}: X_{n+1} \rightarrow X_{n}$, for each $0 \le n$ and each $0 \le i \le n+1$, such that $\hat{d}_{i}\circ \hat{d}_{j} = \hat{d}_{j-1} \circ \hat{d}_{i}$, whenever $i< j$.
\end{definition}

\smallskip

Every abstract simplicial complex forms a semi-simplicial set. Indeed, let $X_{k}$ be the set of all $k$-dimensional simplices and let $d_{i}: X_{k+1} \rightarrow X_{k}$ be the face maps as defined in Definition \ref{facemap}. The condition $\hat{d}_{i}\circ \hat{d}_{j} = \hat{d}_{j-1} \circ \hat{d}_{i}$, whenever $i<j$, is satisfied as shown in Equation (\ref{eqfacemap}). Nonetheless, there exist semi-simplicial sets that cannot be built out of abstract simplicial complexes, as we can see in the following example.

\smallskip

\begin{example}
Consider the collection of sets $\mathcal{X} = \{ X_{0}, X_{1} \}$, where $X_{0} = \{ [ 0 ], [ 1 ] \}$  and $X_{1} = \{ [ 0, 1 ], [1, 0] \}$. The face maps are $\hat{d}_{0}([0,1]) = \hat{d}_{1}([1,0]) = [1]$, $\hat{d}_{1}([0,1]) =  \hat{d}_{0}([1,0]) = [0]$. Thus, $\mathcal{X}$ together with the maps $\hat{d}_{i}$ form a semi-simplicial set. However, it does not form an ASC, since the intersection $[0,1] \cap [1,0] = \{0,1\}$ is not a common face of both directed $1$-simplices.
\end{example}

Similarly to an abstract simplicial complex, a semi-simplicial set has a geometric realization as described in \citep{Milnor}.

\subsection{Directed Flag Complexes}
\label{sec:dfc}


From this point forward, we will be mainly interested in abstract directed simplicial complexes built from directed graphs; nevertheless, let's start by making some observations about ASCs built out of undirected graphs.

There are several ways to build an ASC from a given graph $G = (V, E)$. For instance, we can construct an ASC by considering the neighborhood of each vertex and including all its faces (\textit{neighborhood complex}) \citep{Maletic}; by defining a collection of subgraphs of the power-set of $V$, $\mathcal{S}_{G} \subseteq \mathcal{P}(V)$, such that if $H \in \mathcal{S}_{G}$ and $e \in H$ is an edge, then $H - e \in \mathcal{S}_{G}$ (closed under deletion of edges) \citep{Jonsson}; or by considering the $n$-simplices as the $(n+1)$-cliques of the graph, forming what is known as \textit{clique complex} or \textit{flag complex} \citep{Aharoni}. 

Here we will focus on the study of flag complexes and their directed variants, the \textit{directed flag complexes} \citep{Reimann}, so let's present their formal definitions.

\begin{definition}
Given a graph $G = (V, E)$, its \textit{flag complex} (or \textit{clique complex}), denoted by $\mathrm{Fl}(G)$, is the abstract simplicial complex on the vertex set of $G$ whose $k$-simplices are  subsets of $V$ which span $(k+1)$-cliques of $G$.
\end{definition}

The previous definition only takes into account undirected graphs. Nevertheless, we will introduce a variation of this definition for directed graphs by using the definition of \textit{directed cliques}.

\begin{definition}
A \textit{directed $(k+1)$-clique} is a digraph $ G = (V, E)$, $V = \{ v_{0},..., v_{k} \}$, whose underlying undirected graph is a $(k+1)$-clique and for each $0 \le i < j \le k$, $v_{i}, v_{j} \in V$, and there is a directed edge from $v_{i}$ to $v_{j}$, i.e.  $(v_{i}, v_{j}) \in E$, for all $i < j$.
\end{definition}

Notice that, by definition, every directed $(k+1)$-clique has a source ($v_{0}$) and a sink ($v_{k}$) and is a directed acyclic graph (DAG) (c.f. Figure \ref{dir-clique}).

\begin{example}
Figure \ref{dir-clique} presents examples of directed $(k+1)$-cliques for $k=0,1,2,3,4$. From left to right, we have: a $1$-clique, a directed $2$-clique, a directed $3$-clique, a directed $4$-clique, and a directed $5$-clique.

\begin{figure}[h!]
    \centering
\includegraphics[scale=1.35]{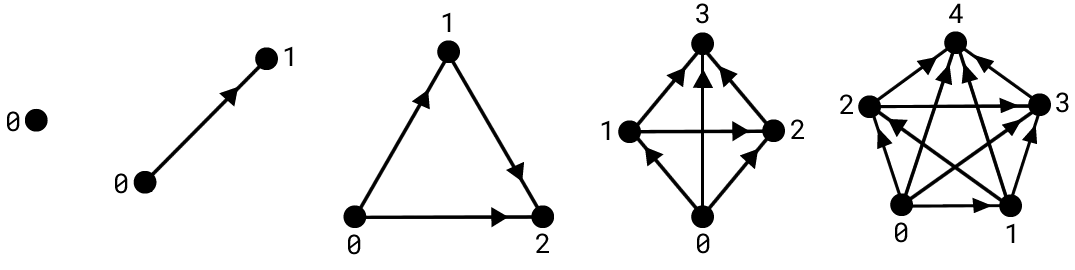}
 \caption{Examples of directed $(k+1)$-cliques, for $k=0,1,2,3,4$.}
 \label{dir-clique}
\end{figure}
\end{example}

The concept of directed flag complexes \citep{Masulli, Reimann} is a variant for digraphs of the concept of flag complexes, however, although the flag complex associated with a graph is, in fact, an example of ASC, a directed flag complex associated with a digraph is actually an example of ADSC, since the digraph may contain double edges, which implies that a subset of its vertex set may span different directed cliques; an exception occurs for digraphs \textit{without double edges}, in which cases the corresponding directed flag complexes are ASCs with a specific order in their simplices.

\begin{definition}
Given a digraph $G = (V, E)$, its \textit{directed flag complex}, denoted by $\mathrm{dFl}(G)$, is the abstract directed simplicial complex whose directed $k$-simplices span directed $(k+1)$-cliques of $G$, i.e. for every $[v_{0},...,v_{k}] \in \mathrm{dFl}(G)$, we have $v_{i} \in V$, $\forall i$, and $(v_{i}, v_{j}) \in E$,  $\forall i < j$.
\end{definition}

The directed simplices of a directed flag complex are not uniquely defined by their set of vertices, since they may differ by the ordering of their vertices. For instance, the directed $2$-simplices $[v_{0}, v_{1}, v_{2}]$ and $[v_{2}, v_{0}, v_{1}]$ have the same set of vertices, but they span different directed $3$-cliques. Again, remember that we are committing an abuse of notation and $i$ and $j$ are indicating the position of the vertices in the simplex, which are independent of the vertex labels, e.g., in the case $[v_{0}, v_{1}]$ we have $(v_{0}, v_{1}) \in E$ but in the case  $[v_{1}, v_{0}]$ we have $(v_{1}, v_{0}) \in E$.

\begin{example}
Figure \ref{fig:double-edge} shows the digraph $G = (V, E)$ which has a double edge between the vertices $0$ and $2$, i.e. $(0,2), (2,0) \in E$. Note that the set $\{0,2\}$ spans two different directed $2$-cliques, $[0,2]$ and $[2,0]$.

\begin{figure}[h!]
\centering
\begin{subfigure}{.33\textwidth}
  \centering
  \includegraphics[scale=0.8]{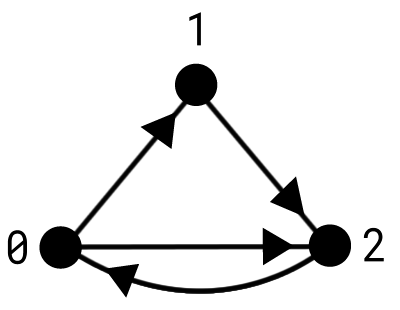}
  \caption{Digraph $G$.}
  \label{fig:double-edge}
\end{subfigure}%
\begin{subfigure}{.33\textwidth}
  \centering
  \includegraphics[scale=0.8]{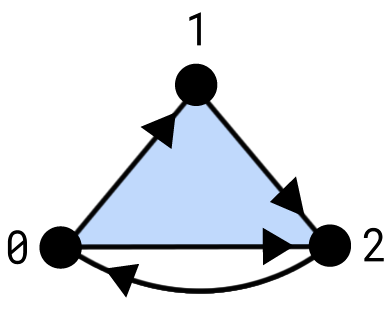}
  \caption{$\mathrm{dFl}(G)$.}
  \label{fig:flag-double}
\end{subfigure}
\caption{A digraph $G$ with a double edge and its directed flag complex $\mathrm{dFl}(G)$.}
\label{fig:digraph-dfl}
\end{figure}
\end{example}


Moreover, note that directed cycles are not considered directed cliques since they do not have a source and a sink.

We say that the flag complex associated with the underlying undirected graph of a digraph is its \textit{underlying flag complex}.


\begin{example}\label{ex:dir-flag-complex}
Figure \ref{fig:flag-complex-3} represents a simple digraph $G$, its underlying flag complex $\mathrm{Fl(G)}$, and its directed flag complex $\mathrm{dFl(G)}$.

\begin{figure}[h!]
\centering
\begin{subfigure}{.23\textwidth}
  \centering
  \includegraphics[scale=0.8]{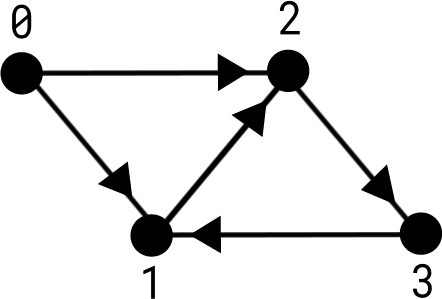}
  \caption{Digraph $G$.}
  \label{fig:dig}
\end{subfigure}%
\begin{subfigure}{.23\textwidth}
  \centering
  \includegraphics[scale=0.8]{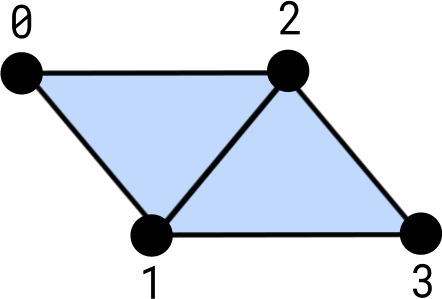}
  \caption{$\mathrm{Fl}(G)$.}
  \label{fig:flag}
\end{subfigure}
\begin{subfigure}{.23\textwidth}
  \centering
  \includegraphics[scale=0.8]{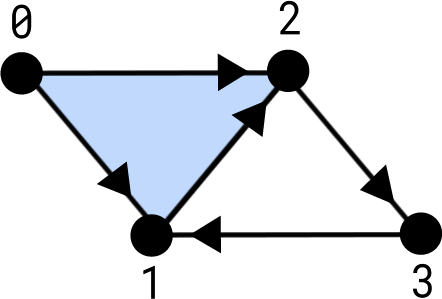}
  \caption{$\mathrm{dFl}(G)$.}
  \label{fig:dflag}
\end{subfigure}
\caption{A digraph $G$ along with its underlying flag complex $\mathrm{Fl(G)}$, and its directed flag complex $\mathrm{dFl(G)}$.}
\label{fig:flag-complex-3}
\end{figure}
\end{example}

Be aware that, in this text, the colors that fill the interior of the cliques do not represent the geometric interior in the sense of Definition \ref{def-geo-simp}, but are used as a visual representation of the corresponding higher-dimensional directed simplices.

\begin{remark}
Directed cliques are combinatorial objects; thus, for instance, given a set of three vertices, we can build six directed $3$-cliques by changing the directions of their edges in a proper way. They are all isomorphic because we can transform one into the other just by reordering their vertices and, since digraph isomorphism is an equivalence relation, the collection of all directed $(k+1)$-cliques form a class of equivalence. Although they are algebraically equivalent, the change in the direction of an edge can be combinatorially relevant, as it can affect the adjacent directed simplices, as we will see later in Subsection \ref{sec:dir-q-analysis}, when we introduce the theory of directed Q-analysis.
\end{remark}

\begin{observation}\label{flag-semi-simp}
In order to build chain complexes and homology groups (see Section \ref{sec:simplicial-homology}), we can associate a semi-simplicial set to a directed flag complex as follows. Given a directed flag complex of a digraph, $\mathrm{dFl}(G)$, let $X_{k}$ be the set of all directed $k$-simplices of $\mathrm{dFl}(G)$. The collection $\{ X_{0},...,X_{\omega(G)-1} \}$, where $\omega(G)$ is the clique number of $G$, together with face maps $d_{i}: X_{k+1} \rightarrow X_{k}$, for each $0 \le i \le k+1$ and each $0 \le k \le \omega(G)-2$, form a simplicial set.
\end{observation}

Govc et al. \citep{Govc} proposed another example of a complex built out of a digraph, called \textit{flag tournaplex}, in which the simplices are considered to be the tournaments of the digraph. A $(k+1)$-tournament is a digraph, without double edges, whose underlying undirected graph is a $(k+1)$-clique. Since the induced subdigraph of a tournament is also a tournament (called a subtournament), in the aforementioned article, an ASC-like structure called \textit{tournaplex} was introduced as a collection $\mathcal{X}$ of tournaments such that if $\sigma \in \mathcal{X}$ and $\tau \subseteq \sigma$ is a subtournament, then $\tau \in \mathcal{X}$; subsequently, a flag tournaplex was defined as the tournaplex built out of a given digraph. We can build a semi-simplicial set from a flag tournaplex in an analogous way as exposed in Observation \ref{flag-semi-simp}.


\subsection{Weighted Directed Flag Complexes}
\label{sec:weighted-dfc}


So far, we've been considering solely directed flag complexes built out of weightless digraphs; nonetheless, many real-world networks present weighted edges and/or weighted vertices, thus it would be convenient to transpose the weights associated with a digraph into its respective directed flag complex. The first formal definition of weighted simplicial complexes was proposed by Dawson \citep{Dawson}, in which the weights were considered to belong to the set of natural numbers. However, his definition was generalized for weights in a commutative ring (with unity) \citep{Ren}, thus we extended Dawson's definition to abstract directed simplicial complexes with weights on a commutative ring.

\smallskip

\begin{definition}\label{def:weig-dfc}
A \textit{weighted abstract directed simplicial complex} is a pair $(\mathcal{X}, \widetilde{\omega})$ consisting of an abstract directed simplicial complex $\mathcal{X}$ and a function $\widetilde{\omega}: \mathcal{X} \rightarrow \mathcal{R}$,  satisfying $\widetilde{\omega}(\tau) | \widetilde{\omega}(\sigma)$ whenever $\tau \subseteq \sigma$, where $\mathcal{R}$ is a commutative ring. Here, the convention $0/0 = 0$ is assumed.
\end{definition}

\smallskip

Furthermore, in previous chapters, all definitions associated with weighted digraphs only considered the weights of their edges, and nothing was said about the weights of their nodes. Nevertheless, in some contexts, a function that produces an edge-weight to node-weight transformation is needed. There are several ways to define such a function; thus, in what follows, we propose an edge-to-node weight function.

\smallskip

\begin{definition}\label{def:weights-edge-node}
Let $G^{\omega} = (V, E, \omega)$ be a weighted digraph, where $ \omega: E  \rightarrow \mathcal{R}$ is its edge-weight function, and $\mathcal{R}$ is a commutative ring. We define a \textit{node-weight function} $\tilde{\omega} : V  \rightarrow \mathcal{R}$ based on the function $\omega$ as follows:

\begin{equation}\label{eq:weight-func2}
\tilde{\omega}(i)  = \max\big(\deg_{\omega}^{-}(i), \deg_{\omega}^{+}(i) \big).
\end{equation}

\end{definition}

\smallskip

The above definition is convenient because it avoids null weights on vertices that are not isolated, and takes the \textit{in} and \textit{out} contributions of the edges into account. However, a drawback is that it might output high values for vertices with just a few strong connections and low values for vertices with a large number of weak connections.

\smallskip

\begin{definition}\label{def:prod-weight}
Let $G^{\omega}$ be a weighted digraph with directed flag complex $\mathrm{dFl}(G^{\omega})$, and let $\mathcal{R}$ be a commutative ring. We define the \textit{product-weight function} $\widetilde{\omega} : \mathrm{dFl}(G^{\omega})  \rightarrow \mathcal{R}$ as
\begin{equation}\label{eq:prod-weight}
\widetilde{\omega}(\sigma^{(n)}) =  \prod_{i = 0}^{n} \tilde{\omega}(i),
\end{equation}

\noindent where $\tilde{\omega}$ is an edge-to-node weight function, and $\tilde{\omega}(i)$ is the weight of the node $i \in \sigma^{(n)} = [0,1,...,n]$.  The pair $(\mathrm{dFl}(G^{\omega}), \widetilde{\omega})$ is called \textit{weighted directed flag complex}.
\end{definition}

\smallskip

Note that the product-weight function satisfies the conditions of the Definition \ref{def:weig-dfc}, therefore a weighted directed flag complex is indeed a weighted ADSC.

\smallskip

\begin{remark}
As defined in Chapter \ref{chap:chap2}, when we use the terminology \textit{weighted digraphs} we are referring to digraphs whose edges are weighted but the nodes aren't, thus, in order to build their weighted directed flag complexes, we must transform the edge weights into node weights, for instance, through the edge-to-node weight function proposed in Definition \ref{def:weights-edge-node}. On the other hand, if only the nodes of a digraph are weighted, then the definition of the product-weight is straightforward.
\end{remark}

\smallskip

Now, we are going to make a digression and discuss briefly how we can construct simplicial complexes when considering the vertices in a metric (or pre-metric) space (see Definition \ref{def:metric}). Different simplicial complexes can be built by establishing different criteria for the formation of simplices based on the (pre-)metric of the space, e.g., Vietoris-Rips complexes, \v{C}ech complexes, Delaunay complexes, and alpha complexes \citep{Edelsbrunner}. In the following, we discuss two types of complexes built in (pre-)metric spaces, namely: the Vietoris-Rips complexes and the Dowker complexes \citep{Chowdhury-2018a}. 

\smallskip

\begin{definition}\label{def:vietoris-rips}
Let $V$ be a set of vertices in a metric space $(X, d)$. Given a real number $\delta > 0$, the \textit{Vietoris-Rips complex} of V, denoted by $\mathfrak{R}_{\delta}(V)$, is an abstract simplicial complex formed by simplices whose diameters are at most $\delta$, i.e.
\begin{equation}
\mathfrak{R}_{\delta}(V) = \{ \sigma \subseteq V : \max_{v, w \in \sigma} d(v, w) \le \delta \}.
\end{equation}
\end{definition}

\smallskip

A drawback of the Vietoris-Rips complex is that it is insensitive to asymmetry, since the metric satisfies the symmetry condition. Also, since we are dealing with weighted directed networks, and since we can define a pre-metric from a weight function (see Definition \ref{def:weig-dist}), we would like to extend the concept of Vietoris-Rips complex to pre-metric spaces. Chowdhury and Mémoli \citep{Chowdhury-2018a} introduced the  \textit{Dowker complexes} for weighted directed networks. In what follows, we present the definition proposed in the aforementioned article for pre-metric spaces.

\smallskip


\begin{definition}\label{def:dowker-complex}
Let $V$ be a set of vertices in a pre-metric space  $(X, d)$ and let $R_{\delta}(V) \subseteq V \times V$ be the following relation:
\begin{equation}
R_{\delta}(V) = \{ (v, w) : d(v, w) \le \delta \},
\end{equation}

\noindent for any $\delta \in \mathbb{R}_{+}$. The \textit{pre-metric Dowker $\delta$-sink complex}, the \textit{pre-metric Dowker $\delta$-source complex} and the \textit{ pre-metric Dowker complex}, denoted, respectively, by $\mathfrak{D}_{si}^{pre}$, $\mathfrak{D}_{so}^{pre}$, and $\mathfrak{D}^{pre}$, are defined by
\begin{equation}
\mathfrak{D}^{pre}_{si}(V; \delta) = \{ \sigma = [v_{0},...,v_{n}] : \exists w \in V \mbox{ such that }  (v_{i}, w) \in R_{\delta}(V), \forall v_{i} \in \sigma \},
\end{equation}
\begin{equation}
\mathfrak{D}^{pre}_{so}(V; \delta) = \{ \sigma = [v_{0},...,v_{n}] : \exists w \in V \mbox{ such that }  (w, v_{i}) \in R_{\delta}(V), \forall v_{i} \in \sigma \},
\end{equation}
\begin{equation}
\mathfrak{D}^{pre}(V; \delta) = \{ \sigma = [v_{0},...,v_{n}]: \max_{v_{i}, v_{j} \in \sigma} d(v_{i}, v_{j}) \le \delta \}.
\end{equation}
\end{definition}

\smallskip

By the previous definition, if $G^{\omega} = (V, E, \omega)$ is a weighted digraph, and if we define the pre-metric $d = D^{\omega}$ such as in Definition \ref{def:weig-dist}, then $, \mathfrak{D}^{pre}_{si}(V; \delta), \mathfrak{D}_{so}^{pre}(V; \delta), \mathfrak{D}^{pre}(V; \delta) \subseteq \mathrm{dFl}(G^{\omega})$, for any real number $\delta > 0$. 

\smallskip


\begin{example}
Let $G^{\omega}=(V,E,\omega)$ be the weighted digraph illustrated in Figure \ref{fig:dowkera}, whose edge weights are not normalized. Let $d = D^{\omega}$ be the pre-metric defined by the formula (\ref{eq:weig-distance}) but replacing $\omega_{ij}^{-1}-1$ with $\omega_{ij}^{-1}$, as explained in Observation \ref{obs:weig-function}. Then, the pre-metric Dowker complex obtained by applying the pre-metric $d$ in $V$ with the condition $d(i,j) \le \delta = 0.17$, $\forall i, j \in V$, is
$$
\mathfrak{D}^{pre}(V; \delta) = \{ [0], [1], [2], [3], [4], [5], [2,3], [1,4], [4,5], [1,5], [1,4,5] \}.
$$

\begin{figure}[h!]
\centering
\begin{subfigure}{.4\textwidth}
  \centering
  \includegraphics[scale=1.5]{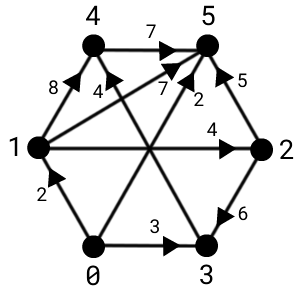}
  \caption{A weighted digraph $G^{\omega}$.}
  \label{fig:dowkera}
\end{subfigure}
\begin{subfigure}{.4\textwidth}
  \centering
  \includegraphics[scale=1.5]{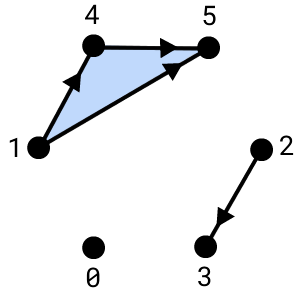}
  \caption{The pre-metric Dowker complex built from the weighted digraph $G^{\omega}$.}
  \label{fig:dowkerb}
\end{subfigure}
\caption{A weighted digraph and its pre-metric Dowker complex for $\delta = 0.17$.}
\label{fig:dowker}
\end{figure}
\end{example}

\subsection{Simplicial Homology}
\label{sec:simplicial-homology}



Historically, Henri Poincaré, in his 1895 paper, \textit{Analysis situs} \citep{Poincare1895}, where he quoted the work of his predecessors Riemann and Betti, was the first person to introduce the concept of homology classes for cell complexes \citep{Dieudonne}. Since then, homology theory has evolved to become one of the main branches of algebraic topology. The basic idea behind homology theory is to define algebraic-topological invariants of topological spaces, which are used to distinguish between these spaces. Homology classes of a topological space are equivalence classes that represent topological invariants, and these classes form the so-called \textit{homology groups} (or just \textit{homologies}, when seen as vector spaces) associated with the space. The dimensions of these groups, called \textit{Betti numbers}, roughly speaking, represent the numbers of ``$n$-dimensional holes" in the space, and are also topological invariants \citep{Hatcher}.

A particular application of homology theory is on simplicial complexes. As already discussed in Subsection \ref{sec:dfc}, we can build simplicial complexes out of the clique complexes (directed clique complexes) of graphs (digraphs without double edges), and then use the algebro-topological framework from homology theory to study the combinatorial and topological structures associated with them.

This section is mainly based on the book \citep{Hatcher}, and we consider homology theory only for simplicial complexes. Moreover, throughout this part, $\mathbb{K}$ will denote a fixed field and $\mathcal{X}$ will denote, unless said otherwise, the directed flag complex of a given digraph \textit{without double edges}.

\smallskip

\begin{definition}\label{def:chains}
For a non-negative integer $k$, we define $C_{k}(\mathcal{X}; \mathbb{K})$ as the $\mathbb{K}$-vector space formed by all formal $\mathbb{K}$-linear combinations of all $k$-dimensional elements of $\mathcal{X}$, and we call these spaces \textit{chain spaces}. The elements of $C_{k}(\mathcal{X}; \mathbb{K})$ are called \textit{$k$-chains} and are denoted by $c = \sum a_{i}\sigma^{(k)}_{i}$, where $a_{i} \in \mathbb{K}$ and $\sigma^{(k)}_{i}$ are directed $k$-simplices.
\end{definition}

\begin{definition}\label{def:boundary-map}
For any integer $n \ge 1$, we define a linear map $\partial_{n} : C_{n}(\mathcal{X}; \mathbb{K}) \rightarrow C_{n-1}(\mathcal{X}; \mathbb{K})$, called \textit{$n$-th boundary map}, by
\begin{equation}\label{eq:boundary-map}
\partial_{n}(\sigma^{(n)}) = \sum_{i=0}^{n} (-1)^{i} \hat{d}_{i}(\sigma^{(n)}) = \sum_{i=0}^{n} (-1)^{i} [v_{0},...,\hat{v}_{i},...,v_{n}],
\end{equation}

\noindent for any directed $n$-simplex $\sigma^{(n)} = [v_{0},...,v_{n}]$, where $\hat{d}_{i}$ is the $i$-th face map. For $n=0$, we define $C_{-1}(\mathcal{X}; \mathbb{K}) := \{0\}$, and $\partial_{0} = 0$ (null map).
\end{definition}

\begin{example}
Consider the directed $2$-simplex $\sigma^{(2)} = [0,1,2]$. Applying the $2$-boundary map on $\sigma^{(2)}$, we obtain the following $1$-chain: $\partial_{2}(\sigma^{(2)})$ $= [1,2] - [0,2] + [0,1]$.

\end{example}

\begin{definition}\label{def:chain-complex}
The \textit{chain complex} of $\mathcal{X}$ is defined as a sequence of chain spaces connected by boundary maps:
\[
   \hdots
   \xrightarrow[]
    {\partial_{n+1}}
    C_{n}(\mathcal{X}; \mathbb{K})
    \xrightarrow[]
    {\partial_{n}}
     C_{n-1}(\mathcal{X}; \mathbb{K})
     \xrightarrow[]
    {\partial_{n-1}}
    \hdots
     \xrightarrow[]
    {\partial_{2}}
    C_{1}(\mathcal{X}; \mathbb{K})
    \xrightarrow[]
    {\partial_{1}}
    C_{0}(\mathcal{X}; \mathbb{K})
    \xrightarrow[]
    {\partial_{0}}
    \{0\}.
\]  

As usual, we denote the chain complex by ($C_{\bullet}(\mathcal{X}; \mathbb{K}), \partial_{\bullet})$.
\end{definition}

\begin{definition}
For an integer $n \ge 0$, the elements of the image $\Ima \partial_{n+1}$ are called \textit{$n$-boundaries}, and the elements of the kernel $\ker \partial_{n}$  are called \textit{$n$-cycles}. Since $C_{-1}(\mathcal{X}; \mathbb{K}) = \{0\}$, every vertex has boundary equal to $0$, thus $\ker \partial_{0} = C_{0}(\mathcal{X}; \mathbb{K})$. The notations $B_{n} = \Ima \partial_{n+1}$ and $Z_{n} = \ker \partial_{n}$ are also commonly used.
\end{definition}


\smallskip

In addition, if an $n$-cycle is not a boundary of any $(n+1)$-chain (i.e. it does not belong to $\Ima \partial_{n+1}$), then it is called an \textit{independent $n$-cycle}.

The following proposition is referred to as the \textit{fundamental lemma of homology}, and it states that the boundary of a boundary is always zero.

\smallskip

\begin{proposition}\label{prop:fundamental-lemma-1}
\textit{For any integer $n \ge 0$, the identity $\partial_{n} \circ \partial_{n+1} = 0$ holds.}
\end{proposition}

A proof for the previous proposition can be found in \citep{Hatcher}, p. 105. Moreover, by this proposition, we have the inclusion $\Ima \partial_{n+1} \subseteq \ker \partial_{n}$, since $\partial_{n}(\partial_{n+1}c) = 0$, for any $(n+1)$-chain $c$.

\smallskip

\begin{definition}\label{def:homology}
Given an integer $n\ge 0$, the \textit{$n$-th homology} of $\mathcal{X}$ (over $\mathbb{K}$) is defined as the quotient vector space
\begin{equation}\label{eq:homology}
H_{n}(\mathcal{X}) = H_{n}(\mathcal{X}; \mathbb{K}) =  \ker \partial_{n} / \Ima \partial_{n+1}.
\end{equation}

The elements of the space $H_{n}(\mathcal{X})$ are called \textit{homology classes} (equivalence classes of independent cycles), and are denoted by $[c]$, for a given chain $c$. Two cycles are said to be \textit{homologous} if they belong to the same homology class. The dimension of $H_{n}(\mathcal{X})$, denoted by $\beta_{n} = \beta_{n}(\mathcal{X}) = \dim H_{n}(\mathcal{X})$, is called \textit{$n$-th Betti number}.
\end{definition}

\smallskip

Notice that $H_{n}(\mathcal{X}) = 0$, for all $n > \dim \mathcal{X}$. Also, the $n$-th Betti number is equal to $\beta_{n} = \dim \ker \partial_{n} - \dim \Ima \partial_{n+1}$, which is equal to the number of ``$n$-dimensional holes" in $\mathcal{X}$. In particular, as $\mathcal{X} = \mathrm{dFl}(G)$ for some given digraph $G$, using an argument analogous to that used for graphs in \citep{Ghrist2014}, the $0$-th Betti number is equal to the number of weakly connected components of $G$. 

\smallskip

\begin{example}\label{ex:holes}
Consider $\mathcal{X} = \{[0], [1], [2], [0,1], [1,2], [2,0]\}$. Applying $\partial_{1}$ on the $1$-chain $c = [0,1] + [1,2] + [2,0]$, we have $\partial_{1}(c)$ $= (1-0) + (2-1) + (0-2) = 0$. Thus, c is a $1$-cycle. However, if $c$ is not a boundary of any $2$-chain, then $c$ is an independent $1$-cycle.
\end{example}

\smallskip

Let $G$ be the digraph present in Figure \ref{fig:hole1} (a cycle of length $3$) and let $\mathcal{X}$ be its directed flag complex. The chain formed by the formal sum of the directed $1$-simplices of $\mathcal{X}$ forms an independent $1$-cycle (Example \ref{ex:holes}), and since $\beta_{1}(\mathcal{X}) = 1$, we say that there is a $1$-dimensional hole associated with $G$. On the other hand, there is no $1$-dimensional hole associated with the digraph present in Figure \ref{fig:hole2} (directed $3$-clique).

\begin{figure}[h!]
\centering
\begin{subfigure}{.35\textwidth}
  \centering
  \includegraphics[scale=0.8]{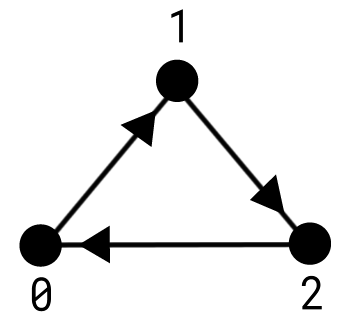}
  \caption{Cycle of length $3$.}
  \label{fig:hole1}
\end{subfigure}%
\begin{subfigure}{.35\textwidth}
  \centering
  \includegraphics[scale=0.8]{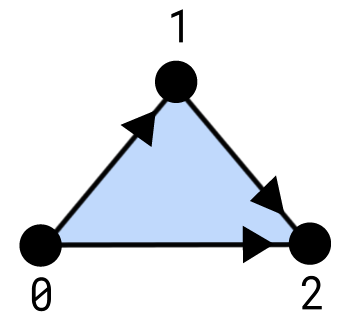}
  \caption{A directed $3$-clique.}
  \label{fig:hole2}
\end{subfigure}
\caption{There is a $1$-dimensional hole associated with a $3$-cycle, but there is no $1$-dimensional hole associated with a directed $3$-clique.}
\label{fig:holes}
\end{figure}

An important topological invariant associated with a simplicial complex is the \textit{Euler characteristic}, which can be defined in terms of Betti numbers as follows.


\smallskip

\begin{definition}\label{def:euler-char}
The \textit{Euler characteristic} of $\mathcal{X}$ is defined as the alternating sum of its Betti numbers:
\begin{equation}\label{eg:euler-char}
\chi(\mathcal{X}) = \sum_{n=0}^{\infty} (-1)^{n} \beta_{n}(\mathcal{X}) = \sum_{n=0}^{\infty} (-1)^{n} \dim H_{n}(\mathcal{X}).
\end{equation}
\end{definition}

\begin{remark}\label{rem:homology-double edges}
As commented in Subsection \ref{sec:dfc}, when a digraph has double edges, and they are taken into account, the corresponding directed flag complex is a semi-simplicial set, and then the corresponding formalism must be applied (see Observation \ref{flag-semi-simp}). In this case, the homologies might be different from those obtained when double edges are ignored.
\end{remark}

\smallskip
Furthermore, as previously mentioned, many real-world networks are weighted, so it would be interesting to take weights into account when calculating homologies. Dawson \citep{Dawson} proposed a generalization of the boundary map for weighted simplicial complexes, and later other authors proposed other generalizations based on Dawson's first proposal \citep{Ren, Wu-2021}. Below, we present the \textit{weighted $n$-th boundary map}, as proposed by Dawson, for a weighted directed flag complex $(\mathcal{X}, \widetilde{\omega})$, which is a direct modification of the formula (\ref{eq:boundary-map}):

\begin{equation}\label{eq:weig-boundary-map}
\partial_{n}^{\omega}(\sigma^{(n)})  = \sum^{n}_{i=0} (-1)^{i} \frac{\widetilde{\omega}(\sigma^{(n)})}{\widetilde{\omega}(\hat{d}_{i}(\sigma^{(n)}))} \hat{d}_{i}(\sigma^{(n)}).
\end{equation}

All previous definitions involving the boundary map are defined analogously for the weighted boundary map, and the identity $\partial_{n}^{\omega} \circ \partial_{n+1}^{\omega} = 0$ also holds for any $n\ge 0$.  In particular, if the weights of the directed simplices are all the same (but non-zero), then both boundary maps are identical. Moreover, since $(\mathcal{X}, \widetilde{\omega})$ is a weighted directed flag complex, where $\widetilde{\omega}$ is the product-weight (\ref{eq:prod-weight}), for any directed $n$-simplex $\sigma^{(n)} = [v_{0},...,v_{n}]$, we have $\widetilde{\omega}(\sigma^{(n)})/\widetilde{\omega}(\hat{d}_{i}(\sigma^{(n)})) = \tilde{\omega}(v_{i})$. Finally, we emphasize that depending on the choice of the weights, the homologies for the weighted and the unweighted cases might be different \citep{Wu-2021}.

\subsection{Persistent Homology}
\label{sec:persistent-homology}

Persistent homology is one of the main tools in the field of topological data analysis for computing topological features of a space at different scales and their persistence across these scales \citep{Edelsbrunner}. It originated with Frosini and collaborators \citep{Frosini} in the study of persistence of $0$-dimensional homology for shape recognition (which was referred to as \textit{size theory}), and it was further developed independently by Edelsbrunner et al. \citep{Edelsbrunner2000, Zomorodian}. An important property of persistent homology is that it is a robust method with respect to small perturbations in the input dataset \citep{Cohen-Steiner}.

This section is mainly based on the book \citep{Edelsbrunner}, and we consider persistent homology only for the case of simplicial complexes. Also, throughout this part, $\mathbb{K}$ will denote a fixed field and $\mathcal{X}$ will denote, unless said otherwise, the directed flag complex of a given digraph \textit{without double edges}.

\subsubsection{Filtrations}


\begin{definition}\label{def:filtration}
Given  a directed flag complex $\mathcal{X}$, a \textit{filtration} of  $\mathcal{X}$ is an indexed family of subcomplexes $\mathcal{X}_{i} \subseteq \mathcal{X}$, $\{\mathcal{X}_{i}\}_{i\in I_{k}^{*}}$, such that $\mathcal{X}_{i} \subseteq \mathcal{X}_{j}$, whenever $i \le j$, and it can be represented as a nested sequence of subcomplexes:
\begin{equation}\label{eq:filtration}
\emptyset = \mathcal{X}_{0} \subseteq \mathcal{X}_{1} \subseteq \mathcal{X}_{2} \subseteq \cdots \subset \mathcal{X}_{k} = \mathcal{X}.
\end{equation}

We say that $\mathcal{X}$ together with a filtration is a \textit{filtered directed flag complex}.
\end{definition}


\smallskip

As we advance in the sequence of a filtration (increasing the indexes), topological features of the simplicial complex, such as independent cycles, may appear and disappear.

\smallskip

\begin{example}\label{ex:filtration}
Consider $\mathcal{X} = \{[0], [1], [2], [0,1], [1,2], [0,2], [0,1,2]\}$.  The subcomplexes $\mathcal{X}_{1} = \{0], [1], [2]\}$, $\mathcal{X}_{2} = \{0], [1], [2], [0,1]\}$, $\mathcal{X}_{3} = \{0], [1], [2], [0,1], [1,2], [0,2]\}$, and $\mathcal{X}_{4} = \mathcal{X}$, satisfy $\mathcal{X}_{1} \subset \mathcal{X}_{2} \subset \mathcal{X}_{3} \subset \mathcal{X}_{4}$, thus they form a filtration of $\mathcal{X}$ (see Figure \ref{fig:filtration}). We let $\mathcal{X}_{0} = \emptyset$ implicit in the filtration .

\begin{figure}[h!]
    \centering
\includegraphics[scale=0.8]{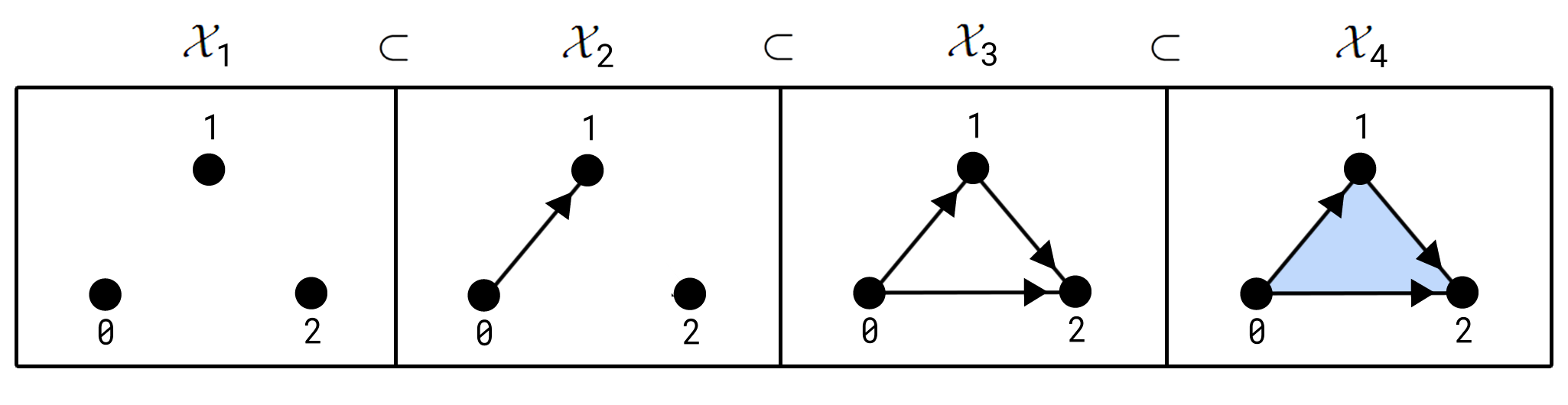}
 \caption{Example of filtration.}
 \label{fig:filtration}
\end{figure}
\end{example}

There are several different ways to create different filtrations for a (unweighted or weighted) simplicial complex \citep{Edelsbrunner, Ren}. One way is by considering its vertices in a metric space, and then producing simplices (and thus subcomplexes) by gradually increasing the threshold on the distance between the vertices, similarly as specified in the Definition \ref{def:vietoris-rips} of the Vietoris-Rips complex. These filtrations are called \textit{metric filtrations}. 

Analogously, in the case where $\mathcal{X}$ is associated with a weighted digraph $G=(V,E)$, we can obtain a filtration of $\mathcal{X}$ by considering the vertex set $V$ in a pre-metric space and then compute pre-metric Dowker complexes $\mathfrak{D}^{pre}(V; \delta)$ (Definition \ref{def:dowker-complex}) by gradually increasing the threshold $\delta$ on the pre-metric obtained from the weight function, that is, the subcomplexes $\mathcal{X}_{i} \subseteq \mathcal{X}$ will be $\mathcal{X}_{i} = \mathfrak{D}^{pre}(V; \delta_{i}) \subseteq \mathcal{X}_{j} = \mathfrak{D}^{pre}(V; \delta_{j})$, such that $\delta_{i} \le \delta_{j}$, whenever $i \le j$. We call this filtration \textit{pre-metric Dowker filtration}, and denote it by $\{ \mathfrak{D}^{pre}(V; \delta_{i}) \subseteq \mathfrak{D}^{pre}(V; \delta_{j}) \}_{\delta_{i} \le \delta_{j}}$. However, if we consider the pre-metric Dowker $\delta$-sink complex or the pre-metric Dowker $\delta$-source complex, the corresponding filtrations $\{ \mathfrak{D}^{pre}_{si}(V; \delta_{i}) \subseteq \mathfrak{D}^{pre}_{si}(V; \delta_{j}) \}_{\delta_{i} \le \delta_{j}}$ and $\{ \mathfrak{D}^{pre}_{so}(V; \delta_{i}) \subseteq \mathfrak{D}^{pre}_{so}(V; \delta_{j}) \}_{\delta_{i} \le \delta_{j}}$, respectively, are sensitive to directionality  \citep{Chowdhury-2016}.

Another way to define a filtration for a weighted directed flag complex $(\mathcal{X}, \widetilde{\omega})$ based on its weights is by establishing thresholds on the product-weight function: $\mathcal{X}_{i} = \{\sigma \in \mathcal{X} : \widetilde{\omega}(\sigma) \le \delta_{i}\}$, for given positive real numbers $\delta_{i}$.



\subsubsection{Persistent Homology}

Persistent homology essentially deals with the study of the ``lifetime persistence" of the topological features of a space. In the following, we present the mathematical theory behind persistent homology.

\smallskip

\begin{definition}\label{def:chain-map}
Let $f:\mathcal{X} \rightarrow \mathcal{X'}$ be a map between two directed flag complexes. For an integer $n \ge 0$, the map $f$ induces a $\mathbb{K}$-linear map 
\begin{alignat}{2}\label{eq:chain-map}
  \widetilde{f}_{n}: C_{n}(\mathcal{X}; \mathbb{K}) & \longrightarrow &  C_{n}(\mathcal{X}'; \mathbb{K}), \\
  \sum_{i} a_{i}\sigma_{i}^{(n)} & \mapsto & \sum_{i} a_{i}f(\sigma_{i}^{(n)}),  \notag 
\end{alignat}

\noindent where $a_{i} \in \mathbb{K}$. The map $\widetilde{f}_{n}$ is called \textit{$n$-th chain map}.
\end{definition}

\smallskip

If $(C_{\bullet}(\mathcal{X}; \mathbb{K}), \partial_{\bullet})$ and $(C_{\bullet}(\mathcal{X}'; \mathbb{K}), \partial_{\bullet}')$ denote the chain complexes associated with $\mathcal{X}$ and $\mathcal{X}'$, respectively, the $n$-th chain map satisfies $\widetilde{f}_{n} \circ \partial_{n+1} = \partial_{n+1}' \circ \widetilde{f}_{n+1}$, for all $n \ge 0$, which implies that $\widetilde{f}_{n}(\Ima \partial_{n+1}) \subseteq \Ima \partial_{n+1}'$ and $\widetilde{f}_{n}(\ker \partial_{n}) \subseteq \ker \partial_{n}'$, i.e. it takes $n$-cycles to $n$-cycles and $n$-boundaries to $n$-boundaries. Accordingly, $\widetilde{f}_{n}$ induces a linear map 
\begin{alignat}{2}\label{eq:induced-func-homology-groups}
  f_{n}: H_{n}(\mathcal{X}) & \rightarrow & H_{n}(\mathcal{X'}),\\
  [c] & \mapsto & [\widetilde{f}_{n}(c)].  \notag 
\end{alignat}

Let $\{\mathcal{X}_{i}\}_{i\in I_{k}^{*}}$ be a filtration of $\mathcal{X}$. For subcomplexes $\mathcal{X}_{i} \subseteq \mathcal{X}_{j}$, $i \le j$, we have a natural inclusion map $\mathcal{X}_{i} \hookrightarrow \mathcal{X}_{j}$ that, based on the previous discussion, induces a linear map between their $n$-th homologies, for each $n \ge 0$:
\begin{equation}
f_{n}^{i,j}: H_{n}(\mathcal{X}_{i}) \rightarrow H_{n}(\mathcal{X}_{j}).
\end{equation}

Consequently, by considering the entire filtration, for each $n \ge 0$ we have a sequence of homologies connected by linear maps:
\begin{equation}
\{0\} = H_{n}(\mathcal{X}_{0}) \xrightarrow{f_{n}^{0,1}}  H_{n}(\mathcal{X}_{1}) \xrightarrow{f_{n}^{1,2}} ... \xrightarrow{f_{n}^{k-1,k}} H_{n}(\mathcal{X}_{k}) = H_{n}(\mathcal{X}).
\end{equation}

\smallskip

Similarly, as mentioned earlier for filtrations, as we advance in the above sequence of homologies, homology classes may appear and disappear or, using the usual terminology within the computational topology literature, they may be ``born" and ``die." The persistent homologies, as defined below, try to capture the ``lifetimes" (persistence) of these homology classes.

\smallskip

\begin{definition}\label{def:per-homologies}
Given $\mathcal{X}$ with a filtration $\{\mathcal{X}_{i}\}_{i\in I_{k}^{*}}$, the \textit{$n$-th persistent homologies}, $H^{i, j}_{n}$, are the images of the linear maps $f_{n}^{i,j}$, for $0 \le i \le j \le k$, i.e.
\begin{equation}\label{eq:per-homologies}
H^{i, j}_{n} = \Ima f_{n}^{i,j} = \ker \partial_{n}(\mathcal{X}_{i}) / (\Ima \partial_{n+1}(\mathcal{X}_{j}) \cap \ker \partial_{n}(\mathcal{X}_{i})).
\end{equation}

The dimensions of $H^{i, j}_{n}$, denoted by $\beta^{i,j}_{n} = \dim H^{i, j}_{n}$, are called \textit{$n$-th persistent Betti numbers}.
\end{definition}

\smallskip

Note that $H^{i, i}_{n} = H_{n}(\mathcal{X}_{i})$, since $\Ima \partial_{n+1}(\mathcal{X}_{i}) \cap \ker \partial_{n}(\mathcal{X}_{i}) = \Ima \partial_{n+1}(\mathcal{X}_{i})$. Also, we say that a homology class $[c] \in H_{n}(\mathcal{X}_{i})$ is \textit{born at} $\mathcal{X}_{i}$ if $[c] \not\in H_{n}^{i-1, i}(\mathcal{X}_{i})$; and, if $[c]$ is born at $\mathcal{X}_{i}$, we say that $[c]$ \textit{dies entering} $\mathcal{X}_{j}$ if $f_{n}^{i, j-1}([c]) \not\in H_{n}^{i-1, j-1}$ but $f_{n}^{i, j}([c]) \in H_{n}^{i-1, j}$. The \textit{persistence} of a homology class that is born at $\mathcal{X}_{i}$ and dies entering $\mathcal{X}_{j}$ is defined as the difference $j-i$.

\smallskip 

\begin{definition}
Considering the $n$-th persistent Betti numbers $\beta^{i,j}_{n}$ of the $n$-th persistent homologies $H^{i, j}_{n}$, for $n \ge 0$ and $0 \le i \le j \le k$, we define the \textit{pairing number}, denoted by $\mu^{i, j}_{n}$, as the number of independent classes of dimension $n$ that are born at $\mathcal{X}_{i}$ and die entering $\mathcal{X}_{j}$, and it is given by
\begin{equation}
\mu^{i, j}_{n} = (\beta^{i, j - 1}_{n} - \beta^{i, j}_{n}) - (\beta^{i-1, j-1}_{n} - \beta^{i-1, j}_{n}),
\end{equation}

\noindent where $(\beta^{i, j - 1}_{n} - \beta^{i, j}_{n})$ represents the number of homology classes that are born at or before $\mathcal{X}_{i}$ and die entering $\mathcal{X}_{j}$, and $(\beta^{i-1, j-1}_{n} - \beta^{i-1, j}_{n})$ represents the number of classes that are born at or before $\mathcal{X}_{i-1}$ and die entering $\mathcal{X}_{j}$.

\end{definition}

\smallskip

The pairing numbers can be used to visualize the \textit{birth} and \textit{death} of homology classes through the \textit{persistence diagram} of a filtration, as formally defined below.

\smallskip

\begin{definition}\label{def:pers-diag}
Given a filtration $\mathcal{F} = \{\mathcal{X}_{i}\}_{i \in I_{k}^{*}}$,  the \textit{$n$-th persistence diagram} of $\mathcal{F}$, denoted by $\mbox{Dgm}_{n}(\mathcal{F})$,  is a multiset of points $(i,j)$, $0 \le i \le j \le k$, with multiplicities $\mu^{i, j}_{n}$, in the extended plane $\bar{\mathbb{R}}^{2}$, where $\bar{\mathbb{R}} = \mathbb{R} \cup \{\infty\}$, such that the vertical distance of $(i,j)$ to the diagonal is equal to $j-i$, i.e.
\def\test#1{#1|xxx}
\begin{equation}
\mbox{Dgm}_{n}(\mathcal{F}) = \dgal{ (i,i) \in \bar{\mathbb{R}}^{2} : \mu^{i, i}_{n} = \infty } \cup \dgal{ (i,j) \in \bar{\mathbb{R}}^{2} : i<j, \mu^{i, j}_{n} < \infty }.
\end{equation}
\end{definition}

\smallskip

As a consequence, the $n$-th persistent Betti numbers can be computed by counting the points (with multiplicities) in the respective $n$-th persistence diagram, as stated by the \textit{fundamental lemma of persistent homology} (FLPH) \citep{Edelsbrunner}.
\smallskip

\begin{proposition}\label{prop:FLPH}
\textbf{(FLPH)} Given a filtration $\{\mathcal{X}_{i}\}_{i \in I_{k}^{*}}$, for $n \ge 0$, the $n$-th persistent Betti number is equal to $\beta_{n}^{i,j} = \sum_{r \le i} \sum_{l > j} \mu^{r, l}_{n} $, for each pair $i,j$, with $0 \le i \le j \le k$.
\end{proposition}

\smallskip

Moreover, from the FLPH, we can create a related function as follows.

\smallskip

\begin{definition}
Given a filtration $\mathcal{F} = \{\mathcal{X}_{i}\}_{i \in I_{k}^{*}}$, we define the \textit{Betti function} as
\def\test#1{#1|xxx}
\begin{equation}
\mathcal{B}_{n}(t) = \# \dgal{ (i,j) \in \mbox{Dgm}_{n}(\mathcal{F}) : i \le t < j}.
\end{equation}

The plot of the Betti function in the plane is called \textit{Betti curve}.
\end{definition}

\smallskip

It's clear that $\mathcal{B}_{n}(t) = \beta_{n}^{i,j}$, for $i \le t < j$. Figure \ref{fig:per-diag} shows the persistence diagram corresponding to the filtration presented in Example \ref{ex:filtration} (Figure \ref{fig:filtration}), and Figure \ref{fig:betti-curve} shows the corresponding Betti curves. 

An alternative to persistence diagrams to visualize the birth and death of topological features in a given filtration are the \textit{persistence barcodes} \citep{Carlsson}, in which, roughly speaking, each bar of length $j-i$ represents a topological feature that is born at $\mathcal{X}_{i}$ and dies entering $\mathcal{X}_{j}$. Figure \ref{fig:per-bar} shows the persistence barcodes corresponding to the aforementioned example.

\smallskip



\begin{figure}[h!]
\centering
\begin{subfigure}{.32\textwidth}
  \centering
  \includegraphics[scale=0.7]{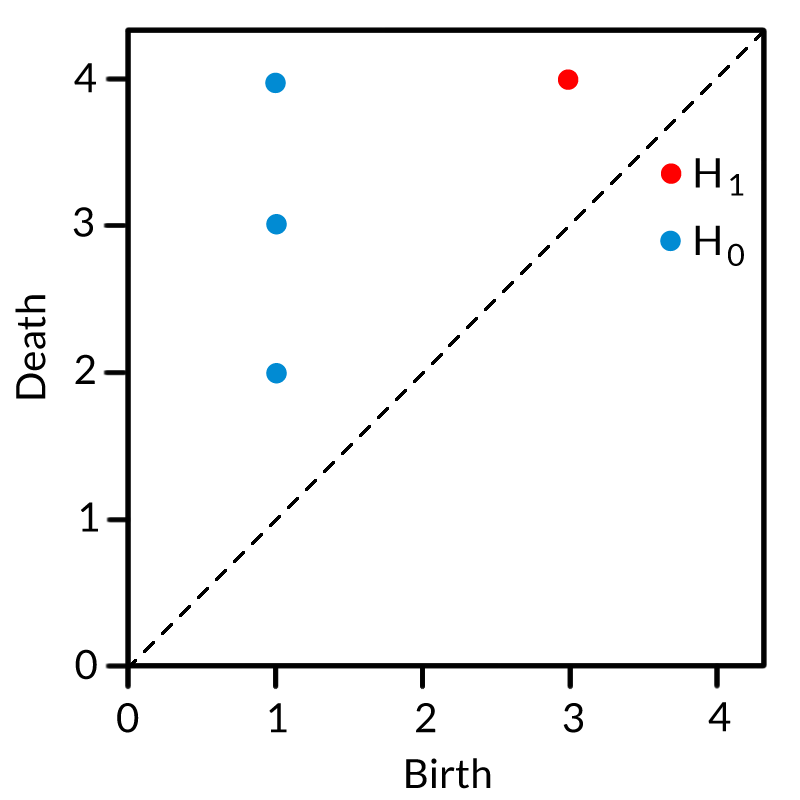}
  \caption{Persistence diagram.}
  \label{fig:per-diag}
\end{subfigure}%
\begin{subfigure}{.32\textwidth}
  \centering
  \includegraphics[scale=0.7]{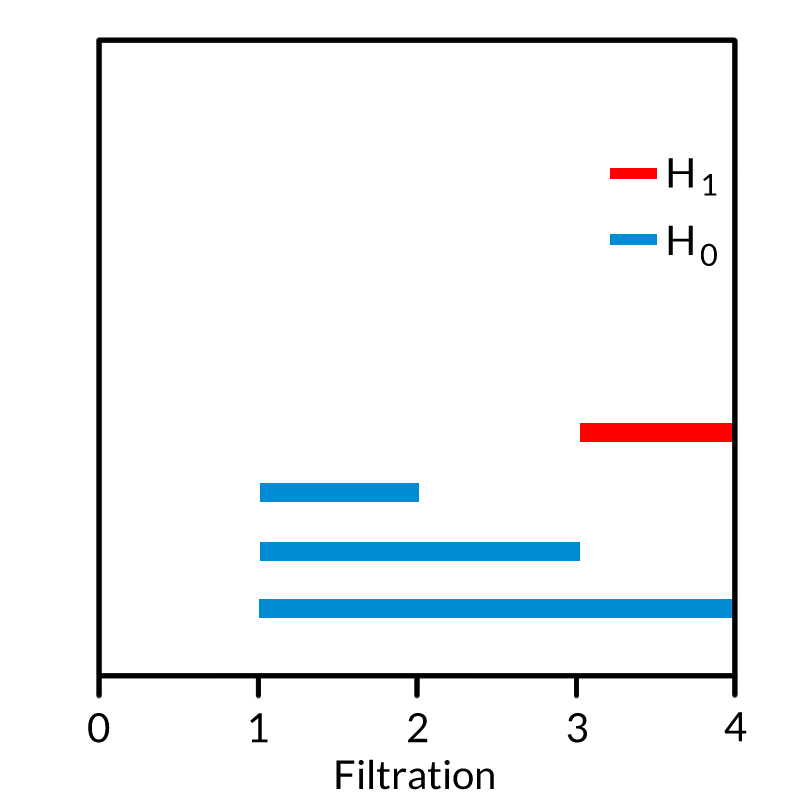}
  \caption{Persistence barcodes.}
  \label{fig:per-bar}
\end{subfigure}
\begin{subfigure}{.32\textwidth}
  \centering
  \includegraphics[scale=0.7]{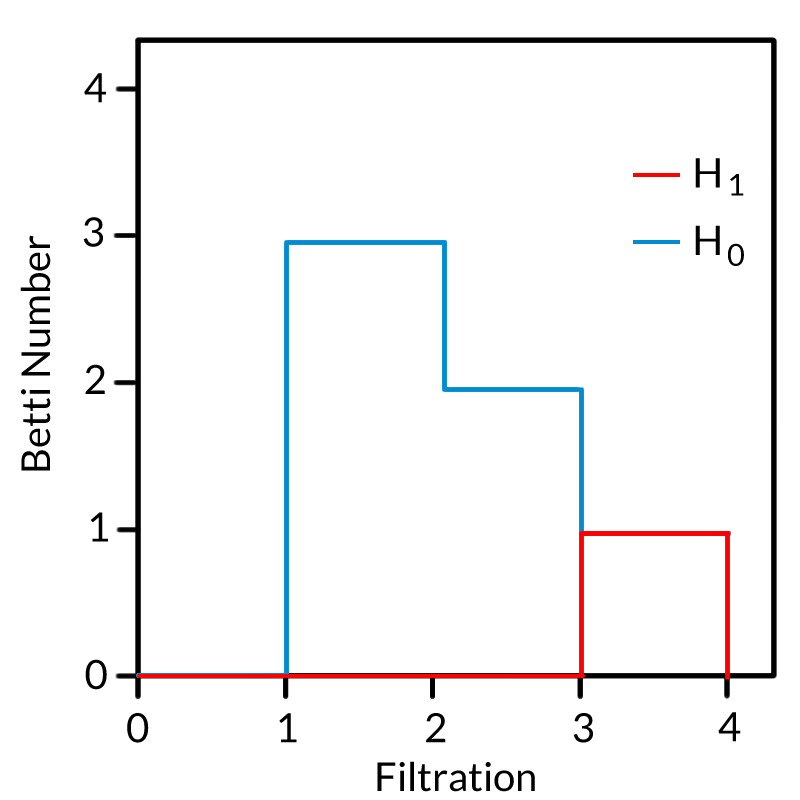}
  \caption{Betti curves.}
  \label{fig:betti-curve}
\end{subfigure}
\caption{The persistence diagram, persistence barcodes, and Betti curves correspond to the filtration of the Example \ref{ex:filtration} (see Figure \ref{fig:filtration}).}
\label{fig:persistence-diag-bar}
\end{figure}

\subsubsection{Distances for Persistence Diagrams and Betti Curves}

One can compare the topology of two different directed flag complexes by comparing the similarity of their persistence diagrams. The most common similarity measures between persistence diagrams are the \textit{bottleneck} and \textit{Wasserstein} distances.

\smallskip

\begin{definition}
Given two persistent diagrams, $P$ and  $Q$, let $\eta : P \rightarrow Q$ denote a bijection (perfect matching) between them, and let $||\cdot||_{\infty}$ denote the $\infty$-norm in the plane $\mathbb{R}^{2}$, i.e. $|| \mathbf{v} ||_{\infty} = \max_{i}(|v_{i}|)$, for $\mathbf{v} = (v_{1}, v_{2}) \in \mathbb{R}^{2}$. The \textit{bottleneck distance} between $P$ and $Q$ is define as
\begin{equation}\label{eq:bottleneck-dist}
d_{W_{\infty}}(P, Q) = \inf_{\eta : P \rightarrow Q} \sup_{x \in P} ||x - \eta(x)||_{\infty} .
\end{equation}
\end{definition}

\smallskip

We can verify that $d_{W_{\infty}}$ is indeed a distance since it satisfies the three conditions present in Definition \ref{def:metric} (formal definition of distance).

\smallskip

\begin{definition}
Given two persistent diagrams, $P$ and $Q$, and a bijection $\eta : P \rightarrow Q$ between them, the \textit{$p$-Wasserstein distance} between $P$ and $Q$ is defined as 
\begin{equation}\label{eq:wasserstein-dist}
d_{W_{p}}(P, Q) =  \inf_{\eta : P \rightarrow Q} \Bigg( \sum_{x \in P} ||x - \eta(x)||_{\infty}^{p} \Bigg)^{1/p}.
\end{equation}
\end{definition}

\smallskip

Note that when $p \rightarrow \infty$, the $p$-Wasserstein distance becomes the bottleneck distance; thus, we can identify the bottleneck distance as the $\infty$-Wasserstein distance.

A notable feature of the persistence diagrams is that, in a space provided with a Wasserstein distance, they are \textit{stable} (or robust) against perturbations (``noise"), i.e. small perturbations in the filtration produce small changes in the respective persistence diagram, as proved in various stability theorems \citep{Cohen-Steiner, Edelsbrunner}.

Furthermore, we can use the $L_{p}$-norm to define a distance between two Betti curves as follows.

\smallskip


\begin{definition}
Given two persistent diagrams, $P$ and $Q$, with respective Betti functions, $\mathcal{B}_{P}(x)$ and $\mathcal{B}_{Q}(x)$, we define de \textit{$p$-Betti distance} between their respective Betti curves as 
\begin{equation}\label{eq:betti-curve-dist}
d_{\mathcal{B}}(\mathcal{B}_{P}, \mathcal{B}_{Q}) = \Bigg( \int_{\mathbb{R}} |\mathcal{B}_{P}(x) - \mathcal{B}_{Q}(x)|^{p} dx \Bigg)^{1/p}.
\end{equation}
\end{definition}

\smallskip

In general, for computational purposes, the most used parameters are $p=1$ or $p=2$.

\subsection{Combinatorial Hodge Laplacian}
\label{sec:comb-hodge-laplacian}


As studied in Chapter \ref{chap:chap2}, the methods of spectral graph theory can capture structural properties of graphs by using eigenvalues and eigenvectors of their matrix representations, such as the adjacency and Laplacian matrices. Likewise, spectral simplicial theory, a generalization of the spectral graph theory for simplicial complexes, can reveal insightful structural information about simplicial complexes through the eigenvalues and eigenvectors associated with their \textit{combinatorial Laplacians} (also called \textit{Hodge Laplacians}, due to their connection with Hodge theory), which are a generalization to higher-orders of the graph Laplacian \citep{Friedman1998, Steenbergen}.

Before introducing the formal definition of the Hodge Laplacian, let's remember that, given a vector space with inner product, $(X, \langle \cdot , \cdot \rangle)$, and a linear operator $f: X \rightarrow X$, the \textit{adjoint operator} of $f$ is a linear operator $f^{*}:X \rightarrow X$ satisfying $\langle f(x), y \rangle = \langle x, f^{*}(y) \rangle$, for all $x, y \in X$.

Just as before, throughout this part, $\mathbb{K}$ will denote a fixed field and $\mathcal{X}$ will denote the directed flag complex of a given digraph \textit{without double edges}.

\smallskip

\begin{definition}
Let $(C_{\bullet}(\mathcal{X}, \mathbb{K}), \partial_{\bullet})$ denote the chain complex of $\mathcal{X}$. Let $\partial_{n}^{*}$ denote the adjoint operator of $\partial_{n}$, for $n \ge 0$. The \textit{(combinatorial) Hodge $n$-Laplacian operator}, $\mathcal{L}_{n}: C_{n}(\mathcal{X}, \mathbb{K}) \rightarrow C_{n}(\mathcal{X}, \mathbb{K}) $, is defined by
\begin{equation}\label{eq:laplacian-operator}
\mathcal{L}_{n} = \partial_{n+1} \circ \partial_{n+1}^{*} + \partial_{n}^{*} \circ \partial_{n}.
\end{equation}

\end{definition}


\smallskip

The higher-order boundary maps $\partial_{n}$ induce matrices $B_{n}$ that can be interpreted as higher-order incidence matrices between the directed simplices and their (co-)faces. For instance, the matrix $B_{1}$ is the vertex-to-arc incidence matrix (see Definition \ref{def:incident-matrix}), and $B_{2}$ is the arc-to-($2$-simplex) incidence matrix, and so on. Thus, let $[\partial_{n}] = B_{n}$ and $[\partial_{n}^{*}] = B_{n}^{T}$ be the matrix representations of the $n$-th boundary operator and its adjoint operator, respectively, the matrix representation of the Hodge $n$-Laplacian operator is given by
\begin{equation}\label{eq:hodge-laplacian}
[\mathcal{L}_{n}] = B_{n+1}B_{n+1}^{T} + B_{n}^{T}B_{n}.
\end{equation}

\smallskip

As commented before, the Hodge $n$-Laplacian is a generalization of the graph Laplacian for simplicial complexes. Indeed, for $n=0$, we have $[\mathcal{L}_{0}] = BB^{T}$ (graph Laplacian). Also, since $\partial_{n} \circ \partial_{n+1} = 0$, we have $B_{n}B_{n+1} = 0$, for all $n \ge 0$. Occasionally, the following notations are used: $\mathcal{L}_{n}^{u} = \partial_{n+1}\circ \partial_{n+1}^{*}$ (\textit{upper Laplacian}) and $\mathcal{L}_{n}^{l} = \partial_{n}^{*} \circ \partial_{n}$ (\textit{lower Laplacian}), thus $\mathcal{L}_{n} = \mathcal{L}_{n}^{u} + \mathcal{L}_{n}^{l}$. 

\smallskip

\begin{observation}
We denote the \textit{$n$-Laplacian spectrum}, $n \ge 0$, of the Hodge $n$-Laplacian matrix by $\{ \mu_{i}^{n} \}_{i=1}^{max}$ and, analogously to the graph Laplacian, we represent the eigenvalues of this spectrum in an increasing order: $\mu_{1}^{n} \le \hdots \le \mu_{max}^{n}$.
\end{observation}

\begin{proposition}\label{prop:hodge-laplacian-eigenvalues}
\textit{The Hodge $n$-Laplacian matrix is positive semi-definite and all its eigenvalues are non-negative.}
\end{proposition}

\smallskip

The proof of the previous proposition is analogous to the proof of the Proposition \ref{prop:PSD} together with the proof of the Proposition \ref{prop:eigenvalues}.

Furthermore, the kernel of the Hodge $n$-Laplacian operator associated with $\mathcal{X}$ is isomorphic to the $n$-th homology of $\mathcal{X}$, i.e.  $ker(\mathcal{L}_{n}) \cong H_{n}(\mathcal{X})$, which implies that the number of zero-eigenvalues of $[\mathcal{L}_{n}]$ is equal to the $n$-th Betti number \citep{Friedman1998}. Accordingly, the Hodge Laplacians provide valuable topological information about the complex.

\smallskip

\begin{example}\label{ex:dfc-hodge}
Consider the digraph $G$ as depicted in Figure \ref{laplacian1}. Its directed flag complex is $\mathcal{X} = \{ [0], [1], [2], [3], [0,1], [0,2], [1,2], [1,3], [2,3], [0,1,2], [1,2,3] \}$. Therefore, the higher-order incidence matrices $B_{n}$ associated with $\mathcal{X}$ are

\begin{figure}[h!]
    \centering
\includegraphics[scale=0.8]{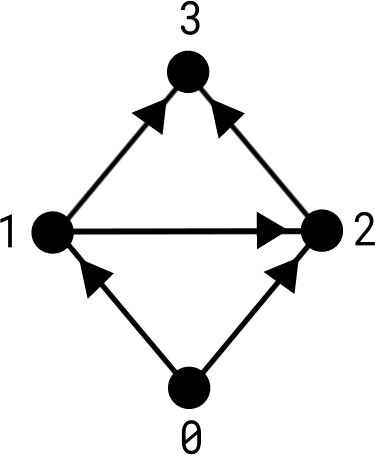}
 \caption{Digraph $G$.}
 \label{laplacian1}
\end{figure}

$$
{\small
B_{1} =  
\bbordermatrix{
  & [01] & [02] & [12] & [13] & [23] \cr
[0]  & 1 & 1 & 0 &  0 & 0 \cr
[1]  & -1 & 0 & 1 & 1 & 0 \cr
[2]  & 0 & -1 & -1 & 0 & 1 \cr
[3]  & 0 & 0 & 0 & -1 & -1 \cr
}, \hspace{0.1in}
B_{2} =  
\bbordermatrix{
  & [012] & [123] \cr
[01]  & 1 & 0 \cr
[02]  & -1 & 0  \cr
[12]  & 1 & 1  \cr 
[13]  & 0 & -1  \cr
[23]  & 0 & 1  \cr
}
}.
$$

\smallskip

Note that $B_{n} = 0$ for all $n \ge 3$. Thus, by the formula \ref{eq:hodge-laplacian}, the Hodge $n$-Laplacians associated with $\mathcal{X}$ are: 

$$
[\mathcal{L}_{0}] = B_{1}B_{1}^{T} = \begin{bmatrix}
2 & -1 & -1 & 0\\
-1 & 3 & -1 & -1 \\
-1 & -1 & 3 & -1 \\
 0 & -1 & -1 & 2 
\end{bmatrix},
$$

$$
[\mathcal{L}_{1}] = B_{2}B_{2}^{T} + B_{1}^{T}B_{1} = 
\begin{bmatrix}
3 & 0 & 0 & -1 & 0 \\
0 & 3 & 0 & 0 & -1 \\
0 & 0 & 4 & 0 & 0 \\
-1 & 0 & 0 & 3 & 0\\
0 & -1 & 0 & 0 & 3
\end{bmatrix}, \hspace{0.1in}
[\mathcal{L}_{2}] = B_{2}^{T}B_{2} = \begin{bmatrix}
3 & 1\\
1 & 3
\end{bmatrix},
$$

\noindent and  $[\mathcal{L}_{n}] = 0$ for all $n \ge 3$.
\end{example}



\section{Path Complexes of Digraphs}
\label{sec:path-complexes}

In the literature, there are different approaches to constructing complexes and homologies from digraphs. For instance, one can build homologies from the directed clique complexes of digraphs, as shown in Subsection \ref{sec:simplicial-homology}; one can also build Hochschild homologies out of the path algebras of digraphs \citep{Happel}. Nevertheless, a novel type of complexes and homologies associated with digraphs that generalizes the concept of directed clique complexes, called \textit{path complexes} and \textit{path homologies}, was introduced and developed by Grigor’yan et al. \citep{Grigoryan-2018, Grigoryan-2013, Grigoryan-2014a, Grigoryan-2015, Grigoryan-2020, Grigoryan-2014b, Grigoryan-2016, Grigoryan-2017}. Some interesting properties of this new formalism are that the chain complexes associated with path complexes might contain more digraph substructures than just directed cliques, and we can have non-trivial path homologies of all dimensions.

In the following, we present the formalism of path complexes and path homologies as presented in the aforementioned articles, preserving the original notation whenever possible, and also the new concept of \textit{$\partial$-invariant directed quasi-cliques}.


\subsection{Path Complexes}

\begin{definition}
Given a finite set of vertices $V$ and an integer $p \ge 0$, an \textit{elementary $p$-path} is any sequence of $p+1$ vertices $i_{k} \in V$, $k=0,...,p$ (not necessarily distinct), which will be denoted by $e_{i_{0}...i_{p}}$. Given a field $\mathbb{K}$, we denote by $\Lambda_{p} = \Lambda_{p}(V)$ the $\mathbb{K}$-vector space formed by all formal $\mathbb{K}$-linear combinations of elementary $p$-paths, i.e. $u = \sum_{i_{0},..., i_{p} \in V} u^{i_{0}...i_{p}}e_{i_{0}...i_{p}} \in \Lambda_{p}$, where $u^{i_{0}...i_{p}} \in \mathbb{K}$. The elements of $\Lambda_{p}$ are called \textit{$p$-paths}. Also, denoting the empty set as an element of $\Lambda_{-1}$ by $e$, we have $\Lambda_{-1} \cong \mathbb{K}$, since the elements of $\Lambda_{-1}$ are multiples of $e$, and we define $\Lambda_{-2} = \{0\}$.
\end{definition}

\begin{definition}
Given a digraph $G = (V,E)$, for any $p \ge 0$ we define the \textit{p-th boundary operator} $\partial_{p} : \Lambda_{p}(V) \rightarrow \Lambda_{p-1}(V)$ by

\begin{equation}\label{eq:boundary-map-path}
\partial_{p}(e_{i_{0}...i_{p}}) = \sum_{q=0}^{p} (-1)^{q} \hat{d}_{q}(e_{i_{0}...i_{p}}),
\end{equation}

\noindent where $\hat{d}_{q}(e_{i_{0}...i_{p}}) = e_{i_{0}...\hat{i}_{q}...i_{p}}$, i.e. $\hat{d}_{q}$ is the function that excludes the $q$-th node $i_{q}$ of the $p$-path $e_{i_{0}...i_{p}}$. In addition, we define $\partial_{-2} := 0$.
\end{definition}

\smallskip

The operator (\ref{eq:boundary-map-path}) satisfies the property $\partial_{p} \circ \partial_{p+1} = 0$, for all $p \ge 0$ (see \citep{Grigoryan-2020}, p. 567, for a proof).

\smallskip

\begin{definition}
An elementary $p$-path $e_{i_{0}...i_{p}}$ is said to be \textit{regular} if $i_{k} \neq i_{k+1}$ for all $k=0,...,p-1$, and it is called \textit{non-regular} otherwise. The elements of the subspace $\mathcal{R}_{p}(V) = \mbox{{span}} \{ e_{i_{0}...i_{p}} : i_{k} \neq i_{k+1}, \forall k=0,...,p-1 \} \subseteq \Lambda_{p}(V)$ are called \textit{regular $p$-paths}.
\end{definition} 

\smallskip

Although Definition \ref{def:path-walk-trail} states that all vertices of a directed path must be different, here we are adopting an abuse of notation and using the expression \textit{elementary path} as a synonym of \textit{walk}, and the expression \textit{regular elementary path} as a synonym of \textit{trail}.

\smallskip

\begin{definition} 
A \textit{path complex} over a finite set $V$ is a non-empty set $\mathcal{P} = \mathcal{P}(V)$ of elementary paths
on $V$ such that for any $n \ge 0$, if $e_{i_{0}...i_{n}} \in \mathcal{P}(V)$, then the truncated paths $e_{i_{0}...i_{n-1}}$ and $e_{i_{1}...i_{n}}$ belong to $\mathcal{P}(V)$ as well.
\end{definition}

\begin{definition} 
Let $G = (V,E)$ be a digraph. A regular elementary $p$-path $e_{i_{0}...i_{p}}$ on $V$ is called \textit{allowed} if the directed edge $(i_{k}, i_{k+1}) \in E$ for any $k=0,..,p-1$, and its called \textit{non-allowed} otherwise. We denote by $\mathcal{A}_{p}(V)$ the $\mathbb{K}$-vector subspace of $\Lambda_{p}(V)$ spanned by allowed elementary $p$-paths, i.e.
\begin{equation}
\mathcal{A}_{p}(V) = \mbox{span} \{ e_{i_{0}...i_{p}} : i_{0}...i_{p} \mbox{ is allowed} \}.
\end{equation}
\end{definition} 

\smallskip

Also, we denote by $\mathcal{P}_{p}(G)$ the set of all allowed $p$-paths and by $\mathcal{P}(G) = \bigcup_{p} \mathcal{P}_{p}(G)$ the path complex associated with the digraph $G=(V,E)$. In particular, $\mathcal{P}_{0}(G) = V$ and $\mathcal{P}_{1}(G) = E$.

Here we emphasize that, from now on, we will consider only path complexes associated with digraphs.

\smallskip

\begin{example}
Consider the digraph $G$ shown in Figure \ref{pathcomplex}. Its path complex is the following set of elementary $p$-paths, with $p=0,1,2,3$: 
$$
\mathcal{P}(G) = \{ e_{0}, e_{1}, e_{2}, e_{3}, e_{01}, e_{02}, e_{12}, e_{13}, e_{23}, e_{012}, e_{123}, e_{013}, e_{023}, e_{0123} \}.
$$

\begin{figure}[h!]
    \centering
\includegraphics[scale=0.82]{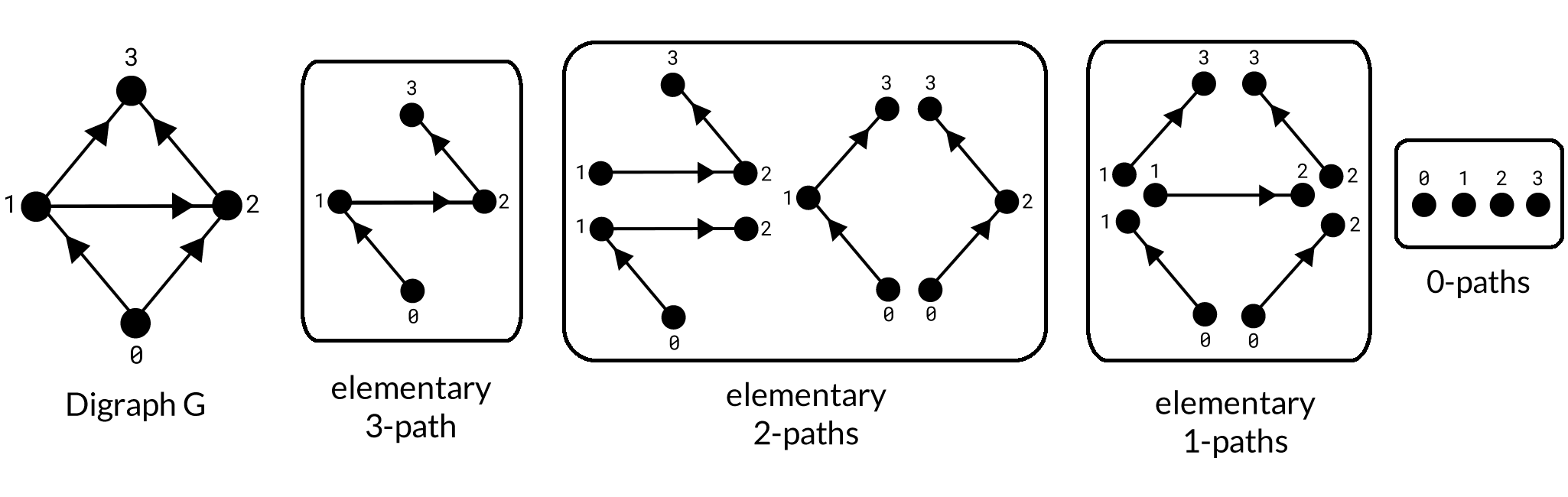}
 \caption{A digraph and its path complex.}
 \label{pathcomplex}
\end{figure}
\end{example}

Note that every $p$-path of a path complex is allowed. Restricting the operator $\partial_{p}$ on the subspace $\mathcal{A}_{p} \subseteq \Lambda_{p}$, $p \ge 0$, we can have $\partial_{p}\mathcal{A}_{p} \not\subset \mathcal{A}_{p-1}$, but we are interested in the case where the inclusion occurs, so we define the following subspace of $\mathcal{A}_{p}$.

\smallskip

\begin{definition}
Given a digraph $G = (V,E)$, consider the following subspace of $\mathcal{A}_{p}(V)$, $p \ge 0$:
\begin{equation}
\Omega_{p} = \Omega_{p}(G) := \{u \in \mathcal{A}_{p} : \partial_{p} u \in \mathcal{A}_{p-1}\}.
\end{equation}

The elements of $\Omega_{p}$ are called \textit{$\partial$-invariant $p$-paths}.
\end{definition}

\smallskip

Notice that $\partial_{p}\Omega_{p} \subseteq \Omega_{p-1}$, for all $p \ge 0$. In fact, by definition, $\partial_{p}u \in \mathcal{A}_{p-1}$ for all $u \in \Omega_{p}$, and since $\partial_{p-1}(\partial_{p}u) = 0 \in \mathcal{A}_{p-2}$, we have $\partial_{p} u \in \Omega_{p-1}$.

\begin{example}
A triangle is a sequence of three vertices $0, 1, 2$ such that the directed edges $(0,1)$, $(1,2)$, and $(0,2)$ exist (this coincides with the definition of a directed $3$-clique, but, as already commented, in this part we'll use the original nomenclature as exposed in \citep{Grigoryan-2013}). A triangle determines a $\partial$-invariant $2$-path, $e_{012} \in \Omega_{2}$, because $e_{012} \in \mathcal{A}_{2}$ and $\partial e_{012} = e_{12} - e_{02} + e_{01} \in \mathcal{A}_{1}$. Other examples of digraphs that determine $\partial$-invariant $2$-paths are double edges and squares (see Figure \ref{fig:inv-2-paths}).

\begin{figure}[h!]
    \centering
\includegraphics[scale=1.3]{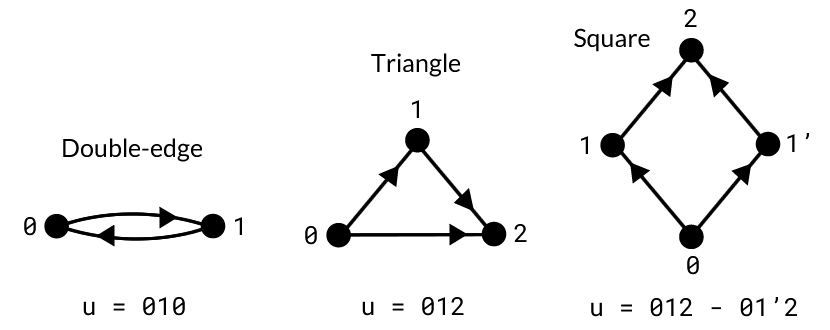}
 \caption{Examples of digraphs containing $\partial$-invariant $2$-paths.}
 \label{fig:inv-2-paths}
\end{figure}
\end{example}

Moreover, Grigor'yan et al. \citep{Grigoryan-2013} proved that the elements of $\Omega_{2}$ are linear combinations of triangles, squares, and double edges.

\subsubsection{Snakes and $\partial$-Invariant Directed Quasi-Cliques}

The presence or absence of $\partial$-invariant paths in a digraph can characterize some of its topological properties, since these paths are related to the \textit{path homology} of the digraph, as we will see in the next section. In view of this, it's worth looking for certain types of subdigraphs that contain $\partial$-invariant paths, such as \textit{snakes} and \textit{$\partial$-invariant directed quasi-cliques}.

\smallskip

\begin{definition}
For a given integer $p \ge 0$, a \textit{$p$-snake} is a digraph $G=(V, E)$, with $V = \{v_{0}, ..., v_{p}\}$, such that its arcs are $(v_{i}, v_{i+1})$ for all $i = 0,...,p-1$, and $(v_{j}, v_{j+2})$, for all $j = 0,...,p-2$.
\end{definition}

\begin{example}  
Figure \ref{fig:snakes} presents examples of $p$-snakes for $p=2,3,4$.
\begin{figure}[h!]
\centering
\begin{subfigure}{.28\textwidth}
  \centering
  \includegraphics[scale=0.69]{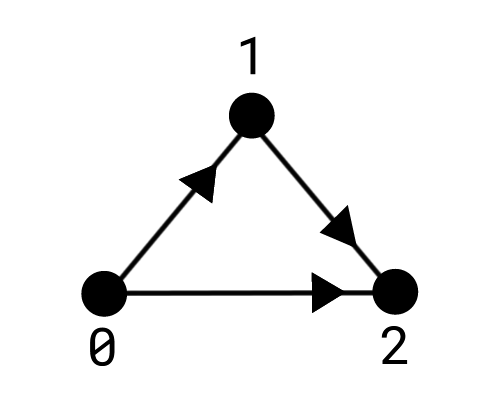}
  \caption{$2$-snake.}
  \label{fig:snakes1}
\end{subfigure}%
\begin{subfigure}{.28\textwidth}
  \centering
  \includegraphics[scale=0.69]{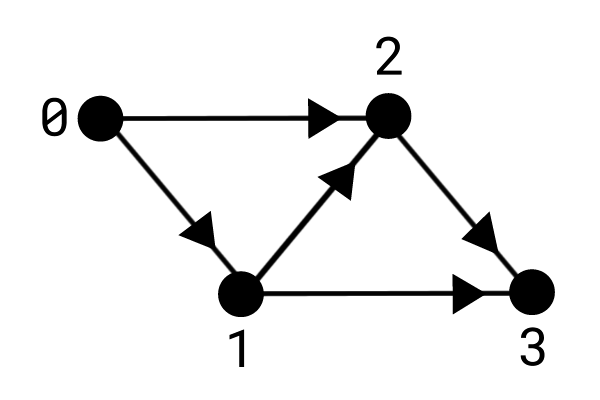}
  \caption{$3$-snake.}
  \label{fig:snakes2}
\end{subfigure}
\begin{subfigure}{.28\textwidth}
  \centering
  \includegraphics[scale=0.69]{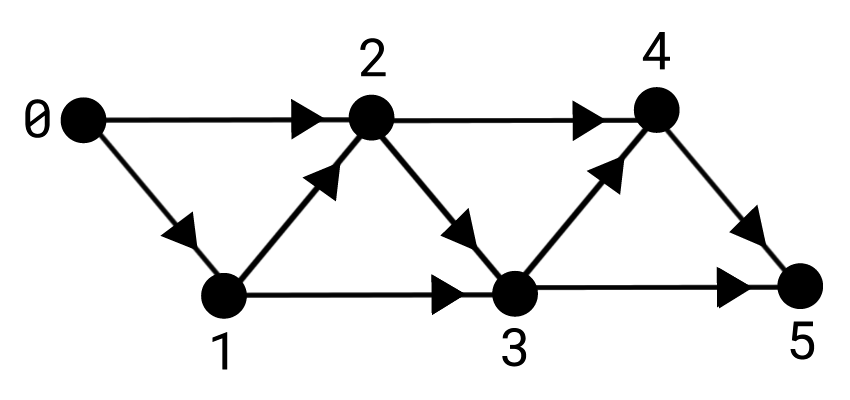}
  \caption{$4$-snake.}
  \label{fig:snakes3}
\end{subfigure}
\caption{Examples of $p$-snakes for $p=2,3,4$.}
\label{fig:snakes}
\end{figure}
\end{example}

In the next definition, we present a directed version of the definition of $\gamma$-quasi-clique (Definition \ref{def:quasi-clique}).

\smallskip

\begin{definition}\label{def:dir-quasi-clique}
Given a digraph $G=(V, E)$, a subdigraph $H = (V', E') \subseteq G$, with $|V'| = m$, is called a \textit{directed $\gamma$-quasi-clique} (or $\gamma$-DQC), for a parameter $0 < \gamma \le 1$, if  $\deg^{tot}_{H}(v) \ge \gamma (m-1)$, for all $v \in V'$.
\end{definition}

\smallskip

Note that the previous definition is equivalent to saying that a digraph is a directed $\gamma$-quasi-clique if its underlying undirected graph is a $\gamma$-quasi-clique.

\smallskip

\begin{definition}
Let $u \in \Omega_{p}$ be a $\partial$-invariant $p$-path. For a parameter $0 < \gamma \le 1$, the \textit{$\partial$-invariant $(u, \gamma)$-DQC} (or just \textit{$(u, \gamma)$-DQC}), is the directed $\gamma$-quasi-clique with the minimum amount of arcs which contains all the paths necessary to make $u$ $\partial$-invariant. In particular, if $u = e_{0...p}$ is an elementary $\partial$-invariant $p$-path, we denote its $(u, \gamma)$-DQC simply by $(p, \gamma)$-DQC.
\end{definition}

\smallskip

It is important to point out that every directed $(p+1)$-clique determines an elementary $p$-path (the concept of \textit{simplex-digraph} as proposed by \citep{Grigoryan-2013} coincides with the definition of directed clique), thus every $(p, \gamma)$-DQC is a subdigraph of a directed $(p+1)$-clique. Moreover, if the number of edges in the digraph $(p, \gamma)$-DQC is equal to $2p-1$, then $(p, \gamma)$-DQC coincides with the $p$-snake, and since the total degree of each node of a $p$-snake is greater than or equal to $2$, we can adopt $\gamma = 1/p$ to every $(p, \gamma)$-DQC. Table \ref{tab:table-inv} summarizes the three types of digraphs discussed here, which contain $\partial$-invariant paths.

\smallskip

\begin{table}[h!]
 \center
    \caption{Digraphs containing $\partial$-invariant paths.}
    \label{tab:table-inv}
    \begin{tabular}{M{4.1cm} M{4.1cm} M{4.1cm}}
      \toprule 
      \textbf{Directed $(p+1)$-Clique} & \textbf{$(u, \gamma)$-DQC} & \textbf{$p$-Snake} \\
      \midrule 
      Contains a $\partial$-invariant elementary $p$-path  & Contains an arbitrary $\partial$-invariant $p$-path $u$ & Contains a $\partial$-invariant elementary $p$-path\\
      \bottomrule 
    \end{tabular}
 
\end{table}

\begin{example}
Figures \ref{fig:DQC1} and \ref{fig:DQC2} present some examples of $(u, \gamma)$-DQC associated with some given $\partial$-invariant paths $u$. Note that the digraphs (a), (b), and (d) in Figure \ref{fig:DQC1} coincide with the snakes (a), (b), and (c) present in Figure \ref{fig:snakes}, respectively.

\begin{figure}[h!]
\centering
\begin{subfigure}{.21\textwidth}
  \centering
  \includegraphics[scale=1.2]{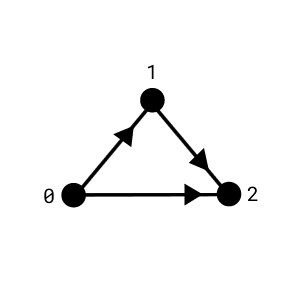}
  \caption{ {\small $(u_{1}, \gamma)$-DQC} }
  \label{fig:DQCa1}
\end{subfigure}%
\begin{subfigure}{.21\textwidth}
  \centering
  \includegraphics[scale=1.2]{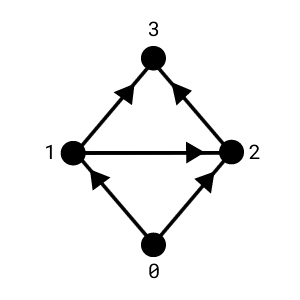}
  \caption{ {\small $(u_{2}, \gamma)$-DQC} }
  \label{fig:DQCb1}
\end{subfigure}
\begin{subfigure}{.21\textwidth}
  \centering
  \includegraphics[scale=1.2]{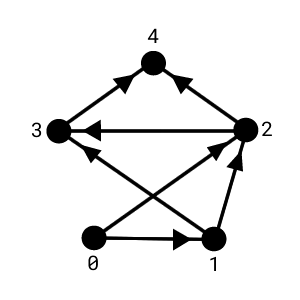}
  \caption{ {\small $(u_{3}, \gamma)$-DQC} }
  \label{fig:DQCc1}
\end{subfigure}
\begin{subfigure}{.21\textwidth}
  \centering
  \includegraphics[scale=1.2]{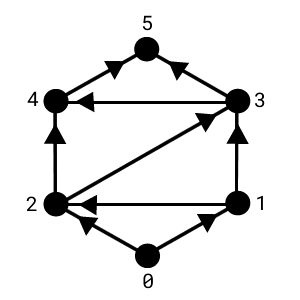}
  \caption{ {\small $(u_{4}, \gamma)$-DQC} }
  \label{fig:DQCd1}
\end{subfigure}
\caption{The ($u_{i}, \gamma)$-DQCs associated with elementary $\partial$-invariant $p$-paths, with $\gamma = 1/p$. (a) $u_{1} = e_{012}$. (b)$u_{2} = e_{0123}$. (c) $u_{3} = e_{01234}$ (d) $u_{4} = e_{012345}$.}
\label{fig:DQC1}
\end{figure}

\begin{figure}[h!]
\centering
\begin{subfigure}{.25\textwidth}
  \centering
  \includegraphics[scale=1.2]{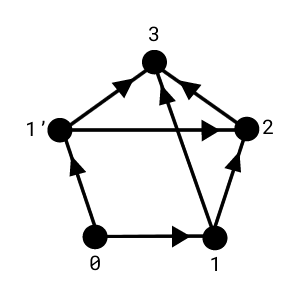}
  \caption{ {\small $(u_{1}, \gamma)$-DQC} }
  \label{fig:DQCa2}
\end{subfigure}%
\begin{subfigure}{.25\textwidth}
  \centering
  \includegraphics[scale=1.2]{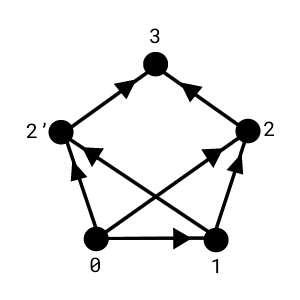}
  \caption{ {\small $(u_{2}, \gamma)$-DQC} }
  \label{fig:DQCb2}
\end{subfigure}
\begin{subfigure}{.25\textwidth}
  \centering
  \includegraphics[scale=1.2]{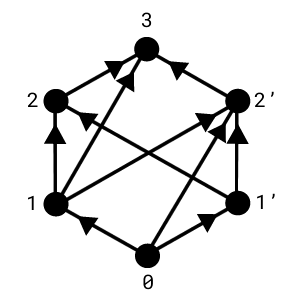}
  \caption{ {\small $(u_{3}, \gamma)$-DQC} }
  \label{fig:DQCc2}
\end{subfigure}
\caption{Examples of $(u, \gamma)$-DQCs associated with $\partial$-invariant $3$-paths. (a) $u = e_{0123} - e_{01'23}$. (b) $u_{2} = e_{0123} - e_{012'3}$. (c) $u_{3} = e_{0123} - e_{012'3} + e_{01'2'3}$.}
\label{fig:DQC2}
\end{figure}

\end{example}

\subsection{Path Homology}

In this part, we extend the concept of simplicial homology introduced for simplicial complexes in Subsection \ref{sec:simplicial-homology} to path complexes.

We begin by pointing out that the boundary operators  (\ref{eq:boundary-map-path}) restricted to the spaces $\Omega_{\bullet}$ satisfy the same properties as the boundary operators (\ref{eq:boundary-map}). Accordingly, by restricting $\partial_{p}$ to $\Omega_{p}$,  $p \ge 0$, we have  $\partial_{p} \circ \partial_{p+1} = 0$, and $\Ima \partial_{p+1} \subseteq \ker \partial_{p}$, thus we can define a chain complex $(\Omega_{\bullet}, \partial_{\bullet})$:

\[
   \hdots
   \xrightarrow[]
    {\partial_{p+1}}
    \Omega_{p}
    \xrightarrow[]
    {\partial_{p}}
     \Omega_{p-1}
     \xrightarrow[]
    {\partial_{p-1}}
    \hdots
     \xrightarrow[]
    {\partial_{2}}
    \Omega_{1}
    \xrightarrow[]
    {\partial_{1}}
    \Omega_{0}
    \xrightarrow[]{}
    \{0\}.
\]  

\smallskip

\begin{definition}\label{def:path-homology}
Given a digraph $G$ and boundary operators $\partial_{p}: \Omega_{p}(G) \rightarrow \Omega_{p-1}(G)$, $p\ge 0$, the \textit{$p$-th path homology} of $G$ is defined as the quotient space
\begin{equation}\label{eq:path-homology}
H_{p}(G) =  \ker \partial_{p} / \Ima \partial_{p+1}.
\end{equation}
\end{definition}

\smallskip

Here we adopt the same nomenclature of Subsection \ref{sec:simplicial-homology} by saying that the $p$-th Betti number corresponds to the dimension of the $p$-th path homology, and we denote $\beta_{p}(G) = \dim H_{p}(G)$. Furthermore, the Euler characteristic of $G$ is defined as the alternating sum of its Betti numbers:
\begin{equation}
\chi(G) = \sum_{p=0}^{\infty} (-1)^{p} \beta_{p}(G) = \sum_{p=0}^{\infty} (-1)^{p} \dim H_{p}(G).
\end{equation}

\smallskip

\begin{observation}
A version of persistent homology for path complexes, the \textit{persistent path homology} (PPH), was introduced in \citep{Chowdhury-2018b}. Lin et al. \citep{Lin-2019} generalized the path homology of digraphs for weighted digraphs by defining weighted path homologies through weighted boundary operators defined in an analogous way to the formula (\ref{eq:weig-boundary-map}), and, in addition, they proved that the corresponding persistent weighted path homology, in the case where the coefficients of the homologies belong to a field, is independent on the weights, but it will depend on the weights if the coefficients belong to a general ring.
\end{observation}

\subsection{Combinatorial Hodge Laplacian of Path Complexes}

In Subsection \ref{sec:comb-hodge-laplacian} we have introduced the Hodge Laplacian operators for simplicial complexes. Nonetheless, one can build Hodge Laplacian operators for path complexes associated with digraphs in an analogous way \citep{Grigoryan-2022}. 

Consider an inner product in the spaces $\Lambda_{p}$, $p \ge 0$. Since $\Omega_{p}$ is a subspace of $\Lambda_{p}$, it inherits the inner product. Let $\partial_{p}: \Omega_{p} \rightarrow \Omega_{p-1}$ be the $p$-th boundary operator and $\partial_{p}^{*}: \Omega_{p-1} \rightarrow \Omega_{p}$ be its adjoint operator. The \textit{Hodge $p$-Laplacian operator}, $\mathcal{L}_{p}: \Omega_{p} \rightarrow \Omega_{p}$, is defined by
\begin{equation}\label{eq:laplacian-operator}
\mathcal{L}_{p} = \partial_{p+1} \circ \partial_{p+1}^{*} + \partial_{p}^{*} \circ \partial_{p}.
\end{equation}

The matrix representation of the Hodge $p$-Laplacian operator is given by
\begin{equation}\label{eq:hodge-laplacian-path}
[\mathcal{L}_{p}] = B_{p+1}B_{p+1}^{T} + B_{p}^{T}B_{p},
\end{equation}

\noindent where $B_{p} = [\partial_{p}]$ is the matrix representation of the operator $\partial_{p}$.

\smallskip

\begin{observation}
We point out that all additional observations from Subsection \ref{sec:comb-hodge-laplacian}  made for the Hodge Laplacians associated with simplicial complexes are equally valid for the case when they are obtained from path complexes.
\end{observation}

\begin{example}\label{ex:path-hodge}
Consider the digraph $G$ as depicted in Figure \ref{laplacian1}. The higher-order incidence matrices $B_{p}$ are:
$$
B_{1} =  
\bbordermatrix{
  & e_{01} & e_{02} & e_{12} & e_{13} & e_{23} \cr
e_{0}  & 1 & 1 & 0 &  0 & 0 \cr
e_{1}  & -1 & 0 & 1 & 1 & 0 \cr
e_{2}  & 0 & -1 & -1 & 0 & 1 \cr
e_{3}  & 0 & 0 & 0 & -1 & -1 \cr
}, \hspace{0.1in}
B_{2} =  
\bbordermatrix{
  & e_{012} & e_{123} & (e_{013} - e_{023}) \cr
e_{01}  & 1 & 0 & 1 \cr
e_{02}  & -1 & 0 & -1 \cr
e_{12}  & 1 & 1 & 0 \cr 
e_{13}  & 0 & -1 & 1 \cr
e_{23}  & 0 & 1 & -1 \cr
},
$$
$$
B_{3} =  
\bbordermatrix{
 & e_{0123} \cr
e_{012}  & -1 \cr
e_{123}  & 1  \cr
(e_{013} - e_{023})  & 1\cr 
}.
$$

\smallskip

Note that $B_{p} = 0$ for all $p \ge 4$. Thus, by Equation (\ref{eq:hodge-laplacian-path}), the matrix representations of the Hodge $p$-Laplacians of $G$ are: 
$$
[\mathcal{L}_{0}] = B_{1}B_{1}^{T} = \begin{bmatrix}
2 & -1 & -1 & 0\\
-1 & 3 & -1 & -1 \\
-1 & -1 & 3 & -1 \\
 0 & -1 & -1 & 2 
\end{bmatrix}, \hspace{0.1in}
[\mathcal{L}_{1}] = B_{2}B_{2}^{T} + B_{1}^{T}B_{1} = 
\begin{bmatrix}
4 & -1&  0&  0& -1\\
-1 &  4&  0& -1&  0\\
0&  0&  4&  0&  0\\
0& -1&  0& 4& -1\\
-1&  0&  0& -1&  4
\end{bmatrix},
$$

$$
[\mathcal{L}_{2}] =  B_{3}B_{3}^{T} + B_{2}^{T}B_{2} = \begin{bmatrix}
4&  0&  1\\
0&  4& -1\\
1& -1&  5
\end{bmatrix}, \hspace{0.1in}
[\mathcal{L}_{3}] =  B_{3}^{T}B_{3} = \begin{bmatrix}
3\\
\end{bmatrix},
$$

\noindent and  $[\mathcal{L}_{p}] = 0$ for all $p \ge 4$.
\end{example}

\smallskip

Comparing Example \ref{ex:dfc-hodge} with Example \ref{ex:path-hodge}, we can verify that the Hodge Laplacians obtained from the directed flag complex of a given digraph may differ from the Hodge Laplacians obtained from the path complex of the same digraph.


\section{Directed Q-Analysis and Directed Higher-Order Adjacencies}
\label{sec:directed-q-analysis}

In this section, we briefly present the main concepts of the classical Q-Analysis and, subsequently, we present a directed version of Q-Analysis by introducing new concepts to deal with directed higher-order connectivity between directed simplices.

\subsection{A Brief Introduction to Q-Analysis}
\label{sec:q-analysis}

The set of ideas that later came to be known as Q-Analysis was first introduced by the physicist R. H. Atkin  \citep{Atkin1974b, Atkin1974a, Atkin1976}. Atkin's initial proposal was to develop mathematical tools to analyze structures associated with relations, specially relations within social systems \citep{Atkin1977}. His idea was to model social networks via simplicial complexes and study the connections among their simplices, highlighting, therefore, the importance of the concept of topological connectivity. Sometimes Q-Analysis is referred to as a ``language of structure." Since Atkin's seminal work, several extensions and applications of Q-Analysis have been proposed \citep{Barcelo, Coombs, Kramer}, and many other new concepts were introduced \citep{Johnson}.

One of the most important concepts coming from Q-Analysis is the concept of \textit{$q$-connectivity} in a simplicial complex, which is, essentially, a notion of \textit{higher-order connectivity}, i.e. connectivity in different levels of organization (or dimensional levels) of the complex. It is important to highlight that classical Q-Analysis was developed to deal solely with the higher-order connectivity without considering any directionality in the connections. 

This section is based mainly on the references \citep{Atkin1977, Earl, Johnson}. Also, throughout this part, $\mathcal{X}$ denotes an arbitrary simplicial complex.

\smallskip

\begin{definition}\label{def:q-nearness}
Given two simplices $\sigma^{(n)}, \tau^{(m)} \in \mathcal{X}$, they are called \textit{$q$-near}, with $0 \le q \le \min(n,m)$, if they share a $q$-face, and in this case we denote $\sigma^{(n)} \sim_{q} \tau^{(m)}$. In particular, we say that an $n$-simplex is $n$-near to itself.
\end{definition}

\begin{definition}\label{def:q-connectivity}
Given two simplices $\sigma^{(n)}, \tau^{(m)} \in \mathcal{X}$, they are called \textit{$q$-connected} if there exists a finite number of simplices $\alpha_{i}^{(n_{i})} \in \mathcal{X}$, let's put $i=1,...,l$, with $0 \le q \le \min(n,m,n_{1},...,n_{l})$, such that
\begin{equation}
\sigma^{(n)} \sim_{q_{0}} \alpha_{1}^{(n_{1})} \sim_{q_{1}} ...  \sim_{q_{l-1}} \alpha_{l}^{(n_{l})} \sim_{q_{l}} \tau^{(m)},
\end{equation}

\noindent where $q \le q_{j}$, for all $j=0,...,l$, and in this case we denote\footnote{As a convention, here we use the symbol $\sim_{q}$ to denote $q$-nearness and its bold version $\bm{\sim}_{\bm{q}}$ to denote $q$-connectivity.} $\sigma^{(n)} \bm{\sim}_{\bm{q}} \tau^{(m)}$. We call this sequence a \textit{chain of $q$-connection}. Also, we say that $\sigma^{(n)}$ and $\tau^{(m)}$ are $q$-connected by a \textit{chain} of length $l$. In particular, an $n$-simplex is said to be $n$-connected to itself by a chain of length $0$. 
\end{definition}

\begin{example}
Consider the simplicial complex shown in Figure \ref{fig:example-q-conn1} (which corresponds to the flag complex of the underlying graph). The simplices $\sigma$ and $\tau$ are $0$-connected by a chain of length $4$: $\sigma \sim_{1} \alpha_{1} \sim_{1} \alpha_{2} \sim_{1} \alpha_{3} \sim_{1} \alpha_{4} \sim_{0} \tau$. At the same time, they are $1$-connected  by a chain of length $5$: $\sigma \sim_{1} \alpha_{1} \sim_{1} \alpha_{5} \sim_{1} \alpha_{6} \sim_{1} \alpha_{7} \sim_{1} \alpha_{8} \sim_{1} \tau$.
\end{example}

\smallskip

Notice that $q$-connectivity does not imply $q$-nearness but the converse is true: two $q$-near simplices are $q$-connected (by a chain of length $0$) but two $q$-connected simplices may not be $q$-near. That is, $q$-nearness is a particular case of $q$-connectivity. Also, by definition, if two simplices are $q$-connected, then they are $q'$-connected for all $q' < q$. 

\begin{definition}
Given a simplicial complex $\mathcal{X}$, we denote by $\mathcal{X}_{q}$ the set of simplices in $\mathcal{X}$ with dimension $\ge q $, i.e.
\begin{equation}
\mathcal{X}_{q} = \{\sigma^{(n)} \in \mathcal{X} : q \le n \}.
\end{equation}
\end{definition}

\smallskip

A notable property of $q$-connectivity is that it is an \textit{equivalence relation} on $\mathcal{X}_{q}$.

\smallskip

\begin{proposition}
\textit{Given a simplicial complex $\mathcal{X}$, the relation  ``is $q$-connected to" ($\bm{\sim}_{\bm{q}}$) is an equivalence relation on the set $\mathcal{X}_{q}$.}
\end{proposition}

\smallskip

The proof of the previous proposition is analogous to the proof of the Proposition \ref{prop:connected}. Moreover, we define the \textit{$q$-connected components} of $\mathcal{X}$ as the elements of the quotient set $\mathcal{X}_{q}/\bm{\sim}_{\bm{q}}$, i.e. the equivalence classes (or \textit{$q$-connectivity classes}).

As a matter of fact, to perform a \textit{Q-analysis} on a simplicial complex means to compute its $q$-connected components for $0 \le q \le \dim \mathcal{X}$, and then summarize the number of these components existing at each level $q$ into a \textit{structure vector}.

\smallskip

\begin{definition}\label{def:structure-vector}
Let $f_{q}$ be the number of $q$-connected components of $\mathcal{X}$ for $0 \le q \le \dim \mathcal{X}$. The \textit{structure vector} (or \textit{first structure vector}) of $\mathcal{X}$ is the tuple

\begin{equation}\label{eq:structure-vector}
\mathbf{f}(\mathcal{X}) = (f_{0},..., f_{\dim \mathcal{X}}).
\end{equation}
\end{definition}

\smallskip

We point out that the nomenclature ``first structure vector" comes from the fact that there are other structure vectors defined for simplicial complexes, such as the ``second and third structure vectors" \citep{Andjelkovic2015}, but they'll be discussed in the next section. Also, since $f_{0}$ denotes the number of $0$-connected components of $\mathcal{X}$, it is equal to the $0$-th Betti number, i.e. $f_{0} = \beta_{0}$.


\smallskip

\begin{example}
Let $\mathcal{X}$ be the simplicial complex shown in Figure \ref{fig:example-q-conn1}. The corresponding structure vector is $\mathbf{f}(\mathcal{X}) = (2,2,11,1)$.
\end{example}

\smallskip

One way to summarize the connectivity between the simplices of a simplicial complex at a given level $q$ is through the idea of \textit{$q$-graph}.

\smallskip

\begin{definition}\label{def:q-graph}
The \textit{$q$-graph} of a simplicial complex $\mathcal{X}$, $0 \le q \le \dim \mathcal{X}$, is a graph in which every vertex corresponds to a simplex in $\mathcal{X}_{q}$ and there exists an edge between two vertices if and only if the simplices corresponding to these vertices are $q$-near.
\end{definition}

\smallskip

Now we present some concepts that are not only related to $q$-connectivity but also related to the complex itself.

\smallskip

\begin{definition}\label{def:stars}
The \textit{$q$-star} of a simplex $\sigma^{(n)} \in \mathcal{X}$, $0 \le q \le n$, is defined as the set of all simplices that are $q$-near with $\sigma^{(n)}$, i.e.
\begin{equation}\label{eq:star}
\mathrm{st}_{q}(\sigma^{(n)}) = \{ \tau^{(m)} \in \mathcal{X} : \sigma^{(n)} \sim_{q} \tau^{(m)} \}.
\end{equation}

When $q = n$, the $n$-star is the set of all the simplices having $\sigma^{(n)}$ as a face, and in this case we use the notation $\mathrm{st}^{*}(\sigma^{(n)}) = \mathrm{st}_{n}(\sigma^{(n)})$. Also, if $\mathcal{F}$ is a simplicial family obtained from $\mathcal{X}$, the set of all simplices that have at least one face in $\mathcal{F}$ is given by 
\begin{equation}
\mathrm{st}^{*}(\mathcal{F}) = \bigcup_{\sigma \in \mathcal{F}} \mathrm{st}^{*}(\sigma).
\end{equation}
\end{definition}

\begin{definition}\label{def:simp-hub}
Given a simplicial family $\mathcal{F}$, the \textit{hub} of $\mathcal{F}$ is the set formed by all the simplices that are common faces of the elements of $\mathcal{F}$, i.e.
\begin{equation}\label{eq:simp-hub}
\mathrm{hub}(\mathcal{F}) = \bigcap_{\sigma \in \mathcal{F}} \sigma.
\end{equation}
\end{definition}

\smallskip


A generalization for simplicial complexes of the idea of neighborhood of a node in a graph is the concept of \textit{link} of a simplex.

\smallskip

\begin{definition}\label{def:link}
The \textit{link} of a simplex $\sigma \in \mathcal{X}$ is defined as the set of all simplices $\tau$ such that $\sigma$ and $\tau$ are disjoint faces of the simplex $\sigma \cup \tau$, i.e.
\begin{equation}\label{eq:link}
\mathrm{lk}(\sigma) = \{ \tau \in \mathcal{X} : \sigma \cap \tau = \emptyset \mbox{ and } \sigma \cup \tau \in \mathcal{X} \}.
\end{equation}
\end{definition}

\begin{example}
Considering the simplicial complex presented in Figure \ref{fig:example-q-conn1}, we have the following examples: the $1$-star of the simplex $\alpha_{1}$ is $\mathrm{st}_{1}(\alpha_{1}) = \{ \sigma, \alpha_{2}, \alpha_{5}  \}$; the hub of the simplicial family $\mathcal{F} = \{ \alpha_{1},\alpha_{2}, \alpha_{3}, \alpha_{5}, \alpha_{6} \}$ is $\mbox{hub}(\mathcal{F}) = \{ 3 \}$; the link of the $1$-simplex $\{12, 13\}$ is the $1$-simplex $\{10, 11\}$.
\end{example}

\begin{figure}[h!]
\centering
  \includegraphics[scale=0.8]{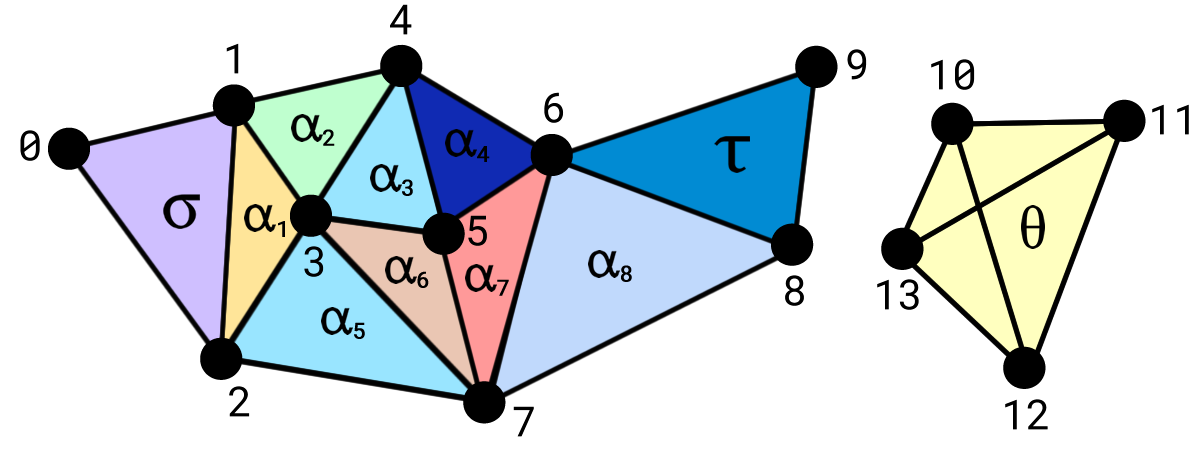}
\caption{A simplicial complex.}
\label{fig:example-q-conn1}
\end{figure}

\subsection{Directed Q-Analysis and Directed Higher-Order Adjacencies}
\label{sec:dir-q-analysis}

As we have seen in the previous section, Atkin's Q-Analysis defines $q$-connectivity between two simplices based solely on the face shared by them and does not say anything about the directionality of this connection. Recently, however, H. Riihimaki \citep{Riihimaki} introduced a formalism to treat the $q$-connectivity between directed simplices taking directionality into account, thus creating a directed analogue of Q-Analysis for directed flag complexes. 

Moreover, unlike adjacencies between vertices in a graph, when we are dealing with simplices, we can distinguish between two types of adjacencies: \textit{lower} and \textit{upper} adjacencies. Lower adjacencies compare how two simplices share their faces, and upper adjacencies tell us how they are nested in other higher-dimensional simplices \citep{Estrada, Goldberg, Serrano2020}. Accordingly, we need appropriate definitions of lower and upper adjacencies for directed simplices, so that the directionality of the connection between them is taken into account, and this is the aim of this section.

In what follows, we present the concept of \textit{$(q, \hat{d_{i}}, \hat{d_{j}})$-connectivity} as introduced in \citep{Riihimaki}, and then we extend the concepts of lower, upper, and general adjacencies as presented in \citep{Serrano2020} to directed simplices. Throughout this part, $\mathrm{dFl}(G)$ will denote the directed flag complex of a given simple digraph $G$ \textit{without double edges}.

\subsubsection{$(q, \hat{d_{i}}, \hat{d_{j}})$-Connectivity Between Directed Simplices}

The direction of the connection between two directed simplices is based on a slightly modified face map as defined below.

\smallskip

\begin{definition}\label{def:mod-face-map}
Given a directed simplex $\sigma^{(n)} = [v_{0},...,v_{n}] \in \mathrm{dFl}(G)$, the $i$-th face map $\hat{d_{i}}$ is defined as 
\begin{equation}\label{eq:face-map}
\hat{d}_{i}(\sigma^{(n)}) = \begin{cases}
[v_{0},...,\hat{v}_{i},...,v_{n}], \mbox{ if } i < n,\\
[v_{0},...,v_{n-1}, \hat{v}_{n}], \mbox{ if } i \ge n.
\end{cases}
\end{equation}

\end{definition}

\smallskip

\begin{definition}\label{def:qij-nearness}
Let $\sigma^{(n)}, \tau^{(m)} \in \mathrm{dFl}(G)$ be two directed simplices. For $0 \le q \le \min(n,m)$, we say that $\sigma^{(n)}$ is \textit{$(q, \hat{d}_{i},\hat{d}_{j})$-near} to $\tau^{(m)}$ if either of the following conditions is satisfied:
\begin{enumerate}
\item $\sigma^{(n)} \subseteq \tau^{(m)}$;
\item $\hat{d_{i}}(\sigma^{(n)}) \supseteq \alpha^{(q)} \subseteq \hat{d_{j}}(\tau^{(m)})$, for some $\alpha^{(q)} \in \mathrm{dFl}(G)$ (i.e. if they share a $q$-face).
\end{enumerate}

\end{definition}

\smallskip

The previous definition established a concept of directionality in the connection between two directed simplices based on \textit{how} their faces are shared. However, a more precise and descriptive definition/notation is needed to make analogies, for the higher-order setting, with some definitions made for directed networks. Accordingly, we propose the following definitions/notations.

\smallskip

\begin{definition}\label{def:qij-nearness-new}
For two directed simplices  $\sigma^{(n)}, \tau^{(m)} \in \mathrm{dFl}(G)$ and for $0 \le q \le \min(n,m)$, we have the following definitions:

\begin{enumerate}
\item $\sigma^{(n)}$ is said to be \textit{in-$q$-near} (or \textit{$(-)$-$q$-near}) to  $\tau^{(m)}$ if they are $(q, d_{i}(\sigma^{(n)}), d_{j}(\tau^{(m)}))$-near with $i \ge j$, for at least one pair $i,j$. In this case, we denote $\sigma^{(n)} \sim_{q}^{-} \tau^{(m)}$.

\item $\sigma^{(n)}$ is said to be \textit{out-$q$-near} (or \textit{$(+)$-$q$-near}) to $\tau^{(m)}$ if they are $(q, d_{i}(\sigma^{(n)}), d_{j}(\tau^{(m)}))$-near with $i \le j$, for at least one pair $i,j$. In this case, we denote $\sigma^{(n)} \sim_{q}^{+} \tau^{(m)}$.

\item $\sigma^{(n)}$ is said to be \textit{bidirectionally $q$-near} (or \textit{$(\pm)$-$q$-near}) to  $\tau^{(m)}$ if $\sigma^{(n)} \sim_{q}^{-} \tau^{(m)}$ and $\sigma^{(n)} \sim_{q}^{+} \tau^{(m)}$. In this case, we denote $\sigma^{(n)} \sim_{q}^{\pm} \tau^{(m)} = \tau^{(m)}  \sim_{q}^{\pm} \sigma^{(n)}$.

\end{enumerate}
\end{definition}

\smallskip

It's clear from the definitions that $\sigma^{(n)} \sim_{q}^{-} \tau^{(m)} = \tau^{(m)} \sim_{q}^{+} \sigma^{(n)} $ and $\sigma^{(n)} \sim_{q}^{+} \tau^{(m)} = \tau^{(m)} \sim_{q}^{-} \sigma^{(n)} $. Also, if $\sigma^{(n)} \subseteq \tau^{(m)}$, by definition, $\sigma^{(n)} \sim_{q}^{\pm} \tau^{(m)}$.

\bigskip

\noindent {\bf Notation:} We use the notation $\sigma^{(n)} \sim_{q}^{\bullet} \tau^{(m)}$, where $\bullet \in \{-,+,\pm \}$, to represent the respective connectivity relation by replacing $\bullet$ with its respective symbol.

\bigskip

Now that we have introduced the concept of $(q, \hat{d}_{i},\hat{d}_{j})$-nearness and established the notations, let's introduce the concept of \textit{$(q, \hat{d_{i}}, \hat{d_{j}})$-connectivity} using the new notations.

\smallskip

\begin{definition}\label{def:dir-q-connectivity}
Given two directed simplices $\sigma^{(n)}, \tau^{(m)} \in \mathrm{dFl}(G)$, we say that $\sigma^{(n)}$ is \textit{$(\bullet)$-$q$-connected} to $\tau^{(m)}$, where $\bullet \in \{-,+\}$, if there exists a finite number of simplices $\alpha_{i}^{(n_{i})} \in \mathrm{dFl}(G)$, let's put $i=1,...,l$, with $0 \le q \le \min(n,m,n_{1},...,n_{l})$, such that
\begin{equation}
\sigma^{(n)} \sim_{q_{0}}^{\bullet} \alpha_{1}^{(n_{1})} \sim_{q_{1}}^{\bullet}  ...  \sim_{q_{l-1}}^{\bullet}  \alpha_{l}^{(n_{l})} \sim_{q_{l}}^{\bullet}  \tau^{(m)},
\end{equation}

\noindent where $q \le q_{j}$, for all $j=0,...,l$, and in this case we denote $\sigma^{(n)} \bm{\sim}_{\bm{q}}^{\bullet} \tau^{(m)}$. We say that $\sigma^{(n)}$ is $(\bullet)$-$q$-connected to $\tau^{(m)}$ by a \textit{directed $(\bullet)$-$q$-chain} of length $l$. We denote $\sigma^{(n)} \bm{\sim}_{\bm{q}}^{\pm} \tau^{(m)}$ if $\sigma^{(n)} \bm{\sim}_{\bm{q}}^{+} \tau^{(m)}$ and $\sigma^{(n)} \bm{\sim}_{\bm{q}}^{-} \tau^{(m)}$ and we say they are \textit{$(\pm)$-$q$-connected}. In particular, a directed $n$-simplex is said to be $(\pm)$-$q$-connected to itself by a directed $(\pm)$-$q$-chain of length $0$. Finally, for $\bullet \in \{-,+\}$, if $\sigma^{(n)}$ is not $(\bullet)$-$q$-connected to $\tau^{(m)}$ for all $\bullet$, we say that they are \textit{$q$-disconnected}.
\end{definition}

\smallskip

Notice that, just as in the case of $q$-connectivity, if two directed simplices are $(\bullet)$-$q$-near, $\bullet \in \{-,+\}$, then they are $(\bullet)$-$q$-connected by a directed $(\bullet)$-chain of length $0$. Moreover, if two directed simplices are $(\bullet)$-$q$-connected, then they are $(\bullet)$-$q'$-connected for all $q' < q$.

\smallskip

\begin{example}
Consider the directed flag complex shown in Figure \ref{fig:dfc-example}. We have the following relations:

\smallskip

1) $\theta \sim_{0}^{\pm} \tau$, since $\hat{d}_{0}(\theta) \supseteq  [2] \subseteq \hat{d}_{2}(\tau)$ and $\hat{d}_{2}(\theta) \supseteq  [2] \subseteq \hat{d}_{1}(\tau)$.

2) $\sigma \sim_{0}^{\pm} \tau$, since $\hat{d}_{1}(\sigma) \supseteq  [2] \subseteq \hat{d}_{2}(\tau)$ and $\hat{d}_{2}(\sigma) \supseteq  [2] \subseteq \hat{d}_{1}(\tau)$.

3) $\theta \sim_{1}^{+} \sigma$, since $\hat{d}_{0}(\theta) \supseteq  [2,6] \subseteq \hat{d}_{2}(\sigma)$.

4) $\alpha \sim_{1}^{+} \theta$, since $\hat{d}_{0}(\alpha) \supseteq  [1,6] \subseteq \hat{d}_{1}(\theta)$.

5) $\alpha \bm{\sim}_{\bm{0}}^{+} \tau$, since $\alpha \sim_{1}^{+} \theta \sim_{0}^{+} \tau$.

\smallskip

\begin{center}
\begin{tabular}{ c }
$\hat{d}_{i}(\alpha)$ \\
\hline
$\hat{d}_{0}(\alpha) = [1,6]$ \\
$\hat{d}_{1}(\alpha) = [0,6]$ \\
$\hat{d}_{2}(\alpha) = [0,1]$ \\
\end{tabular}
\quad
\begin{tabular}{ c }
$\hat{d}_{i}(\theta)$ \\
\hline
$\hat{d}_{0}(\theta) = [2,6]$ \\
$\hat{d}_{1}(\theta) = [1,6]$ \\
$\hat{d}_{2}(\theta) = [1,2]$ \\
\end{tabular}
\quad
\begin{tabular}{ c }
$\hat{d}_{i}(\sigma)$ \\
\hline
$\hat{d}_{0}(\sigma) = [6,7]$ \\
$\hat{d}_{1}(\sigma) = [2,7]$ \\
$\hat{d}_{2}(\sigma) = [2,6]$ \\
\end{tabular}
\quad
\begin{tabular}{ c }
$\hat{d}_{i}(\tau)$ \\
\hline
$\hat{d}_{0}(\tau) = [3,4,5]$ \\
$\hat{d}_{1}(\tau) = [2,4,5]$ \\
$\hat{d}_{2}(\tau) = [2,3,5]$ \\
$\hat{d}_{3}(\tau) = [2,3,4]$ \\
\end{tabular}
\end{center}

\begin{center}
\begin{tabular}{ c }
$(q, \hat{d}_{i}(\theta), \hat{d}_{j}(\tau))$ \\
\hline
$(0, \hat{d}_{0}, \hat{d}_{1})$ \\
$(0, \hat{d}_{0}, \hat{d}_{2})$ \\
$(0, \hat{d}_{0}, \hat{d}_{3})$ \\
$(0, \hat{d}_{2}, \hat{d}_{1})$ \\
$(0, \hat{d}_{2}, \hat{d}_{2})$ \\
$(0, \hat{d}_{2}, \hat{d}_{3})$ \\
\end{tabular}
\quad
\begin{tabular}{ c }
$(q, \hat{d}_{i}(\sigma), \hat{d}_{j}(\tau))$ \\
\hline
$(0, \hat{d}_{1}, \hat{d}_{1})$ \\
$(0, \hat{d}_{1}, \hat{d}_{2})$ \\
$(0, \hat{d}_{1}, \hat{d}_{3})$ \\
$(0, \hat{d}_{2}, \hat{d}_{1})$ \\
$(0, \hat{d}_{2}, \hat{d}_{2})$ \\
$(0, \hat{d}_{2}, \hat{d}_{3})$ \\
\end{tabular}
\quad
\begin{tabular}{ c }
$(q, \hat{d}_{i}(\theta), \hat{d}_{j}(\sigma))$ \\
\hline
$(0, \hat{d}_{0}, \hat{d}_{0})$ \\
$(0, \hat{d}_{0}, \hat{d}_{1})$ \\
$(0, \hat{d}_{0}, \hat{d}_{2})$ \\
$(0, \hat{d}_{1}, \hat{d}_{0})$ \\
$(0, \hat{d}_{1}, \hat{d}_{2})$ \\
$(0, \hat{d}_{2}, \hat{d}_{1})$ \\
$(0, \hat{d}_{2}, \hat{d}_{2})$ \\
$(1, \hat{d}_{0}, \hat{d}_{2})$ \\
\end{tabular}
\quad
\begin{tabular}{ c }
$(q, \hat{d}_{i}(\alpha), \hat{d}_{j}(\theta))$ \\
\hline
$(0, \hat{d}_{0}, \hat{d}_{0})$ \\
$(0, \hat{d}_{0}, \hat{d}_{1})$ \\
$(0, \hat{d}_{0}, \hat{d}_{2})$ \\
$(0, \hat{d}_{1}, \hat{d}_{0})$ \\
$(0, \hat{d}_{1}, \hat{d}_{1})$ \\
$(0, \hat{d}_{2}, \hat{d}_{1})$ \\
$(0, \hat{d}_{2}, \hat{d}_{2})$ \\
$(1, \hat{d}_{0}, \hat{d}_{1})$ \\
\end{tabular}
\end{center}

\begin{figure}[h!]
\centering
  \includegraphics[scale=0.8]{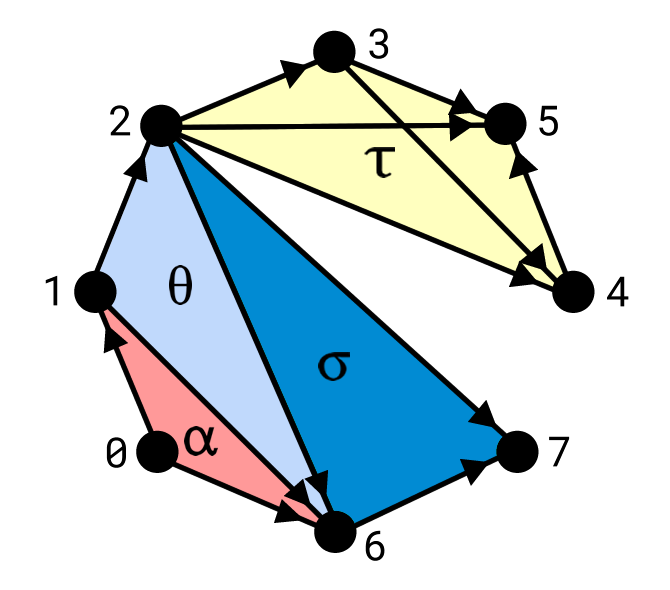}
\caption{A directed flag complex.}
\label{fig:dfc-example}
\end{figure}

\end{example}

\subsubsection{Lower, Upper, and General Adjacencies}

In what follows, we extend the definitions of lower, upper, and general adjacencies as exposed in \citep{Serrano2020} to directed simplices.

\smallskip

\begin{definition}
For two directed simplices $\sigma^{(n)}, \tau^{(m)} \in \mathrm{dFl}(G)$ and for $0 \le q \le \min(n,m)$, we have the following definitions:

\begin{enumerate}
\item $\sigma^{(n)}$ is \textit{lower $(\bullet)$-$q$-adjacent} to  $\tau^{(m)}$, where $\bullet \in \{-, +, \pm\}$, if and only if $\sigma^{(n)}$ is $(\bullet)$-q-near to $\tau^{(m)}$ , i.e.
$$
 \sigma^{(n)} \sim_{L_{q}}^{\bullet} \tau^{(m)} \iff \sigma^{(n)}  \sim_{q}^{\bullet} \tau^{(m)}.
$$

\item $\sigma^{(n)}$ is \textit{strictly lower $(\bullet)$-$q$-adjacent} to $\tau^{(m)}$, where $\bullet \in \{-, +, \pm\}$, if and only if $\sigma^{(n)} \sim_{L_{q}}^{\bullet} \tau^{(m)}$ and $\sigma^{(n)}$ is not ($\star$)-$(q+1)$-near to $\tau^{(m)}$, for all $\star \in \{-,+,\pm\}$, i.e.
$$
 \sigma^{(n)} \sim_{L_{q^{*}}}^{\bullet} \tau^{(m)} \iff \sigma^{(n)} \sim_{L_{q}}^{\bullet} \tau^{(m)} \mbox{ and }  \sigma^{(n)} \not\sim_{L_{q+1}}^{\star} \tau^{(m)}, \mbox{ } \forall \star.
$$
\end{enumerate}
\end{definition}

\smallskip

We point out that lower $(\bullet)$-$q$-adjacency and $(\bullet)$-$q$-nearness are exactly the same definitions, thus we propose \textit{lower $(\bullet)$-$q$-nearness} as an alternative nomenclature of $(\bullet)$-$q$-nearness. Also, note that for vertices they are only lower $(\bullet)$-$0$-adjacent to themselves.

Besides the lower adjacency, that is, how two simplices share their faces, we could also think about how simplices are nested in other simplices of greater dimensions, and for this, we have the idea of \textit{upper adjacency}. To extend upper adjacency to directed simplices, we need to define precisely how the directionality is taken into account between two directed simplices that are faces of other directed simplices of greater dimensions, thus, based on Definition \ref{def:qij-nearness}, we propose the definition of \textit{upper $(p, \hat{d}_{i}, \hat{d}_{j})$-nearness} as follows.

\smallskip

\begin{definition}
Let $\sigma^{(n)}, \tau^{(m)} \in \mathrm{dFl}(G)$ be two directed simplices. For $n,m < p \le \dim \mathrm{dFl}(G)$, $\sigma^{(n)}$ is said to be \textit{upper $(p, \hat{d}_{i}, \hat{d}_{j})$-near} to $\tau^{(m)}$ if the following condition is true:
$$
\sigma^{(n)} = \hat{d_{i}}(\Theta^{(n+1)}) \subseteq \Theta^{(p)} \supseteq  \hat{d_{j}}(\Theta^{(m+1)}) = \tau^{(m)}, \mbox{ for some } \Theta^{(n+1)}, \Theta^{(m+1)} \subseteq \Theta^{(p)} \in \mathrm{dFl}(G).
$$
\end{definition}

\smallskip

\begin{definition}
For two directed simplices $\sigma^{(n)}, \tau^{(m)} \in \mathrm{dFl}(G)$ and for $n,m  < p \le \dim \mathrm{dFl}(G)$, we have the following definitions:

\begin{enumerate}
\item $\sigma^{(n)}$ is \textit{upper $(-)$-$p$-adjacent} to $\tau^{(m)}$  if and only if $\sigma^{(n)}$ is upper $(p, \hat{d}_{i}, \hat{d}_{j})$-near to $\tau^{(m)}$ and $i \ge j$, for at least one pair $i,j$, i.e.
$$
 \sigma^{(n)} \sim_{U_{p}}^{-} \tau^{(m)} \iff \sigma^{(n)} \mbox{ is upper } (p, \hat{d}_{i}, \hat{d}_{j})\mbox{-near to } \tau^{(m)} \mbox{ with } i \ge j.
$$

\item $\sigma^{(n)}$ is  \textit{upper $(+)$-$p$-adjacent} to $\tau^{(m)}$  if and only if $\sigma^{(n)}$ is upper $(p, \hat{d}_{i}, \hat{d}_{j})$-near to $\tau^{(m)}$ and $i \le j$, for at least one pair $i,j$, i.e.
$$
 \sigma^{(n)} \sim_{U_{p}}^{+} \tau^{(m)} \iff \sigma^{(n)} \mbox{ is upper } (p, \hat{d}_{i}, \hat{d}_{j})\mbox{-near to } \tau^{(m)} \mbox{ with } i \le j.
$$

\item $\sigma^{(n)}$ is  \textit{upper $(\pm)$-$p$-adjacent} to $\tau^{(m)}$  if and only if $\sigma^{(n)}$ is upper $(-)$-$p$-adjacent and upper $(+)$-$p$-adjacent to $\tau^{(m)}$, i.e.
$$
 \sigma^{(n)} \sim_{U_{p}}^{\pm} \tau^{(m)} \iff \sigma^{(n)} \sim_{U_{p}}^{-} \tau^{(m)} \mbox{ and }  \sigma^{(n)} \sim_{U_{p}}^{+} \tau^{(m)}.
$$

\item $\sigma^{(n)}$ is  \textit{strictly $(\bullet)$-$p$-upper adjacent} to $\tau^{(m)}$, where $\bullet \in \{-,+,\pm\}$, if and only if $\sigma^{(n)} \sim_{U_{p}}^{\bullet} \tau^{(m)}$ and $\sigma^{(n)}$ is not upper $(\star)$-$(p+1)$-adjacent to $\tau^{(m)}$, for all $\star \in \{-,+,\pm\}$, i.e.
$$
 \sigma^{(n)} \sim_{U_{p^{*}}}^{\bullet} \tau^{(m)} \iff \sigma^{(n)} \sim_{U_{p}}^{\bullet} \tau^{(m)} \mbox{ and } \sigma^{(n)} \not\sim_{U_{p+1}}^{\star} \tau^{(m)}, \mbox{ } \forall \star.
$$
\end{enumerate}
\end{definition}

\smallskip

Note that if $(v,u)$ is an arc in the graph $G$, then $v \sim_{U_{1}}^{+} u$. On the other hand, if $(u,v)$ is an arc in $G$, then $v \sim_{U_{1}}^{-} u$.

Now that we have introduced the lower and upper adjacencies, let's define the \textit{general adjacencies}, which take both lower and upper adjacencies into account.

\smallskip

\begin{definition}\label{def:general-adj}
For two directed simplices $\sigma^{(n)}, \tau^{(m)} \in \mathrm{dFl}(G)$, we have the following definitions:

\begin{enumerate}
\item $\sigma^{(n)}$ is \textit{$(\bullet)$-$q$-adjacent} to $\tau^{(m)}$, where $\bullet \in \{-,+,\pm\}$, if and only if $\sigma^{(n)}$ is strictly lower $(\bullet)$-$q$-adjacent to $\tau^{(m)}$ and $\sigma^{(n)}$ is not upper $(\star)$-$p$-adjacent to $\tau^{(m)}$, with $p = n+m-q$, for all $\star \in \{-, +, \pm\}$, i.e.
$$
 \sigma^{(n)} \sim_{A_{q}}^{\bullet} \tau^{(m)} \iff \sigma^{(n)} \sim_{L_{q^{*}}}^{\bullet} \tau^{(m)}  \mbox{ and } \sigma^{(n)} \not\sim^{\star}_{U_{p}} \tau^{(m)}, \mbox{ } \forall \star.
$$

\item $\sigma^{(n)}$ is \textit{maximal $(\bullet)$-$q$-adjacent} to $\tau^{(m)}$, where $\bullet \in \{-,+,\pm\}$, if and only if  $\sigma^{(n)}$ is $(\bullet)$-$q$-adjacent to $\tau^{(m)}$ and $\sigma^{(n)}$ is not a face of any other directed simplex which is $(\star)$-$q$-adjacent to $\tau^{(m)}$,  for all $\star \in \{-, +, \pm\}$, i.e.
$$
\sigma^{(n)} \sim_{A_{q^{*}}}^{\bullet} \tau^{(m)} \iff \sigma^{(n)} \sim_{A_{q}}^{\bullet} \tau^{(m)} \mbox{ and } \sigma^{(n)} \not\subset \sigma^{(r)}, \mbox{ } \forall \sigma^{(r)} \mbox{ : } \sigma^{(r)} \sim_{A_{q}}^{\star} \tau^{(m)}, \mbox{ } \forall \star.
$$
\end{enumerate}
\end{definition}

\begin{observation}
The quantity $p=m+n-q$ comes from the fact that if $\sigma^{(n)} \sim_{L_{q^{*}}}^{\bullet} \tau^{(m)}$, then they share $(q+1)$ vertices and thus the smallest directed simplex which might contain $\sigma^{(n)}$ and $\tau^{(m)}$ as faces must have $(n+1)+(m+1) - (q+1)$ vertices, i.e. must have a dimension equal to $n + m - q$.
\end{observation}

\smallskip

As will be proved in the next proposition, if a maximal directed simplex is strictly lower $(\bullet)$-$q$-adjacent to another maximal directed simplex, then this adjacency is actually a maximal $(\bullet)$-$q$-adjacency. Let's first introduce some useful definitions. 

\smallskip

\begin{definition}\label{def:set-maximal-simplices}
The set of all maximal directed simplices of $\mathrm{dFl}(G)$ is defined by
\begin{equation}
\mathrm{dFl}^{\ast}(G) = \{ \sigma^{(n)} \in \mathrm{dFl}(G) : \sigma^{(n)} \mbox{ is maximal } \}.
\end{equation}

The set of all maximal directed simplices of $\mathrm{dFl}(G)$ whose dimensions are greater than or equal to $0 \le q \le \dim \mathrm{dFl}(G)$ is defined by
\begin{equation}
\mathrm{dFl}^{\ast}_{q}(G) = \{ \sigma^{(n)} \in \mathrm{dFl}^{\ast}(G) : q \le n \}.
\end{equation}
\end{definition}

\begin{proposition}\label{prop:maximal-adjacency}
\textit{For two maximal directed simplices $\sigma^{(n)}, \tau^{(m)} \in \mathrm{dFl}^{\ast}(G)$ and for $\bullet \in \{-,+,\pm \}$, we have the following equivalence:}
$$
\sigma^{(n)} \sim^{\bullet}_{L_{q^{*}}} \tau^{(m)} \iff \sigma^{(n)} \sim^{\bullet}_{A_{q^{*}}} \tau^{(m)}.
$$
\end{proposition}

\begin{proof}
Suppose $\sigma^{(n)} \sim^{\bullet}_{L_{q^{*}}} \tau^{(m)}$. Since $\sigma^{(n)}, \tau^{(m)} \in \mathrm{dFl}^{\ast}(G)$, both simplices are not faces of any other simplices in $\mathrm{dFl}(G)$, then $\sigma^{(n)} \sim^{\bullet}_{A_{q^{*}}} \tau^{(m)}$. On the other hand, if $\sigma^{(n)} \sim^{\bullet}_{A_{q^{*}}} \tau^{(m)}$, by definition, $\sigma^{(n)} \sim^{\bullet}_{L_{q^{*}}} \tau^{(m)}$.
\end{proof}

\smallskip

\subsubsection{Directed Simplicial $q$-Walks and $q$-Distances} 

In Definition \ref{def:dir-q-connectivity}, we implicitly defined $(\bullet)$-$q$-connectivity in terms of \textit{lower} $(\bullet)$-$q$-adjacency; therefore, from now on, we will refer to this connectivity as \textit{lower $(\bullet)$-$q$-connectivity}. In what follows, we introduce the concept of \textit{maximal $(\bullet)$-$q$-connectivity}, which is the basis for defining the idea of \textit{maximal directed simplicial $q$-walk}.

\smallskip

\begin{definition}\label{def:maximal-q-connectivity}
Given two directed simplices $\sigma^{(n)}, \tau^{(m)} \in \mathrm{dFl}(G)$, we say that $\sigma^{(n)}$ is \textit{maximal $(\bullet)$-$q$-connected} to $\tau^{(m)}$, where $\bullet \in \{-,+\}$, if there exists a finite number of simplices $\alpha_{i}^{(n_{i})} \in \mathrm{dFl}(G)$, let's put $i=1,...,l$, with $0 \le q \le \min(n,m,n_{1},...,n_{l})$, such that
\begin{equation}
\sigma^{(n)} \sim_{A_{q_{0}^{*}}}^{\bullet} \alpha_{1}^{(n_{1})} \sim_{A_{q_{1}^{*}}}^{\bullet}  ...  \sim_{A_{q_{l-1}^{*}}}^{\bullet}  \alpha_{l}^{(n_{l})} \sim_{A_{q_{l}^{*}}}^{\bullet}  \tau^{(m)},
\end{equation}

\noindent where $q \le q_{j}$, for all $j=0,...,l$, and in this case we denote $\sigma^{(n)} \bm{\sim}_{\bm{A_{q^{*}}}}^{\bullet} \tau^{(m)}$. We say that there is a \textit{maximal directed simplicial $q$-walk} of length $l$ from $\sigma^{(n)}$ to $\tau^{(m)}$ if $\sigma^{(n)} \bm{\sim}_{\bm{A_{q^{*}}}}^{+} \tau^{(m)}$ or from $\tau^{(m)}$ to $\sigma^{(n)}$ if $\sigma^{(n)} \bm{\sim}_{\bm{A_{q^{*}}}}^{-} \tau^{(m)}$. We denote $\sigma^{(n)} \bm{\sim}_{\bm{A_{q^{*}}}}^{\pm} \tau^{(m)}$ if $\sigma^{(n)} \bm{\sim}_{\bm{A_{q^{*}}}}^{+} \tau^{(m)}$ and $\sigma^{(n)} \bm{\sim}_{\bm{A_{q^{*}}}}^{-} \tau^{(m)}$. Moreover, for $\bullet \in \{-,+, \pm\}$, if $\sigma^{(n)}$ is not maximal $(\bullet)$-$q$-connected to $\tau^{(m)}$ for all $\bullet$,  we say that they are \textit{maximal $q$-disconnected}.
\end{definition}

\begin{remark}\label{rem:lower_case}
If we replace the maximal adjacency $A_{q^{*}}$ by the lower adjacency $L_{q}$ throughout Definition~\ref{def:maximal-q-connectivity}, we obtain the analogous notion of \textit{lower $(\bullet)$-$q$-connectivity}. All remaining conventions (directed simplicial $q$-walks, the $(-, +, \pm)$ notations, and $q$-disconnectedness) carry over with $A_{q^{*}}$ replaced by $L_{q}$ and  ``maximal'' replaced by ``lower'' throughout.
\end{remark}

\smallskip   
    
Notice that if $\sigma \bm{\sim}_{\bm{A_{q^{*}}}}^{+} \tau$ (respec. $\bm{L_{q}}$) with $\sigma = \tau$, then we have a \textit{maximal} (respec. \textit{lower}) \textit{directed simplicial $q$-cycle}.

\smallskip

\begin{definition}
Given two directed simplices $\sigma, \tau \in \mathrm{dFl}(G)$, the (\textit{lower}) \textit{maximal directed simplicial $q$-distance} from $\sigma$ to $\tau$, denoted by $\vec{d}_{q}(\sigma, \tau)$, is equal to the length of the shortest (lower) maximal directed simplicial $q$-walk from $\sigma$ to $\tau$. If $\sigma$ and $\tau$ are $q$-disconnected, then we define $\vec{d}_{q}(\sigma, \tau) = \infty$.
\end{definition}

\begin{remark}
When not explicitly stated, we may use the notations $\vec{d}_{q}^{L}(\sigma, \tau)$ and $\vec{d}_{q}^{A}(\sigma, \tau)$ to distinguish between the lower and maximal cases.
\end{remark}
 
\smallskip

As a matter of fact, $\vec{d}_{q}$ is a quasi-distance (see Definition \ref{def:metric}), since the property of symmetry is not necessarily satisfied.

\subsubsection{Weakly and Strongly $q$-Connected Components and Structure Vectors}

As already commented in Subsection \ref{sec:q-analysis}, performing a Q-analysis on a simplicial complex consists of calculating its $q$-connected components and then constructing its structure vector. However, similarly when we are dealing with digraphs, for directed flag complexes we have two types of ``$q$-connected components," namely: \textit{weakly $q$-connected components} and \textit{strongly $q$-connected components}. Before defining these kinds of ``$q$-connected components," let's define one more subset of $\mathrm{dFl}(G)$.

\smallskip

\begin{definition}\label{def:set-q-simplices}
The set of the directed simplices of $\mathrm{dFl}(G)$ whose dimension is greater or equal to $0 \le q \le \dim \mathrm{dFl}(G)$ is defined by
\begin{equation}
\mathrm{dFl}_{q}(G) = \{ \sigma^{(n)} \in \mathrm{dFl}(G) : q \le n \}.
\end{equation}
\end{definition}

\smallskip

In analogy with what was exposed in \citep{Serrano2020}, a directed simplex is not maximal $(\bullet)$-$q$-connected to itself, $\bullet \in \{-, +, \pm\}$, thus this relation is not reflexive. Accordingly, in order to obtain an equivalence relation, we introduce the following relation:
\begin{equation}\label{eq:strong-q-connected-rel}
(\sigma^{(n)}, \tau^{(m)}) \in S^{s}_{q} \iff \begin{cases}
\sigma^{(n)} \bm{\sim}_{\bm{A_{q^{*}}}}^{\pm} \tau^{(m)}, \\
\mbox{or } \sigma^{(n)} = \tau^{(m)}.
\end{cases}
\end{equation}

One can verify that (\ref{eq:strong-q-connected-rel}) is indeed an equivalence relation by following the steps of the proof of Proposition \ref{prop:wcc-scc}.

\smallskip

\begin{definition}\label{def:strongly-q-connected-comp}
The  \textit{maximal strongly $q$-connected components} of $\mathrm{dFl}(G)$ are the equivalence classes of the quotient set $\mathcal{K}^{s}_{q} = \mathrm{dFl}_{q}(G)/ S^{s}_{q}$, which are the equivalence classes of maximal $(\pm)$-$q$-connected directed simplices.
\end{definition}

\smallskip

Furthermore, if we disregard the directionality of the connections between the directed simplices in the relation (\ref{eq:strong-q-connected-rel}), we obtain exactly the maximal $q$-connectivity as defined in \citep{Serrano2020} for undirected simplices; thus, we define the following equivalence relation:
\begin{equation}\label{eq:weak-q-connected-rel}
(\sigma^{(n)}, \tau^{(m)}) \in S^{w}_{q} \iff \begin{cases}
\sigma^{(n)} \bm{\sim}_{\bm{A_{q^{*}}}} \tau^{(m)}, \\
\mbox{or } \sigma^{(n)} = \tau^{(m)},
\end{cases}
\end{equation}

\noindent where $\sigma^{(n)} \bm{\sim}_{\bm{A_{q^{*}}}} \tau^{(m)}$ denotes the maximal $q$-connectivity between $\sigma^{(n)}$ and $\tau^{(m)}$. One can verify that (\ref{eq:weak-q-connected-rel}) is indeed an equivalence relation by following the steps of the proof of Proposition \ref{prop:connected}.

\smallskip

\begin{definition}\label{def:weakly-q-connected-comp}
The  \textit{maximal weakly $q$-connected components} of $\mathrm{dFl}(G)$ are the equivalence classes of the quotient set $\mathcal{K}^{w}_{q} = \mathrm{dFl}_{q}(G)/ S^{w}_{q}$, which are the equivalence classes of maximal $q$-connected directed simplices.
\end{definition}

\smallskip

\begin{remark}
The previous definitions apply equally to the lower ($\pm$)-$q$-adjacency, since it's reflexive by definition.
\end{remark}

\smallskip

In Subsection \ref{sec:q-analysis} we have defined the (first) structure vector (Definition \ref{def:structure-vector}) based on the number of $q$-connected components of a simplicial complex $\mathcal{X}$, nonetheless, Andjelkovic et al. \citep{Andjelkovic2015} described two additional different structure vectors, the \textit{second} and the \textit{third} structure vectors, which are based, respectively, on the number of simplices in the set $\mathcal{X}_{q}$ and the ``degree of connectedness" among the simplices at level $q$. In what follows, we extend these vectors for the directed case.

\smallskip

\begin{definition}\label{def:dir-structure-vectors}
Let $\mathrm{dFl}(G)$ be a directed flag complex with $N = \dim \mathrm{dFl}(G)$. We define the following structure vectors associated with $\mathrm{dFl}(G)$:

\begin{enumerate}
\item The \textit{first weak/strong structure vector} is $F^{1} = (F^{1}_{0},..., F^{1}_{N})$, where $F^{1}_{q}$ is the number of maximal/lower weakly/strongly $q$-connected components.

\item The \textit{second structure vector} is $F^{2} = (F^{2}_{0},..., F^{2}_{N})$, where $F^{2}_{q}$ is the number of directed simplices in the set $\mathrm{dFl}_{q}(G)$.

\item The \textit{third weak/strong structure vector} is $F^{3} = (F^{3}_{0},..., F^{3}_{N})$, where $F^{3}_{q} = 1 - F^{1}_{q}/F^{2}_{q}$, with $F^{1}_{q}$ being the $q$-th element of the first weak or strong structure vector. The quantity $F^{3}_{q}$ can be interpreted as the ``degree of directed connectedness" among the directed simplices at level $q$.

\end{enumerate}
\end{definition}

\smallskip

Notice that we have two different ``first" structure vectors, one for the maximal/lower weakly $q$-connected components and one for the maximal/lower strongly $q$-connected components, and the same occurs for the ``third" structure vector, but we have only one ``second" structure vector.

\subsubsection{Maximal and Lower $q$-Digraphs}

Previously, we have defined $\mathrm{dFl}^{\ast}_{q}(G)$ as the set of all maximal directed simplices with dimensions greater than or equal to $q$, so when we advance in the level $q$, that is, \textit{change the level of organization} of the complex (see Figure \ref{fig:DFC-level-of-organization}), we obtain a new perspective on its higher-order topology, and, consequently, we may gain insights about the higher-order topology (or the clique organization) of the underlying network.

\smallskip

\begin{figure}[h!]
\centering
\begin{subfigure}{.23\textwidth}
  \centering
  \includegraphics[scale=0.46]{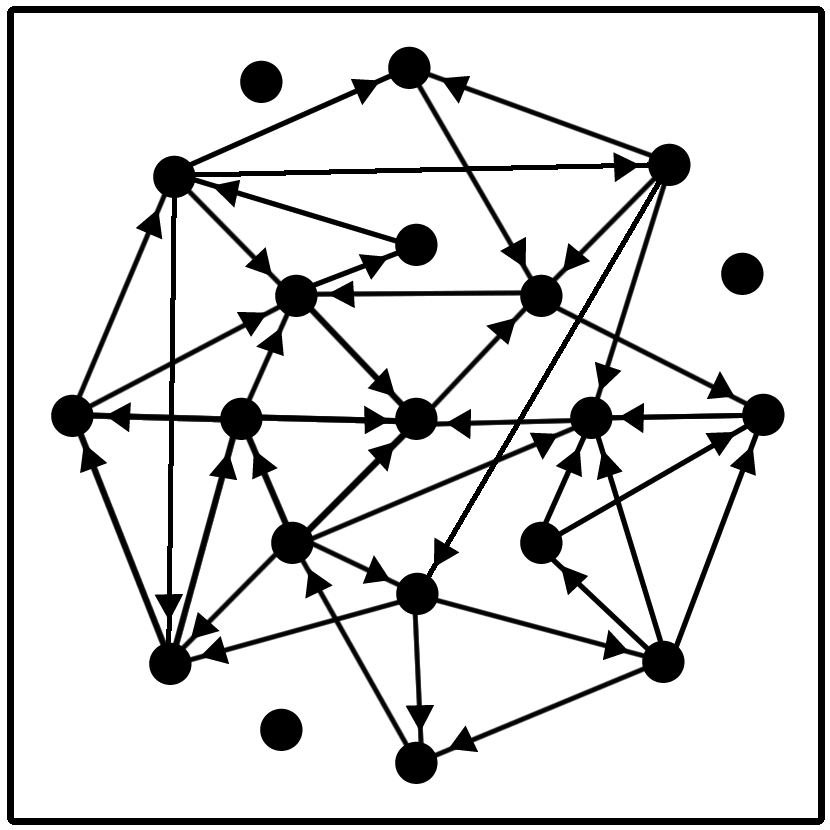}
  \caption{Digraph $G$.}
  \label{high-adj1}
\end{subfigure}%
\begin{subfigure}{.22\textwidth}
  \centering
  \includegraphics[scale=0.46]{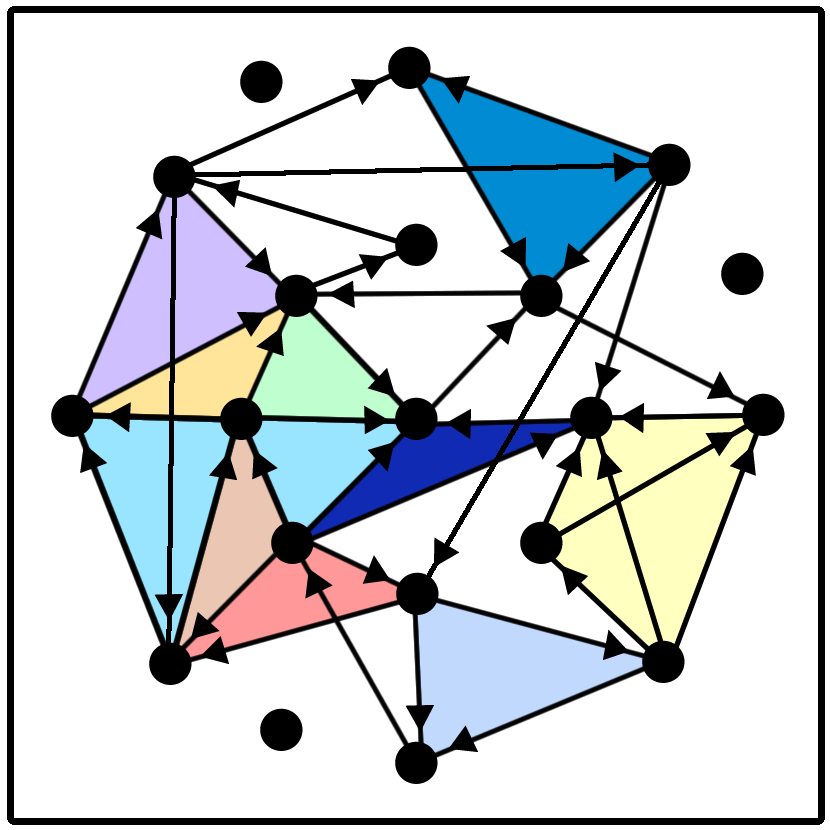}
  \caption{$\mathrm{dFl}^{\ast}_{0}(G)$.}
  \label{high-adj2}
\end{subfigure}
\begin{subfigure}{.22\textwidth}
  \centering
  \includegraphics[scale=0.46]{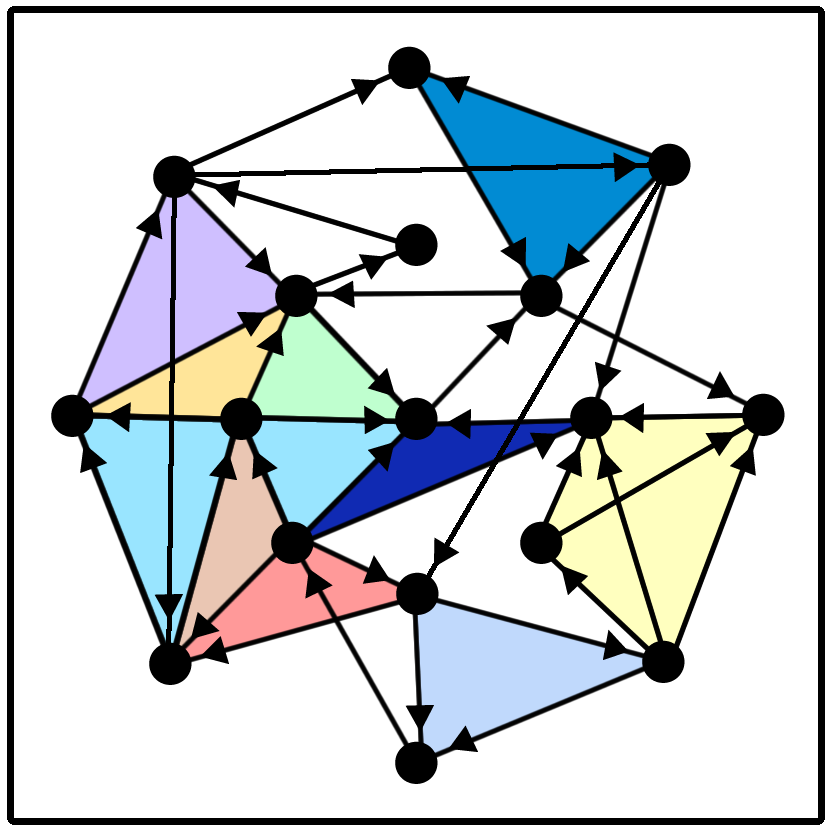}
  \caption{$\mathrm{dFl}^{\ast}_{1}(G)$.}
  \label{high-adj2}
\end{subfigure}
\begin{subfigure}{.22\textwidth}
  \centering
  \includegraphics[scale=0.46]{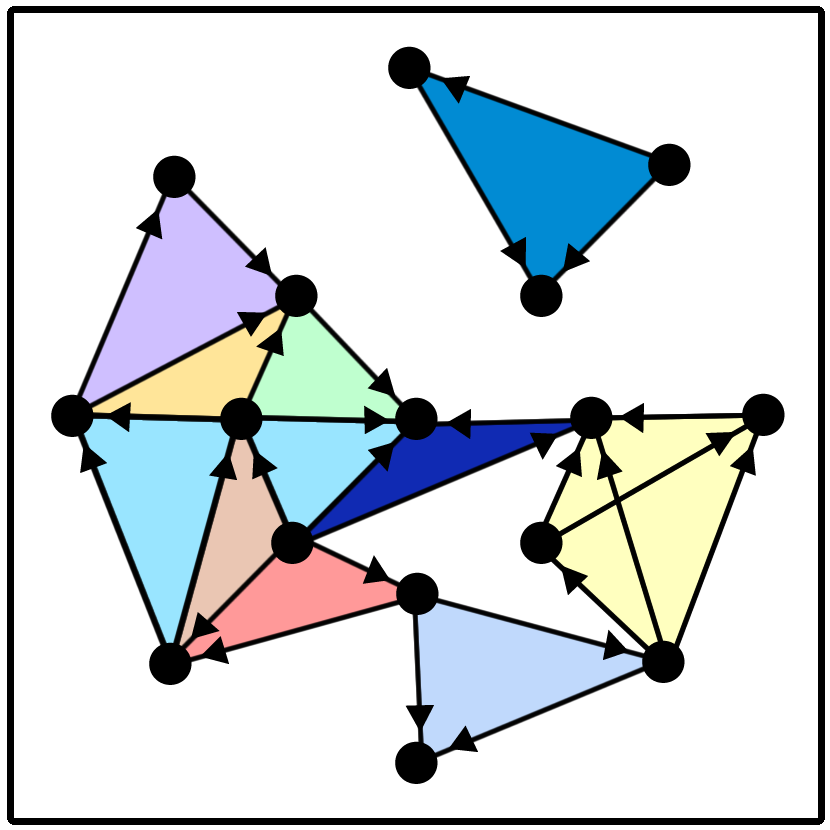}
  \caption{$\mathrm{dFl}^{\ast}_{2}(G)$.}
  \label{high-adj1}
\end{subfigure}%
\caption{Graphical representation of the maximal directed simplices for each level $q=0,1,2$.}
\label{fig:DFC-level-of-organization}
\end{figure}

\smallskip

Accordingly, we can go further and consider the directed higher-order connectivity among the simplices in each level $q$, and then define a ``higher-order digraph" for each of these levels. A definition of ``higher-order digraph" was already introduced in \citep{Caputi} with the concept of \textit{$q$-digraph}; however, since we have introduced different types of adjacencies, here we formalize two definitions: \textit{maximal} and \textit{lower $q$-digraph}.

\smallskip

\begin{definition}\label{def:q-digraph}
The \textit{maximal} (respec. \textit{lower}) \textit{$q$-digraph} of $\mathrm{dFl}(G)$, denoted by $\mathcal{G}_{q}^{A}$ (respec. $\mathcal{G}_{q}^{L}$), is the digraph whose vertices are the simplices of $\mathrm{dFl}^{\ast}_{q}(G)$ and for each pair $\sigma, \tau \in \mathrm{dFl}^{\ast}_{q}(G)$ there is a directed edge from $\sigma$ to $\tau$ if $\sigma \sim^{+}_{A_{q^{*}}} \tau$ (respec. $\sigma \sim^{+}_{L_{q}} \tau$), with $0 \le q \le \dim \mathrm{dFl}(G)$.
\end{definition}

\smallskip

In addition, when $\sigma \sim^{+}_{A_{q^{*}}} \tau$ (respec. $\sigma \sim^{+}_{L_{q}} \tau$), for $\sigma, \tau \in \mathrm{dFl}^{\ast}_{q}(G)$, we say that there is a \textit{maximal} (respec. \textit{lower}) \textit{$q$-arc} from $\sigma$ to $\tau$ and in this case we denote $(\sigma, \tau)_{A}$ (respec. $(\sigma, \tau)_{L}$). When the context is clear, we will simply denote $(\sigma, \tau)$ and call it a \textit{$q$-arc}. Also, we may use the notation $\mathcal{G}^{A}_{q} = (\mathcal{V}_{q}, \mathcal{E}^{A}_{q})$ (respec. $\mathcal{G}^{L}_{q} = (\mathcal{V}_{q}, \mathcal{E}^{L}_{q})$), where $\mathcal{V}_{q} = \mathrm{dFl}^{\ast}(G)$ and $\mathcal{E}^{A}_{q}$ (respec. $\mathcal{E}^{L}_{q}$) is the set of all maximal (respec. lower) $q$-arcs $(\sigma, \tau)$.

Analogously to graph theory, the maximal and lower $q$-digraphs can be represented in terms of maximal and lower \textit{$q$-adjacency matrices}, defined as follows.

\smallskip

\begin{definition}\label{def:q-adjacency-matrix}
Let $\mathcal{G}^{A}_{q}$ be the maximal $q$-digraph of $\mathrm{dFl}(G)$. The \textit{maximal $q$-adjacency matrix} of $\mathcal{G}^{A}_{q}$, denoted by $\mathcal{H}^{A}_{q} = \mathcal{H}_{q}^{A}(\mathcal{G}_{q})$, is a real square matrix whose entries are given by
\begin{equation}\label{eq:q-adjacency-matrix}
\big( \mathcal{H}^{A}_{q}\big)_{ij} = \begin{cases}
1, \mbox{ if } \sigma_{i} \sim^{+}_{A_{q^{*}}} \sigma_{j},\\
0, \mbox{ if } i = j \mbox{ or } \sigma_{i} \not\sim^{+}_{A_{q^{*}}} \sigma_{j}.
\end{cases}
\end{equation}

Similarly, we define the \textit{lower $q$-adjacency matrix} of $\mathcal{G}^{L}_{q}$ by
\begin{equation}\label{eq:lower-q-adjacency-matrix}
\big( \mathcal{H}^{L}_{q}\big)_{ij} = \begin{cases}
1, \mbox{ if } \sigma_{i} \sim^{+}_{L_{q}} \sigma_{j},\\
0, \mbox{ if } i = j \mbox{ or } \sigma_{i} \not\sim^{+}_{L_{q}} \sigma_{j}.
\end{cases}
\end{equation}
\end{definition}

\begin{example}
Consider the directed flag complex shown in Figure \ref{fig:dfc1}. Figures \ref{fig:q-digraph1}, \ref{fig:q-digraph2}, and \ref{fig:q-digraph3} represent its respective lower $q$-digraphs $\mathcal{G}^{L}_{q}$ for $q=0,1,2$.

\begin{figure}[h!]
\centering
\begin{subfigure}{.23\textwidth}
  \centering
  \includegraphics[scale=0.46]{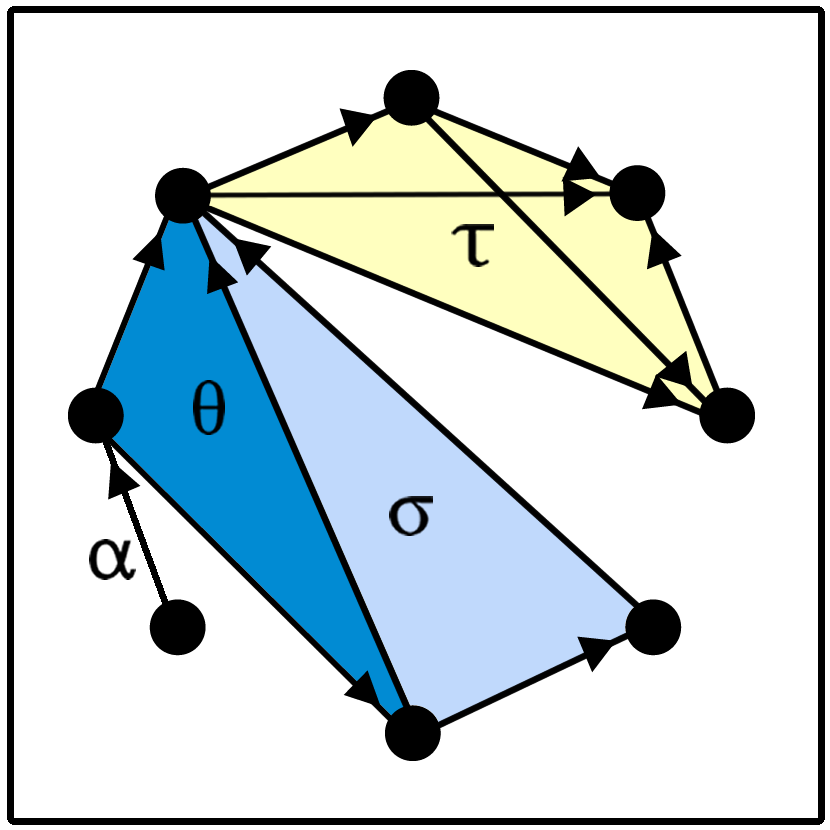}
  \caption{$\mathrm{dFl}(G)$.}
  \label{fig:dfc1}
\end{subfigure}%
\begin{subfigure}{.22\textwidth}
  \centering
  \includegraphics[scale=0.46]{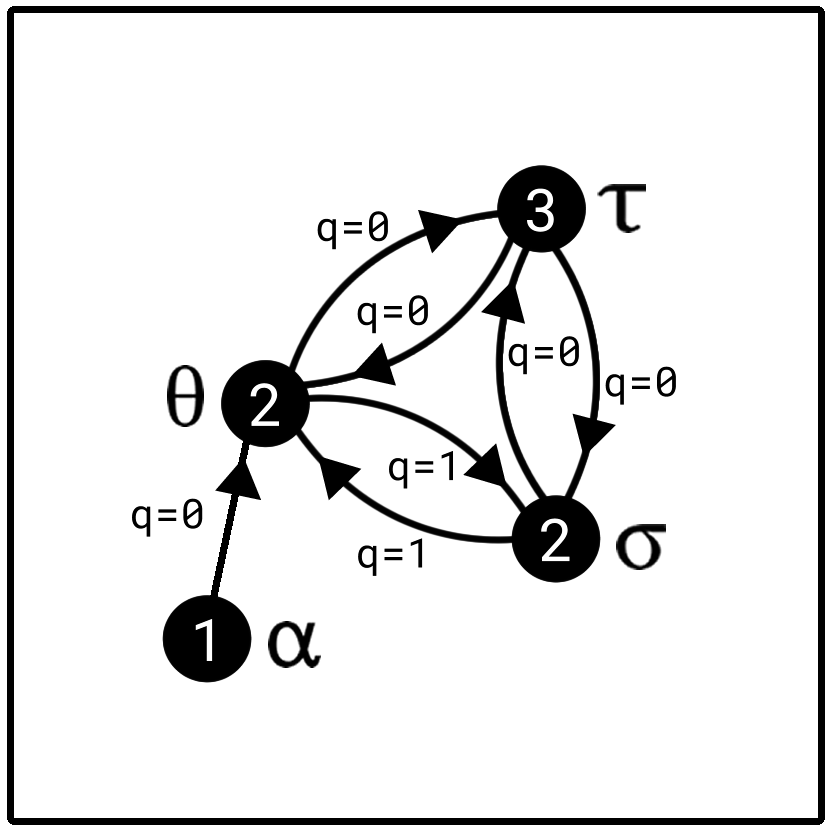}
  \caption{$\mathcal{G}^{L}_{0}$.}
  \label{fig:q-digraph1}
\end{subfigure}
\begin{subfigure}{.22\textwidth}
  \centering
  \includegraphics[scale=0.46]{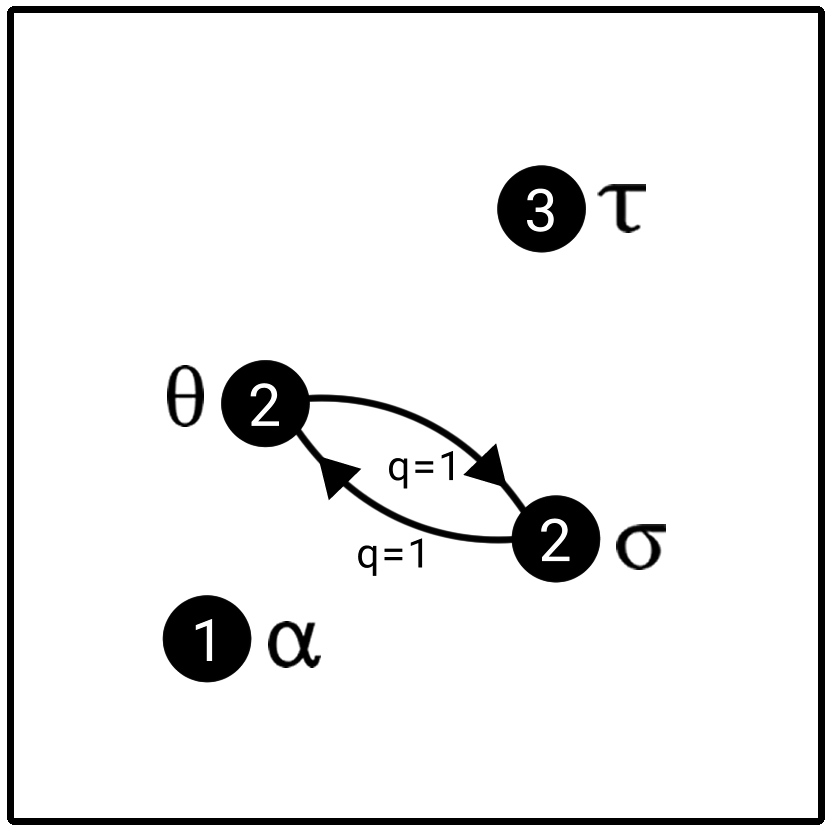}
  \caption{$\mathcal{G}^{L}_{1}$.}
  \label{fig:q-digraph2}
\end{subfigure}
\begin{subfigure}{.22\textwidth}
  \centering
  \includegraphics[scale=0.46]{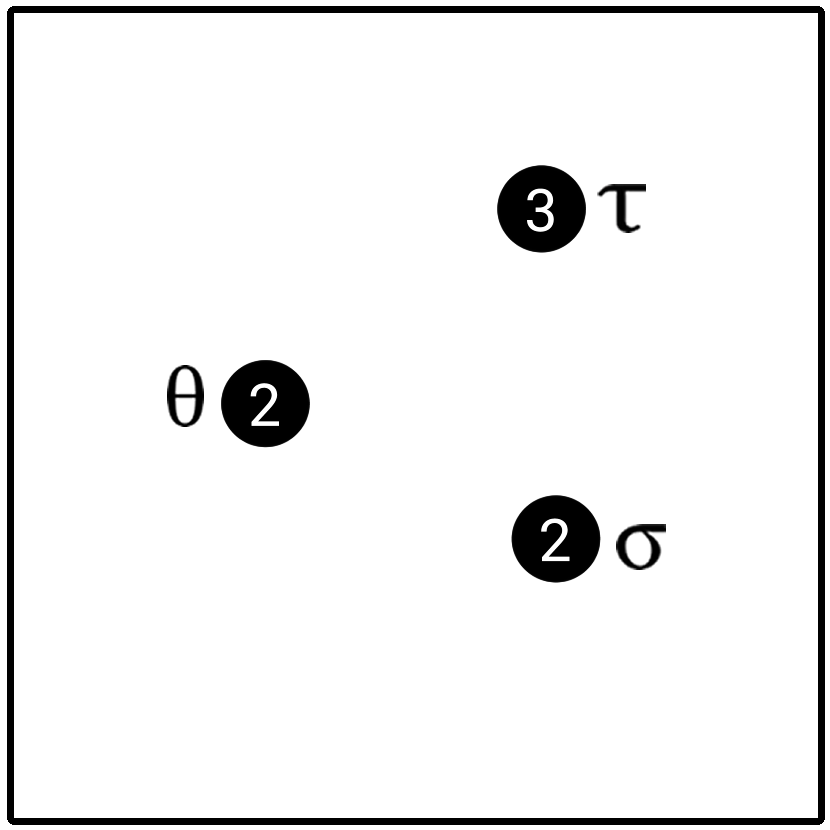}
  \caption{$\mathcal{G}^{L}_{2}$.}
  \label{fig:q-digraph3}
\end{subfigure}%
\caption{A directed flag complex and its respective lower $q$-digraphs, for $q=0,1,2$. The numbers inside the nodes represent the dimensions of the respective directed simplices.}
\label{fig:q-levels-digraphs}
\end{figure}
\end{example}

\smallskip

Note that, by Proposition \ref{prop:maximal-adjacency}, we can replace the maximal $q$-adjacency in the expression (\ref{eq:q-adjacency-matrix}) with the strictly lower $q$-adjacency. From now on, we will adopt the generic notations $\mathcal{G}_{q} = (\mathcal{V}_{q},\mathcal{E}_{q})$ and $\mathcal{H}_{q} = \mathcal{H}_{q}(\mathcal{G}_{q})$ designating both the lower and maximal variants.

\smallskip

\begin{observation}
The weakly and strongly connected components of a maximal/lower $q$-digraph $\mathcal{G}_{q}$ are equivalent to the maximal/lower weakly and strongly $q$-connected components as defined in Definition \ref{def:weakly-q-connected-comp} and Definition \ref{def:strongly-q-connected-comp}, respectively, where the set $\mathrm{dFl}(G)$ is replaced by $\mathrm{dFl}^{*}(G)$. Moreover, the largest weakly $q$-connected component of $\mathcal{G}_{q}$ is called its \textit{giant $q$-component}.
\end{observation}

\smallskip

In the case where we have a weighted directed flag complex obtained from a weighted digraph, by definition, the corresponding maximal/lower $q$-digraphs will be node-weighted digraphs, since their nodes represent directed simplices. Accordingly, in order to obtain edge-weighted digraphs, we need a node-to-edge weight function. In the literature, there is a myriad of methods to transform a node-weighted digraph into an edge-weighted digraph \citep{DelaCruzCabrera}, however, since the relations in a digraph can be non-symmetric, we would like a non-symmetric transformation function, thus, for a given $(\mathrm{dFl}(G), \widetilde{\omega})$ and for a given $q$-arc $(\sigma_{i}, \sigma_{j})$, we consider the following node-to-edge weight functions:
\begin{equation}\label{eq:node-to-edge1}
f(\widetilde{\omega}(\sigma_{i}), \widetilde{\omega}(\sigma_{j})) = \widetilde{\omega}(\sigma_{i}),
\end{equation}
\begin{equation}\label{eq:node-to-edge2}
f(\widetilde{\omega}(\sigma_{i}), \widetilde{\omega}(\sigma_{j})) = \widetilde{\omega}(\sigma_{j}).
\end{equation}

\smallskip

This leads us to extend our definitions of maximal/lower $q$-digraphs to the weighted case as follows.

\smallskip


\begin{definition}\label{def:weig-q-digraph}
Given a weighted digraph $G^{\omega}$, the \textit{weighted maximal} (respec. \textit{lower}) \textit{$q$-digraph} of $(\mathrm{dFl}(G^{\omega}), \widetilde{\omega})$, denoted by $\mathcal{G}_{q}^{\widetilde{\omega}}$, is the digraph whose vertices are the simplices of $\mathrm{dFl}^{\ast}_{q}(G^{\omega})$ and for each pair $\sigma, \tau \in \mathrm{dFl}^{\ast}_{q}(G^{\omega})$ there is a weighted arc from $\sigma$ to $\tau$ if $\sigma \sim^{+}_{A_{q^{*}}} \tau$ (respec. $\sigma \sim^{+}_{L_{q}} \tau$), with $0 \le q \le \dim \mathrm{dFl}(G^{\omega})$, such that the weight of the arc is given by a node-to-edge weight function.
\end{definition}

\smallskip

In addition, we may use the notation $\mathcal{G}_{q}^{\widetilde{\omega}} = (\mathcal{V}_{q}, \mathcal{E}_{q}, \widetilde{\omega})$, where $\mathcal{V}_{q} = \mathrm{dFl}_{q}^{\ast}(G^{\omega})$, $\mathcal{E}_{q}$ is the set of all $q$-arcs $(\sigma, \tau)$, and $\widetilde{\omega}$ is the product-weight function.

\smallskip

\begin{definition}
Let $\mathcal{G}_{q}^{\widetilde{\omega}}$ be the weighted maximal (respec. lower) $q$-digraph of $(\mathrm{dFl}(G^{\omega}), \widetilde{\omega})$. The \textit{weighted maximal} (respec. \textit{lower}) \textit{$q$-adjacency matrix} of $\mathcal{G}_{q}^{\widetilde{\omega}}$, denoted by $\mathcal{H}^{\widetilde{\omega}}_{q} = \mathcal{H}^{\widetilde{\omega}}_{q}(\mathcal{G}^{\widetilde{\omega}}_{q})$, is a real square matrix whose entries are given by
\begin{equation}\label{eq:weig-q-adjacency-matrix}
\big(\mathcal{H}^{\widetilde{\omega}}_{q}\big)_{ij} = \begin{cases}
f(\widetilde{\omega}(\sigma_{i}), \widetilde{\omega}(\sigma_{j})), \mbox{ if } \sigma_{i} \sim^{+}_{A_{q^{*}}} \sigma_{j} \mbox{ (respec. } \sigma_{i} \not\sim^{+}_{L_{q}} \sigma_{j}),\\
0, \mbox{ if } i = j \mbox{ or } \sigma_{i} \not\sim^{+}_{A_{q^{*}}} \sigma_{j} \mbox{ (respec. } \sigma_{i} \not\sim^{+}_{L_{q}} \sigma_{j}),
\end{cases}
\end{equation}

\noindent where $f$ is some non-symmetric node-to-edge weight function.
\end{definition}

\begin{example}
Consider the (non-normalized) weighted digraph $G^{\omega}$ as depicted in Figure \ref{fig:example-weig-digraph-dfc1}. Figure \ref{fig:example-weig-digraph-dfc2} depicts its weighted directed flag complex $(\mathrm{dFl}(G^{\omega}), \widetilde{\omega})$ in which the simplices are shown with their respective (non-normalized) weights. The weight function $\widetilde{\omega}$ is the product-weight function (\ref{eq:prod-weight}) such that the edge-to-node function is given by Definition \ref{def:weights-edge-node}. Thus, considering the node-to-edge weight function (\ref{eq:node-to-edge1}), the corresponding weighted lower $0$-digraph $\mathcal{G}_{0}^{\widetilde{\omega}}$ is the weighted digraph presented in Figure \ref{fig:example-weig-digraph-dfc3}. Also, both lower $0$-adjacency matrices of $\mathcal{G}_{0}$ and $\mathcal{G}_{0}^{\widetilde{\omega}}$ were computed in (\ref{eq:two-q-adjancecy-matrices}).

\begin{figure}[h!]
\centering
\begin{subfigure}{.33\textwidth}
  \centering
  \includegraphics[scale=0.75]{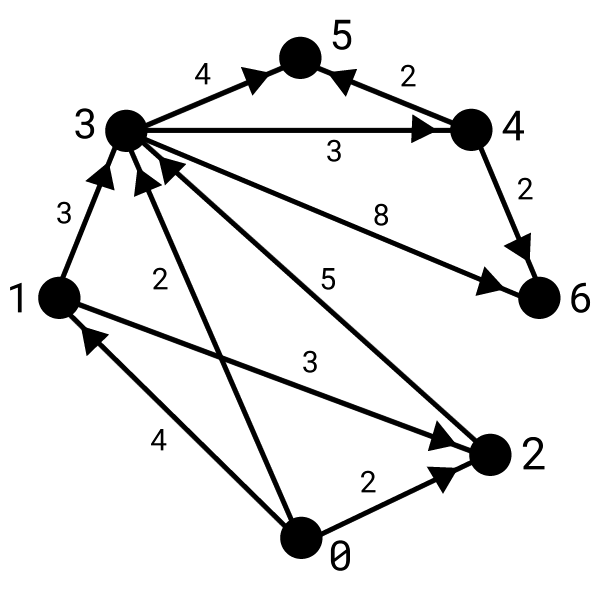}
  \caption{$G^{\omega}$.}
  \label{fig:example-weig-digraph-dfc1}
\end{subfigure}%
\begin{subfigure}{.33\textwidth}
  \centering
  \includegraphics[scale=0.75]{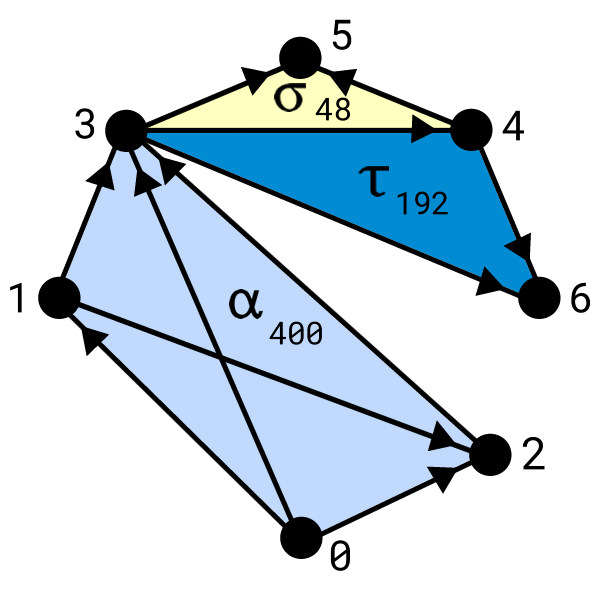}
  \caption{$(\mathrm{dFl}(G^{\omega}), \widetilde{\omega})$.}
  \label{fig:example-weig-digraph-dfc2}
\end{subfigure}%
\begin{subfigure}{.33\textwidth}
  \centering
  \includegraphics[scale=0.8]{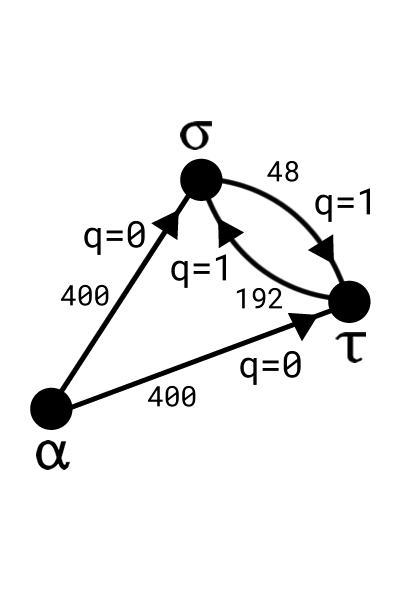}
  \caption{$\mathcal{G}^{\widetilde{\omega}}_{0}$.}
  \label{fig:example-weig-digraph-dfc3}
\end{subfigure}
\caption{A weighted digraph together with its weighted directed flag complex and its weighted lower $0$-digraph.}
\label{fig:digraph-and-dfc}
\end{figure}

\begin{equation}\label{eq:two-q-adjancecy-matrices}
\mathcal{H}_{0} =  
\bbordermatrix{
& \sigma & \tau & \alpha \cr
\sigma & 0  & 1 & 0\cr 
\tau   & 1  & 0 & 0\cr 
\alpha & 1  & 1 & 0\cr 
}, \hspace{0.12in}
\mathcal{H}^{\widetilde{\omega}}_{0} =
\bbordermatrix{
& \sigma & \tau & \alpha \cr
\sigma & 0  & 48 & 0\cr 
\tau   & 192  & 0 & 0\cr 
\alpha & 400  & 400 & 0\cr 
}.
\end{equation}
\end{example}

\medskip

Furthermore, the directed simplicial $q$-distance between two vertices of a maximal/lower $q$-digraph can be written in terms of the entries of its maximal/lower $q$-adjacency matrix, $\mathcal{H}_{q} = (h_{\sigma \tau})$:
\begin{equation}\label{eq:simp-weig-distance-entries}
\vec{d}_{q}(\sigma, \tau) = \sum_{\sigma', \tau' \in s^{q}_{\sigma \rightarrow \tau}} h_{\sigma' \tau'},
\end{equation}

\noindent where $s^{q}_{\sigma \rightarrow \tau}$ is the shortest directed simplicial $q$-walk from $\sigma$ to $\tau$. Similarly, for a weighted maximal/lower $q$-digraph, we can define the \textit{weighted directed simplicial $q$-distance} in terms of the entries of its weighted maximal/lower $q$-adjacency matrix: 
\begin{equation}\label{eq:simp-weig-distance-entries}
\vec{d}_{q}^{\omega}(\sigma, \tau) = \sum_{\sigma', \tau' \in s^{q}_{\sigma \rightarrow \tau}(F)} F\big(f(\widetilde{\omega}(\sigma'), \widetilde{\omega}(\tau')\big),
\end{equation}

\noindent where $f$ is some non-symmetric node-to-edge weight function, $F$ is a weight-to-distance function, and $s^{q}_{\sigma \rightarrow \tau}(F)$ is the shortest directed simplicial $q$-walk from $\sigma$ to $\tau$ with respect to $F$.

\subsubsection{Stars, Hubs, and Links}

In this part, we extend the definitions of stars, hubs, and links introduced in Subsection \ref{sec:q-analysis} to directed simplices.

\smallskip

\begin{definition}
Given a directed simplex $\sigma^{(n)} \in \mathrm{dFl}(G)$, for $0 \le q \le n$ and for $\bullet \in \{ -,+,\pm \}$, we have the following definitions:

\begin{enumerate}
\item The \textit{lower $(\bullet)$-$q$-star} of $\sigma^{(n)}$ is the set defined by
\begin{equation}\label{def:dir-stars}
\mathrm{st}_{L_{q}} ^{\bullet}(\sigma^{(n)}) = \{ \tau^{(m)} \in \mathrm{dFl}(G) : \sigma^{(n)}  \sim^{\bullet}_{L_{q}}  \tau^{(m)}  \}.
\end{equation}

\item The \textit{strictly lower $(\bullet)$-$q$-star} of $\sigma^{(n)}$ is the set defined by
\begin{equation}\label{eq:lower-star}
\mathrm{st}_{L_{q^{*}}} ^{\bullet}(\sigma^{(n)}) = \{ \tau^{(m)} \in \mathrm{dFl}(G) : \sigma^{(n)}  \sim^{\bullet}_{L_{q^{*}}} \tau^{(m)}  \}.
\end{equation}

\item The \textit{upper $(\bullet)$-$q$-star} of $\sigma^{(n)}$ is the set defined by
\begin{equation}
\mathrm{st}_{U_{q}} ^{\bullet}(\sigma^{(n)}) = \{ \tau^{(m)} \in \mathrm{dFl}(G) : \sigma^{(n)}  \sim^{\bullet}_{U_{q}}  \tau^{(m)}  \}.
\end{equation}

\item The \textit{strictly upper $(\bullet)$-$q$-star} of $\sigma^{(n)}$ is the set defined by
\begin{equation}
\mathrm{st}_{U_{q^{*}}} ^{\bullet}(\sigma^{(n)}) = \{ \tau^{(m)} \in \mathrm{dFl}(G) : \sigma^{(n)}  \sim^{\bullet}_{U_{q^{*}}} \tau^{(m)}  \}.
\end{equation}

\item The \textit{$(\bullet)$-$q$-star} of $\sigma^{(n)}$ is the set defined by
\begin{equation}
\mathrm{st}_{A_{q}} ^{\bullet}(\sigma^{(n)}) = \{ \tau^{(m)} \in \mathrm{dFl}(G) : \sigma^{(n)}  \sim^{\bullet}_{A_{q}}  \tau^{(m)}  \}.
\end{equation}

\item The \textit{maximal $(\bullet)$-$q$-star} of $\sigma^{(n)}$ is the set defined by
\begin{equation}
\mathrm{st}_{A_{q^{*}}} ^{\bullet}(\sigma^{(n)}) = \{ \tau^{(m)} \in \mathrm{dFl}(G) : \sigma^{(n)}  \sim^{\bullet}_{A_{q^{*}}}  \tau^{(m)}  \}.
\end{equation}

\end{enumerate}
\end{definition}

\smallskip

Notice that $\mathrm{st}^{\pm}(\sigma) = \mathrm{st}^{+}(\sigma) \cap \mathrm{st}^{-}(\sigma)$, for any of the $q$-stars defined above.

\smallskip

The definition of the hub of a simplicial family of directed simplices obtained from $\mathrm{dFl}(G)$ is exactly the same as the Definition \ref{def:simp-hub}, since it is defined as the set of faces that are shared by the elements of the family, regardless of the direction of the connection between them. We formalize this fact as follows.

\smallskip

\begin{definition}\label{def:dir-hub}
Let $\mathcal{F}(G)$ denote a simplicial family of directed simplices obtained from $\mathrm{dFl}(G)$. The \textit{hub} of $\mathcal{F}(G)$ is the set formed by all directed simplices that are common faces of the elements of $\mathcal{F}(G)$, i.e.

\begin{equation}
\mathrm{hub}(\mathcal{F}(G)) = \bigcap_{\sigma \in \mathcal{F}(G)} \sigma.
\end{equation}

\end{definition}

\smallskip

Analogously to Definition \ref{def:link} (link of a simplex), we can generalize the idea of in- and out-neighborhood of a node in a digraph to directed simplices through the concepts of \textit{in-} and \textit{out-link}. 

\smallskip

\begin{definition}\label{def:dir-link}
The \textit{in-link} and \textit{out-link} of a directed simplex $\sigma^{(n)} \in \mathrm{dFl}(G)$ are defined, respectively, by
\begin{equation}\label{eq:in-link}
\mathrm{lk}^{-}(\sigma^{(n)}) = \{ \tau^{(m)} \in \mathrm{dFl}(G) |\sigma^{(n)} \cap \tau^{(m)} = \emptyset, \sigma^{(n)}\sim^{-}_{U_{p}}  \tau^{(m)}  \},
\end{equation}
\begin{equation}\label{eq:out-link}
\mathrm{lk}^{+}(\sigma^{(n)}) = \{ \tau^{(m)} \in \mathrm{dFl}(G) |\sigma^{(n)} \cap \tau^{(m)} = \emptyset, \sigma^{(n)} \sim^{+}_{U_{p}} \tau^{(m)}  \},
\end{equation}

\noindent where $p=n+m+1$.
\end{definition}

\smallskip


Notice that if $\sigma$ is a simplex in the underlying flag complex of $ \mathrm{dFl}(G)$, then $\mathrm{lk}(\sigma) = \mathrm{lk}^{-}(\sigma) \cup \mathrm{lk}^{+}(\sigma) - (\mathrm{lk}^{-}(\sigma) \cap \mathrm{lk}^{+}(\sigma))$.

\smallskip

\begin{example}
Considering the directed flag complex depicted in Figure \ref{fig:dfc-example-stars-degree}, we have the following examples: the in- and out-link of the arc $[0,3]$ are $\mathrm{lk}^{-}([0,3]) = \{ [4], [1,9] \}$ and $\mathrm{lk}^{+}([0,3]) = \{ [1,9] \}$; the maximal $(\bullet)$-$1$-star of $\sigma_{1}$ are $\mathrm{st}_{A_{1^{*}}}^{-}(\sigma_{1}) = \{ \sigma_{2} \}$ and $\mathrm{st}_{A_{1^{*}}}^{+}(\sigma_{1}) = \emptyset$; the maximal $(\bullet)$-$1$-star of $\sigma_{6}$ are $\mathrm{st}_{A_{1^{*}}}^{-}(\sigma_{6}) = \{  \sigma_{7} \}$ and $\mathrm{st}_{A_{1^{*}}}^{+}(\sigma_{6}) = \{ \sigma_{5}, \sigma_{3} \}$.

\begin{figure}[h!]
\centering
  \includegraphics[scale=0.8]{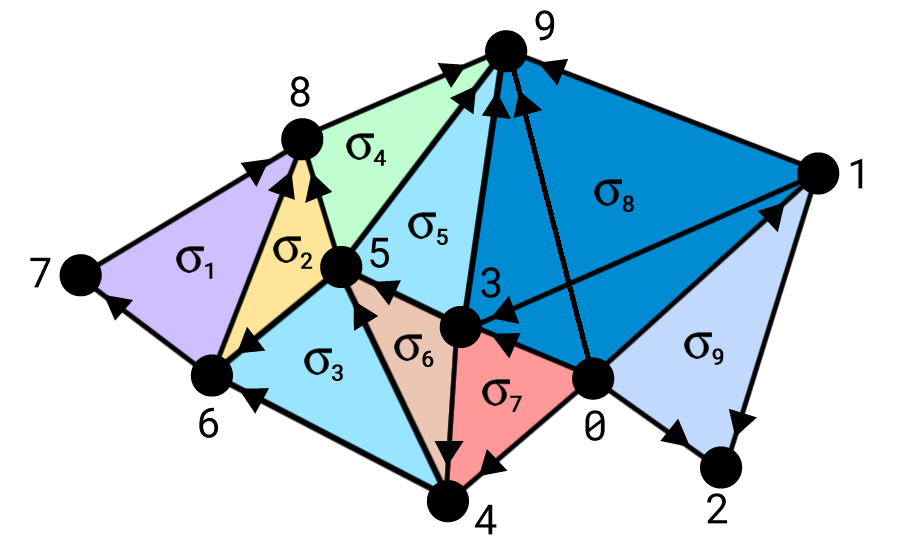}
\caption{A directed flag complex.}
\label{fig:dfc-example-stars-degree}
\end{figure}

\end{example}

\subsubsection{Lower, Upper, and General Degrees}

Based on the previous definitions of $q$-stars, in this last part, we extend the definitions of lower, upper, and general degrees to directed simplices.

\smallskip

\begin{definition}
Given a directed simplex $\sigma^{(n)} \in \mathrm{dFl}(G)$, for $0 \le q \le n$ and for $\bullet \in \{ -,+,\pm \}$, we have the following definitions:

\begin{enumerate}
\item The \textit{lower $(\bullet)$-$q$-degree} of $\sigma^{(n)}$ is defined by 
\begin{equation}
 \deg^{\bullet}_{L_{q}}(\sigma^{(n)}) = |\mathrm{st}_{L_{q}}^{\bullet}(\sigma^{(n)})| = \# \{ \tau^{(m)} \in \mathrm{dFl}(G) : \sigma^{(n)} \sim_{L_{q}}^{\bullet} \tau^{(m)} \}.
\end{equation}

\item The \textit{strictly lower $(\bullet)$-$q$-degree} of $\sigma^{(n)}$ is defined by
\begin{equation}
 \deg^{\bullet}_{L_{q^{*}}}(\sigma^{(n)}) = |\mathrm{st}_{L_{q^{*}}}^{\bullet}(\sigma^{(n)})| = \# \{ \tau^{(m)} \in \mathrm{dFl}(G) : \sigma^{(n)} \sim_{L_{q^{*}}}^{\bullet} \tau^{(m)} \}.
\end{equation}

\item The \textit{upper $(\bullet)$-$q$-degree} of $\sigma^{(n)}$ is defined by 
\begin{equation}
 \deg^{\bullet}_{U_{q}}(\sigma^{(n)}) = |\mathrm{st}_{U_{q}}^{\bullet}(\sigma^{(n)})| = \# \{ \tau^{(m)} \in \mathrm{dFl}(G) : \sigma^{(n)} \sim_{U_{q}}^{\bullet} \tau^{(m)} \}.
\end{equation}

\item The \textit{strictly upper $(\bullet)$-$q$-degree} of $\sigma^{(n)}$ is defined by
\begin{equation}
 \deg^{\bullet}_{U_{q^{*}}}(\sigma^{(n)}) = |\mathrm{st}_{U_{q^{*}}}^{\bullet}(\sigma^{(n)})| = \# \{ \tau^{(m)} \in \mathrm{dFl}(G) : \sigma^{(n)} \sim_{U_{q^{*}}}^{\bullet} \tau^{(m)} \}.
\end{equation}

\item The \textit{$(\bullet)$-$q$-degree} of $\sigma^{(n)}$ is defined by 
\begin{equation}
\deg^{\bullet}_{A_{q}}(\sigma^{(n)}) = |\mathrm{st}_{A_{q}}^{\bullet}(\sigma^{(n)})| = \# \{ \tau^{(m)} \in \mathrm{dFl}(G) : \sigma^{(n)} \sim_{A_{q}}^{\bullet} \tau^{(m)} \}.
\end{equation}

\item The \textit{maximal $(\bullet)$-$q$-degree} of $\sigma^{(n)}$ is defined by 
\begin{equation}
\deg^{\bullet}_{A_{q^{*}}}(\sigma^{(n)}) = |\mathrm{st}_{A_{q^{*}}}^{\bullet}(\sigma^{(n)})| = \# \{ \tau^{(m)} \in \mathrm{dFl}(G) : \sigma^{(n)} \sim_{A_{q^{*}}}^{\bullet} \tau^{(m)} \}.
\end{equation}

\end{enumerate}
\end{definition}

\smallskip


Notice that if $\sigma$ is a simplex in the underlying flag complex of $\mathrm{dFl}(G)$, then its $q$-degree, for any of the lower, upper, and general adjacencies, is equal to
\begin{equation}
\deg(\sigma) =  \deg^{-}(\sigma) +  \deg^{+}(\sigma) -  \deg^{\pm}(\sigma).
\end{equation}

\begin{example}
Considering the directed flag complex depicted in Figure \ref{fig:dfc-example-stars-degree}, we have the following examples: $\deg^{-}_{L_{1}}(\sigma_{1}) = 1$ and $\deg^{+}_{L_{1}}(\sigma_{1}) = 0$;  $\deg^{-}_{L_{1}}(\sigma_{6}) = 2$ and $\deg^{+}_{L_{1}}(\sigma_{6}) = 2$.
\end{example}

\smallskip

It's important to note that if we are considering solely the elements of $\mathrm{dFl}_{q}^{\ast}(G)$, i.e. the maximal directed simplices, then the maximal/lower $(\bullet)$-$q$-degrees, $\bullet \in \{-, +\}$, can be written in terms of the entries of the $q$-adjacency matrix of the maximal/lower $q$-digraph, $\mathcal{H}_{q} = (h_{\sigma\tau})$:
\begin{equation}\label{eq:simp-weig-in-degree}
\deg^{-}_{q}(\sigma) = \sum_{\tau \in \mathrm{st}^{-}_{q}(\sigma)} h_{\tau \sigma},
\end{equation}

\begin{equation}\label{eq:simp-weig-out-degree}
\deg^{+}_{q}(\sigma) = \sum_{\tau \in \mathrm{st}^{+}_{q}(\sigma)} h_{\sigma \tau},
\end{equation}

\noindent where the generic notations $\deg^{\bullet}_{q}$ and $\mathrm{st}^{\bullet}_{q}$ designate either the maximal or the lower variants. On the other hand, in the case where we have a weighted directed flag complex $(\mathrm{dFl}(G^{\omega}), \widetilde{\omega})$, the \textit{weighted maximal/lower $(\bullet)$-$q$-degrees}, $\bullet \in \{-, +\}$, can be written as:
\begin{equation}\label{eq:simp-weig-in-degree}
\deg^{\omega, -}_{q}(\sigma) = \sum_{\tau \in \mathrm{st}^{-}_{q}(\sigma)} f(\widetilde{\omega}(\tau), \widetilde{\omega}(\sigma)),
\end{equation}

\begin{equation}\label{eq:simp-weig-out-degree}
\deg^{\omega, +}_{q}(\sigma) = \sum_{\tau \in \mathrm{st}^{+}_{q}(\sigma)} f(\widetilde{\omega}(\sigma), \widetilde{\omega}(\tau)),
\end{equation}

\noindent where $f$ is some non-symmetric node-to-edge weight function and $\deg^{\omega, \bullet}_{q}$ is a generic notation for both variants.






\chapter[Quantitative Approaches to Digraph-Based Complexes]{Quantitative Approaches to Digraph-Based Complexes}
\label{chap:chap5}

\epigraph{En un mot, pour tirer la loi de l'expérience, if faut généraliser; c'est une nécessité qui s'impose à l'observateur le plus circonspect. (In one word, to draw the rule from experience, one must generalize; this is a necessity that imposes itself on the most circumspect observer.)}{--- Henri Poincaré \citep{PoincareValue}}

\bigskip

In this chapter, we extend several quantifiers and similarity comparison methods defined for (weighted and unweighted) graphs in Chapter \ref{chap:chap2}, such as distance-related measures, measures of centrality, segregation, spectrum-related measures, graph kernels, etc., which take into account their abstract, algebraic, and topological properties, to (weighted and unweighted) digraph-based complexes, specially for directed flag complexes, and, in addition, we introduce some new measures.

Specifically, we extend all of the quantitative methods defined previously for graphs to maximal/lower $q$-digraphs obtained from directed flag complexes, so we can quantify the higher-order connectivity and topological properties of the directed network at each level $q$, i.e. we can analyze quantitatively each level of organization of the directed cliques in the network taking into account their directed higher-order connectivities.


\section{Quantitative Graph Theory and Beyond}
\label{sec:QGT-beyond}

In this first section, we present some conceptual descriptions and examples related to quantitative graph theory; subsequently, we discuss some recently introduced quantitative methods for simplicial complexes and their importance as a basis for a quantitative theory of digraph-based complexes.

\subsection{Quantitative Graph Theory}

According to Dehmer et al. \citep{Dehmer-2017}, the ``quantitative graph theory (QGT) deals with the quantification of structural aspects of graphs, instead of characterizing graphs only descriptively." The QGT can be considered a new branch within graph theory \citep{Dehmer-2017, Dehmer-2014}, and the reason why its introduction was necessary is that most of the methods from classical graph theory for graph analysis are descriptive.

Moreover, we can divide the QGT into two broad categories, namely: \textit{graph characterization} and \textit{comparative graph analysis} \citep{Dehmer-2014}. In what follows, we present a description of these categories and provide some examples.

\bigskip
\noindent \textbf{Graph Characterization:}
Graph characterization is concerned with describing some property of the network through a local or global numerical graph invariant. Examples of numerical graph invariants are: distance-based measures (characteristic path length (\ref{eq:CPL}), eccentricity (\ref{eq:eccentricity})); centrality measures (closeness centrality (\ref{eq:closeness-centrality}), betweenness centrality (\ref{eq:betweenness-centrality})); complexity/information-theoretic measures (graph entropy (\ref{eq:distribution-entropy})); and spectrum-based measures (graph energy (\ref{eq:graph-energy})).

\bigskip
\noindent \textbf{Comparative Graph Analysis:}
Comparative graph analysis is concerned with methods to compare the structural similarity or structural distance between two or more graphs. As already commented in Section \ref{sec:graph-similarity}, there are two classes of methods of graph similarity comparison: \textit{statistical comparison methods} and \textit{distance-based comparison algorithms}. Examples of the first class are: applying numerical invariants in different groups of graphs and comparing them using statistical tests; and examples of the second class are: graph edit distances; graph kernels; and structure distances.

\bigskip

In the past two paragraphs, we presented just a few examples among a myriad of measures and methods used to characterize and compare graphs in several different scientific disciplines, and it is up to the reader to look over the references if they want to learn more \citep{Dehmer-2017, Dehmer-2014}. 

\smallskip

\subsection{Beyond QGT: Quantitative Simplicial Theory}

In view of the recent developments of quantitative methods for the analysis of simplicial complexes, we can draw an analogy with QGT and propose a ``quantitative simplicial theory," with the same conceptual characteristics described in the previous section for QGT, i.e. as the field which ``deals with the quantification of structural aspects of simplicial complexes, instead of characterizing simplicial complexes only descriptively."

Making another analogy with QGT, we can divide the ``quantitative simplicial theory" into two categories: \textit{simplicial characterization} and \textit{simplicial similarity comparison methods}. The description of these categories is analogous to those of the QGT; thus, below we present some examples and the respective references of the methods belonging to each of these categories:

\bigskip
\noindent \textbf{Simplicial Characterization:} simplical distances and simplical eccentricities \citep{Johnson, Serrano2020}; simplicial centralities \citep{Estrada, Serrano2020}; simplicial clustering coefficients \citep{Maletic, Serrano2020}; simplicial entropies \citep{Baccini, Dantchev, Maletic2012}; discrete curvatures \citep{Yamada2023}; simplicial energies \citep{Knill}.

\bigskip
\noindent \textbf{Simplicial Similarity Comparison Methods:} distances between persistent diagrams \citep{Edelsbrunner}; distances between vectorized persistence summaries \citep{Fasy}; simplicial kernels \citep{Martino, Zhang2020}; distances between structure vectors.

\bigskip

Most of these simplicial quantitative methods can be extended to directed flag complexes, and some of them can also be extended to path complexes, as we will see in the next sections.


\section{Simplicial Characterization Measures}

In this section, we introduce novel simplicial analogues of the digraph measures (graph invariants) presented in Section \ref{sec:graph-measures} to maximal and lower $q$-digraphs (henceforth simply referred to as $q$-digraphs) associated with directed flag complexes. We use a similar textual organizational structure as used in the aforementioned section, that is, we start by presenting the distance-based simplicial measures, followed by the simplicial centralities, simplicial segregation measures, simplicial entropies, discrete curvatures, and finally the spectrum-related simplicial measures. It is worth pointing out that all these measures are based on the directed high-order connectivity of the complex in a certain way.

\paragraph{Conventions and notations.} Throughout the next subsections, $\mathrm{dFl}(G)$ will denote the directed flag complex of a given simple digraph $G$ \textit{without double edges}, and $\mathcal{G}_{q} = (\mathcal{V}_{q}, \mathcal{E}_{q})$ will denote its (maximal or lower) $q$-digraph with (maximal or lower) $q$-adjacency matrix (henceforth simply referred to as $q$-adjacency matrix) $\mathcal{H}_{q} = (h_{\sigma \tau})$, for $0 \le q \le \dim \mathrm{dFl}(G)$. For each of the measures that we will define in the sections below, corresponding maximal and lower variants are obtained analogously by replacing the adjacency $A_{q^{*}}$ with $L_{q^{*}}$ and vice versa (e.g., $\mathcal{H}_{q}^{A}$ by $\mathcal{H}_{q}^{L}$, $\vec{d}^{A}_{q}$ by $\vec{d}^{L}_{q}$, and $\deg^{\bullet}_{A_{q^{*}}}$ by $\deg^{\bullet}_{L_{q^{*}}}$), and we will compress this using the generic notations as described in Table~\ref{tab:lower-maximal-conventions} below.

{\renewcommand{\arraystretch}{1.35}
\begin{table}[h!]
\centering
\caption{Notation conventions for maximal and lower variants of the simplicial measures.}
\begin{tabular}{ccc}
\hline
\textbf{Generic notation} & \textbf{Maximal variant} & \textbf{Lower variant}\\
\hline
$\mathcal{G}_{q}$        & $\mathcal{G}_{q}^{A}$        & $\mathcal{G}_{q}^{L}$ \\
$\mathcal{E}_{q}$        & $\mathcal{E}_{q}^{A}$        & $\mathcal{E}_{q}^{L}$ \\
$(\sigma, \tau)$  & $(\sigma, \tau)_{A}$   & $(\sigma, \tau)_{L}$  \\
$\mathcal{H}_{q}$        & $\mathcal{H}_{q}^{A}$        & $\mathcal{H}_{q}^{L}$ \\
$\vec{d}_{q}$            & $\vec{d}_{q}^{A}$            & $\vec{d}_{q}^{L}$ \\
$\mathrm{st}^{\bullet}_{q}$ & $\mathrm{st}^{\bullet}_{A_{q^{*}}}$ & $\mathrm{st}^{\bullet}_{L_{q}}$ \\
$\deg^{\bullet}_{q}$ & $\deg^{\bullet}_{A_{q^{*}}}$ & $\deg^{\bullet}_{L_{q}}$ \\
\hline
\end{tabular}
\label{tab:lower-maximal-conventions}
\end{table}
}

\subsection{Distance-Based Simplicial Measures}
\label{sec:simp-distance-measures}

In this part, we extend the distance-based measures as defined for digraphs in Subsection \ref{subsec:Distance-Related-Measures} to $q$-digraphs. These new simplicial measures can be seen as measures of higher-order global integration of a directed network, that is, how the network is integrated at various levels of organization.

\subsubsection{Average Shortest Directed Simplicial q-Walk Length}

The \textit{average shortest directed simplicial q-walk length} of $\mathcal{G}_{q}$ is defined as the simplicial analogue of the directed version of the average shortest path length (\ref{eq:CPL}), i.e.
\begin{equation}\label{eq:average-shortest-simp-q-walk}
\vec{L}_{q}(\mathcal{G}_{q}) = \frac{1}{|\mathcal{V}_{q}|} \sum_{\sigma \in \mathcal{V}_{q}}\frac{\sum_{\tau \in \mathcal{V}_{q}, \tau \neq \sigma} \vec{d}_{q}(\sigma, \tau)}{|\mathcal{V}_{q}|-1} = \sum_{\substack{\sigma, \tau\in \mathcal{V}_{q} \\ \sigma\neq \tau}} \frac{\vec{d}_{q}(\sigma, \tau)}{|\mathcal{V}_{q}|(|\mathcal{V}_{q}|-1)}.
\end{equation}

For a weighted $q$-digraph $\mathcal{G}^{\widetilde{\omega}}_{q}$, the weighted version of the formula (\ref{eq:average-shortest-simp-q-walk}) is obtained by replacing $\vec{d}_{q}$ with $\vec{d}_{q}^{\omega}$. Also, for computational purposes, we may consider $\vec{d}_{q}(\sigma, \tau) = 0$ instead of $\vec{d}_{q}(\sigma, \tau) = \infty$ when $(\sigma, \tau) \not\in \mathcal{E}_{q}$, otherwise $|\mathcal{V}_{q}|$ must be replaced with the order of the giant $q$-component.

\subsubsection{Directed Simplicial q-Eccentricity}

In the literature, there are different ways to define the eccentricity of a simplex \citep{Atkin1977, Johnson}. Here, however, we extend the definition of simplicial eccentricity as proposed in \citep{Serrano2020} to directed simplices. Following the formula (\ref{eq:eccentricity}), for a strongly $q$-connected $q$-digraph $\mathcal{G}_{q}$, we define the \textit{directed simplicial $q$-eccentricity} of $\sigma \in \mathcal{V}_{q}$ as the maximum directed simplicial $q$-distance from $\sigma$ to any other $\tau \in \mathcal{V}_{q}$, i.e.
\begin{equation}\label{eq:simp-eccentricity}
\vec{\mathrm{ecc}}_{q}(\sigma) = \displaystyle \max_{\tau \in \mathcal{V}_{q}} \vec{d}_{q}(\sigma, \tau).
\end{equation}

Considering the formula (\ref{eq:simp-eccentricity}), the \textit{directed simplicial $q$-diameter} and the \textit{directed simplicial $q$-radius} of $\mathcal{G}_{q}$ are defined in an analogous way to the formulas (\ref{eq:diameter}) and (\ref{eq:radius}), respectively.

\subsubsection{Directed Simplicial Global q-Efficiency}

The \textit{directed simplicial global q-efficiency} of $\mathcal{G}_{q}$ is defined as the simplicial analogue of the directed version of the global efficiency (\ref{eq:global-efficiency}), i.e.
\begin{equation}\label{eq:simp-global-efficiency}
   \vec{E}_{glob}^{q}(\mathcal{G}_{q}) = \frac{1}{|\mathcal{V}_{q}|} \sum_{\sigma \in \mathcal{V}_{q}}\frac{\sum_{\tau \in \mathcal{V}_{q}, \tau \neq \sigma} \vec{d}_{q}^{-1}(\sigma, \tau)}{|\mathcal{V}_{q}|-1}.
\end{equation}
  
For a weighted $q$-digraph $\mathcal{G}^{\widetilde{\omega}}_{q}$, the weighted version of the formula (\ref{eq:simp-global-efficiency}) is obtained by replacing $\vec{d}_{q}$ with $\vec{d}_{q}^{\omega}$.



\subsubsection{Simplicial q-Communicability}

The \textit{simplicial $q$-communicability} between two directed simplices $\sigma, \tau \in \mathcal{V}_{q}$ is defined as the simplicial analogue of the communicability between two vertices in a digraph (\ref{eq:communicability-2}), and therefore it can be written in terms of the powers of the $q$-adjacency matrix as
\begin{equation}\label{eq:simp-communicability}
CM_{q}(\sigma, \tau) = \sum^{\infty}_{k=0} \frac{(\mathcal{H}_{q}^{k})_{\sigma\tau}}{k!} = (\exp(\mathcal{H}_{q}))_{\sigma\tau}.
\end{equation}

Be aware that the entire discussion on the properties of the adjacency matrices of digraphs made in Subsection \ref{sec:AGT} are equally valid for the $q$-adjacency matrices, therefore $(\mathcal{H}_{q}^{k})_{\sigma\tau}$ represents the number of directed simplicial $q$-walks of length $k$ from $\sigma$ to $\tau$.

\subsubsection{Simplicial q-Returnability}

The \textit{simplicial $q$-returnability} of $\mathcal{G}_{q}$ is defined as the simplicial analogue of the returnability of a digraph (\ref{eq:communicability-2}), and therefore it can be written in terms of the powers of the $q$-adjacency matrix as
\begin{equation}\label{eq:simp-returnability}
K_{r,q}(\mathcal{G}_{q}) = \sum^{\infty}_{k=2} \frac{\mathrm{Tr}(\mathcal{H}_{q}^{k})}{k!} = \mathrm{Tr}(\exp(\mathcal{H}_{q})) - |\mathcal{V}_{q}|.
\end{equation}

Moreover, we can define the \textit{relative simplicial $q$-returnability} as the simplicial analogue of the relative returnability (\ref{eq:relative-returnability}), i.e.
\begin{equation}\label{eq:relative-returnability}
K_{r, q}'(\mathcal{G}_{q}) = \frac{\mathrm{Tr}(\exp(\mathcal{H}_{q})) - |\mathcal{V}_{q}|}{\mathrm{Tr}(\exp(\mathcal{H}_{q}')) - |\mathcal{V}_{q}|},
\end{equation}

\noindent where $\mathcal{H}_{q}'$ is the $q$-adjacency matrix of the underlying $q$-graph.

\subsection{Simplicial Centrality Measures}
\label{sec:simp-centrality-measures}

In this part, we extend the measures of centrality as defined for digraphs in Subsection \ref{subsec:centrality-measures} to $q$-digraphs. These new simplicial measures can be interpreted as measures that try to quantify the ``importance," ``influence," or ``centrality" of a directed clique within the higher-order settings, that is, its ``centrality" at various levels of organization of the network.

\subsubsection{Simplicial Degree Centralities}

The \textit{simplicial in-$q$-degree centrality} of $\sigma \in \mathcal{V}_{q}$ is defined as the simplicial analogue of the in-degree centrality of a node (\ref{eq:in-degree-centrality}), i.e.
\begin{equation}\label{eq:simp-in-degree}
 C_{\deg_{q}}^{-}(\sigma) = \frac{\deg^{-}_{q}(\sigma) }{|\mathcal{V}_{q}| - 1}.
\end{equation}

In the same way, the \textit{simplicial out-$q$-degree centrality} of $\sigma$ is defined as the simplicial analogue of the out-degree centrality of a node (\ref{eq:out-degree-centrality}), i.e.
\begin{equation}\label{eq:simp-out-degree}
 C_{\deg_{q}}^{+}(\sigma) = \frac{\deg^{+}_{q}(\sigma) }{|\mathcal{V}_{q}| - 1}.
\end{equation}

Also, if $\sigma$ is a simplex in the corresponding underlying $q$-graph, then its $q$-degree centrality can be written as
\begin{equation}
 C_{\deg_{q}}(\sigma) = C_{\deg_{q}}^{-}(\sigma) + C_{\deg_{q}}^{+}(\sigma) - \frac{\deg^{\pm}_{q}(\sigma)}{|\mathcal{V}_{q}| - 1}.
\end{equation}

To obtain the \textit{weighted simplicial in-$q$-degree centrality} we simply replace $\deg^{-}_{q}$ with $\deg^{\omega, -}_{q}$ (see formula (\ref{eq:simp-weig-in-degree})), and to obtain the \textit{weighted simplicial out-$q$-degree centrality} we replace $\deg^{+}_{q}$ with $\deg^{\omega, +}_{q}$ (see formula (\ref{eq:simp-weig-out-degree})).

\begin{remark}
For the maximal variant, since $\sigma \in \mathrm{dFl}_{q}^{\ast}(G)$, the maximal $q$-degree in the formulas (\ref{eq:simp-in-degree}) and  (\ref{eq:simp-out-degree}) can be replace by the strictly lower $q$-degree.
\end{remark}

\subsubsection{Directed Simplicial Closeness Centrality}

The \textit{directed simplicial $q$-closeness centrality} of $\sigma \in \mathcal{V}_{q}$ is defined as the simplicial analogue of the directed version of the closeness centrality (\ref{eq:closeness-centrality}), i.e.
\begin{equation}\label{eq:simp-closeness}
\vec{Cl}_{q}(\sigma) = \frac{1}{\sum_{\substack{\tau \in \mathcal{V}_{q} \\ \tau \neq \sigma}} \vec{d}_{q}(\sigma, \tau)}.
\end{equation}

The formula (\ref{eq:simp-closeness}) is defined for weakly $q$-connected $q$-digraphs, thus if $N_{q}$ is the order of 
the giant $q$-component of $\mathcal{G}_{q}$, we define the \textit{normalized directed simplicial $q$-closeness centrality} by
\begin{equation}\label{eq:normalized-simp-closeness}
\vec{Cl}_{q}(\sigma) = \frac{N_{q} - 1}{\sum_{\substack{\tau \in \mathcal{V}_{q} \\ \tau \neq \sigma}} \vec{d}_{q}(\sigma, \tau)}.
\end{equation}

For the weighted case, we simply replace $\vec{d}_{q}$ with $\vec{d}_{q}^{\omega}$.

\smallskip

\begin{remark}
For the corresponding underlying $q$-graph of $\mathcal{G}_{q}$, the formula (\ref{eq:simp-closeness}) corresponds to the $q$-closeness centrality as introduced in \citep{Serrano2020} for simplicial complexes.
\end{remark}

\subsubsection{Directed Simplicial Harmonic Centrality}

The \textit{directed simplicial $q$-harmonic centrality} of $\sigma \in \mathcal{V}_{q}$ is defined as the simplicial analogue of the directed version of the harmonic centrality (\ref{eq:harmonic-centrality}), i.e.
\begin{equation}\label{eq:simp-harmonic}
\vec{HC}_{q}(\sigma) = \sum_{\substack{\tau \in \mathcal{V}_{q} \\ \tau \neq \sigma}} \frac{1}{\vec{d}_{q}(\sigma, \tau)},
\end{equation}

\noindent where the convention $1/\infty = 0$ is adopted. Unlike the directed simplicial $q$-closeness centrality, the directed simplicial $q$-harmonic centrality can be computed for disconnected $q$-digraphs. Also, for the weighted case, we simply replace $\vec{d}_{q}$ with $\vec{d}_{q}^{\omega}$.

\subsubsection{Directed Simplicial Betweenness Centrality}

The \textit{directed simplicial $q$-betweenness centrality} of $\sigma \in \mathcal{V}_{q}$ is defined as the simplicial analogue of the directed version of the between centrality (\ref{eq:betweenness-centrality}), i.e.
\begin{equation}\label{eq:simp-betweenness}
\vec{B}_{q}(\sigma) = \sum_{\substack{\tau, \tau' \in \mathcal{V}_{q} \\ \tau' \neq \tau \neq \sigma}} \frac{\vec{l}^{q}_{\tau'\tau}(\sigma)}{\vec{l}^{q}_{\tau'\tau}},
\end{equation}

\noindent where $\vec{l}^{q}_{\tau'\tau}(\sigma)$ is the number of shortest directed simplicial $q$-walks from $\tau'$ to $\tau$ passing through $\sigma$, and $\vec{l}^{q}_{\tau'\tau}$ is the total number of shortest directed simplicial $q$-walks from $\tau'$ to $\tau$.

The formula (\ref{eq:simp-betweenness}) is defined for weakly $q$-connected $q$-digraphs, thus if $N_{q}$ is the order of the giant $q$-component of $\mathcal{G}_{q}$, we define the \textit{normalized directed simplicial $q$-betweenness centrality} of $\sigma$ by 
\begin{equation}\label{eq:normalized-simp-betweenness}
\vec{B}_{q}(\sigma) = \frac{1}{(N_{q} - 1)(N_{q}-2)} \sum_{\substack{\tau, \tau' \in \mathcal{V}_{q} \\ \tau' \neq \tau \neq \sigma}} \frac{\vec{l}^{q}_{\tau'\tau}(\sigma)}{\vec{l}^{q}_{\tau'\tau}}.
\end{equation}

For the weighted case, we simply consider the weighted version $\vec{l}^{\omega, q}_{\tau'\tau}(\sigma)$, where the shortest directed simplicial $q$-walks are computed in relation to a weight-to-distance function.

\begin{remark}
For the corresponding underlying $q$-graph of $\mathcal{G}_{q}$, the formula (\ref{eq:simp-betweenness}) corresponds to the $q$-betweenness centrality as introduced in \citep{Serrano2020} for simplicial complexes.
\end{remark}

\subsubsection{Simplicial Reaching Centrality}

The \textit{simplicial local $q$-reaching centrality} of $\sigma \in \mathcal{V}_{q}$ is defined as the simplicial analogue of the local reaching centrality (\ref{eq:local-reaching-centrality}), i.e.
\begin{equation}\label{eq:simp-local-reaching-centrality}
C_{R, q}(\sigma) = \frac{r_{\mathcal{G}_{q}}(\sigma)}{|\mathcal{V}_{q}|-1},
\end{equation}

\noindent where $r_{\mathcal{G}_{q}}(\sigma)$ is the number of vertices in $\mathcal{G}_{q}$ which are reachable from $\sigma$. Define $C_{R, q}^{max} = \max_{\sigma \in \mathcal{V}_{q}} C_{R, q}(\sigma)$. The \textit{simplicial global $q$-reaching centrality} is defined as the simplicial analogue of the global reaching centrality (\ref{eq:global-reaching-centrality}), i.e. 
\begin{equation}\label{eq:simp-global-reaching-centrality}
GRC_{q}(\mathcal{G}_{q}) = \frac{\sum_{\sigma \in \mathcal{V}_{q}} [C_{R, q}^{max} - C_{R, q}(\sigma)]}{|\mathcal{V}_{q}|-1}.
\end{equation}

\smallskip

\begin{example}\label{ex:simp-centralities}
Consider the directed flag complex presented in Figure \ref{fig:example-centrality}. Table \ref{tab:table5} presents the values of the following simplicial centralities (lower variants) for all maximal simplices of this complex, for $q=0,1$: in/out-$q$-degree centrality, (non-normalized) directed simplicial $q$-betweenness centrality, and directed simplicial $q$-harmonic centrality. Note that at the level $q=1$ we only have directed $q$-connectivity between $\sigma_{1}$ and $\sigma_{2}$, $\sigma_{1}$ and $\theta_{2}$, and $\theta_{1}$ and $\theta_{2}$. Also, note that $\theta_{1}$ has the largest $\vec{B}_{0}$ and the largest $C^{+}_{deg_0}$, suggesting that it may be the most central simplex in the complex at the level $q=0$, but at the level $q=1$, $\theta_{2}$ may be the most central, since it has the largest $\vec{B}_{1}$ and the largest $C^{-}_{deg_1}$.

\begin{figure}[h!]
\centering
  \includegraphics[scale=0.7]{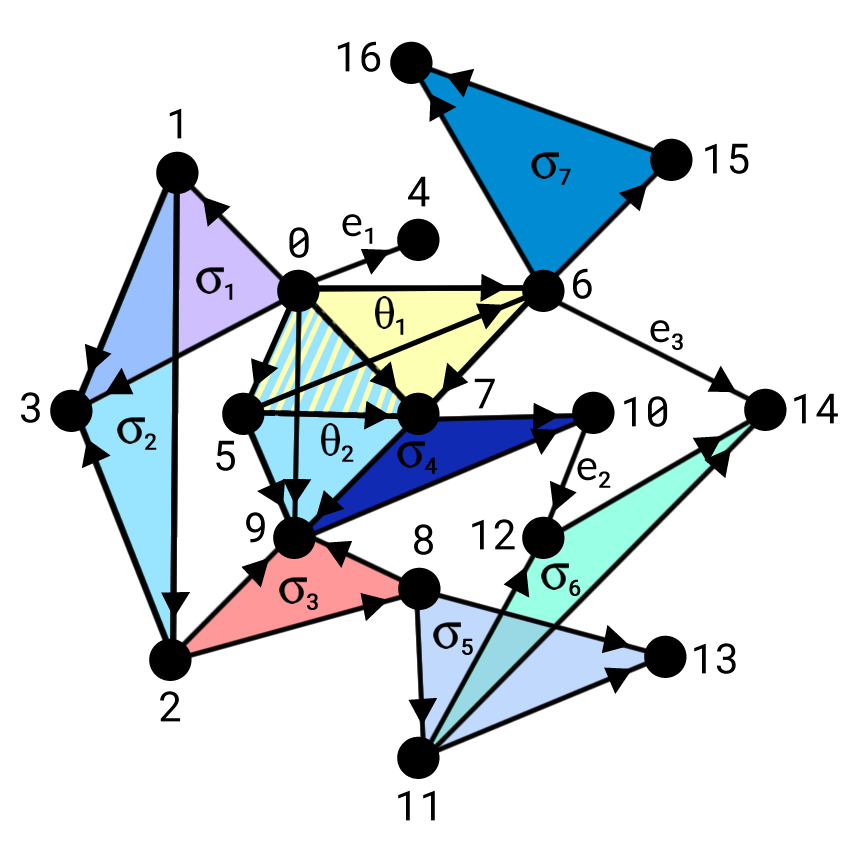}
\caption{A directed flag complex.}
\label{fig:example-centrality}
\end{figure}

\begin{table}[h!]
  \center
    \caption{Directed simplicial centralities for $q=0,1$.}
    \label{tab:table5}
    \begin{tabular}{ c c c c c c c c c }
      \toprule 
\textbf{  } &  ${\bf C^{-}_{deg_{0}} }$ & ${\bf C^{+}_{deg_{0}} }$ & ${\bf C^{-}_{deg_{1}} }$ & ${\bf C^{+}_{deg_{1}} }$ & $\vec{\mathbf{B}}_{0}$ & $\vec{\mathbf{B}}_{1}$ & $\vec{\mathbf{HC}}_{0}$ & $\vec{\mathbf{HC}}_{1}$  \\
      \midrule 
$e_{1}$ & 0.27 & 0.27 & 0.0 & 0.0 & 1.0 & 0.0 & 6.25 & 0.0\\
$e_{2}$ & 0.18 & 0.09 & 0.0 & 0.0 & 1.5 & 0.0 & 5.83 & 0.0\\
$e_{3}$ & 0.27 & 0.27 & 0.0 & 0.0 & 17.34 & 0.0 & 6.42 & 0.0\\
$\sigma_{1}$ & 0.27 & 0.36 & 0.0 & 0.09 & 5.5 & 0.0 & 6.25 & 0.0\\
$\sigma_{2}$ & 0.18 & 0.09 & 0.09 & 0.0 & 1.0 & 0.0 & 5.75 & 1.0\\
$\sigma_{3}$ & 0.36 & 0.36 & 0.0 & 0.0 & 26.17 & 0.0 & 6.99 & 0.0\\
$\sigma_{4}$ & 0.27 & 0.36 & 0.09 & 0.09 & 15.17 & 0.0 & 6.58 & 1.5\\
$\sigma_{5}$ & 0.18 & 0.18 & 0.0 & 0.0 & 9.8 & 0.0 & 5.83 & 0.0\\
$\sigma_{6}$ & 0.27 & 0.27 & 0.0 & 0.0 & 18.67 & 0.0 & 6.34 & 0.0\\
$\sigma_{7}$ & 0.18 & 0.18 & 0.0 & 0.0 & 0.0 & 0.0 & 5.58 & 0.0\\
$\theta_{1}$ & 0.45 & 0.54 & 0.0 & 0.09 & 34.17 & 0.0 & 7.49 & 0.0\\
$\theta_{2}$ & 0.45 & 0.36 & 0.18 & 0.09 & 13.67 & 1.0 & 7.58 & 2.0\\
      \bottomrule 
    \end{tabular}
  
\end{table}

\end{example}

\subsection{Simplicial Segregation Measures}
\label{sec:simp-segregation-measures}

In this part, we extend the measures of segregation as defined for digraphs in Subsection \ref{subsec:segration-measures} to $q$-digraphs. These new simplicial measures can be interpreted as measures that attempt to quantify the tendency of directed cliques to segregate into higher-order clusters or higher-order communities in a directed network.


\subsubsection{Directed Simplicial Clustering Coefficients}

The \textit{average directed simplicial $q$-clustering coefficient} of $\mathcal{G}_{q}$ is defined as the simplicial analogue of the average directed $q$-clustering coefficient (\ref{eq:dir-clustering-coef}), i.e.
\begin{equation}\label{eq:simp-average-cc}
\vec{\bar{C}}_{q}(\mathcal{G}_{q}) = \frac{1}{|\mathcal{V}_{q}|} \sum_{\sigma \in \mathcal{V}_{q}} \frac{\vec{T}_{q}(\sigma)}{\deg_{q}^{tot}(\sigma)(\deg_{q}^{tot}(\sigma) - 1) - 2\deg_{q}^{\pm}(\sigma)},
\end{equation}

\noindent where $\deg_{q}^{tot}(\sigma) = \deg_{q}^{-}(\sigma) + \deg_{q}^{+}(\sigma)$ and $\vec{T}_{q}(\sigma)$ is the number of directed triangles at the level $q$ containing $\sigma$, which can be written in terms of the $q$-adjacency matrix entries in an analogous way to the formula (\ref{eq:dir-number-triangles}):
\begin{equation}\label{eq:dir-simp-number-triangles}
\vec{T}(\sigma)  = \frac{1}{2} \sum_{\tau', \tau \in \mathcal{V}_{q}} (h_{\sigma \tau'} + h_{\tau' \sigma})(h_{\sigma \tau} + h_{\tau \sigma})(h_{\tau' \tau} + h_{\tau \tau'}).
\end{equation}

Moreover, Maletic et al. \citep{Maletic} proposed a simplicial clustering coefficient of a simplex based on its dimension and on the dimensions of the simplices of its neighborhood; we can adapt this coefficient for a simplex $\sigma^{(n)}$ in the underlying $q$-graph of $\mathcal{G}_{q}$ as
\begin{equation}\label{eq:simp-clustering-coef-under}
C_{q}(\sigma^{(n)}) = \sum_{\tau^{(m)} \in \mathrm{st}_{q}(\sigma)} \frac{2^{1+f_{\sigma\tau}} - 1}{2^{n} + 2^{m} - 1},
\end{equation}

\noindent where $f_{\sigma\tau}$ is the dimension of the face shared between $\sigma$ and $\tau$, and $\mathrm{st}_{q}(\sigma) = \mathrm{st}^{-}_{q}(\sigma) \cup \mathrm{st}^{+}_{q}(\sigma) - (\mathrm{st}^{-}_{q}(\sigma) \cap \mathrm{st}^{+}_{q}(\sigma))$ (the directed stars are computed considering the corresponding $\sigma$ in $\mathcal{G}_{q}$). 

Accordingly, in what follows, we propose directed variants of this simplicial coefficient for a directed simplex $\sigma^{(n)}$ in $\mathcal{G}_{q}$. The \textit{simplicial $(\bullet)$-$q$-clustering coefficient}, with $\bullet \in \{ -,+,\pm \}$, is defined by
\begin{equation}\label{eq:in-out-clustering-coef}
\vec{C}_{q}^{\bullet}(\sigma^{(n)}) = \sum_{\tau^{(m)} \in \mathrm{st}^{\bullet}_{q}(\sigma)} \frac{2^{1+f_{\sigma\tau}} - 1}{2^{n} + 2^{m} - 1}.
\end{equation}

Also, the simplicial clustering coefficient (\ref{eq:simp-clustering-coef-under}) can be written in terms of the simplicial $(\bullet)$-$q$-clustering coefficients as
\begin{equation}
C_{q}(\sigma) = \vec{C}_{q}^{-}(\sigma) + \vec{C}_{q}^{+}(\sigma) - \vec{C}_{q}^{\pm}(\sigma).
\end{equation}

\smallskip






\subsubsection{Directed Simplicial Rich-Club Coefficients}

For an integer $k \ge 0$, the \textit{simplicial in-$q$-degree rich-club coefficient} of $\mathcal{G}_{q}$ is defined as the simplicial analogue of the in-degree rich-club coefficient (\ref{eq:in-rich-club-coef}), i.e.
\begin{equation}\label{eq:simp-in-rich-club}
\phi_{q, k}^{-}(\mathcal{G}_{q}) = \frac{{E}_{>k}^{-}(q)}{F_{>k}^{-}(q)(F_{>k}^{-}(q) - 1)},
\end{equation}

\noindent where $F^{-}_{>k}(q)$ is the number of vertices $\sigma \in \mathcal{V}_{q}$ having $\deg_{q}^{-}(\sigma) > k$, and $E^{-}_{>k}(q)$ is the number of $q$-arcs connecting those $F^{-}_{>k}(q)$ vertices.

Analogously, the \textit{simplicial out-$q$-degree rich-club coefficient} of $\mathcal{G}_{q}$ is defined as the simplicial analogue of the out-degree rich-club coefficient (\ref{eq:out-rich-club-coef}), i.e.
\begin{equation}\label{eq:simp-out-rich-club}
\phi_{q, k}^{+}(\mathcal{G}_{q}) = \frac{{E}_{>k}^{+}(q)}{F_{>k}^{+}(q)(F_{>k}^{+}(q) - 1)},
\end{equation}

\noindent where $F^{+}_{>k}(q)$ is the number of vertices $\sigma \in \mathcal{V}_{q}$ having $\deg_{q}^{+}(\sigma) > k$, and $E^{+}_{>k}(q)$ is the number of $q$-arcs connecting those $F^{+}_{>k}(q)$ vertices.

\subsubsection{Directed Simplicial Local Efficiency}

The \textit{directed simplicial local $q$-efficiency} of $\sigma \in \mathcal{V}_{q}$ is defined as the simplicial analogue of the directed version of the local efficiency (\ref{eq:local-efficiency}), i.e.
\begin{equation}\label{eq:simp-local-efficiency}
 \vec{E}^{q}_{loc}(\sigma) = \frac{1}{|\mathcal{V}_{q}|} \sum_{\sigma \in \mathcal{V}_{q}} \vec{E}^{q}_{glob}(\mathcal{G}_{q}(\sigma)),
\end{equation}

\noindent where $\mathcal{G}_{q}(\sigma)$ is the induced subdigraph of $\mathcal{G}_{q}$ formed by elements of the set $\mathrm{st}^{-}_{q}(\sigma) \cup \mathrm{st}^{+}_{q}(\sigma) - \mathrm{st}^{\pm}_{q}(\sigma) \cap \mathrm{st}^{+}_{q}(\sigma))$, excluding $\sigma$. For the weighted case, we simply replace $\vec{E}^{q}_{glob}$ with its weighted version.


\subsection{Simplicial Entropies}
\label{sec:simp-entropies}

In the literature, there are different ways to define entropies for simplicial complexes \citep{Baccini, Dantchev, Maletic2012}. However, here we propose novel simplicial entropies associated with $q$-digraphs, namely: the \textit{simplicial $q$-structural entropy} and the \textit{simplicial in/out-$q$-degree distribution entropies}.

\subsubsection{Simplicial Structural Entropy}

Let us start by introducing a new concept. The \textit{relative $q$-communicability} of a directed simplex $\sigma \in \mathcal{V}_{q}$, denoted by $RCM_{q}(\sigma)$, is defined as the fraction of all directed simplicial $q$-walks in $\mathcal{G}_{q}$ that starts in $\sigma$, i.e.
\begin{equation}\label{eq:relative-q-communicability}
RCM_{q}(\sigma) = \frac{\sum_{\tau \in \mathcal{V}_{q}} CM_{q}(\sigma, \tau)}{\sum_{\tau' \in \mathcal{V}_{q}} \sum_{\tau'' \in \mathcal{V}_{q}} CM_{q}(\tau', \tau'')}.
\end{equation}

\smallskip

The \textit{simplicial $q$-structural entropy} of $\mathcal{G}_{q}$ is defined as the Shannon entropy of the relative $q$-communicabilities:
\begin{equation}\label{eq:simp-struc-entropy}
H^{str}_{q}(\mathcal{G}_{q}) = - \sum_{\sigma \in \mathcal{V}_{q}}  RCM_{q}(\sigma) \log_{2}  RCM_{q}(\sigma).
\end{equation}

\smallskip

This measure can be roughly interpreted as the ``degree of higher-order structural disorder" in the network for each level $q$. 


\subsubsection{Simplicial Degree Distribution Entropy}

Let $\delta^{-}_{q}(k)$ be the number of directed simplices having in-$q$-degree $k$ and $\delta^{+}_{q}(k)$ be the number of directed simplices having out-$q$-degree $k$. The \textit{in-$q$-degree distribution} and the \textit{out-$q$-degree distribution} are defined, respectively, by 

\begin{equation}
p^{-}_{q}(k) = \frac{\delta^{-}_{q}(k)}{|\mathcal{V}_{q}|},
\end{equation}

\begin{equation}
p^{+}_{q}(k) = \frac{\delta^{+}_{q}(k)}{|\mathcal{V}_{q}|}.
\end{equation}

We define the \textit{simplicial in-$q$-degree distribution entropy} of $\mathcal{G}_{q}$ as the simplicial analogue of the in-degree distribution entropy (\ref{eq:in-distribution-entropy}), i.e.
\begin{equation}\label{eq:in-dist-entropy}
H^{-}_{q}(\mathcal{G}_{q}) = - \sum_{k=1}^{n}  p_{q}^{-}(k) \log_{2} p_{q}^{-}(k).
\end{equation}

Similarly, the \textit{simplicial out-$q$-degree distribution entropy} of $\mathcal{G}_{q}$ is defined as the simplicial analogue of the out-degree distribution entropy (\ref{eq:out-distribution-entropy}), i.e.
\begin{equation}\label{eq:out-dist-entropy}
H^{+}_{q}(\mathcal{G}_{q}) = - \sum_{k=1}^{n}  p_{q}^{+}(k) \log_{2} p^{+}_{q}(k).
\end{equation}

Roughly speaking, these entropies can be interpreted as measures of the ``degree of higher-order disorder" (\textit{in relation to} the inner or outer higher-order flux in the case of $H^{-}$ or $H^{+}$, respectively) in the network for each level $q$. Also, similarly to the in/out-degree distributions, $H^{-}_{q}$ and $H^{+}_{q}$ reach their minimum when all directed simplices have the same in/out-$q$-degree, and reach their maximum when $p^{-}_{q}(k)=p^{+}_{q}(k)=1/(|\mathcal{V}_{q}|-1)$, for all $k$.



\subsection{Forman-Ricci Curvature}
\label{sec:discrete-curvatures}



In Riemannian geometry, there are several notions of curvature associated with Riemannian manifolds \citep{Manfredo}. One of these curvatures is the Ricci curvature, which tries to quantify the ``non-flatness" of a Riemannian manifold, or in which ``degree" it deviates from being locally Euclidean \citep{Manfredo, Eidi}. Robin Forman \citep{Forman2003} was the first person to propose a discrete notion of Ricci curvature for cell complexes (known as \textit{Forman-Ricci curvature}). Since then, a myriad of generalized Ricci curvatures for discrete structures have been proposed, such as discrete curvatures for graphs \citep{Saucan, Sreejith2016}, digraphs \citep{Sreejith2017}, hypergraphs \citep{Eidi, Leal}, and simplicial complexes \citep{Yamada2023}. 

Sreejith et al. \citep{Sreejith2016} proposed a version of Forman's discrete version of Ricci curvature for (unweighted and weighted) undirected graphs, and it was later extended to (unweighted and weighted) directed graphs \citep{Sreejith2017}. These generalized curvatures are edge-centric local measures of geometrical characterization that can help us gain insights into the network organization.

In what follows, we present the mathematical formalism of the Forman-Ricci curvature for directed graphs as it was introduced in \citep{Sreejith2017}. Let  $e = (v_{1}, v_{2})$ be an arc of a digraph $G$, and let $\omega(e)$, $\omega(v_{1})$, and $\omega(v_{1})$ denote the weights associated with $e$, $v_{1}$, and $v_{2}$, respectively. The \textit{Forman-Ricci curvature} of the arc $e$ is defined as
\begin{equation}\label{eq:frc-digraph}
  F(e) = \omega(e) \Bigg( \frac{\omega(v_{1})}{ \omega(e) } - \sum_{e_{v_1} \sim e} \frac{\omega(v_{1})}{\sqrt{ \omega(e)\omega(e_{v_1})}} \Bigg) + \omega(e) \Bigg( \frac{\omega(v_{2})}{ \omega(e) } - \sum_{e_{v_2} \sim e}  \frac{\omega(v_{2})}{\sqrt{ \omega(e)\omega(e_{v_2})}} \Bigg),
\end{equation}

\medskip

\noindent where $e_{v_1} \sim e$ represents the arcs whose heads coincide with $v_1$ (i.e., the incoming arcs at node $v_1$) and $e_{v_2} \sim e$ represents the arcs whose tails coincide with $v_2$ (i.e. the outgoing arcs at node $v_2$), excluding the arc $e$.

In view of this, we define the \textit{$q$-Forman-Ricci curvature} for a $q$-arc $E=(\sigma_{1}, \sigma_{2})$ in a $q$-digraph $\mathcal{G}_{q}$ as a straightforward extension of the formula (\ref{eq:frc-digraph}), i.e.
\begin{equation}\label{eq:frc-dir-simp}
  F_{q}(E) = \omega(E) \Bigg( \frac{\omega(\sigma_{1})}{ \omega(E) } - \sum_{E_{\sigma_1} \sim E} \frac{\omega(\sigma_{1})}{\sqrt{ \omega(E)\omega(E_{\sigma_1})}} \Bigg) + \omega(E) \Bigg( \frac{\omega(\sigma_{2})}{ \omega(E) } - \sum_{E_{\sigma_2} \sim E}  \frac{\omega(\sigma_{2})}{\sqrt{ \omega(E)\omega(E_{\sigma_2})}} \Bigg),
\end{equation}

\noindent where $\omega(\sigma_{1})$, $\omega(\sigma_{2})$, and $\omega(\sigma)$ are the weights associated with $\sigma_1$, $\sigma_1$, and $\omega(E)$, respectively, and $E_{\sigma_1} \sim E$ represents the $q$-arcs whose heads coincide with $\sigma_1$ (i.e. the incoming $q$-arcs at node $\sigma_1$) and $E_{\sigma_2} \sim E$ represents the $q$-arcs whose tails coincide with $\sigma_2$ (i.e. the outgoing $q$-arcs at node $\sigma_2$), excluding the $q$-arc $E$.

Now we can construct a straightforward extension of the \textit{in} and \textit{out} Forman-Ricci curvatures of a node in a digraph, as proposed in \citep{Sreejith2017}, to a node in a $q$-digraph as follows. Let $E_{I, \sigma}$ and $E_{O, \sigma}$ denote the set of all incoming $q$-arcs of $\sigma$ (i.e. the $q$-arcs whose heads coincide with $\sigma$) and the set of all outgoing $q$-arcs of $\sigma$ (i.e. the $q$-arcs whose tails coincide with $\sigma$), respectively. The \textit{in-$q$-Forman-Ricci curvature} and the  \textit{out-$q$-Forman-Ricci curvature} of a directed simplex $\sigma$ are defined, respectively, by 
\begin{equation}\label{eq:in-forman-ricci}
F_{q}^{-}(\sigma) = \sum_{E \in E_{I, \sigma}} F_{q}(E),
\end{equation}
  
\begin{equation}\label{eq:out-forman-ricci}
F_{q}^{+}(\sigma) = \sum_{E \in E_{O, \sigma}} F_{q}(E).
\end{equation}
  
In addition, we define the \textit{total $q$-flow} through $\sigma$ as 
\begin{equation}\label{eq:frc-total}
F_{q}(\sigma) = F_{q}^{-}(\sigma) + F_{q}^{+}(\sigma).
\end{equation}

It's important to note that for a weighted $q$-digraph the weights associated with the nodes are given by the product-weight function (\ref{eq:prod-weight}), and the weights associated with the $q$-arcs are given by some non-symmetric node-to-edge function, such as the functions (\ref{eq:node-to-edge1}) and (\ref{eq:node-to-edge2}). On the other hand, for an unweighted $q$-digraph, there are several ways to assign weights to its nodes and $q$-arcs, for instance, we can identify the weight of a node with its in- or out-degree, and use the function (\ref{eq:node-to-edge1}) to transform the node weights into $q$-arc weights.

\smallskip

\begin{example}
Figure \ref{fig:frc-img-1} presents a directed flag complex and its $q$-digraphs, for $q=0,1$. We computed the $q$-Forman-Ricci curvature for the $q$-arcs $E_{12} = (\sigma_{1},\sigma_{2})$ and $E_{23} = (\sigma_{2},\sigma_{3})$, and the in- and out-$q$-Forman-Ricci curvatures for the nodes $\sigma_{1}$, $\sigma_{2}$, and $\sigma_{3}$. The node weights were obtained through the function (\ref{eq:weight-func2}), where we considered the arc weights of the underlying digraph equal to $1$, and the $q$-arc weights were obtained through the node-to-edge function (\ref{eq:node-to-edge1}).

\begin{figure}[h!]
    \centering
\includegraphics[scale=0.98]{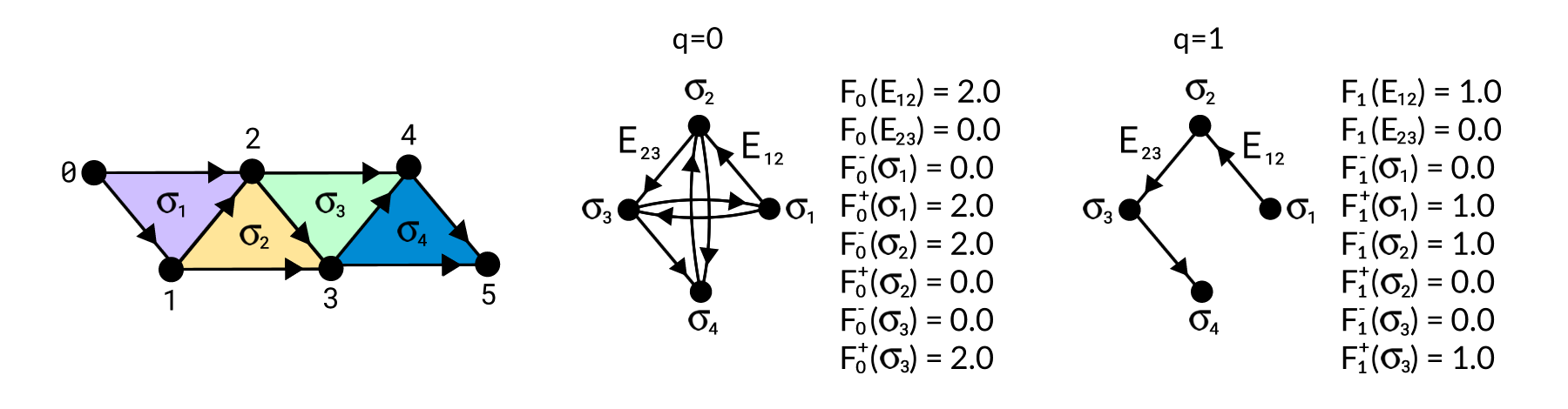}
    \caption{$q$-Forman-Ricci curvatures for $q$-arcs and in- and out-$q$-Forman-Ricci curvatures for nodes in the $q$-digraphs associated with the directed flag complex shown on the left side, for $q=0,1$.}
 \label{fig:frc-img-1}
\end{figure}
\end{example}

Furthermore, from formula (\ref{eq:frc-dir-simp}), we can infer that $q$-arcs connecting nodes with high in-$q$-degrees to nodes with high out-$q$-degrees have highly negative curvature values.

  
 \smallskip



\subsection{Spectrum-Related Simplicial Measures}
\label{sec:simp-spectrum-measures}

In this last part, we extend the spectrum-based measures as defined for digraphs in Subsection \ref{subsec:spectrum-measures} to directed flag complexes and some of them to path complexes. 
These new simplicial measures are related to the spectrum of the $q$-adjacency matrices or the spectrum of the Hodge Laplacian matrices, and they can be interpreted as measures that try to quantify the global ``higher-order structures" of the network through its ``higher-order spectra."






\subsubsection{Simplicial Energy}


Let $\{ \varsigma_{k}^{q} \}_{k}$ be the \textit{singular $q$-values} of $\mathcal{H}_{q}$ (i.e., the square roots of the eigenvalues of the matrix $\mathcal{H}_{q}^{T}\mathcal{H}_{q}$). We define the \textit{simplicial $q$-energy} of $\mathcal{G}_{q}$ as the simplicial analogue of the graph energy (\ref{eq:graph-energy}), i.e. as the trace norm of the matrix $(\mathcal{H}_{q}\mathcal{H}_{q}^{T} )^{1/2}$: 
\begin{equation}\label{eq:simp-graph-energy}
\varepsilon_{q}(\mathcal{G}_{q}) = || ( \mathcal{H}_{q}\mathcal{H}_{q}^{T} )^{1/2}||_{*} = \mathrm{Tr}( (\mathcal{H}_{q}\mathcal{H}_{q}^{T})^{1/2}) = \mathrm{Tr}( (\mathcal{H}_{q}^{T}\mathcal{H}_{q})^{1/2}) = \sum_{k} \varsigma_{k}^{q}.
\end{equation}

A different concept of energy of simplicial complexes related to their Euler characteristics was defined by O. Knill \citep{Knill}.


\subsubsection{Directed Simplicial Katz Centralities}

Let $\mathcal{H}_{q} \in \mathbb{R}^{n_{q} \times n_{q}}$. We define the \textit{simplicial in-$q$-Katz centrality} of $\sigma \in \mathcal{V}_{q}$ as the simplicial analogue of the directed version of the Katz centrality which considers the incoming arcs (\ref{eq:dir-katz-centrality-in}), i.e.
\begin{equation}\label{eq:simp-dir-katz-centrality-in}
K^{-}_{q}(\sigma) = \big[ \langle 1| (I_{n_{q}} - \alpha \mathcal{H}_{q})^{-1}  \big]_{\sigma},
\end{equation}

\noindent and we define the \textit{simplicial out-$q$-Katz centrality} of $\sigma \in \mathcal{V}_{q}$ as the simplicial analogue of the directed version of the Katz centrality that considers the outgoing arcs (\ref{eq:dir-katz-centrality-out}), i.e.
\begin{equation}\label{eq:simp-dir-katz-centrality-out}
K^{+}_{q}(\sigma) = \big[ (I_{n_{q}} - \alpha \mathcal{H}_{q})^{-1} |1 \rangle  \big]_{\sigma},
\end{equation}

\noindent where $|1\rangle = (1,...,1)^{T}$ and $I_{n_{q}}$ is the $n_{q} \times n_{q}$ identity matrix. The subscript $\sigma$ in the brackets represents the position corresponding to $\sigma$ in the vector inside the brackets. 

The attenuation factor must be $\alpha \neq 1/\lambda_{1}^{q}$, where $\lambda_{1}^{q}$ is the largest eigenvalue of $\mathcal{H}_{q}$. Analogously to the graph Katz centrality, typically the attenuation factor is chosen to be $\alpha < 1/\lambda_{1}^{q}$.

\subsubsection{Simplicial Eigenvector Centralities}

The \textit{right simplicial $q$-eigenvector centrality} of $\sigma \in \mathcal{V}_{q}$ is defined as the simplicial analogue of the eigenvector centrality (\ref{eq:eigenvector-centrality}), i.e. as the entry corresponding to $\sigma$ of the right eigenvector associated with the largest eigenvalue ($\lambda_{1}^{q}$) of $\mathcal{H}_{q}$:
\begin{equation}\label{eq:right-simp-eigenvector-centrality}
C_{e, r}^{q}(\sigma) = \big( v^{q}_{1} \big)_{\sigma} = \Bigg( \frac{1}{\lambda_{1}^{q}} \mathcal{H}_{q} v^{q}_{1} \Bigg)_{\sigma}.
\end{equation}

The \textit{left simplicial $q$-eigenvector centrality} of $\sigma$ is defined in an analogous way, but considering the left eigenvector, i.e.
\begin{equation}\label{eq:left-simp-eigenvector-centrality}
C_{e, l}^{q}(\sigma) = \big( v^{q}_{1} \big)_{\sigma} = \Bigg( \frac{1}{\lambda_{1}^{q}} \mathcal{H}^{T}_{q} v^{q}_{1} \Bigg)_{\sigma}.
\end{equation}

Similarly as observed for the graph eigenvector centrality in Subsection \ref{subsec:spectrum-measures}, the Perron-Frobenius theorem (Theorem \ref{theo:perron-frobenius}) guarantees that the right and left eigenvectors associated with $\lambda_{1}^{q}$ are non-negative, thus the right and left simplicial $q$-eigenvalue centralities of every $\sigma$ are non-negative.


\subsubsection{Simplicial Spectral Entropy}


Let $\mathcal{X}$ be a path complex or a directed flag complex associated with a digraph without double edges. Let $\{\mu_{i}^{n}\}_{i}$ be the spectrum of the Hodge $n$-Laplacian matrix of $\mathcal{X}$, $[\mathcal{L}_{n}]$, with multiplicities, and let
\begin{equation}\label{eq:n-spectral-prob}
p(\mu_{i}^{n}) = \frac{\mu_{i}^{n}}{\sum_{i} \mu_{i}^{n}}
\end{equation}

\noindent be the \textit{$n$-eigenvalue probabilities}, i.e. the contribution of $\mu_{i}^{n}$ in the Hodge $n$-Laplacian spectrum (Proposition \ref{prop:hodge-laplacian-eigenvalues} guarantees that $\mu_{i}^{n} \ge 0$, $\forall i$). Assuming the conventions $0/0=0$ and $0 \log_{2} 0 = 0$, we define the \textit{simplicial $n$-spectral entropy} of $\mathcal{X}$ as the Shannon entropy of the $n$-eigenvalue probabilities, i.e.
\begin{equation}\label{eq:simp-spectral-entropy}
S_{n}(\mathcal{X}) = - \sum_{i} p(\mu_{i}^{n}) \log_{2} p(\mu_{i}^{n}).
\end{equation}

\smallskip

Alternative definitions of spectral entropies associated with simplicial complexes were proposed by Maletić and Rajković \citep{Maletic2012} and Baccini et al. \citep{Baccini}. 

\smallskip




\section{Simplicial Similarity Comparison Methods}
\label{sec:simp-distances-kernels}

In this section, we introduce several similarity comparison methods, such as \textit{topological structure distances} and \textit{simplicial kernels}, for digraph-based complexes based on their algebraic/topological/structural properties, such as the number of weakly/strongly $q$-connected components, the number of directed simplices of certain dimensions, Betti numbers, etc. Also, we propose a \textit{simplicial spectral distance} to compare complexes via their Hodge Laplacian spectra.

Throughout this part, all directed flag complexes are considered to be associated with simple digraphs \textit{without double edges}.

\bigskip
\subsection{Topological Structure Vectors and Structure Distances}

In Subsection \ref{sec:dir-q-analysis}, we introduced structure vectors associated with directed flag complexes, each capturing a specific structural property of the complex. Here we introduce additional structure vectors that take into account other topological features of these complexes, for instance, the Betti numbers $\beta_{n} = \dim H_{n}$ (which are associated with the topology of the network) and the lengths of the bars in the persistence barcodes, and we present specific structure vectors for path complexes; subsequently, we propose a general formula for similarity comparison between two digraph-based complexes based on these novel structure vectors.

\smallskip

\begin{definition}\label{def:topological-structure-vectors-1}
Given a directed flag complex $\mathcal{X}$, with $\dim \mathcal{X} = N$, we define five \textit{topological structure vectors} associated with it, namely:

\begin{enumerate}
\item The \textit{$1$st topological structure vector} is defined as $\mathrm{Str}_{1}(\mathcal{X}) = (s^{1}_{0},..., s^{1}_{N})$, where $s^{1}_{i}$ is the number of directed $i$-simplices contained in $\mathcal{X}$, $i=0,...,N$, i.e. $\mathrm{Str}_{1}(\mathcal{X})$ is equal to the second structure vector as introduced in Definition \ref{def:dir-structure-vectors}.

\item The \textit{$2$nd topological structure vector} is defined as $\mathrm{Str}_{2}(\mathcal{X}) = (s^{2}_{0},..., s^{2}_{N})$, where $s^{2}_{q}$ is the number of weakly $q$-connected components of $\mathcal{X}$, $q=0,...,N$, i.e. $\mathrm{Str}_{2}(\mathcal{X})$ is equal to the first weak structure vector as introduced in Definition \ref{def:dir-structure-vectors}.

\item The \textit{$3$rd topological structure vector} is defined as $\mathrm{Str}_{3}(\mathcal{X}) = (s^{3}_{0},..., s^{3}_{N})$, where $s^{3}_{q}$ is the number of strongly $q$-connected components of $\mathcal{X}$, $q=0,...,N$, i.e. $\mathrm{Str}_{3}(\mathcal{X})$ is equal to the first strong structure vector as introduced in Definition \ref{def:dir-structure-vectors}.
   
\item The \textit{$4$th topological structure vector} is defined as $\mathrm{Str}_{4}(\mathcal{X}) = (s^{4}_{0}, ..., s^{4}_{max})$, where $s^{4}_{i} = \beta_{i}$ is the $i$-th Betti number of $\mathcal{X}$, $i=0,..., max$.

\item The \textit{$5$th topological structure vector} is defined as $\mathrm{Str}_{5}(\mathcal{X}) = (s^{5}_{0}, ..., s^{5}_{max})$, where $s^{5}_{i}$ is the average bar length in the $i$-th persistence barcode of $\mathcal{X}$, $i=0,...,max$.

\end{enumerate}
\end{definition}

\smallskip

We could consider more structure vectors (such as the third weak/strong structure vector introduced in Definition \ref{def:dir-structure-vectors}), but here we will only deal with these five vectors.

In the case where $\mathcal{X}$ is path complex, since we did not define any kind of ``higher-order connectivity" between its elementary $p$-paths, we do not have an equivalent for path complexes of the $2$nd and $3$rd topological structure vectors as defined above. Accordingly, we define different topological structure vectors for path complexes as follows.

\smallskip

\begin{definition}\label{def:topological-structure-vectors-2}
Let $\mathcal{P}$ be a path complex and let's denote by $N$ the length of its largest elementary $p$-path. We define five \textit{topological structure vectors} associated with $\mathcal{P}$, namely:

\begin{enumerate}

\item The \textit{$1$st topological structure vector} is defined as $\mathrm{Str}_{1}(\mathcal{P}) = (s^{1}_{0},..., s^{1}_{N})$, where $s^{1}_{i}$ is the number of elementary $i$-paths in $\mathcal{P}$, $i=0,...,N$.

\item The \textit{$2$nd topological structure vector} is defined as $\mathrm{Str}_{2}(\mathcal{P}) = (s^{2}_{0},..., s^{2}_{N})$, where $s^{2}_{i}$ is the number of $(i, 1/i)$-DQCs associated with the elementary $i$-paths of $\mathcal{P}$, $i=0,...,N$.

\item The \textit{$3$rd topological structure vector} is defined as $\mathrm{Str}_{3}(\mathcal{P}) = (s^{3}_{0},..., s^{3}_{N})$, where $s^{3}_{i} = \dim\Omega_{i}$, $i=0,...,N$.
   
\item The \textit{$4$th topological structure vector} is defined as $\mathrm{Str}_{4}(\mathcal{P}) = (s^{4}_{0},..., s^{4}_{max})$, where $s_{i}^{4} = \beta_{i}$ is the $i$-th Betti number, i.e. the dimension of the $i$-th path homology of $\mathcal{P}$, $i=0,...,max$.

\item The \textit{$5$th topological structure vector} is defined as $\mathrm{Str}_{5}(\mathcal{P}) = (s^{5}_{0},..., s^{5}_{max})$, where $s^{5}_{i}$ is the average bar length in the $i$-th persistence barcode of $\mathcal{P}$, $i=0,...,max$.
    
\end{enumerate}
\end{definition}

\smallskip

Due to the computational cost of finding $\partial$-invariant $p$-paths, in the $2$nd topological structure vector only the elementary $p$-paths of $\mathcal{P}$ were considered.

In addition, note that if $G$ is a digraph without double edges and $\mathcal{X}$ and $\mathcal{P}$ are its directed flag complex and its path complex, respectively, the $k$-th entry of the vector $(\mathrm{Str}_{2}(\mathcal{P}) - \mathrm{Str}_{1}(\mathcal{X}))$ is equal to the number of $\partial$-invariant elementary $k$-paths that do not belong to any directed $k$-clique.

\smallskip

\begin{remark}
For a weighted directed flag complex, if one wishes to take the weights into account, it is enough to compute the premetric Dowker complex (see Definition \ref{def:dowker-complex}).
\end{remark}

\smallskip

We now propose a general formula for similarity comparison between two directed flag complexes, or between two path complexes, based on the topological structure vectors defined above.

\smallskip

\begin{definition}
Given two path complexes or two directed flag complexes, $\mathcal{X}_{1}$ and $\mathcal{X}_{2}$, let $\mathrm{Str}_{n}(\mathcal{X}_{1})=(s^{1,n}_{0},...,s^{1,n}_{N^{n}_{1}})$ and $\mathrm{Str}_{n}(\mathcal{X}_{2})=(s^{2,n}_{0},...,s^{2,n}_{N^{n}_{2}})$ be their respective $n$-th topological structure vectors. Without loss of generality, suppose $N_{1}^{n} \ge N_{2}^{n}$. Then consider the new vector $\mathrm{Str}^{\ast}_{n}(\mathcal{X}_{2})=(s^{2,n}_{0},...,s^{2,n}_{N^{n}_{1}})$, such that $s^{2,n}_{k}=0$ for all $N^{n}_{2} < k \le N^{n}_{1}$. Let $|| \cdot ||_{2}$ denote the Euclidean norm (or $2$-norm) in $\mathbb{R}^{N^{n}_{1}}$. We define the (normalized) \textit{$n$-th topological structure distance} between $\mathcal{X}_{1}$ and $\mathcal{X}_{2}$ by

\begin{equation}\label{eq:topological-distance}
\widehat{T}^{n}_{tsd}(\mathcal{X}_{1}, \mathcal{X}_{2}) = \begin{cases}
T^{n}_{tsd}(\mathcal{X}_{1}, \mathcal{X}_{2}), \mbox{ if } \mathcal{X}_{1} \neq \emptyset \mbox{ or } \mathcal{X}_{2} \neq \emptyset,\\
0, \mbox{ if } \mathcal{X}_{1} = \mathcal{X}_{2} = \emptyset,
\end{cases}
\end{equation}

\noindent where

\begin{equation}\label{eq:topological-distance2}
T^{n}_{tsd}(\mathcal{X}_{1}, \mathcal{X}_{2}) = \frac{||\mathrm{Str}_{n}(\mathcal{X}_{1}) - \mathrm{Str}^{\ast}_{n}(\mathcal{X}_{2})||_{2}}{|| \mathrm{Str}_{n}(\mathcal{X}_{1}) ||_{2} + || \mathrm{Str}^{\ast}_{n}(\mathcal{X}_{2}) ||_{2}}.
\end{equation} 

\smallskip

The vectors $\mathrm{Str}_{n}$ correspond to those defined in Definition \ref{def:topological-structure-vectors-1} when $\mathcal{X}_{1}$ and $\mathcal{X}_{2}$ are directed flag complexes, and correspond to those defined in Definition \ref{def:topological-structure-vectors-2} when $\mathcal{X}_{1}$ and $\mathcal{X}_{2}$ are path complexes.
\end{definition}

\smallskip

It can be easily verified that $0 \le \widehat{T}_{tsd}^{n} \le 1$. Indeed, since all entries of the vectors $\mathrm{Str}_{n}$ are non-negative, by the triangular inequality of the Euclidean norm, we have
$||\mathrm{Str}_{n}(\mathcal{X}_{1}) - \mathrm{Str}^{\ast}_{n}(\mathcal{X}_{2}) ||_{2} \le ||\mathrm{Str}_{n}(\mathcal{X}_{1})||_{2} + ||\mathrm{Str}^{\ast}_{n}(\mathcal{X}_{2})||_{2}$.



\subsection{Simplicial Kernels}


As we already discussed in Section \ref{sec:graph-similarity}, graph kernels are distance-based algorithms that produce a similarity score as the output of the comparison between two graphs. Martino et al. \citep{Martino} proposed four kernels for simplicial complexes, namely: \textit{histogram cosine kernel}, \textit{weighted Jaccard kernel}, \textit{edit kernel}, and \textit{stratified edit kernel}. The first two kernels are based on the count of simplices belonging simultaneously to both simplicial complexes being compared, and the last two kernels are based on the counting of edit operations. In what follows, we propose adaptations of these four simplicial kernels to digraph-based complexes. 

\smallskip

\begin{definition}
Given two path complexes or two directed flag complexes, $\mathcal{X}_{1}$ and $\mathcal{X}_{2}$, let $\mathrm{Str}_{1}(\mathcal{X}_{1})=(s^{1,1}_{0},...,s^{1,1}_{N^{1}_{1}})$ and $\mathrm{Str}_{1}(\mathcal{X}_{2})=(s^{2,1}_{0},...,s^{2,1}_{N^{1}_{2}})$ be their respective $1$st topological structure vectors. Without loss of generality, suppose $N_{1}^{1} \ge N_{2}^{1}$. Then consider the new vector $\mathrm{Str}^{\ast}_{1}(\mathcal{X}_{2})=(s^{2,1}_{0},...,s^{2,1}_{N^{1}_{1}})$, such that $s^{2,1}_{k}=0$ for all $N^{1}_{2} < k \le N^{1}_{1}$. Let $\langle \cdot, \cdot \rangle$ denote the Euclidean inner product in $\mathbb{R}^{N^{1}_{1}}$. The (normalized) \textit{histogram cosine kernel} (HCK) is defined by
\begin{equation}\label{eq:HCK}
K_{HC}(\mathcal{X}_{1}, \mathcal{X}_{2}) = \frac{\langle \mathrm{Str}_{1}(\mathcal{X}_{1}), \mathrm{Str}_{1}^{\ast}(\mathcal{X}_{2}) \rangle}{\sqrt{\langle \mathrm{Str}_{1}(\mathcal{X}_{1}), \mathrm{Str}_{1}(\mathcal{X}_{1}) \rangle} \sqrt{\langle \mathrm{Str}_{1}^{\ast}(\mathcal{X}_{2}), \mathrm{Str}_{1}^{\ast}(\mathcal{X}_{2}) \rangle}}.
\end{equation}

The $1$st topological structure vectors correspond to those defined in Definition \ref{def:topological-structure-vectors-1} when $\mathcal{X}_{1}$ and $\mathcal{X}_{2}$ are directed flag complexes, and correspond to those defined in Definition \ref{def:topological-structure-vectors-2} when $\mathcal{X}_{1}$ and $\mathcal{X}_{2}$ are path complexes.
\end{definition}

\smallskip

\begin{definition}
Given two path complexes or two directed flag complexes, $\mathcal{X}_{1}$ and  $\mathcal{X}_{2}$, the \textit{Jaccard kernel} is defined as
\begin{equation}\label{eq:jaccard-kernel}
K_{J}(\mathcal{X}_{1}, \mathcal{X}_{2}) = 
\begin{cases}
1 - \frac{|\mathcal{X}_{1} \cap \mathcal{X}_{2}|}{|\mathcal{X}_{1} \cup \mathcal{X}_{2}|}, \mbox{ if } \mathcal{X}_{1} \neq \emptyset \mbox{ or } \mathcal{X}_{2} \neq \emptyset,\\
0, \mbox{ if } \mathcal{X}_{1} = \mathcal{X}_{2} = \emptyset,
\end{cases}
\end{equation}

\noindent where $|\mathcal{X}_{1} \cap \mathcal{X}_{2}|$ represents the cardinality of the intersection and $|\mathcal{X}_{1} \cup \mathcal{X}_{2}|$ represents the cardinality of the union.
\end{definition}

\smallskip

The Jaccard kernel is a normalized similarity distance. In fact, since $|\mathcal{X}_{1} \cap \mathcal{X}_{2}| \le |\mathcal{X}_{1} \cup \mathcal{X}_{2}|$, we have $0 \le K_{J} \le 1$, and $K_{J} = 0$ when $\mathcal{X}_{1} = \mathcal{X}_{2}$ and $K_{J} = 1$ when $\mathcal{X}_{1} \cap \mathcal{X}_{2} = \emptyset$.

\smallskip

\begin{definition}
Given two path complexes or two directed flag complexes, $\mathcal{X}_{1}$ and  $\mathcal{X}_{2}$, let $e(\mathcal{X}_{1}, \mathcal{X}_{2})$ be an edit distance between them (i.e. a distance based on the number of changes necessary to convert one into the other - see Section \ref{sec:graph-similarity}). We define the (normalized) \textit{edit kernel} as
\begin{equation}\label{eq:edit-kernel1}
K_{E}(\mathcal{X}_{1}, \mathcal{X}_{2}) = 
\begin{cases}
1 - \bar{e}(\mathcal{X}_{1}, \mathcal{X}_{2}), \mbox{ if } \mathcal{X}_{1} \neq \emptyset \mbox{ or } \mathcal{X}_{2} \neq \emptyset,\\
0, \mbox{ if } \mathcal{X}_{1} = \mathcal{X}_{2} = \emptyset,
\end{cases}
\end{equation}

\noindent where
\begin{equation}\label{eq:edit-kernel2}
\bar{e}(\mathcal{X}_{1}, \mathcal{X}_{2}) = \frac{ 2e(\mathcal{X}_{1}, \mathcal{X}_{2}) }{ |\mathcal{X}_{1}| + |\mathcal{X}_{2}| + e(\mathcal{X}_{1}, \mathcal{X}_{2}) }.
\end{equation}
\end{definition}

\smallskip

\begin{definition}
Given two directed flag complexes (or two path complexes), $\mathcal{X}_{1}$ and  $\mathcal{X}_{2}$, let $\mathcal{D}$ denote the set of all different dimensions (or lengths) of the simplices (or elementary paths) present in these two complexes. The \textit{stratified edit kernel} (SEK) is defined as
\begin{equation}\label{eq:stratified-edit-kernel}
K_{SE}(\mathcal{X}_{1}, \mathcal{X}_{2}) = \frac{1}{|\mathcal{D}|}\sum_{k \in \mathcal{D}}  K_{E}(\mathcal{X}_{1}^{k}, \mathcal{X}_{2}^{k}),
\end{equation}

\noindent where $\mathcal{X}_{1}^{k} \subseteq \mathcal{X}_{1}$ and $\mathcal{X}_{2}^{k} \subseteq \mathcal{X}_{2}$ are the subsets of all directed $k$-simplices (or elementary $k$-paths) in the respective complexes, and $K_{E}$ is the edit kernel (\ref{eq:edit-kernel1}).
\end{definition}

\smallskip

\subsection{Simplicial Spectral Distance}

The distances between two digraph-based complexes defined in the previous two sections are mainly based on the structural properties of these complexes. Now we present a novel distance based on their Hodge Laplacian spectra, i.e. a \textit{higher-order spectral distance}, which is also connected to their topological structures via their spectra (see Subsection \ref{sec:comb-hodge-laplacian}).

Let $\mathcal{X}_{1}$ and $\mathcal{X}_{2}$ be two directed flag complexes (or two path complexes) with Hodge $n$-Laplacian matrices $[\mathcal{L}_{n}^{1}]$ and $[\mathcal{L}_{n}^{2}]$, respectively. Let $\mu^{n}_{0}(\mathcal{X}_{1}) \le... \le \mu^{n}_{N_{1}}(\mathcal{X}_{1})$ and $\mu^{n}_{0}(\mathcal{X}_{2}) \le... \le \mu^{n}_{N_{2}}(\mathcal{X}_{2})$ be the spectra of $[\mathcal{L}_{n}^{1}]$ and $[\mathcal{L}_{n}^{2}]$, respectively. Without loss of generality, suppose $N_{1} \le N_{2}$. The \textit{simplicial $n$-spectral distance} between $\mathcal{X}_{1}$ and $\mathcal{X}_{1}$ is defined as

\begin{equation}\label{eq:simp-spectral-distance}
\widehat{D}_{\mathcal{L}_{n}}(\mathcal{X}_{1}, \mathcal{X}_{2}) = \begin{cases}
D_{\mathcal{L}_{n}}(\mathcal{X}_{1}, \mathcal{X}_{2}), \mbox{ if } [\mathcal{L}_{n}^{1}] \neq 0 \mbox{ or } [\mathcal{L}_{n}^{2}] \neq 0,\\
0, \mbox{ if } [\mathcal{L}_{n}^{1}] = [\mathcal{L}_{n}^{2}] = 0 ,
\end{cases}
\end{equation}
\noindent where
\begin{equation}
D_{\mathcal{L}_{n}}(\mathcal{X}_{1}, \mathcal{X}_{2}) = \frac{1}{S_{n}(\mathcal{X}_{1}, \mathcal{X}_{2})} \Bigg( \sum^{N_{1}}_{k=0} | \mu^{n}_{k}(\mathcal{X}_{1}) - \mu^{n}_{k}(\mathcal{X}_{2})| + \sum^{N_{2}}_{k=N_{1}+1} |\mu^{n}_{k}(\mathcal{X}_{2})| \Bigg),
\end{equation}
\noindent with the normalization term
\begin{equation}
S_{n}(\mathcal{X}_{1}, \mathcal{X}_{2}) = \sum^{N_{1}}_{k=0} |\mu^{n}_{k}(\mathcal{X}_{1})| + \sum^{N_{2}}_{k=0} |\mu^{n}_{k}(\mathcal{X}_{2})|.
\end{equation}

\smallskip

It can be easily verified that $0 \le \widehat{D}_{n} \le 1$. Indeed, by Proposition \ref{prop:hodge-laplacian-eigenvalues}, for all $n$, all eigenvalues of the Hodge $n$-Laplacian matrices are non-negative, thus  $| \mu^{n}_{k}(\mathcal{X}_{1}) - \mu^{n}_{k}(\mathcal{X}_{2})| \le | \mu^{n}_{k}(\mathcal{X}_{1})| + |\mu^{n}_{k}(\mathcal{X}_{2})|$, for all $k$.

Moreover, notice that the distance (\ref{eq:simp-spectral-distance}) compare solely the spectra of the Hodge $n$-Laplacians of $\mathcal{X}_{1}$ and $\mathcal{X}_{2}$ for a same order $n$. 

\section{Examples with Random Digraphs}

In this last part, we present some examples of applications of the simplicial characterization measures and simplicial similarity comparison distances to random digraphs built through the four random digraph models introduced in Section \ref{sec:random-graphs}.

\smallskip

\begin{example}
Figure \ref{fig:ex-random-digraphs} presents four sparse (density $\le 0.5$) random digraphs of same order $n=20$ obtained through four different random models, namely: $k$-regular (KR) model with parameter $k=6$; Erdős-Rényi $G(n,p)$ (ER) model with parameter $p=0.2$;  Watts-Strogatz (WS) model with parameters $k = 10$ and $p = 0.2$; and Barabási-Albert (BA) model with parameters $\alpha =0.41, \beta=0.54, \gamma=0.05, \delta_{in}=0.2,$ and $\delta_{out} = 0.0$. Also, Figure \ref{fig:ex-random-digraphs} depicts the respective $q$-digraphs associated with these four digraphs, for $q=0,1,2$. As a convention, let's refer to the original digraphs as $(-1)$-digraphs, i.e. we will use $q=-1$ to designate the original networks.

We applied eighteen simplicial characterization measures (lower variants) (described in Table \ref{tab:measures-analysis})\footnote{We used the parameter $k=6$ for the simplicial in- and out-$q$-degree rich-club coefficients.} to these four random digraphs and their respective $q$-digraphs, $q=0,1,2$, whose results are presented in Table \ref{tab:ex1-meas}, and we performed a pairwise comparison between the original digraphs using nine different simplicial similarity comparison distances (described in Table \ref{tab:distances-analysis})\footnote{For the bottleneck, Wasserstein, Betti, 4th, and 5th topological structure distances we only considered the 0-th Betti numbers. For the simplicial $n$-spectral distance, we considered $n=1$.}, whose results are presented in Table \ref{tab:ex1-dist}. To compute the measures, since some of them depend on the order (N) of the network (i.e., are \textit{N-dependent} \citep{vanWijk}), we used a fixed N equal to the order of the largest network among all levels $q$.

The idea of this example is to present a simple numerical analysis of the results just to get a sense of how the simplicial measures behave at each level $q$ for digraphs with different topologies, and not to perform a rigorous statistical analysis. With that said, inspecting the results, we obtained some interesting insights:

\begin{itemize}
\item Distance-based simplicial measures: The BA network presents the lowest values for the global $q$-efficiency (measure 2) and for the average shortest $q$-walk length (measure 1) in relation to other networks at all levels $q$, which might suggest that information is propagated less efficiently in this network than in the other networks, including at higher-order levels of network topological organization.

\item Simplicial centrality measures: The in- and out-$q$-degree centralities (measures 4 and 5, respectively), as expected, are greater at the level $q=0$ for all networks and become smaller as we increase the level $q$ since the $q$-digraphs become more disconnected (sparse). However, the global $q$-reaching centrality (measure 8) increases for some networks, which might suggest that higher-order structures have a greater influence on the information flow.

\item Simplicial segregation measures: The ER and BA networks have high in-$q$-rich-club coefficients (measure 10) in the original networks, which might suggest the presence of densely connected nodes with high in-degrees; however, this property is not preserved at other levels $q$. Also, the out-$q$-rich-club coefficients (measure 11) are equal to zero for all the original networks and at level  $q=2$, but they are greater than zero at levels $q=0$ and $q=1$. 


\item Simplicial similarity comparison distances: For most of the distances, the pairs KR-BA, ER-BA, and WS-BA produced values closest to 1, which might suggest that these pairs have the most different topologies compared to each other. However, most distances produced values close to 0 for the pairs KR-WS and ER-WS, which might suggest that their topologies are similar.
\end{itemize}
\end{example}

\begin{figure}[h!]
    \centering
\includegraphics[scale=0.67]{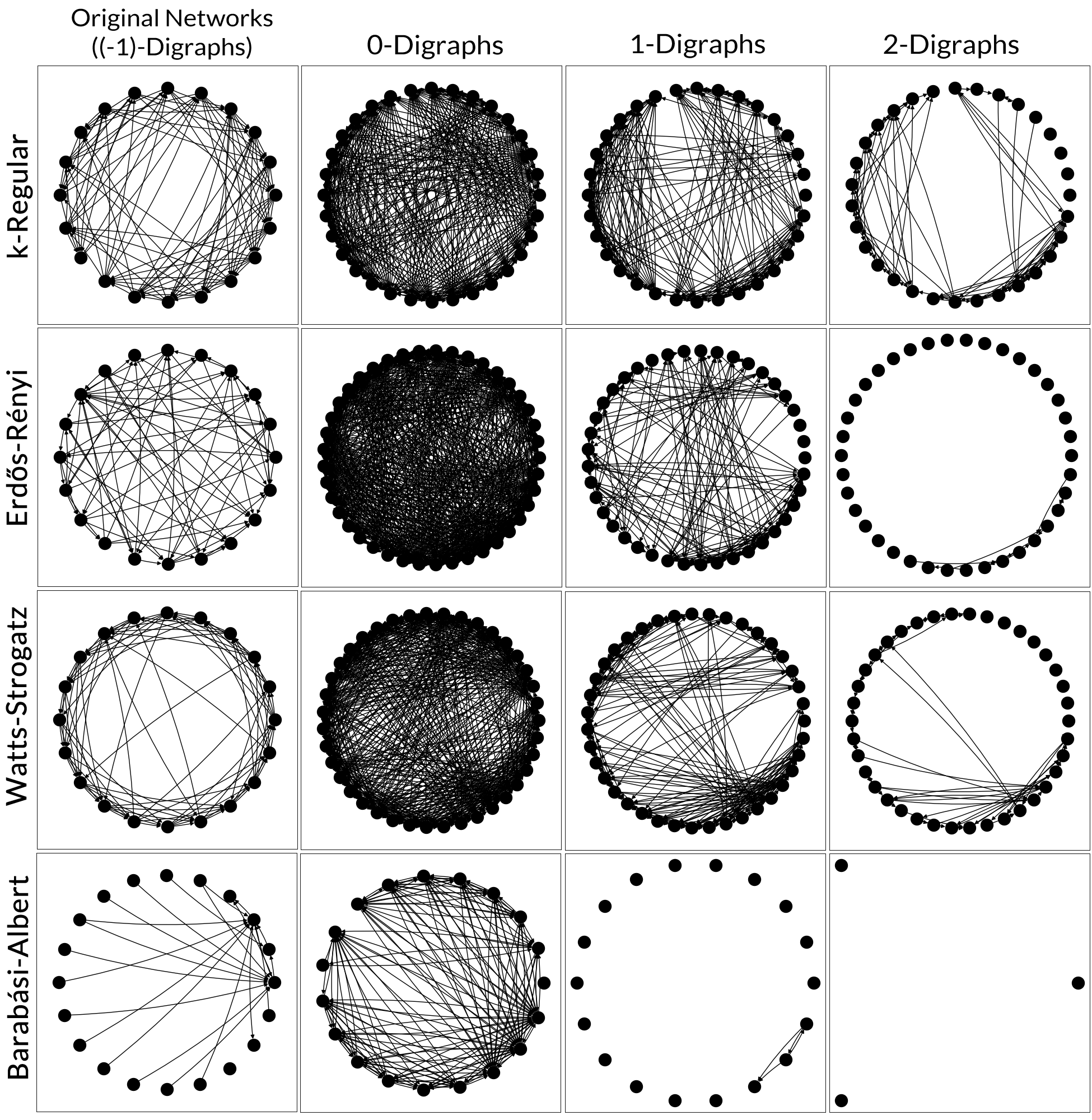}
 \caption{Random digraphs and their respective $q$-digraphs  ($q=0,1,2$).}
 \label{fig:ex-random-digraphs}
\end{figure}

\begin{table}[h!]
\center
\tiny
    \caption{Results of the simplicial characterization measures for the random digraph models at each level $q=-1, 0,1,2$. See Table \ref{tab:measures-analysis} for the measure ids.}
    \label{tab:ex1-meas}
    \begin{tabular}{c c c c c c c c c c c c c c c c c c c c}
     \toprule 
      \multicolumn{1}{c}{} & \multicolumn{4}{c}{\textbf{$q=-1$}} & \multicolumn{4}{c}{\textbf{$q=0$}} & \multicolumn{4}{c}{\textbf{$q=1$}} & \multicolumn{4}{c}{\textbf{$q=2$}} \\
\cmidrule(rl){2-5} \cmidrule(rl){6-9} \cmidrule(rl){10-13} \cmidrule(rl){14-17} \textbf{Measure} & \textbf{KR} & \textbf{ER} & \textbf{WS} & \textbf{BA}  & \textbf{KR} & \textbf{ER} & \textbf{WS} & \textbf{BA} & \textbf{KR} & \textbf{ER} & \textbf{WS} & \textbf{BA} & \textbf{KR} & \textbf{ER} & \textbf{WS} & \textbf{BA}\\
      \midrule 
1 & $2.03$ & $1.66$ & $2.22$ & $0.49$ & $1.57$ & $1.55$ & $1.58$ & $1.54$ & $2.35$ & $3.12$ & $2.6$ & $0.67$ & $1.94$ & $1.0$ & $6.06$  & $0.0$\\
2 & $0.09$ & $0.07$ & $0.08$ & $0.02$ & $0.29$ & $0.49$ & $0.43$ & $0.09$ & $0.21$ & $0.23$ & $0.27$ & $0.0$ & $0.08$ & $0.01$ & $0.08$  & $0.0$\\
3 & $0.0$ & $0.0$ & $0.01$ & $0.01$ & $0.2$ & $0.26$ & $0.18$ & $0.24$ & $0.16$ & $0.31$ & $0.13$ & $0.21$ & $0.1$ & $0.19$ & $0.12$  & $0.0$\\
$4^{*}$ & $0.12$ & $0.16$ & $0.14$ & $0.18$ & $0.37$ & $0.57$ & $0.45$ & $0.33$ & $0.24$ & $0.2$ & $0.22$ & $0.04$ & $0.14$ & $0.06$ & $0.08$  & $0.0$\\
$5^{*}$  & $0.12$ & $0.18$ & $0.1$ & $0.06$ & $0.37$ & $0.55$ & $0.43$ & $0.33$ & $0.24$ & $0.16$ & $0.18$ & $0.04$ & $0.1$ & $0.04$ & $0.06$  & $0.0$\\
$6^{*}$  & $12.3$ & $12.5$ & $12.6$ & $12.3$ & $24.5$ & $34.0$ & $30.0$ & $16.5$ & $19.1$ & $21.1$ & $22.5$ & $2.0$ & $9.83$ & $4.0$ & $9.02$  & $0.0$\\
$7^{*}$  & $53.3$ & $66.9$ & $45.7$ & $36.0$ & $50.6$ & $54.6$ & $44.7$ & $48.9$ & $251$ & $301$ & $157$ & $0.0$ & $169$ & $6.0$ & $348$  & $0.0$\\
8 & $0.24$ & $0.27$ & $0.24$ & $0.07$ & $0.23$ & $0.15$ & $0.17$ & $0.23$ & $0.24$ & $0.2$ & $0.18$ & $0.04$ & $0.36$ & $0.07$ & $0.27$  & $0.0$\\
9 & $0.12$ & $0.09$ & $0.1$ & $0.03$ & $0.36$ & $0.5$ & $0.4$ & $0.29$ & $0.34$ & $0.23$ & $0.25$ & $0.04$ & $0.17$ & $0.0$ & $0.06$  & $0.0$\\
10 & $0.0$ & $12.83$ & $0.0$ & $11.0$ & $0.44$ & $0.48$ & $0.43$ & $0.58$ & $0.67$ & $2.84$ & $1.06$ & $0.0$ & $0.0$ & $0.0$ & $0.0$  & $0.0$\\
11 & $0.0$ & $0.0$ & $0.0$ & $0.0$ & $0.44$ & $0.45$ & $0.43$ & $0.75$ & $0.55$ & $13.25$ & $1.68$ & $0.0$ & $0.0$ & $0.0$ & $0.0$  & $0.0$\\
12 & $4.28$ & $4.3$ & $4.45$ & $5.11$ & $4.95$ & $5.2$ & $5.26$ & $4.07$ & $4.75$ & $4.9$ & $5.11$ & $5.49$ & $4.52$ & $5.46$ & $5.11$  & $0.0$\\
13 & $1.51$ & $2.01$ & $1.8$ & $0.66$ & $2.9$ & $3.77$ & $3.17$ & $1.82$ & $2.69$ & $2.94$ & $3.2$ & $0.38$ & $2.03$ & $0.78$ & $1.96$  & $0.0$\\
14 & $1.21$ & $1.73$ & $1.46$ & $1.26$ & $2.78$ & $3.97$ & $3.09$ & $1.54$ & $2.76$ & $2.76$ & $2.81$ & $0.24$ & $2.15$ & $0.82$ & $1.87$  & $0.0$\\
15 & -$360$ & -$408$ & -$258$ & -$381$ & -$126$ & -$782$ & -$759$ & -$60$ & -$103$ & -$28$ & -$74$ & $0.0$ & -$69$ & $0.0$ & -$15$  & $0.0$\\
16 & -$360$ & -$785$ & -$298$ & -$471$ & -$139$ & -$790$ & -$750$ & -$63$ & -$134$ & -$25$ & -$87$ & $0.0$ & -$132$ & $0.0$ & -$22$  & $0.0$\\
17 & $33.4$ & $30.2$ & $33.0$ & $9.23$ & $78.7$ & $106$ & $104$ & $33.3$ & $61.6$ & $63.1$ & $71.3$ & $2.73$ & $37.2$ & $8.49$ & $34.6$  & $0.0$\\
$18^{*}$  & $2.41$ & $2.25$ & $2.18$ & $1.98$ & $1.65$ & $1.67$ & $1.66$ & $29.3$ & $8.76$ & $3.65$ & $4.59$ & $1.22$ & $2.15$ & $1.33$ & $1.51$  & $0.0$\\
      \bottomrule 
    \end{tabular}
{\footnotesize $^{*}$  The results correspond to the maximum values obtained among all nodes. }
\end{table}

\begin{table}[h!]
\footnotesize
  \center
    \caption{Results of pairwise comparisons of the random digraphs via the simplicial distances. See Table \ref{tab:distances-analysis} for the distance ids.}
    \label{tab:ex1-dist}
    \begin{tabular}{c c c c c c c}
     \toprule 
     \textbf{Distance} & \textbf{KR-ER} & \textbf{KR-WS} & \textbf{KR-BA} & \textbf{ER-WS} & \textbf{ER-BA} & \textbf{WS-BA} \\
      \midrule 
1 & $0.0$ & $0.0$ & $0.5$ & $0.0$ & $0.5$ & $0.5$\\
2 & $0.0$ & $0.0$ & $0.707$ & $0.0$ & $0.707$ & $0.707$\\
3 & $0.0$ & $0.0$ & $1.0$ & $0.0$ & $1.0$ & $1.0$\\
4 & $0.559$ & $0.362$ & $0.895$ & $0.254$ & $0.677$ & $0.794$\\
5 & $0.0$ & $0.0$ & $1.0$ & $0.0$ & $1.0$ & $1.0$\\
6 & $0.0$ & $0.0$ & $0.0$ & $0.0$ & $0.0$ & $0.0$\\
7 & $0.762$ & $0.898$ & $0.393$ & $0.953$ & $0.741$ & $0.582$\\
8 & $0.972$ & $0.981$ & $0.988$ & $0.939$ & $0.982$ & $0.986$\\
9 & $0.432$ & $0.299$ & $0.874$ & $0.408$ & $0.740$ & $0.845$\\
      \bottomrule 
    \end{tabular}

\end{table}

\begin{example}
In this example, we performed a statistical analysis on four global simplicial characterization measures (lower variants) (namely,  average shortest directed simplicial $q$-walk length (measure 1)),  directed simplicial global simplicial $q$-efficiency (measure 2),  simplicial global $q$-reaching centrality (measure 8), and  average directed simplicial $q$-clustering coefficient (measure 9) (measure ids are in accordance with Table \ref{tab:measures-analysis})) computed on $q$-digraphs, for $q=-1,0,1,2$, obtained from samples produced by Monte Carlo simulations in the four random digraph models described in Section \ref{sec:random-graphs} to investigate whether the properties associated with the network topology of each model persist at each level $q$. Due to the limitations of computation time and processing power, we limited ourselves to small and sparse networks. The number of nodes was set to $n = 20$ and, to avoid dense digraphs, we limited the parameters in such a way that the digraph densities were always $\le 0.5$. We ran a Monte Carlo simulation on the parameters of each digraph model using uniform distributions:

\begin{itemize}
\item k-Regular (KR): $k \sim U(1,10)$ (integer part);

\item Erdős-Rényi $G(n,N)$ (ER): $N \sim U(0,190)$ (integer part);

\item Watts-Strogatz (WS): $k=10$, $p \sim U(0,1)$; 

\item Barabási-Albert (BA): $\delta_{in} \sim U(0,1)$, $\delta_{out} = 0; \alpha = 0.41, \beta = 0.54, \gamma = 0.05$.
\end{itemize}

The mean of each measure was obtained from 50 simulations carried out for each model, repeated 20 times. Since some of the measures are N-dependent, we used a fixed N equal to the order of the largest network among all $q$. Figure \ref{fig:ex2-random-digraphs} presents the means and standard deviations obtained for each measure. Subsequently, we compared the means obtained for each random model, for each level $q=-1,0,1,2$, and for each measure, using ANOVA. We found statistically significant differences ($p<0.05$) in all four measures at all levels $q$, except for measure 1 at level $q=0$ ($F=0.1513, p=0.135$). Also, as a post hoc test, we performed a Bonferroni's multiple comparison test and we found statistically significant differences ($p<0.05$) for all pairs of models, for all measures, at all levels $q$, except for the pair ER-KR in measure 1 ($q=1$: $p=0.069$; $q=2$: $p=0.778$), measure 8 ($q=-1$: $p=0.596$, $q=2$: $p=0.151$), and measure 9 ($q=-1$: $p=0.416$; $q=1$: $p=0.052$; $q=2$: $p=0.324$), and for the pair KR-WS in measure 1 ($q=-1$: $p=0.361$) and measure 2 ($q=0$: $p=0.542$).

Moreover, we notice that the ER networks presented a global  $q$-reaching centrality significantly larger than the other network models at levels $q=0,1,2$, which might suggest that the higher-order structures of ER networks propagated information more efficiently than the higher-order structures of the other networks. Also, for measures 8 and 9, it was not possible to find any significant difference between the KR and ER networks in the original networks ($q=-1$), as well as between the KR and WS networks in measures 1 and 2, but it was possible to find significant differences between these networks at other higher-order levels $q$.

We may conclude that the differences between the global topological properties for each random model are, in general, preserved in higher-order structural and connectivity levels, but may differ in different ways at each level. Furthermore, we saw that for some measures it is not possible to identify topological/structural, and/or functional differences between one network model and another, but these differences may be evident at higher-order levels of topological organization (clique organization), and this is one of the advantages of taking the $q$-digraphs associated with the networks into account in the analysis.
\end{example}

\begin{figure}[h!]
    \centering
\includegraphics[scale=0.66]{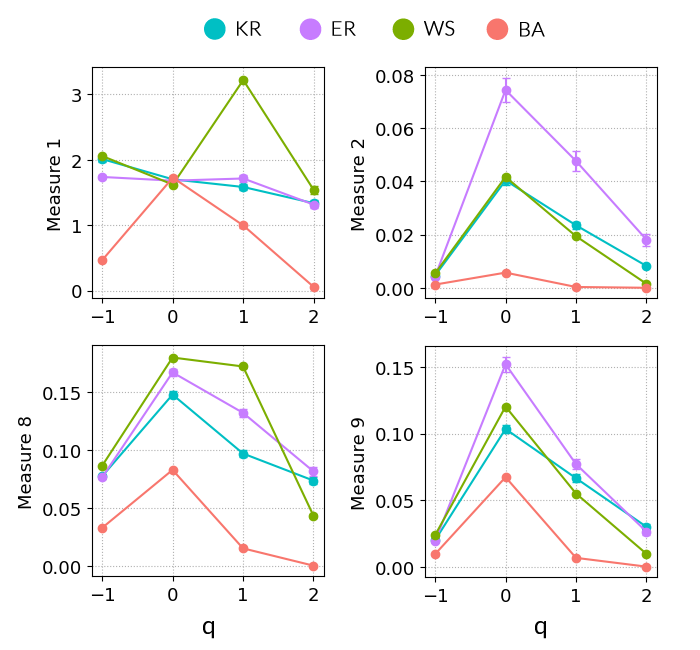}
 \caption{Means and standard deviations of the four simplicial measures computed for the $q$-digraphs, $q=-1,0,1,2$, corresponding to each random digraph model. See Table \ref{tab:measures-analysis} for the measure ids.}
 \label{fig:ex2-random-digraphs}
\end{figure}

\part[Brain Connectivity Networks and a Quantitative Graph/ Simplicial Analysis of Epileptic Brain Networks]{Brain Connectivity Networks and a Quantitative Graph/Simplicial Analysis of Epileptic Brain Networks}
\label{partII}

\chapter[Brain Connectivity Networks]{Brain Connectivity Networks}
\label{chap:chap3}

\epigraph{Since all models are wrong the scientist cannot obtain a ``correct" one by excessive elaboration. On the contrary following William of Occam he should seek an economical description of natural phenomena.}{--- George E. P. Box \citep{Box}}

\bigskip


It is a well-known fact that brain activities depend on several different brain areas rather than being isolated occurrences in one or more distinct brain regions. Neuroanatomical structures located in different regions interact to process information, thus creating structural and functional brain networks. The signals from these neurophysiological activities can be captured by techniques such as functional magnetic resonance imaging (fMRI) and electroencephalography (EEG), which can then be subjected to mathematical and statistical methods, such as partial directed coherence (PDC), which can estimate the causal relationship between two regions by characterizing the influence that one exerts on the other through the concept of Granger causality; these connections can be represented by directed graphs, and thus can be analyzed by tools from graph theory and computational algebraic topology.

In this chapter, we discuss the general aspects of brain connectivity networks, starting with the biophysical principles of brain signals, going through the common methods for acquiring these signals, especially EEG, and then discussing brain connectivity estimators, with special attention to PDC and its variants. Afterwards, we discuss the different types of brain connectivity and briefly discuss the applications of graph theoretical analysis (GTA) and topological data analysis (TDA) in modern network neuroscience research.

\section{Biophysical Principles of Brain Signals}
\label{sec:brain-signals}

The human brain is part of a larger system called the \textit{central nervous system} (CNS). The \textit{nervous system} consists of the \textit{peripheral nervous system} (PNS) and the CNS, and its basic units are the \textit{neurons}  \citep{Kandel}. Neurons are specialized cells in the nervous system that are responsible for transmitting information. Each neuron is composed of a cell body (soma) and a nucleus, dendrites, and an axon, as is schematically represented in Figure \ref{fig:neuron}. To be more precise, Figure \ref{fig:neuron} presents just a general schematic representation of a neuron, since neurons may vary substantially in size and shape, and they may be classified according to their specific morphology. For instance, some neurons have axons that are only a fraction of millimeters long and others can have axons that are many centimeters long; some neurons known as \textit{pyramidal neurons} (discovered by Santiago Ramón y Cajal (1852 - 1934)) have a pyramid-shaped cell body, and they form the most common class of neurons in the neocortex, accounting for approximately 75-90$\%$ of all neurons \citep{Trappenberg}. 

The transmission of information between neurons occurs through connections called \textit{synapses}. The neurons that receive inputs from other neurons are called \textit{postsynaptic} neurons and the neurons that give inputs to other neurons are called \textit{presynaptic} neurons, i.e. postsynaptic neurons receive inputs through synapses from presynaptic neurons. It is clear from the previous passage that synaptic connections are not symmetrical. Moreover, synapses can be electrical or chemical. Electrical synapses allow the direct flow of ionic currents from one neuron to another. Chemical synapses, however, involve the release of \textit{neurotransmitters} from the synaptic vesicles of the presynaptic neuron into the \textit{synaptic cleft} (the small gap between the neurons). These neurotransmitters then bind to receptors on the postsynaptic neuron, producing changes in the membrane potential of this neuron (known as \textit{postsynaptic potentials} (PSPs)), leading to the initiation (in which case the PSPs are known as \textit{excitatory} PSPs) or inhibition (in which case the PSPs are known as \textit{inhibitory} PSPs) of a new action potential \citep{Dayan, Kandel}.

An \textit{action potential} (or \textit{spike}) is an electrical impulse that travels down the axon of a neuron. When a neuron emits an action potential, it is said to be ``firing." The \textit{resting potential} is the potential inside the neuron when it is not firing, and it is approximately -70$mV$. An action potential is initiated when the membrane potential of the neuron reaches a certain threshold, causing an influx of sodium ions ($\mathrm{Na}^{+}$) into the cell. This influx of positively charged sodium ions results in a rapid increase of the membrane potential (depolarization), generating an action potential, followed by a rapid decrease (repolarization) of this potential, reaching values lower than the resting potential for a brief period of time (hyperpolarization), and then returning to the resting potential (see Figure \ref{fig:action-potential}) \citep{Trappenberg}. Once the action potential reaches the end of the axon, it triggers the release of neurotransmitters into the synaptic cleft, which bind to receptors on the postsynaptic neuron and lead to the transmission of information \citep{Dayan}.


\begin{figure}[h!]
\centering
\begin{subfigure}{.5\textwidth}
  \centering
  \includegraphics[scale=0.85]{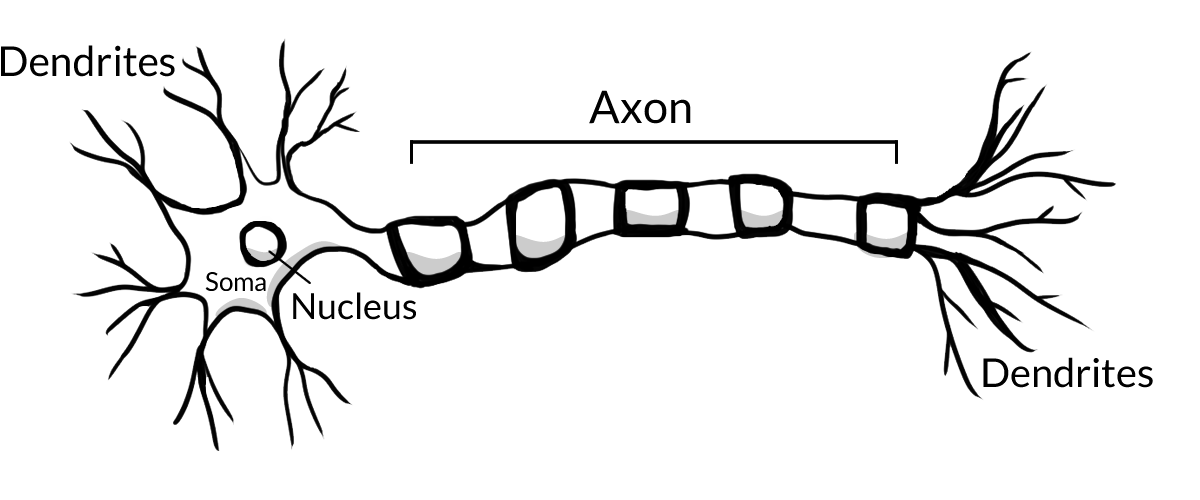}
  \caption{Neuron (schematic representation).}
  \label{fig:neuron}
\end{subfigure}
\begin{subfigure}{.4\textwidth}
  \centering
  \includegraphics[scale=0.9]{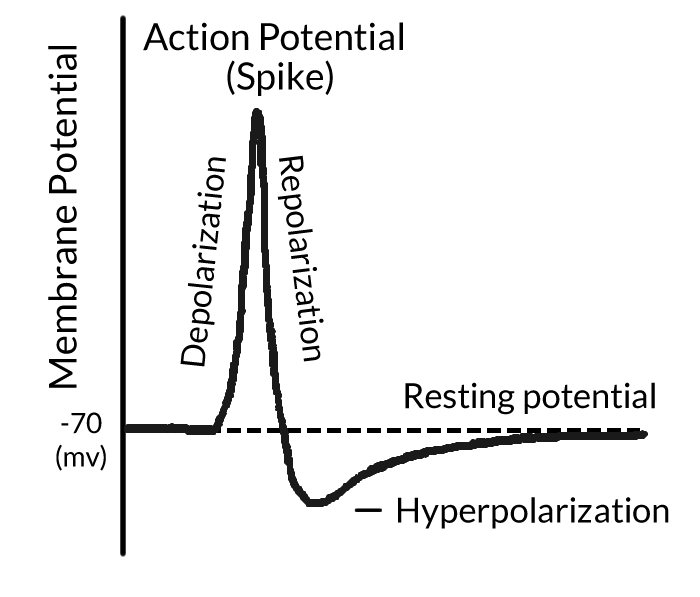}
  \caption{Action potential.}
  \label{fig:action-potential}
\end{subfigure}
\caption{Graphical representation of a neuron and an action potential.}
\label{fig:neuron-ap}
\end{figure}

In the previous paragraphs, we discussed very briefly the biophysical bases of brain activity at the cellular level. However, the brain has a multiscale organization (multiple levels of organization), from molecules to synapses, cells, neuronal networks, and brain areas \citep{Churchland, Frohlich}. This is an important feature of the brain because, although there are (invasive) techniques to detect signals from just a few neurons, we are often interested in the functioning of the brain at a global level. We can study brain functions by utilizing the signals produced by the activity of populations of neurons, which can be recorded by non-invasive devices. The acquisition of brain signals is a fundamental aspect of neuroscience research and clinical diagnosis.

Numerous methods have been developed to measure the neurophysiological activity of the brain, all of which rely on specific signals resulting from the dynamics of brain processes, for example: \textit{electroencephalography} (EEG) involves placing electrodes on the scalp to detect electrical activity \citep{Schomer};  \textit{magnetoencephalography} (MEG) measures the magnetic fields produced by brain's electrical activity, and it requires highly sensitive devices called SQUIDs (superconducting quantum interference devices) \citep{Papanicolaou}; \textit{functional magnetic resonance imaging} (fMRI) detects changes in blood flow and oxygenation levels in the brain \citep{Huettel}; \textit{positron emission tomography} (PET) is a nuclear imaging technique that requires the injection of a radioactive tracer into the bloodstream for subsequent measurement of the tracer's distribution and concentration in the brain \citep{Heurling}.

\section{Electroencephalography}
\label{sec:electroencephalography}

Electroencephalography (EEG) is a non-invasive and economically feasible technique that allows the recording of the brain's electrical activity through electrodes attached to the scalp. In a 1929 report, titled \textit{On the Electroencephalogram of Man}, the German neuropsychiatrist Hans Berger (1873–1941) described for the first time a successful observation of a human EEG, and was later considered the father of EEG \citep{Libenson, Schomer}.

The electrical signals captured by the electrodes are generated by large populations of neurons, since only these large populations are able to generate electric potentials strong enough to be captured by scalp EEG because of the distance between the scalp electrodes and the neurons as well as the high electrical resistance of the medium between them (e.g., skin, skull, dura mater, and cerebral spinal fluid). It is thought that the extracellular potentials captured by EEG are produced by postsynaptic potentials (PSPs) of pyramidal neurons \citep{Schomer}. 


The effects caused by the recording of electrical potentials at a distance from their sources, as occurs in EEG recordings, are known as \textit{volume conduction} \citep{Rutkove}. Generally speaking, a medium, for example, skin, skull, dura mater, and cerebrospinal fluid, will occupy the space between an electrode placed on the scalp and a cerebral source of electrical potential. Electrical signals conducted across this medium result in volume conduction, since the electrical impulses dissipate and refract as a result of this medium conducting them. Volume conduction effects are present in all EEG recordings \citep{vandenBroek}, and they might affect the resulting EEG signals.




The standard system for the arrangement of electrodes on the head is known as \textit{international 10-20 system}\footnote{There are alternative systems to the 10-20 system, such as the 10-10 system, that can accommodate extra electrodes for a more detailed EEG.}. It is implemented by positioning electrodes on the scalp using standard coordinates that are determined relatively to three reference points on the skull, namely: the \textit{nasion} (the deepest point between the nose and the forehead), the \textit{inion} (the lowest posterior point of the skull above the neck), and the \textit{pre-auricular points} (commonly denoted by A1 and A2). These electrodes are spaced in ratios of $10\%$ or $20\%$ of the length of the distances between these three reference points. The electrode locations are indicated by a combination of a \textit{letter} and a \textit{number}. The letter indicates the corresponding cortical area (Fp: pre-frontal, F: frontal, T: temporal, P: parietal, O: occipital, C: central; the letter ``z'' is used to indicate the electrodes in the center of the head (Fz, Cz, and Pz)). The number indicates the relative position; odd and even numbers correspond to the left and right hemispheres, respectively \citep{Frohlich, Libenson}. Figure \ref{fig:eeg} presents the configuration of 21 electrodes corresponding to the international 10-20 system.


EEG recordings are differential recordings, in which the difference between the voltages in two electrodes is measured by a differential amplifier. In other words, the differential amplifier receives the voltages from two electrodes and returns the difference between them, highly amplified. The patterns of chosen electrode pairs (channels) can vary, and these variations are called \textit{montages}. Each signal is the amplified difference between the voltages of two electrodes, for example, in a longitudinal bipolar montage, we have the pairs T5-O1, Fz-Cz, etc. (see Figure \ref{fig:eeg-ref1}), and in a referential (or monopolar) montage, e.g., in an ipsilateral ear montage (see Figure \ref{fig:eeg-ref2}), we have the pairs T5-Ref, Fz-Ref, etc., where ``Ref" is the electrode of reference (e.g., A1 or A2).

\begin{figure}[h!]
\centering
\begin{subfigure}{.45\textwidth}
  \centering
  \includegraphics[scale=0.85]{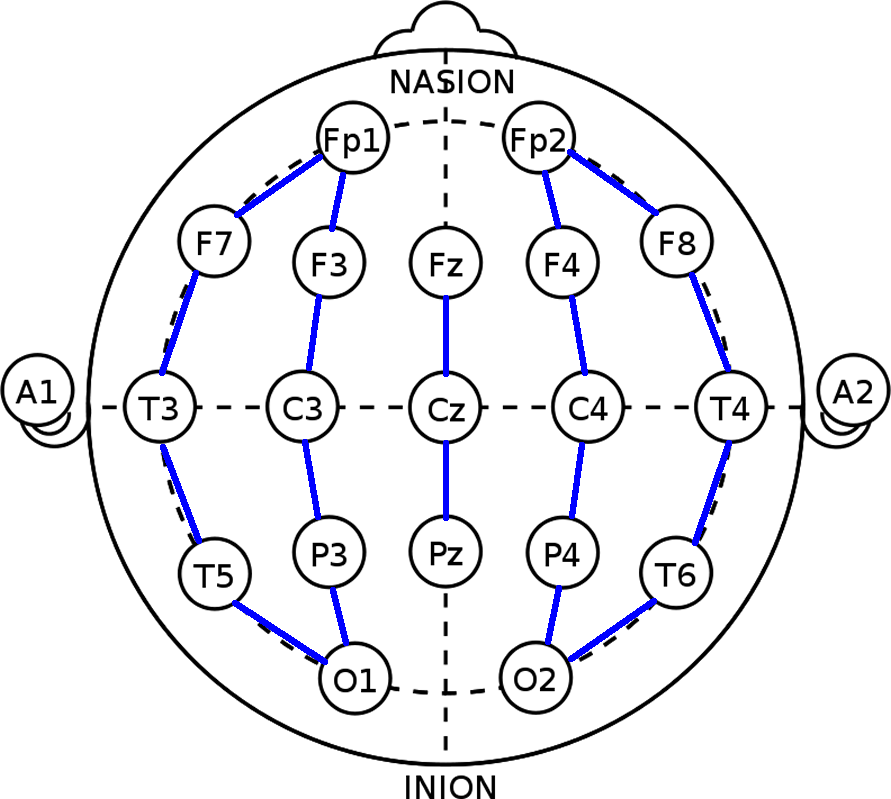}
  \caption{Longitudinal bipolar montage.}
  \label{fig:eeg-ref1}
\end{subfigure}%
\begin{subfigure}{.45\textwidth}
  \centering
  \includegraphics[scale=0.85]{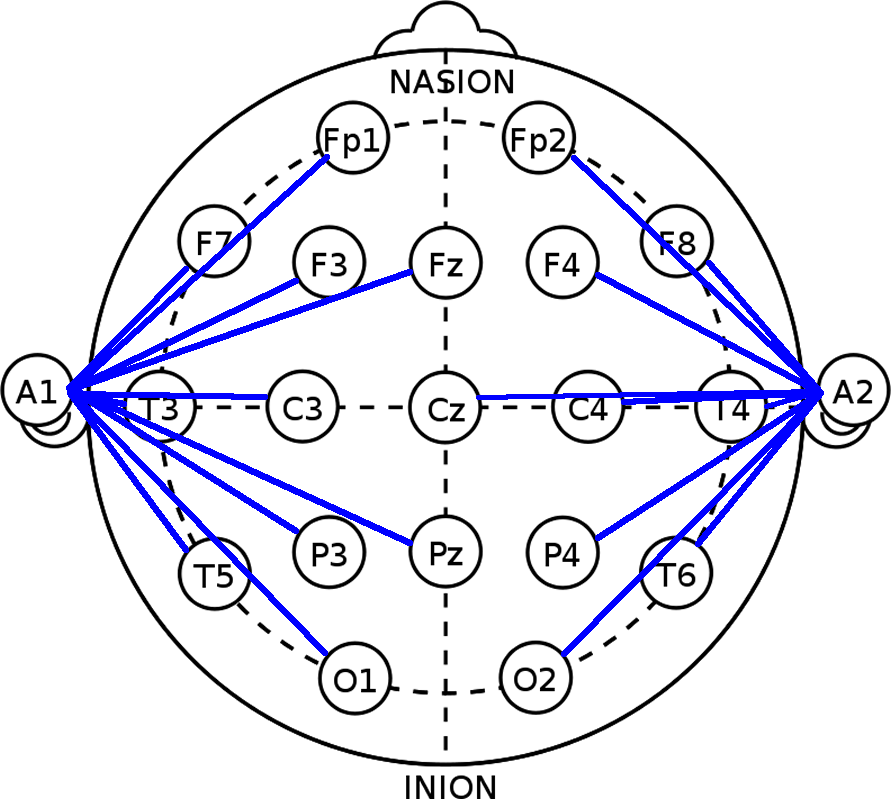}
  \caption{Ipsilateral ear reference montage. }
  \label{fig:eeg-ref2}
\end{subfigure}
\caption{International 10-20 electrode system, with 21 electrodes.}
\label{fig:eeg}
\end{figure} 

The frequencies of EEG signals can be divided into five frequency bands: delta [0.1 Hz, 4 Hz), theta [4 Hz, 8 Hz), alpha [8 Hz, 14 Hz), beta [14 Hz, 30 Hz), and gamma ($\ge$ 30 Hz), where [ , ) denotes semi-closed intervals, i.e. we do not include the highest frequency of each interval. In the literature, we can find other slightly different ranges for these frequency bands \citep{Chiang, Cociu, Quraan, Schomer}, but here we will consider this convention. Typically, the clinically significant EEG frequency components have a range of 0.1 to 100 Hz \citep{Schomer}.




In the previous paragraphs, we described the non-invasive EEG modality known as \textit{scalp EEG}. However, two invasive EEG modalities that require surgically implanted intracranial electrodes are the \textit{electrocorticography} (ECoG) and the \textit{stereoelectroencephalography} (sEEG) \citep{Parvizi}. Moreover, scalp EEG has high temporal resolution, which allows the measurement of electrical activities that occur in the order of milliseconds; however, it has low spatial resolution, while ECoG and sEEG have high temporal and spatial resolutions.


In clinical practice, it is common to use Video-EEG, which is the combination of EEG with video recording to allow simultaneous monitoring of brain electrical activity and the patient's behavior.

\subsection{Stationarity of Signals}

When we are analyzing real-world signals, such as electrophysiological signals, a concept that must be taken into consideration is the concept of \textit{stationarity}.

Stationarity \citep{Marple} refers to a property of a signal that remains constant over time. A stationary (wide-sense stationary) signal has statistical properties that do not change with time, such as its \textit{mean} and \textit{autocorrelation}. This means that the behavior of the signal is predictable and can be analyzed using statistical tools. For instance, if a signal is stationary, its power spectral density can be estimated using the autocorrelation function, which provides a measure of the signal's frequency content. 

In practice, however, most of the signals are non-stationary. In particular, EEG signals are essentially non-stationary due to the time-varying nature of the underlying brain processes \citep{Klonowski, vonBunau}. One way to tackle the problem of non-stationarity is to consider short segments of the signal as locally stationary (short-window methods), thus the methods that require stationarity can be applied from segment to segment.

\subsection{EEG Artifacts}


An \textit{artifact} in an EEG is a component of the signal that is not a brain signal, i.e., it is a component of the signal that does not correspond to a true cerebral activity \citep{Libenson}. EEG artifacts may be caused by electrical signals originating from the patient's extracerebral physiological activities (\textit{physiological artifacts}) or by electrical signals originating from external sources, such as equipment and the environment (\textit{non-physiological artifacts}); common examples of these two types of artifacts are:

\begin{itemize}
\item \textbf{Physiological artifacts:} Eye blinks, eye movements, tongue movement (glossokinetic artifact), cardiac activity, and muscle contractions.

\item \textbf{Non-physiological artifacts:} Electromagnetic interference from AC power lines (line frequencies typically are 50 Hz or 60 Hz), electrode detachment, EEG cable movement, and interference from other external electronic equipment.
\end{itemize}

The presence of artifacts in the EEG data may reduce the statistical power of the analysis; therefore, it is essential in the EEG preprocessing stage, before the analysis \textit{per se}, to identify and remove or attenuate artifacts. There are several techniques used for artifact handling, for instance, low-frequency filters (to remove very low frequencies, e.g., frequencies $<$ 1 Hz), high-frequency filters (to remove high frequencies, e.g., frequencies $>$ 70 Hz), notch filters (to eliminate signals at a specific frequency, e.g., 50 Hz or 60 Hz frequencies), regression methods, wavelet transform, principal component analysis (PCA), and independent component analysis (ICA) (to remove eye blinks, eye movements, and muscle artifacts) \citep{Jiang}.

\subsection{Clinical Uses and Limitations of EEG and Video-EEG}

As already discussed previously, EEG is a tool to assess the brain's electrical activity, and Video-EEG combines EEG with video recording to monitor both the brain's electrical activity and the patient's behavior. Despite the diverse clinical uses, in what follows, we focus on describing some important uses and limitations of EEG and Video-EEG in the diagnosis and treatment of epilepsy \citep{Montenegro}:

\begin{itemize}
\item \textbf{Diagnosing epilepsy:} EEG is commonly used to diagnose epilepsy and differentiate it from other neuronal disorders. Video-EEG can provide additional information, such as the timing and duration of seizures, which can help in the diagnosis and treatment of epilepsy.

\item \textbf{Monitoring seizures:} EEG and Video-EEG can be used to monitor seizures and track their frequency, duration, and severity. This information can be used to adjust medication dosages and evaluate the effectiveness of treatment, or as presurgical evaluation for patients with drug-resistant epilepsy (DRE).

\end{itemize}

Nevertheless, there are some limitations in the use of EEG and Video-EEG, such as:

\begin{itemize}

\item \textbf{Limited spatial resolution and sensitivity:} EEG provides a measure of overall brain activity but has limited spatial resolution. It cannot identify the precise location of epileptiform activity occurring far from the cranial vault, i.e. epileptiform activity originating from deep brain structures.


\item \textbf{False-positives:} EEG can produce false-positive results, indicating the presence of epileptiform activity when none is present. 

\item \textbf{Inconvenience:} EEG and Video-EEG require the placement of electrodes on the scalp, which can be uncomfortable and time-consuming for patients.

\end{itemize}

EEG and Video-EEG are useful tools in clinical practice, especially in the diagnosis and treatment of epilepsy, but they have some limitations that must be considered when interpreting results.

\section{From Brain Signals to Brain Connectivity}
\label{sec:brain-signals-connectivity}

Brain signals captured by brain mapping modalities, such as fMRI, and electrophysiological methods, such as EEG, provide information about the complex and dynamic activities of the brain. These activities involve the interaction of several brain areas, connected by both anatomical and functional associations \citep{Sporns-2010}. Brain connectivity comprises the description of how different brain structures interact and influence or are influenced by each other over time \citep{Friston-1994, Friston-2011}, encompassing \textit{structural connectivity} (anatomical connections), \textit{functional connectivity} (statistical dependencies), and \textit{effective connectivity} (causal influences) \citep{Sporns-Scholar}. 


The connections and interactions between different brain areas can be measured and interpreted through connectivity estimators. These estimators are mathematical and statistical methods that can be computed from data obtained from fMRI, EEG, MEG, etc. Different approaches can be used, consequently, we can have directed and non-directed, bivariate and multivariate, linear and non-linear, time-domain and frequency-domain connectivity estimators \citep{Chiarion}. 

The most common connectivity estimators are correlation and coherence. They are examples of linear, bivariate estimators. Examples of non-linear, bivariate approaches include mutual information, transfer entropy, and phase synchronization \citep{Pereda}. All these measures can be used to infer functional connectivity. Several estimators such as transfer entropy, Granger causality index (GCI) \citep{Brovelli, Geweke}, directed coherence (DC) \citep{Saito}, partial directed coherence (PDC) \citep{Baccala-2001}, and directed transfer function (DTF) \citep{Kaminski1}, are based on the concept of Granger causality, which is a cause-effect relation concept (the past values of one time series can predict current values of another), suggesting a directional influence or causality between brain regions, and thus they can be applied to infer effective connectivity. GCI, PDC, and DTF are directed, model-based, multivariate estimators, and whereas GCI is defined in the time-domain, DTF and PDC are defined in the frequency-domain. Moreover, information-theoretic forms of PDC and DTF, the information PDC (iPDC) and the information DTF (iDTF), which reinterpreted these estimators in terms of the mutual information, were introduced in \citep{Takahashi2, Takahashi}. Another very popular model-based estimator used to infer causal interactions (but not based on Granger causality), and therefore used to infer effective connectivity, is the dynamic causal modeling (DCM) \citep{Friston2003}, which is a Bayesian model-based approach formulated in terms of stochastic differential equations or ordinary differential equations. As is reasonable to expect, all of the previous estimators have limitations and are influenced by the quality of the data, data variability, volume conduction effects, and data preprocessing steps \citep{Chiarion}.

As already discussed in Chapter \ref{chap:chap1}, in this thesis we are interested in effective connectivity networks obtained through the iPDC estimator, thus in the next sections we will introduce all the mathematical formalism behind the PDC and its variants, beginning with the theory of vector autoregressive models and the concept of Granger causality.

\section{A Brief Introduction to Multivariate Autoregressive Models}
\label{sec:MVAR}

This section is based on the textbook by H. Lütkepohl \citep{Lutkepohl}; however, here we adopt the notation and some observations presented in \citep{Sameshima1}. 


\textit{Multivariate autoregressive models} (MVAR models), also called \textit{vector autoregressive models} (VAR models), are stochastic process models used to capture the relationship between various quantities as they change over time. Formally, they are defined as follows.

\smallskip

\begin{definition} 
A \textit{VAR model of order} $p \in \mathbb{N}$, denoted by VAR($p$), is given by
\begin{equation}\label{var1}
X(n) = \nu + A_1 X(n-1) + A_2 X(n-2) + \cdots + A_p X(n-p) + W(n),
\end{equation}

\noindent where $p$ represents the number of lags (past values) in the model, $X(n) = (x_{1}(n),...,$ $x_{N}(n))^{T}$ is an $N \times 1$ random vector, $n \in \mathbb{Z}$, $A_{i}$ are $N \times N$ fixed coefficient matrices,  $\nu = (\nu_{1},..., \nu_{N})^{T}$ is an $N \times 1$ fixed intercept vector, and $W(n) = (w_{1}(n),..., w_{N}(n))^{T}$ is an $N$-dimensional innovation process (or white noise), i.e. $E(W(n)) = 0$, $E(W(n)W(n)^{T}) = \Sigma_{W}$ and $E(W(n)W(m)^{T}) = 0$, for $n \neq m$, such that $E(\cdot)$ is the mean and $\Sigma_{W}$ is the covariance matrix of $W(n)$. Here, $\Sigma_{W}$ is assumed to be nonsingular.
\end{definition}

\smallskip

Explicitly, Equation (\ref{var1}) is written as
\begin{equation}
\begin{bmatrix}x_{1}(n) \\ \vdots \\ x_{N}(n)\end{bmatrix} = \begin{bmatrix}\nu_{1} \\ \vdots \\ \nu_{N}\end{bmatrix} + \sum_{r=1}^{p} A_{r} \begin{bmatrix}x_{1}(n - r) \\ \vdots \\ x_{N}(n - r)\end{bmatrix} + \begin{bmatrix}w_{1}(n) \\ \vdots \\ w_{N}(n)\end{bmatrix},
\end{equation}

\noindent where 

\begin{equation}
 A_{r} =\begin{bmatrix}
a_{11}(r) &  \cdots & a_{1N}(r) \\
\vdots &  \ddots & \vdots \\
a_{N1}(r)  & \cdots & a_{NN}(r) \end{bmatrix}.
\end{equation}

\medskip

The coefficients $a_{ij}(r)$ of the matrices $A_{r}$ represent the linear interaction effect of $x_{j}(n - r)$ on $x_{i}(n)$. The stochastic innovation processes $w_{i}(n)$ represent the part of the dynamic behavior that cannot be predicted from the past observations of the processes. Therefore, $w_{i}(n)$ are not time correlated but can exhibit instantaneous correlations among each other, which are described by their covariance matrix \citep{Sameshima1}:

\begin{equation}
 \Sigma_{W} = \begin{bmatrix}
\sigma_{11}(r) &  \cdots & \sigma_{1N}(r) \\
\vdots &  \ddots & \vdots \\
\sigma_{N1}(r)  & \cdots & \sigma_{NN}(r) \end{bmatrix},
\end{equation}

\noindent where $\sigma_{ij} = \mathrm{Cov}(w_{i}(n), w_{j}(n))$.

A multivariate time series can be regarded as a (finite) realization of a vector stochastic process, thus it can be modeled by a VAR($p$) model. There are numerous algorithms to estimate the VAR($p$) model parameters from the input multivariate time series, such as the Arfit, multichannel Levinson, Viera-Morf, and Nuttall–Strand algorithms \citep{Marple, Schlogl}. The quality of the fitting, however, depends on the model order selection. This involves determining an optimal number for $p$, i.e. an optimal number of lags to include in the model. Choosing an appropriate $p$ is essential, since a very small $p$ can lead to loss of information about the series, and a large $p$ can cause overfitting. There are several criteria that can be used for this purpose, including: Akaike’s information criterion (AIC), Hannan–Quinn’s criterion, and Bayesian–Schwarz’s criterion \citep{Lutkepohl}.

Now, let us formally define the concept of \textit{stationarity} for $X(n)$ as a multivariate time series.

\smallskip

\begin{definition} 
We say that a multivariate time series $X(n) = X_{n} = (x_{1}(n),...,$ $x_{N}(n))^{T}$ is \textit{stationary} if
\begin{equation}
E(X_{n}) = \mu = (\mu_{1},...,\mu_{N})^{T}, \hspace{0.1in} \forall n, \hspace{0.1in}  \mbox{and}
\end{equation}
\begin{equation}
E[(X_{n} - \mu)(X_{n - h} - \mu)^{T}] = E[(X_{n + h} - \mu)(X_{n} - \mu)^{T}], \hspace{0.1in} \forall n \hspace{0.1in}  \mbox{and} \hspace{0.1in}  h = 0,1,2,....
\end{equation}

That is, if its first (mean) and second moments are \textit{time invariant}.
\end{definition}

\section{Granger Causality}
\label{sec:granger}

In a 1969 paper \citep{Granger}, the econometrician Clive Granger introduced the concept of causality for stationary stochastic processes that can easily be extended to time series \citep{Lutkepohl}. The concept of \textit{Granger causality} (or \textit{G-causality}) lies in the idea that the cause cannot come after the effect. For time series, if we say that the series $x(n)$ ``Granger-causes" (or ``G-causes") the series $y(n)$, then the past values of $x(n)$ must contain information that helps to predict $y(n)$, in addition to the information contained in the past values of $y(n)$:
\begin{equation}
y(n) = \alpha + \beta_1 y(n-1) + \beta_2 y(n-2) + \cdots + \delta_1 x(n-1) + \delta_2 x(n-2) + \cdots + w(n).
\end{equation}

An interesting property of Granger causality is that \textit{it is not reciprocal}, i.e. $x(n)$ Granger-cause $y(n)$ does not imply that $y(n)$ Granger-cause $x(n)$. This makes it a suitable approach to directional causal processes, such as brain connectivity.

Formally, for time series, Granger causality is defined as follows. Let $x_{n} = x(n)$ and $y_{n} = y(n)$ be two time series and let $\{\Omega_{n}, n = 0, 1, 2,...\}$ be a set of relevant information accumulated up to $n$ (including $n$), containing at least $x_{n}$ and $y_{n}$. Set $\bar{\Omega}_{n} = \{\Omega_{s} : s < n\}$, $\bar{\bar{\Omega}}_{n} = \{\Omega_{s} : s \le n\}$, and let $\bar{x}_{n}$, $\bar{y}_{n}$, and $\bar{\bar{x}}_{n}$ be similar definitions. Given the information set $B$, let $P_{n}(y | B)$ be the predictor of $y_{n}$ with the minimum mean square error (MSE); let $\sigma^{2}(y | B)$ be the corresponding MSE of the predictor.

\smallskip

\begin{definition}
We say that $x_{n}$ \textit{Granger-causes} $y_{n}$ if

\begin{equation}
\sigma^{2}(y_{n} | \bar{\Omega}_{n}) < \sigma^{2}(y_{n} | \bar{\Omega}_{n} - \bar{x}_{n}).
\end{equation}

In other words, $y_{n}$ can be better predicted if we use all available information about the past of both $x_{n}$ and $y_{n}$.
\end{definition}

\begin{definition}
We say that $x_{n}$ \textit{instantaneously causes} $y_{n}$ in the Granger sense if
\begin{equation}
\sigma^{2}(y_{n} | \bar{\Omega}_{n}, \bar{\bar{x}}_{n}) < \sigma^{2}(y_{n} | \bar{\Omega}_{n}).
\end{equation}

That is, the present value of $y_{n}$ is better predicted if the present value of $x_{n}$ is taken into account.
\end{definition}

\begin{observation}\label{Obs-Test-GC}
Granger causality can be tested through linear prediction models \citep{Lutkepohl}. For instance, for a VAR($p$) model (\ref{var1}), testing for the existence of Granger causality from $x_{i}(n)$ to $x_{j}(n)$ is equivalent to verifying the hypothesis
\begin{equation}\label{GC-test}
H: a_{ij}(r) = 0, \hspace{0.2in} \forall r = 1,..., p,
\end{equation}

\noindent i.e. $x_{i}(n)$ Granger-causes $x_{j}(n)$ if there is at least some coefficient $a_{ij}(r)$ different from zero. Moreover, note that $a_{ij}(r) = 0$ does not imply that $a_{ji}(r) = 0$, due to the non-reciprocity of Granger causality.
\end{observation}

\section{Partial Directed Coherence and its Variants}
\label{sec:PDC}

In this section, we introduce the mathematical formalism behind the partial directed coherence and its variants, specifically the generalized partial directed coherence and the information partial directed coherence, together with their asymptotic properties. The main reference for this part is \citep{Sameshima1}.

\subsection{Partial Directed Coherence}
\label{subsec:pdc}

The concept of \textit{partial directed coherence} (PDC) was first proposed by Baccalá and Sameshima \citep{Baccala3, Baccala-2001, Baccala5} as a multivariate generalization of the \textit{directed coherence} (DC) proposed by Saito and Harashima \citep{Saito}. PDC is a quantifier based on the Granger causality concept, and can be considered a representation of Granger causality in the frequency-domain. 

Formally, PDC is defined as follows. Consider a VAR($p$) model given by
\begin{equation}\label{Var-1}
X(n) = \sum^{p}_{r=1} A_{r}X(n - r) + W(n),
\end{equation}

\noindent where $X(n) = (x_{1}(n),..., x_{N}(n))^{T}$ is a stationary multivariate time series (e.g., $X(n)$ may represents $N$ channels of EEG signals in time $n$) and $W(n) = (w_{1}(n),..., w_{N}(n))^{T}$ is an $N$-dimensional stationary Gaussian innovation process (with covariance matrix $\Sigma_{W}$). The order $p$ may be estimated by one of the order selection criteria mentioned in Section \ref{sec:MVAR}.

By properly obtaining the coefficients $a_{ij}(r)$ of the matrices $A_{r}$ (see Section \ref{sec:MVAR}), we can define a frequency-domain representation of (\ref{Var-1}) by defining the matrix $A(f)$ as follows:
\begin{equation}
A(f) = \sum^{p}_{r=1} A(r)e^{-i2\pi fr}.
\end{equation}

Defining $\bar{A}(f) = I - A(f) = [\bar{a_{1}}(f)\bar{a_{2}}(f)...\bar{a_{m}}(f)]$, where $\bar{a}_{j}(f)$ represents the $j$-th column of the matrix $\bar{A}(f)$, the entries of $\bar{A}(f)$ are given by
\begin{equation}
\bar{A}_{ij}(f) =
\begin{cases}
1 - \sum^{p}_{r=1} a_{ij}(r)e^{-i2\pi fr}, &\text{if i = j,}\\
 - \sum^{p}_{r=1} a_{ij}(r)e^{-i2\pi fr} , &\text{otherwise.}\\
\end{cases}
\end{equation}

\smallskip

\begin{definition}
The \textit{partial directed coherence} (PDC) from $j$ to $i$ at frequency $f$ is defined by
\begin{equation}\label{PDC}
\pi_{ij}(f) := \frac{\bar{A}_{ij}(f)}{\sqrt{\bar{a}_{j}^{H}(f)\bar{a}_{j}(f)}},
\end{equation}

\noindent where the superscript $H$ denotes the Hermetian transpose.
\end{definition}

\smallskip

Note that because of the dependence on $a_{ij}(r)$ in Equation (\ref{PDC}), the nullity of $\pi_{ij}(f)$ at a given frequency implies the lack of G-causality from $j$ to $i$. Also, the expression $\pi_{ij}(f)$ denotes the \textit{direction} and \textit{intensity} of the information flow from $j$ to $i$ at frequency $f$, and it satisfies the following normalization properties:
\begin{equation}\label{No1}
0 \le |\pi_{ij}(f)|^{2} \le 1,
\end{equation}
\begin{equation}\label{No2}
\sum_{i = 1}^{N} |\pi_{ij}(f)|^{2}  = 1, \hspace{0.2in} \forall j = 1,..., N.
\end{equation}

\smallskip

\begin{observation}
In the previous paragraphs we introduced the PDC as a quantifier based on the coefficients of a VAR model, however, the PDC does not depend on this specific model, but can be formulated through other multivariate models, such as the \textit{vector moving average} (VMA) model and the \textit{vector autoregressive moving average} (VARMA) model \citep{Baccala2022}.
\end{observation}

\subsection{Generalized PDC}

A scaling-invariant version of PDC, called \textit{generalized partial directed coherence} (gPDC), was introduced by  Baccalá et al. \citep{Baccala8}. As stated in their article, the central problem of connectivity analysis is to analyze the hypothesis
\begin{equation}\label{GC-test}
H_{0}: \pi_{ij}(f) = 0,
\end{equation} 

\noindent whose rejection implies the existence of a directed connection from $x_{j}(n)$ to $x_{i}(n)$, which cannot be explained by other series observed simultaneously. Considering a scenario where a series $y(n)$ Granger-causes $x(n)$, if $y(n)$ is amplified by a constant $\alpha$, and taking $u(n) = \alpha y(n)$, we would eventually have $|\pi_{xu}(f)|^{2} \rightarrow 0$, as $\alpha$ grows. To solve this problem, a generalization of the PDC, which makes it invariant to eventual gains affecting a time series, was defined as follows.

\smallskip

\begin{definition}
The \textit{generalized partial directed coherence} (gPDC) from $j$ to $i$ at frequency $f$ is defined by
\begin{equation}\label{gPDC}
\pi_{ij}^{(w)}(f) := \frac{\frac{1}{\sigma_{i}^{2}}\bar{A}_{ij}(f)}{\sqrt{\bar{a}_{j}^{H}(f) \diag{\sigma_{i}^{2}}^{-1}\bar{a}_{j}(f)}},
\end{equation}

\noindent where $\diag{\sigma_{i}^{2}}$ is the diagonal matrix of the variances $\sigma_{i}^{2}$ of the innovation processes $w_{i}(n)$.
\end{definition}

Equation (\ref{gPDC}) preserves the normalization properties (\ref{No1}) and (\ref{No2}).

\subsection{Information PDC}

The \textit{information partial directed coherence} (iPDC), introduced by Takahashi et al. \citep{Takahashi2, Takahashi}, is a modification of the PDC expression that formalizes the relationship between PDC and information flow. In the following, we expose the formal definition of the iPDC according to the above-mentioned articles.

Let $x = \{x(n)\}_{n \in \mathbb{Z}}$ and $y = \{ y(n)\}_{n \in \mathbb{Z}}$ be two discrete-time stochastic processes. We can evaluate the relationship between $x$ and $y$ through the \textit{mutual information rate} (MIR) (see \citep{Covers} for more details on information theory), which compares the joint probability density with the product of the marginal probability densities of $x$ and $y$:
\begin{equation}\label{mir1}
\mathrm{MIR}(x, y) = \lim_{m\to\infty} \frac{1}{m + 1} E\bigg[ \log \frac{dP(x(1),..., x(m), y(1),..., y(m))}{dP(x(1),..., x(m))dP(y(1),..., y( m))} \bigg],
\end{equation}

\noindent where $dP$ denotes the probability density function. Note that if $\mathrm{MIR}(x, y) = 0$, then $x$ and $y$ are independent.

Let $\omega = 2\pi f$ (angular frequency) and let $S_{xx}(\omega)$ and $S_{yy}(\omega)$ be the auto-spectrum of $x$ and $y$, respectively, and $S_{xy}(\omega)$ the cross-spectrum. The \textit{coherence} between $x$ and $y$ is given by
\begin{equation}\label{coer}
C_{xy}(\omega) = \frac{S_{xy}(\omega)}{\sqrt{S_{xx}(\omega) S_{yy}(\omega)}}.
\end{equation}

For stationary Gaussian processes, Equation (\ref{mir1}) is related to the coherence $C_{xy}$ through the expression
\begin{equation}\label{mir2}
\mbox{MIR}(x, y) = - \frac{1}{4\pi} \int_{-\pi}^{\pi} \log(1 - |C_{xy}(\omega)|^ {2})\,d\omega.
\end{equation}

Now, consider a VAR model given by (\ref{Var-1}), with $p = +\infty$, and such that $\Sigma_{W} = E(w(n)w(n)^{T})$ is the positive-definite covariance matrix of the $N$-dimensional stationary Gaussian process $W(n)$. A sufficient condition for the existence of this VAR model, with $p = +\infty$, is that the spectral density matrix associated with $x$ is invertible at all frequencies and uniformly bounded above and below. Hence, the entries of the matrix $\bar{A}(\omega) = I - A(\omega) = [\bar{a}_{1}(\omega)\bar{a}_{2}(\omega) ...\bar{a}_{m}(\omega)]$ are given by
\begin{equation}
\bar{A}_{ij}(\omega) =
\begin{cases}
1 - \sum^{+\infty}_{r=1} a_{ij}(r)e^{-i\omega r}, &\text{if i = j,}\\
 - \sum^{+\infty}_{r=1} a_{ij}(r)e^{-i\omega r} , &\text{otherwise.}\\
\end{cases}
\end{equation}

\smallskip

\begin{definition}
The \textit{informational partial directed coherence} (iPDC) from $j$ to $i$ is defined by
\begin{equation}\label{iPDC}
\iota\pi_{ij}(\omega) := \frac{\sigma_{i}^{-1/2}\bar{A}_{ij}(\omega)}{\sqrt{ \bar{a}_{j}^{H}(\omega) \Sigma_{W}^{-1}\bar{a}_{j}(\omega)}},
\end{equation}

\noindent where $\omega \in [-\pi, \pi)$ and $\sigma_{i} = E(w_{i}^{2}(n))$.
\end{definition}

\begin{theorem}
\textit{Let $X(n)$ be the stationary multivariate time series satisfying the VAR model (\ref{Var-1}). Then we have}
\begin{equation}\label{iPDC2}
\iota\pi_{ij}(\omega) = C_{w_{i}\eta_{j}}(\omega),
\end{equation}

\noindent \textit{where $\eta_{j}(n) = x_{j}(n) - E[x_{j}(n)|\{x_{r}(m), r \neq j, m\in \mathbb{Z}\}]$ is the partial process associated with $x_{j}(n)$, given the remaining series $\{x_{r}(m)\}_{r \neq j, m \in \mathbb{Z}}$.}
\end{theorem}

\smallskip

The previous theorem shows that the iPDC from $j$ to $i$ measures the amount of common information between the partial process $\eta_{ij}$ and the innovation $w_{i}$. In fact, replacing (\ref{iPDC2}) in Equation (\ref{mir2}) we get
\begin{equation}\label{mir3}
\mbox{MIR}(w_{i}, \eta_{j}) = - \frac{1}{4\pi} \int_{-\pi}^{\pi} \log(1 - |\iota\pi_{ij}(\omega)|^{2})\,d\omega.
\end{equation}

\smallskip

\subsection{General Expression for all PDC Variants}

As noticed by Baccalá et al. \citep{Baccala9}, the PDC from $j$ to $i$ (\ref{PDC}), and its two variants, gPDC (\ref{gPDC}) and iPDC (\ref{iPDC}), can all be obtained from the same general formula for a given frequency $f$:
\begin{equation}\label{general}
\pi_{ij}(f) = \frac{1}{s} \frac{\bar{A}_{ij}(f)}{\sqrt{ \bar{a}_{j}^{H}(f ) S \bar{a}_{j}(f)}},
\end{equation}

\noindent where the variables $s$ and $S$ are given according to the Table \ref{tab:unified-pdc}.

\begin{table}[h!]
  \begin{center}
    \caption{Variables $s$ and $S$ according to the PDC type.}
    \label{tab:unified-pdc}
    \begin{tabular}{c c c c}
      \toprule 
Variable & PDC & gPDC & iPDC\\
      \midrule 
$s$ & 1 & $\sigma_{i}^{-1/2}$ & $\sigma_{i}^{-1/2}$\\
$S$ & $I$ & $\diag{\sigma_{i}^{2}}^{-1}$ & $\Sigma_{W}^{-1}$ \\
      \bottomrule 
    \end{tabular}
  \end{center}
\end{table}

Sometimes the notation $\diag{\sigma_{i}^{2}} = (I \odot \Sigma_{W})$ is used, where ``$\odot$" denotes the Hadamard product of matrices (element-wise product).

Other forms of PDC have been introduced over time, such as the \textit{time-varying generalized orthogonalized PDC} (tv-gOPDC) \citep{Omidvarnia}, which tries to reduce the effect of volume conduction, and the \textit{total PDC} (tPDC) \citep{Baccala11}, which takes into account the instantaneous Granger causality.

\subsection{Asymptotic Properties of the PDC and its Variants}
\label{sec:asymp-pdc}

As stated in \citep{Sameshima1}, the problem of statistical connectivity inference, whether performed in the time or frequency-domain, actually involves two distinct problems:

\begin{itemize}

\item \textit{The connectivity detection problem}: To detect the presence of significant connectivity at a given frequency.
    
\item \textit{The connectivity quantification problem}: To determine the confidence interval of the estimated value when it is significant at a given frequency.
    
\end{itemize}

In \citep{Takahashi-2007}, it was shown that both problems can be rigorously examined from the perspective of asymptotic statistics.

The availability of confidence intervals based on a single trial makes it possible to consistently compare connection strengths under various experimental situations without the need to do repeated experiments based on ANOVA for inference \citep{Sameshima1}. In view of this, Baccalá et al. \citep{Baccala9} worked to determine the asymptotic behavior of the three previously introduced forms of PDC, showing that significant values of PDC (gPDC, iPDC) are asymptotically Gaussian, and this normality is not verified when there is no connectivity.

In what follows, we summarize the results obtained in \citep{Baccala9}, omitting the proofs and some details. Consider the VAR model (\ref{Var-1}). Let $n_{s}$ be the number of observed data points. Let $\theta^{T} = \alpha^{T} \sigma^{T}$ be the vector of parameters, where $\sigma = \mbox{vec}(\Sigma_{W})$ and $\alpha = \mbox{vec}(A_{1}...A_{p})$, with $\mbox{vec}(\cdot)$ denoting the \textit{vectorization operator} (i.e., the operator that converts a matrix into a column vector, by stacking all columns of the matrix). Note that the vector $\theta$ incorporates the dependence of $a_{ij}$ and $\sigma_{ij}$ according to the chosen PDC type.

The \textit{confidence intervals} and the \textit{limit of the null hypothesis} for the general form of the PDC (\ref{general}) can be calculated by dividing its parameter dependence on the parameter vector $\theta$, considering its decomposition into numerator and denominator:

\begin{equation}\label{frac-pdc}
|\pi_{ij}(f)|^{2} = \pi(\theta) = \frac{\pi_{n}(\theta)}{\pi_{d}(\theta)},
\end{equation}

\noindent where the subscripts $n$ and $d$ indicate ``numerator" and ``denominator", respectively. Thus, the following results are valid:

\begin{itemize}

\item \textit{Confidence Intervals}: For a large $n_{s}$, Equation (\ref{frac-pdc}) is asymptotically normal, i.e.
\begin{equation}\label{assymp}
\sqrt{n_{s}}(|\hat{\pi}_{ij}(f)|^{2} - |\pi_{ij}(f)|^{2}) \rightarrow \mathcal{N}(0, \gamma^{2}(f)),
\end{equation}

\noindent where $\gamma^{2}(f)$ is a frequency-dependent variance which depends on the PDC type.
    
 \item \textit{Null hypothesis threshold}: The variance $\gamma^{2}(f)$ is identically zero under the null hypothesis
\begin{equation}\label{H0}
H_{0}: |\pi_{ij}(f)|^{2} = 0,
\end{equation} 

\noindent therefore, Equation (\ref{assymp}) is no longer valid, i.e. asymptotic normality is no longer satisfied. This requires the consideration of the next term of the Taylor expansion of the asymptotic expression of (\ref{frac-pdc}), which has a dependency $\mathcal{O}(n_{s}^{-1})$. The resulting distribution corresponds to a linear combination of at least two $\chi_{1}^{2}$, with appropriate and frequency-dependent multiplication coefficients:
\begin{equation}\label{assymp}
n_{s}\bar{a}_{j}^{H}(f)S\bar{a}_{j}(f) (|\hat{\pi}_{ij}(f)|^{2} - |\pi_{ij}(f)|^{2}) \xrightarrow[]{d} \sum_{k=1}^{q} l_{k}(f)\chi_{1}^{2},
\end{equation}

\noindent where the coefficients $l_{k}(f)$ depend only on the numerator $\pi_{n}(\theta)$, and ``$\xrightarrow[]{d}$" \hspace{0.01in} designates the convergence in distribution.

\end{itemize}

 That is, when the null hypothesis (\ref{H0}) is not rejected (lack of connectivity), the PDC tends to a distribution $\chi_{1}^{2}$, and when (\ref{H0}) is rejected, the PDC tends to a normal distribution.

\subsection{Examples with Simulations}
\label{sec:examples-pdc}

Given a multivariate time series $X(n) = (x_{1}(n),..., x_{N}(n))$, by calculating the PDC between $x_{i}(n)$ and $x_{j}(n)$, for a given frequency $f$, for all $i$ and $j$, a weighted connectivity digraph can be constructed by drawing an arc from $x_j$ to $x_i$ (which correspond to the nodes) if and only if $|\pi_{ij}(f)|^{2} \neq 0$ (where $\pi_{ij}(f)$ represents the general formula (\ref{general})), where the arc weights are equal to the values $|\pi_{ij}(f)|^{2}$. The same applies to DTF.

In the following example, which is based on the article \citep{Baccala10}, we applied the iPDC and iDTF estimators in a toy model consisting of seven ($N=7$) linear difference equations as exposed in \citep{Baccala3}. As commented previously, like PDC, DTF is a multivariate estimator based on the concept of Granger causality, which can describe the causal influence of one time series on another at a given frequency, and an information-theoretic form of DTF, the iDTF, was introduced in \citep{Takahashi2, Takahashi}. The aim of this example is twofold: 1) to present how iPDC can be used to create a connectivity digraph; 2) to compare the differences in the connections generated by these estimators, which are direct and indirect in nature for iDFT, and only direct for iPDC. We emphasize that in this text we do not discuss in details the DTF/iDTF estimators, and it's up to the reader to consult \citep{Kaminski1, Takahashi} for more details.

\smallskip

\begin{example}\label{ex:example-idtf-ipdc}
Figures \ref{fig:idtf} and \ref{fig:ipdc} present the directed graphs produced by the iDTF and iPDC estimator (only statistically significant ($p < 0.01$) connections were considered), respectively, when applied to the following VAR model \citep{Baccala3}:
\begin{equation}
\begin{cases}
x_{1}(n) = 0.95\sqrt{2}x_{1}(n-1) - 0.9025x_{1}(n-2) + 0.5x_{5}(n-2) + w_{1}(n)\\
x_{2}(n) = -0.5x_{1}(n-1) + w_{2}(n)\\
x_{3}(n) = 0.4x_{1}(n-4) - 0.4x_{2}(n-2)  + w_{3}(n)\\
x_{4}(n) = -0.5x_{3}(n-1) + 0.25\sqrt{2}x_{4}(n-1) + 0.25\sqrt{2}x_{5}(n-1) + w_{4}(n)\\
x_{5}(n) = -0.25\sqrt{2}x_{4}(n-1) -0.25\sqrt{2}x_{5}(n-1) + w_{5}(n)\\
x_{6}(n) = 0.95\sqrt{2}x_{6}(n-1) - 0.9025\sqrt{2}x_{6}(n-2) + w_{6}(n)\\
x_{7}(n) = -0.1x_{6}(n-2) + + w_{7}(n)\\
\end{cases}
\end{equation}

\begin{figure}[h!]
\centering
\begin{subfigure}{.41\textwidth}
  \centering
  \includegraphics[scale=0.55]{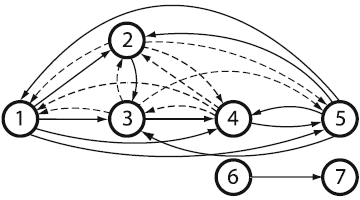}
  \caption{iDTF network.}
  \label{fig:idtf}
\end{subfigure}%
\begin{subfigure}{.41\textwidth}
  \centering
  \includegraphics[scale=0.55]{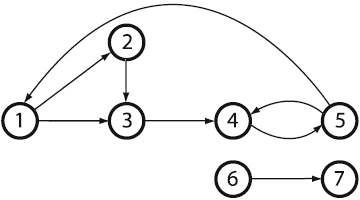}
  \caption{ iPDC network.}
  \label{fig:ipdc}
\end{subfigure}
\caption{Examples of iDTF and iPDC networks. Dashed lines represent weak connections (adapted from \citep{Baccala10}).}
\label{fig:graphs}
\end{figure}

The iDTF is an influenceability (reachability) estimator, i.e., it is unable to differentiate between direct and indirect interactions, whereas the iPDC is a directed connectivity estimator, i.e., it can detect direct interactions.
\end{example}

\section{Brain Connectivity Networks}
\label{sec:brain-connectivity}

In this last section, we describe the different types of brain connectivity networks, presenting some examples; afterwards, we present some concerns related to brain connectivity networks estimated from neurophysiological data, and we conclude with a literature review and some considerations about the state of the art of graph theoretical analysis and topological data analysis in modern network neuroscience research.

\subsection{The Different Types of Brain Connectivity Networks}

Brain connectivity refers to the ways in which different brain structures interact and influence or are influenced by each other during sensory, motor, or cognitive tasks. As mentioned previously, brain connectivity encompasses three main types of connectivity, namely: \textit{structural connectivity}, \textit{functional connectivity}, and \textit{effective connectivity} \citep{Friston-1994, Friston-2011, Horwitz, Lang, Sporns-Scholar, Sporns-2010}.


\textit{Structural (anatomical) connectivity} refers to a set of anatomical connections that link neural elements. These connections vary in scale, from large-scale networks of interregional routes to small circuits of individual neurons \citep{Sporns-2010}. They form the \textit{connectome} through which synapses between neighboring neurons and nerve fibers connect spatially distinct brain regions \citep{Sporns-2012, Sporns-2013, Sporns-2005}. At shorter time scales (seconds to minutes), anatomical connections can be considered as static, but at longer time scales (hours to days), they can be thought of as plastic or dynamic \citep{Sporns-2010}. Structural networks of the brain can be estimated through techniques such as diffusion tensor imaging (DTI) and fiber tractography \citep{Mukherjee}.

\textit{Functional connectivity} refers to the temporal interdependence of activation patterns of anatomically separate brain areas \citep{Sporns-2010}. It displays the statistical relationships between different and remote-located populations of neurons. It is dependent on statistical metrics, such as correlation, covariance, or spectral coherence \citep{Lang}. Since statistical relationships are highly reliant on time, the basis of functional connectivity analysis is time series data of neurophysiological activity, which can be extracted from EEG, MEG, fMRI, or other techniques. Unlike structural connectivity, functional connectivity does not necessarily depend on the anatomical connections of neuronal units. As mentioned in Section \ref{sec:brain-signals-connectivity}, examples of functional connectivity network estimators include correlation, coherence, mutual information, transfer entropy, and phase synchronization.

\textit{Effective connectivity} reflects the causal relationships between activated brain regions by characterizing the influence that one brain structure has on another \citep{Sporns-2010}. It may depict the directed effects inside a neuronal network by combining structural and effective connections. Effective connectivity is also time-dependent. Moreover, methods based on Granger causality may be used to infer causality from time series data of neurophysiological signals. As discussed in Section \ref{sec:brain-signals-connectivity}, examples of effective connectivity network estimators include transfer entropy, GCI, DC, PDC, DTF, and DCM.

Although the previous classification between functional and effective connectivity categories is widely accepted in the literature, it may not be accurate enough to describe multivariate models. Accordingly, Baccalá and Sameshima \citep{Baccala10} proposed an alternative classification based on two novel connectivity categories: \textit{G-connectivity} and \textit{G-influentiability}. G-connectivity (G stands for Granger) describes the direct, immediate, and active coupling between brain structures, but excludes active interactions that occur through intermediate (indirect) structures, and may be estimated by PDC (see Figure \ref{fig:effective-a}). G-influentiability, on the other hand, describes both direct and indirect active connections, and may be estimated by DTF (reachability) (see Figure \ref{fig:effective-b}). This description allows us to classify connectivity networks according to the nature of their links, as shown in Table \ref{tab:G-connectivity}.


\smallskip

\begin{table}[h!]
  \begin{center}
    \caption{G-connectivity/G-influentiability classification.}
    \label{tab:G-connectivity}
    \begin{tabular}{c c c}
      \toprule 
      & \textbf{Direct} & \textbf{Indirect} \\
      \midrule 
      Active & $PDC \neq 0$ &  $PDC = 0$ and $DTF \neq 0$ \\
      Inactive & $PDC = 0$  &  $DTF =0$   \\
      \bottomrule 
    \end{tabular}
  \end{center}
\end{table}

The dynamics of functional and effective connectivity networks may be analyzed through sliding window techniques. These techniques include selecting a temporal window of a specified length and using the data within it to estimate the connectivity network through a chosen measure. A (discrete) temporal network is created by shifting the window in time by a certain number of data points and repeating the procedure. This may be thought of as a quantification of the dynamics of the measure's behavior \citep{Hutchison}.

\smallskip

\begin{example}\label{ex:example1-structural-conn}
In this example, we constructed structural and functional connectivity networks based on results extracted from the article \citep{Cociu}. The authors used EEG, fMRI, and DTI data from three patients diagnosed with autism spectrum disorder (ASD) to investigate how structural brain networks correlate with functional brain networks. Here we used three regions of interest (ROIs) determined in their study, namely: the Precuneus/Posterior Cingulate Cortex (PCUN/PCC), the Left Parietal Cortex (LPC), and the Right Parietal Cortex (RPC). The functional connectivities between these ROIs (nodes) were estimated through several functional connectivity measures from EEG signals, for five frequency bands ($\delta$(1-4 Hz), $\theta$(4-8 Hz), $\alpha$(8-12 Hz), $\beta$(12-30 Hz), $\gamma$(30-45 Hz)), and we chose the results corresponding to the \textit{coherence} in the delta band (COH-$\delta$). The structural connectivities were obtained through DTI analysis, from which the number of white matter fibre tracts connecting the ROIs was estimated. Figure \ref{fig:structural-a} presents a structural connectivity network, along with its weighted adjacency matrix, in which the edge weights correspond to the number of tracts between the ROIs for to the results obtained for patient 1 (subject 1). Figure \ref{fig:structural-b} presents a functional connectivity network, along with its weighted adjacency matrix, in which the edge weights correspond to the coherence (COH-$\delta$) between the ROIs for the results obtained for the same patient.

\begin{figure}[h!]
\centering
\begin{subfigure}{.5\textwidth}
  \centering
  \includegraphics[scale=0.4]{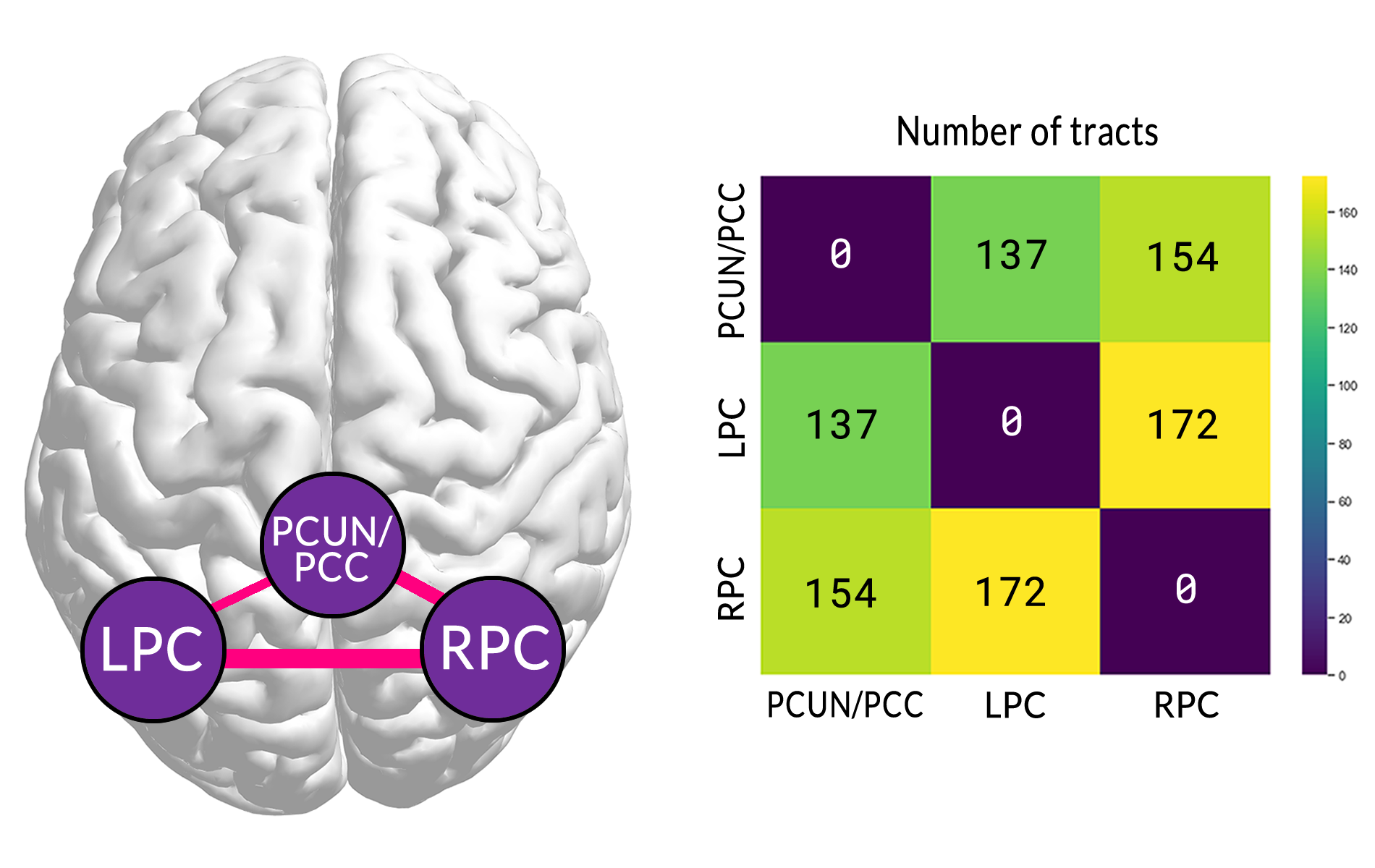}
  \caption{Structural connectivity.}
  \label{fig:structural-a}
\end{subfigure}%
\begin{subfigure}{.5\textwidth}
  \centering
  \includegraphics[scale=0.4]{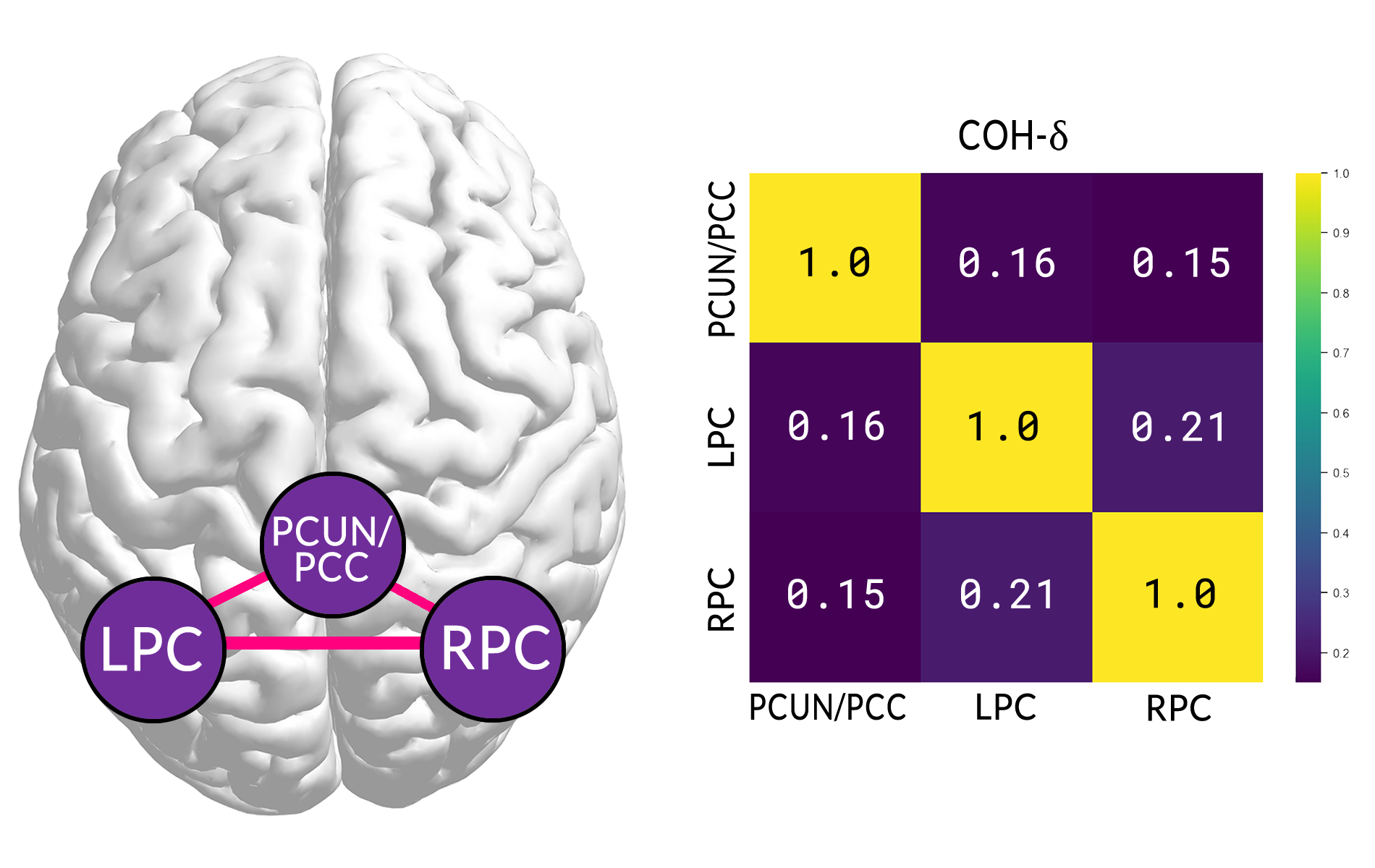}
  \caption{Functional connectivity.}
  \label{fig:structural-b}
\end{subfigure}
\caption{Structural and functional connectivity networks corresponding to patient 1 (subject 1). Structural connectivities correspond to the number of tracts between the ROIs, and the functional connectivities correspond to the coherence in the delta band (COH-$\delta$) between the ROIs (edge thickness is proportional to the magnitude of the connectivity measure) (3D model generated with the BrainNet Viewer software \citep{Xia}).}
\label{fig:structural-conn}
\end{figure}
\end{example}

\begin{example}\label{ex:example2-g-conn}
As commented in Example \ref{ex:example-idtf-ipdc}, PDC and DTF are \textit{directed} connectivity estimators, and thus it is possible to construct a connectivity digraph when applied to a multivariate time series. In this example, we constructed G-connectivity and G-influentiability networks based on data used in the article \citep{Baccala7} (the data were provided by the authors). The authors used EEG data from nine patients diagnosed with left/right mesial temporal lobe epilepsy (MTLE) to assess the seizure lateralization through graph theoretical analysis (GTA) of connectivity networks estimated via PDC. The records were performed using a total of 29 electrodes placed according to the international 10–20 system in referential montage, with a sampling rate of 200 Hz. We chose four channels (F7-Ref, F8-Ref, T5-Ref, T6-Ref) and computed the iPDC and iDTF estimators on a 10s epoch (patient 1), starting at seizure onset, in the delta frequency band (1-4 Hz). The iPDC and iDTF networks were estimated via the MATLAB package asympPDC (see Appendix \ref{appendix:software}), using the Nuttall-Strand algorithm to estimate the parameters of the VAR model, with a fixed order set at $p=2$, and the statistically significant connections were estimated asymptotically (see Subsection \ref{sec:asymp-pdc}) with a significance level of $0.1\%$. Figures \ref{fig:effective-a} and \ref{fig:effective-b} present the G-connectivity network obtained via iPDC and the G-influentiability network obtained via iDTF, respectively, along with their respective weighted adjacency matrices.


\begin{figure}[h!]
\centering
\begin{subfigure}{.5\textwidth}
  \centering
  \includegraphics[scale=0.4]{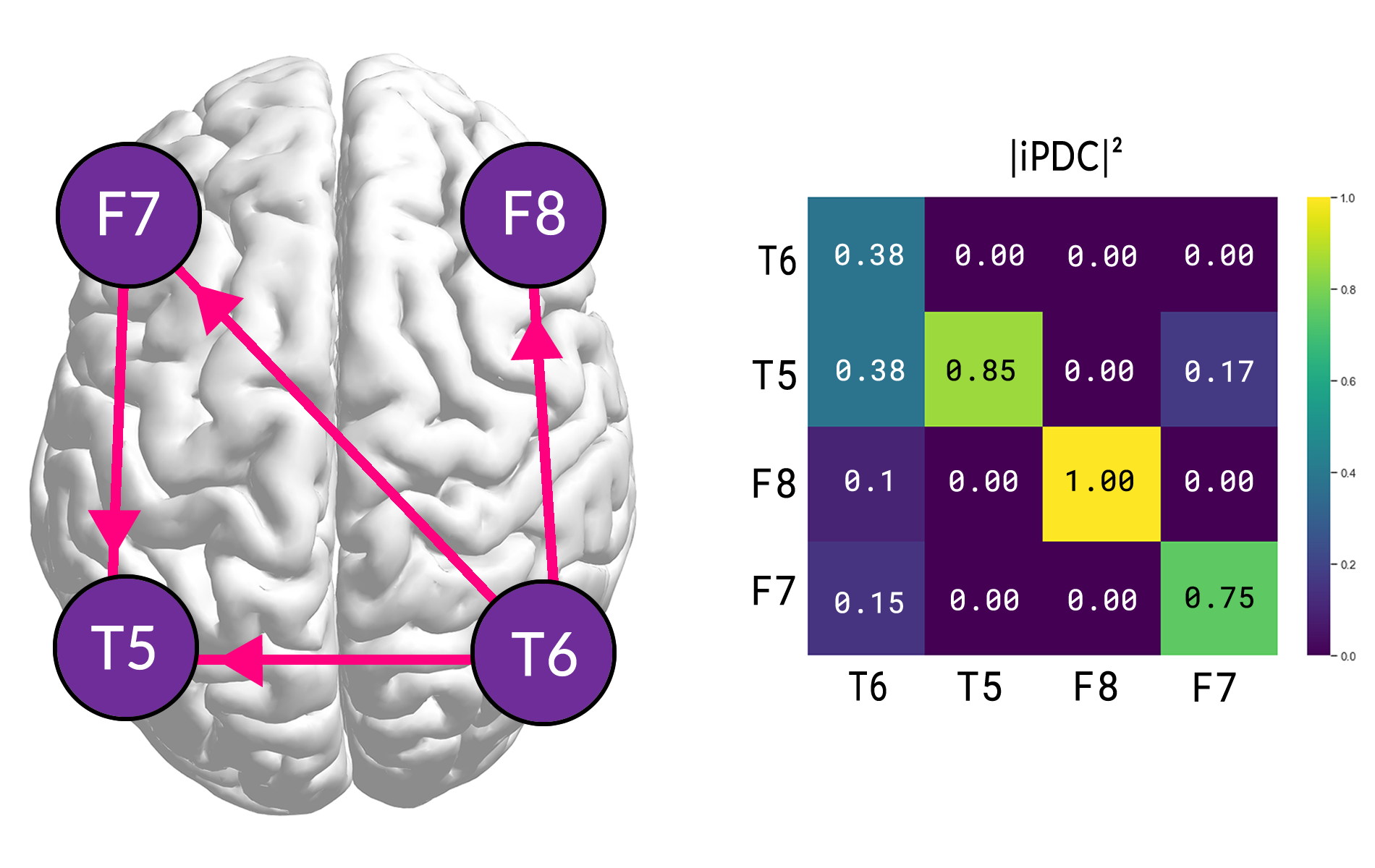}
  \caption{G-connectivity network (iPDC).}
  \label{fig:effective-a}
\end{subfigure}%
\begin{subfigure}{.5\textwidth}
  \centering
  \includegraphics[scale=0.4]{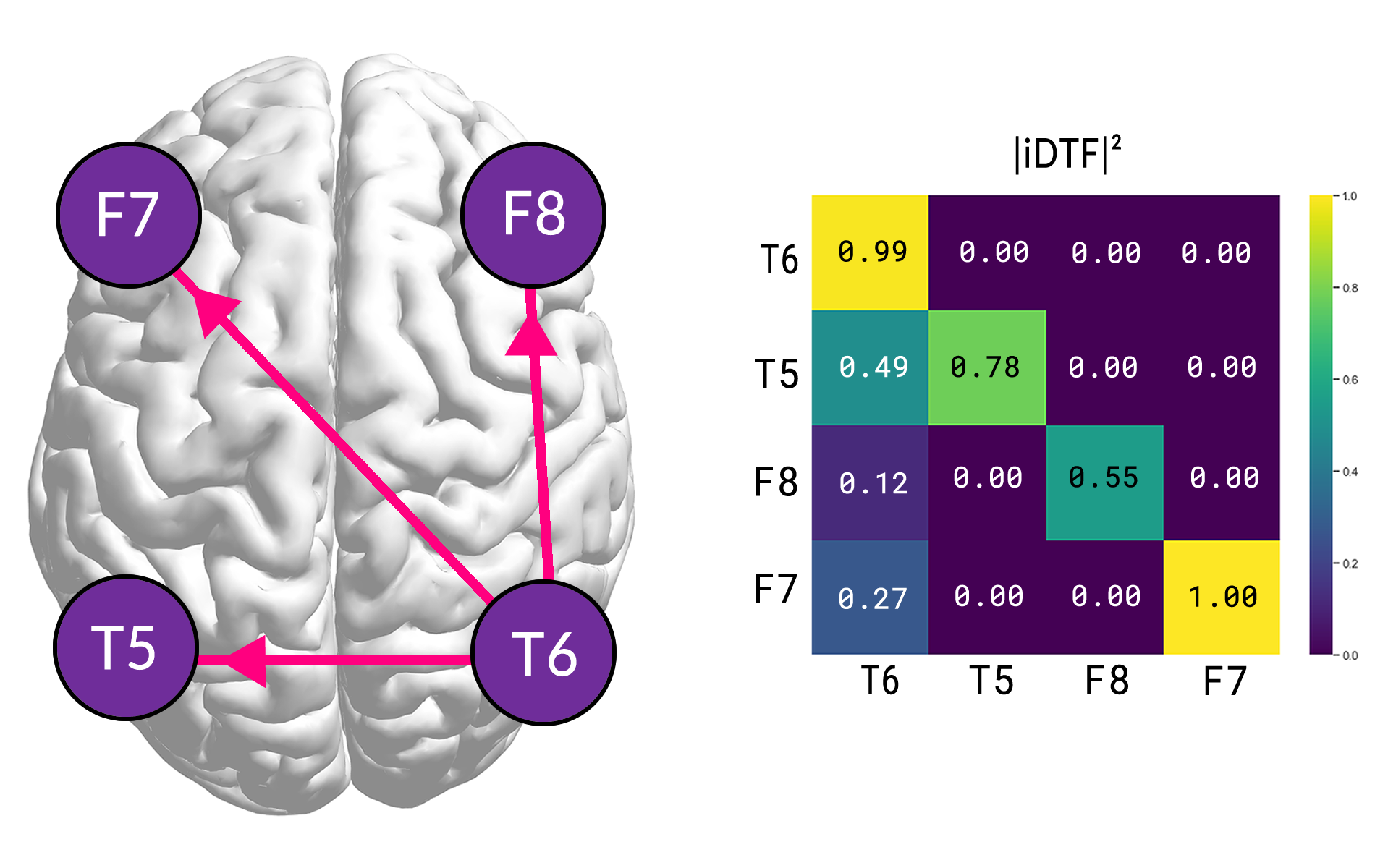}
  \caption{G-influentiability network (iDTF).}
  \label{fig:effective-b}
\end{subfigure}
\caption{G-connectivity (iPDC) and G-influentiablility (iDTF) digraphs, along with their weighted adjacency matrices (3D model generated with the BrainNet Viewer software \citep{Xia}).}
\label{fig:effective-conn}
\end{figure}
\end{example}

\subsection{Concerns about Brain Connectivity Networks}

There are some concerns about the estimation and validity of brain connectivity networks that should be considered in their analysis. Below, we list some of the main points that must be taken into consideration when obtaining and interpreting brain connectivity \citep{Chiarion, Lang, Sameshima1}.

\begin{itemize}

\item  \textbf{Variability of data extraction:} As briefly discussed in Section \ref{sec:brain-signals}, there are several techniques to detect and measure the neurophysiological activity of the brain, such as fMRI, EEG, and MEG. These techniques vary widely in terms of spatial and temporal resolution, and the specific aspects of brain activity that they measure. Therefore, the data acquisition method directly impacts the interpretation of the connectivity networks subsequently estimated.

\item  \textbf{Data preprocessing:} Data preprocessing procedures, such as motion correction, high pass filtering, and RETROICOR correction (RTC) for fMRI data, and downsampling and choice of reference for EEG data, may influence connectivity estimates. For instance, Výtvarová et al. \citep{Vytvarova} found that the strength of the connectivities in functional connectivity networks obtained from fMRI data was statistically significantly affected by high-pass filtering, and that the topology of these networks was affected by RTC.


\item \textbf{Selection of connectivity estimator:} As already discussed in Section \ref{sec:brain-signals-connectivity}, there are numerous measures for estimating brain connectivity. Measures such as correlation, coherence, phase synchronization, GCI, DC, PDC, DTF, and DCM may yield different or even conflicting results under similar conditions. Therefore, the choice of the connectivity estimator can have a decisive impact on the interpretation of the results.

\item \textbf{Volume conduction:} 
In Section \ref{sec:electroencephalography} we already discussed the phenomenon of volume conduction and how it affects EEG signals. Nevertheless, its effects may also impact the results of other techniques, such as MEG. Volume conduction effects can lead to false assumptions about the interactions between different brain regions. For instance, considering a functional connectivity network estimated via correlation from EEG signals, two regions might appear to be connected because their signal patterns are correlated. However, this correlation might be due to the signals from one region spreading to the electrodes near another region, rather than any true functional interaction between them. The effects of volume conduction are inevitable, but there are methods capable of reducing them \citep{He}. In particular, both DTF and PDC are affected by volume conduction \citep{Brunner}.

\item \textbf{Generalizability:} Brain connectivity networks are often estimated from data obtained from relatively small populations and in specific contexts (whether from healthy individuals or individuals with some neurological disorder); therefore, it is essential to take this limitation into account in the applicability of results to general populations and contexts.

\end{itemize}

All of the points mentioned above impact the analysis of brain connectivity networks; therefore, they should be taken into consideration when interpreting the results.

\subsection{Applications of Brain Connectivity Networks Analysis}
\label{sec:analysis-brain-networks}

As we saw in the previous sections, brain connectivity networks can be represented as graphs (sometimes referred to as \textit{brain graphs} \citep{Bullmore2}), where the nodes represent different brain structures and the edges represent specific associations between them (e.g., anatomical connections, statistical dependencies, or causal influences). Thus, they can be analyzed through algorithms and quantitative methods from graph theory and network science \citep{Bullmore1, Fallani3, Rubinov-2010, Sporns-2010, Stam-2007}. In recent years, graph theoretical analysis (GTA) of brain networks has proven to be a powerful analytic tool, and has been widely used in the area of network neuroscience \citep{Farahani, Fornito-2016}.

In Chapter \ref{chap:chap1}, we have already briefly discussed the applications of graph theory in the analysis of brain networks. Many of the GTA have shown that brain structural, functional, and effective networks may present, depending on the context, several non-trivial topological and organizational properties \citep{Zalesky-2012}, such as hierarchical organization \citep{Bassett-2008}, clustering, modularity \citep{Rubinov-2010, Sporns-2010}, presence of hubs (which may be interpreted as important brain regions) \citep{Fallani, Sporns-2007}, presence of structural and functional network motifs \citep{Sporns-2004c, Sporns-2010}, and small-world organization \citep{Achard, Bassett-2017, Papo-2016, Sporns-2004a}.

Moreover, graph theoretical analysis of brain networks has played an important role in the search for \textit{biomarkers} for different aspects of brain functioning. According to the National Institutes of Health (NIH) \citep{Biomarkers}, a biomarker is ``a characteristic that is objectively measured and evaluated as an indicator of normal biological processes, pathogenic processes or pharmacologic responses to a therapeutic intervention." Alterations in structural or functional connectivity of the brain may be used as biomarkers to detect structural and functional abnormalities in network connectivity patterns, which might be indicators of neurological disorders or specific cognitive processes \citep{Ward}. As mentioned in the systematic review carried out by Farahani et al. \citep{Farahani} of GTA studies of brain networks using fMRI data, graph theory applications in human cognition include to find biomarkers (in this case, fMRI-based biomarkers) for the human intelligence (e.g., shorter characteristic path lengths in functional networks were associated with better intellectual performances \citep{vandenHeuvel-2009}), working memory (e.g., better performance of working memory was associated with lower local efficiency \citep{Stanley}), aging brain (e.g., global efficiency was essentially unchanged over the lifespan, whereas local efficiency and the rich-club coefficient increased until adulthood in healthy individuals and decreased with age \citep{Farahani}), and behavioral performance in natural environments (e.g., small-world organizations were maintained in individuals in both normal and hyperthermia conditions, whereas decreased clustering coefficients, local efficiency, and small-worldness indices were observed in heat-exposed individuals \citep{Qian}), and applications in neurological disorders include to find biomarkers for disorders such as epilepsy (e.g., disruption of global integration and local segregation was observed in patients with chronic epilepsy \citep{Vlooswijk}), Alzheimer's disease (AD) (e.g., patients with AD showed a significantly decreased clustering coefficient and characteristic path length compared to healthy individuals \citep{deHaan}), multiple sclerosis (MS) (e.g., patients diagnosed with MS showed decreased global efficiency relative to healthy individuals \citep{Liu2016}), autism spectrum disorders (ASD) (e.g., in individuals diagnosed with ASD, modularity, clustering coefficient, and local efficiency are relatively reduced compared with healthy individuals \citep{Keown}), and attention-deficit/hyperactivity disorder (ADHD) (e.g., patients with ADHD showed increased local efficiency and decreased global efficiency compared to healthy individuals \citep{Wang-2009}). Be aware that the previous network properties were obtained for (\textit{undirected}) functional networks estimated through different methods.

Besides the fMRI-based biomarkers, particularly in the study of neuropathologies, the discovery of EEG-based biomarkers (electrophysiological biomarkers) is of special interest due to the portability and low cost of EEG compared to fMRI. Among the neuropathologies studied and the topological and functional characteristics found in connectivity networks obtained from EEG data, we can mention, for example \citep{Liu, Stam-2014}: patients with Parkinson's disease (PD) showed increased connectivity strength in the theta band compared with healthy controls \citep{Utianski}; patients with epilepsy showed abnormally regular functional networks relative to healthy controls \citep{Horstmann}; patients with schizophrenia (SCZ) showed a decreased clustering coefficient, decreased global and local efficiency, and increased average characteristic path length compared to healthy controls \citep{Yin2017}. Again, the previous network properties were obtained for (\textit{undirected}) functional networks estimated through different methods.


Recently explored approaches to the study of brain networks have involved concepts from computational (algebraic) topology and computational geometry, such as simplicial complexes, homology, homotopy, Betti numbers, and persistent homology (PH). The set of all these tools in data analysis, as we already mentioned in Chapter \ref{chap:chap1}, is identified by the umbrella term topological data analysis (TDA) \citep{Chazal}. Applications in network neuroscience include, for example, assessing neural functions and structures through clique topology and PH of clique complexes built out of functional or structural brain networks \citep{Giusti2016, Giusti2015, Petri, Reimann, Sizemore}, to characterize functional brain networks of patients with ADHD and ASD through PH \citep{Lee1-AD, Lee2-AD}, to characterize \citep{Merelli} and detect \citep{Fernandez, Piangerelli, Sun2023, Wang} epileptic seizures via PH, and to assess the dynamics of functional brain connectivity through persistence vineyards \citep{Yoo}. Also, concepts from Q-Analysis have been used to study functional connectivity networks \citep{Tadic} and human connectomes \citep{Andjelkovic2020}.

In summary, understanding the dynamic nature of brain networks is essential for understanding brain function and dysfunction in healthy individuals and individuals diagnosed with neurological disorders. As we have seen, there are a myriad of techniques, including EEG, fMRI, MEG, and mathematical, statistical, and computational methods, to estimate and study the organization, topology, and dynamics of brain networks, with the aim of identifying patterns of neural activity that are associated with specific brain processes. 

In the next chapter, we will discuss in more depth the role of brain connectivity in individuals with epilepsy and how GTA and TDA can be useful in the study of epileptic brain networks.

\chapter[Epilepsy as a Disorder of Brain Connectivity]{Epilepsy as a Disorder of Brain Connectivity}
\label{chap:epilepsy}

\epigraph{Hippocrates left behind him only a single discussion
of the function of the brain and the nature of consciousness. It was included in a lecture delivered to an audience
of medical men on epilepsia, the affliction that we still
call epilepsy. Here is an excerpt from this lecture, this
amazing flash of understanding: ``Some people say that
the heart is the organ with which we think and that it
feels pain and anxiety. But it is not so. Men ought to
know that from the brain and from the brain only arise
our pleasures, joys, laughter and tears. Through it, in particular, we think, see, hear and distinguish the ugly from
the beautiful, the bad from the good, the pleasant from
the unpleasant.... To consciousness the brain is messenger."
And again, he said: ``The brain is the interpreter of
consciousness." In another part of his discussion he remarked,
simply and accurately, that epilepsy comes from
the brain ``when it is not normal."}{--- Wilder Penfield \citep{Penfield}}

\bigskip

Disorders of brain connectivity encompass a wide range of conditions affecting neural connectivity, communication, and information processing in the brain. As we discussed in the previous chapter, several abnormalities in brain connectivity networks, whether structural, functional, or effective, have been observed in brain diseases such as Parkinson's disease, Alzheimer's disease, multiple sclerosis, schizophrenia, and epilepsy. The discovery of links between the topology of brain networks and neurological disorders has increasingly encouraged researchers to use mathematical methods from graph theory, network science, and computational topology to understand the changes in the dynamics of brain connectivity that underlie the symptoms and progression of specific neuropathologies, which could potentially lead to new diagnostic methods and therapeutic approaches.

In this chapter, we discuss in more detail the neuropathology of epilepsy, which is considered one of the most prominent brain network disorders. We discuss its general characteristics, the different types of epilepsy, the classification of epileptic seizures, the etiology, epidemiology, diagnosis, and treatment of epilepsy, and, finally, how it is characterized as a network disorder by exploring the topological/structural and functional properties of epileptic networks through methods from graph theory and computational topology.

\section{An Introduction to Epilepsy}
\label{sec:epilepsy-intro}

In this section, we present the main characteristics related to epilepsy, the different types of epilepsy, the classification of epileptic seizures, the etiology, epidemiology, diagnosis, and treatment of epilepsy.

\subsection{General Aspects of Epilepsy}

Epilepsy is one of the most common neurological disorders, present in populations all over the globe. It is a disorder of the central nervous system (CNS) typified by recurrent and non-induced seizures (defined as a transient period of excessively synchronous (hypersynchronous) abnormal neuronal activity manifested in the brain in a localized or generalized way) that occur over an interval of time, and which may occur spontaneously \citep{Alarcon, Frohlich, Wasade}.  Epilepsy is defined by the number and frequency of non-induced seizures and can be classified as \textit{focal} (seizure onset involves a specific area of the brain), \textit{generalized} (seizure onset involves both cerebral hemispheres), \textit{combined generalized and focal}, or \textit{unknown} (see next subsection), depending on the area of its origin in the brain \citep{Fisher-2017}.

The first phase of a seizure is known as \textit{pre-ictal phase} (or \textit{aura}), and it occurs immediately before the \textit{ictal phase}, which corresponds to the seizure itself. The \textit{post-ictal phase} is the period right after the ictal phase. The \textit{interictal period} corresponds to the period between seizures. However, determining the precise length of the pre-ictal, ictal, or post-ictal phase is not very clear and may vary depending on each case \citep{Bandarabadi}. Although some crises may last for a short time, most of them last from seconds to minutes.

The \textit{epileptogenic zone} (EZ) refers to the cortical regions involved in the genesis and propagation of the epileptiform activity. The EZ is most likely equivalent to the \textit{seizure focus}, which is defined as the location in the brain from whence the seizure began \citep{Nadler}. For focal epilepsies, the side (left or right cerebral hemisphere) of seizure onset (or the side of seizure focus) is known as the \textit{lateralization} of the seizure. The most common form of focal epilepsy is the \textit{temporal lobe epilepsy} (TLE).

The identification of graphoelements in EEG recordings (patterns in the EEG signal) is a commonly used method in the study of epilepsy. For instance, the majority of individuals with epilepsy exhibit typical \textit{interictal epileptiform discharges} (IEDs), often known as spike ($<70 \mu$s duration), spike-and-wave, or sharp-wave ($70–200 \mu$s duration) discharges \citep{StLouis}; in TLE, an EEG pattern that has been consistently observed is the \textit{temporal intermittent rhythmic delta activity} (TIRDA), which is characterized by an intermittent ($\ge $ 3s), rhythmic, 1-4Hz activity in the anterior temporal region \citep{Fox}.

The diagnosis of epilepsy is complicated by the fact that many of the signs and symptoms of epilepsy occur over brief, irregular periods of time (e.g., seizures, IDEs). This means that in clinical EEG examinations, brain activity may appear normal; even so, EEG is the most commonly used method to confirm the diagnosis of epilepsy. Moreover, the definitive diagnosis and choice of therapy is based on the analysis of the EEG by a specialist, which can take hours; for this and other reasons, the availability of automated systems capable of detecting epilepsy accurately (and eventually in real-time) would be of great value in clinical practice \citep{Covert, Wasade}. 

For patients diagnosed with epilepsy, antiepileptic drugs (AEDs) are the main form of treatment. However, some patients present drug-resistant epilepsy (DRE) \citep{Rao}, which is a pharmacoresistant form of the disease that represents one-third of epilepsies \citep{Fattorusso} and, in these cases, a surgical intervention may be necessary to resect or disconnect the supposed EZ \citep{Rosenow}.

\subsection{A Closer Look into the Neuropathology of Epilepsy}

In the previous subsection, we briefly discussed the general aspects related to the neuropathology of epilepsy. Now, we present in more detail the classification of the different types of seizures and epilepsies, the etiology, epidemiology, diagnosis, and treatment of epilepsy.

\subsubsection{Classification of Seizure Types and Epilepsies}

Epileptic seizures may be classified according to the type of onset, level of awareness, and responsiveness \citep{Fisher-2017a}. According to the International League Against Epilepsy (ILAE), epileptic seizures can be classified into \textit{focal onset}, \textit{generalized onset}, and \textit{unknown onset} \citep{Fisher-2017}:

\begin{itemize}
 \item \textbf{Focal onset (aware/impaired awareness):} These seizures originate from a specific area of the brain and may or may not involve loss of consciousness (impaired awareness). Seizures of this type can additionally be classified into \textit{motor onset} (automatisms, atonic, clonic, epileptic spasms, hyperkinetic, myoclonic, tonic); \textit{non-motor onset} (autonomic, behavior arrest, cognitive, emotional, sensory); and \textit{focal to bilateral tonic-clonic.}
 
\item  \textbf{Generalized onset:} These seizures involve both hemispheres of the brain from the onset and typically result in loss of consciousness (impaired awareness). Seizures of this type can additionally be classified into: \textit{motor onset} (tonic-clonic, clonic, myoclonic, myoclonic-tonic-clonic, myoclonic-atonic, epileptic spasms); and \textit{non-Motor (absence)} (typical, atypical, myoclonic, eyelid myoclonia).

\item \textbf{Unknown onset:} Seizures where their specific origin or onset cannot be determined with certainty. Seizures of this type can additionally be classified into: \textit{motor} (tonic-clonic, epileptic spasms); \textit{non-motor} (behavior arrest); and \textit{unclassified}.
\end{itemize}

In addition, the epilepsy type is classified into four categories, namely: \textit{generalized epilepsy}, \textit{focal epilepsy}, \textit{combined generalized and focal epilepsy}, and \textit{unknown epilepsy}.

\subsubsection{Etiology}

The causes of epilepsy can vary widely. Since there are many distinct processes governing the electrical activity of neurons, there exists a wide variety of possibilities that can disturb these processes, which can lead to a variety of reasons for epilepsy and seizures.  Accordingly, six etiologic categories for epilepsy have been established by the ILAE Task Force \citep{Scheffer}, namely: \textit{structural, genetic, infectious, metabolic, immunological}, and \textit{unknown}.

\begin{itemize}
\item \textbf{Structural etiology:} Structural etiology corresponds to the case in which abnormalities are found on structural neuroimaging that, together with results obtained from electrophysiological tests, suggest that the abnormalities are the cause of the patient's seizures.

\item \textbf{Genetic etiology:} Genetic etiology refers to the case in which the disorder is thought to be directly caused by a known or suspected genetic abnormality.

\item \textbf{Infectious etiology:} Infectious etiology refers to the case in which epilepsy results from a known infection. This is the most common etiology.

\item \textbf{Metabolic etiology:} Metabolic etiology refers to the case in which epilepsy results from a suspected metabolic disorder.

\item \textbf{Immunological etiology:} Immunological etiology refers to the case in which epilepsy results from an immunological disorder in the patient's organism.

\item \textbf{Unknown etiology:} Unknown etiology refers to the case in which the cause of the epilepsy is not known. 
\end{itemize}

\subsubsection{Epidemiology}

The incidence of epilepsy worldwide is contained in the range between 0.5 and $1.5\%$ and is higher in developing countries, mainly due to poor prevention conditions and lack of access to suitable treatments. The prevalence follows this $1\%$ trend and is also higher in developing countries.

Although the occurrence of a spontaneous seizure (i.e., without an identified cause) does not guarantee the diagnosis of epilepsy, the risk of a second seizure is around $40\%$. After a second seizure, the risk increases to almost $100\%$. In addition, up to $30\%$ of the patients diagnosed with epilepsy may present a pharmacoresistant form of the disorder (i.e., DRE), and $10\%$ may need a surgical resection or disconnection of the EZ \citep{Frohlich}.

\subsubsection{Diagnosis}

The first step towards diagnosing epilepsy is the occurrence of at least one non-inducted (unprovoked) seizure. Afterwards, the next steps of the diagnosis include medical history (e.g., family history of epilepsy, history of brain infection, traumatic brain injury, febrile seizures, etc.), physical examination (blood pressure test, skin exam, checking for signs of infection, cancer, etc.), EEG, and neuroimaging (e.g., MRI) \citep{Milligan}. Based on the results of these tests and evaluations, a neurologist or epileptologist can make a diagnosis of epilepsy and recommend an appropriate treatment.


\subsubsection{Treatment}

Epilepsy, to this day, is a neuropathology that cannot be cured. Nevertheless, there are several treatments available to control seizures in patients diagnosed with the disease, allowing them to live an unrestricted life. The three most common treatments for epilepsy are: administration of antiepileptic drugs (pharmacological treatment); surgical intervention to resect or disconnect the EZ (neurosurgical treatment); and neurostimulation \citep{Frohlich}.

\section{Grapho-Topological Characteristics of Epileptic Brain Connectivity Networks}

Epilepsy is, notoriously, a brain connectivity network disorder, characterized by a direct relation between abnormal network dynamics and clinical manifestations \citep{Frohlich}. A fact that has been constantly observed in the literature is that patients diagnosed with epilepsy present alterations in brain connectivity networks, whether structural, functional, or effective, when compared to healthy individuals \citep{Farahani, Liu, Stam-2014}. Many of these findings come from results obtained through graph theoretical analysis (GTA) and topological data analysis (TDA) of these connectivity networks, and that is the main reason why these methods have been widely used to try to answer questions such as: How do epileptic brain networks differ from healthy brain networks? How does brain network topology change in the ictal phase? How can epileptic networks be characterized? How to detect and predict epileptic seizures?

In what follows, we summarize some results obtained from studies that indicate the association of epilepsy with changes in brain networks, particularly associated with (undirected) structural and functional networks estimated from MRI/fMRI, EEG, sEEG, or DTI data.

\begin{itemize}
\item Bernhardt et al. \citep{Bernhardt} estimated structural connectivity networks through MR-based cortical thickness correlations from MRI data obtained from 122 patients with drug-resistant temporal lobe epilepsy (TLE) (52 males and 70 females; age range: 17-62) and 47 healthy controls (23 males and 24 females; age range: 18-66). They found that patients with TLE showed changes in the distribution of hubs, increased path lengths, increased clustering coefficients, and increased vulnerability to targeted attacks compared with healthy controls.

\item Liao et al. \citep{Liao} estimated functional connectivity networks through Pearson's correlation from resting-state fMRI data obtained from 18 patients with mesial temporal lobe epilepsy (MTLE) and 27 healthy controls. They found a pattern of significantly increased local connectivity and decreased global connectivity (typical characteristics of regular networks) in patients with MTLE compared with healthy controls.

\item Ponten et al. \citep{Ponten}, estimated functional connectivity networks through synchronization likelihood from sEEG data obtained from 7 patients diagnosed with MTLE. They found that, during the ictal phase, a change in the topology of the networks occurred, with an increase in the clustering coefficient and the characteristic path length, especially in the delta, theta, and alpha bands, which might suggest a change from a small-world organization (which seems to be characteristic of healthy individuals) to a more regular organization.

\item Bonilha et al. \citep{Bonilha} estimated structural connectivity networks from DTI data of 12 patients with MTLE (5 with left MTLE and 7 with right MTLE) and 26 healthy controls. They found that, in the thalamus, ipsilateral insula, and superior temporal region, patients with MTLE showed increased degree, local efficiency, clustering coefficient, and limbic network clustering compared with healthy controls.

\item Vlooswijk et al. \citep{Vlooswijk} estimated functional connectivity networks through Pearson's correlation from fMRI data obtained from 41 patients diagnosed with chronic epilepsy (20 males and 21 females; age range: 22-63) and 23 healthy controls (9 males and 14 females; age range: 18-58). They found a disruption of global integration and a disruption of local segregation in patients with epilepsy, and efficient small-world properties (high clustering coefficients and short characteristic path lengths) in healthy controls.

\item Horstmann et al. \citep{Horstmann} estimated functional connectivity networks through cross-correlation and mean phase coherence from EEG and MEG data obtained from 21 patients diagnosed with drug-resistant epilepsy (9 males and 12 females) and 23 healthy controls (12 males and 11 females). They found that epileptic brain networks showed abnormally regular functional networks relative to healthy controls. 


\item Merelli et al. \citep{Merelli} estimated functional connectivity networks through Pearson's correlation from EEG data obtained from 10 patients diagnosed with focal epilepsy (5 males and 5 females; age range: 0.5-19). From these networks, they constructed clique complexes and computed weighted persistent entropy for the pre-ictal and ictal phases. They found that the analysis of the persistent entropy can detect the transition between the pre-ictal and ictal phases.
\end{itemize}

Next, we present some results obtained for (directed) effective connectivity networks estimated from EEG data.

\begin{itemize}
\item Hu et al. \citep{Hu} estimated G-connectivity networks through PDC from EEG data obtained from 10 patients diagnosed with focal epilepsy (5 males and 5 females; age range: 0.5-19). They found that, in the delta band, the total degree (the sum of in-degree and out-degree) at the center lobe during the ictal phase was significantly lower compared with the interictal period, and the clustering coefficients were significantly increased in the frontal, parietal, and temporal lobes during the ictal phase compared with the interictal period.

\item Baccalá et al. \citep{Baccala7} estimated G-connectivity networks through PDC from EEG data obtained from 9 patients with left / right MTLE. They found that channels in the hemisphere corresponding to the seizure focus belong to strongly connected subdigraphs, suggesting a possible graph-based biomarker to identify laterality.

\item Coito et al. \citep{Coito} estimated G-connectivity networks through PDC from EEG obtained from 16 patients with TLE (eight with left TLE (LTLE) and eight right TLE (RTLE); 12 males and 4 females; age range: 15-56). They found significantly different patterns between the networks of the LTLE and RTLE groups: ipsilateral predominance in LTLE and bilateral predominance in RTLE.

\end{itemize}

As we have seen, different data acquisition methods, connectivity estimators, and network quantifiers may provide different interpretations of the same phenomenon. However, it is notorious across these different analysis modalities that there are discrepancies between connectivity networks of healthy individuals and individuals diagnosed with some type of epilepsy, as well as alterations in the brain networks during the ictal phase compared with other periods. Therefore, it is worth searching for graph-based biomarkers for the characterization of epileptic brain networks.

\chapter[Quantitative Graph/Simplicial Analysis of Epileptic Brain Networks]{Quantitative Graph/Simplicial Analysis of Epileptic Brain Networks}
\label{chap:chap6}

\epigraph{Le hasard n'est que la mesure de notre ignorance. Les phénomènes fortuits sont, par définition, ceux dont nous ignorons les lois. \\
(Chance is only the measure of our ignorance. Fortuitous phenomena are, by definition, those whose laws we are ignorant of.)}{--- Henri Poincaré \citep{Poincare}}

\bigskip


In this chapter, we apply some of the simplicial characterization measures and simplicial similarity comparison distances introduced in Chapter \ref{chap:chap5} to directed flag complexes built out of G-connectivity networks estimated through iPDC from EEG signals of patients diagnosed with left temporal lobe epilepsy. 

More specifically, we applied the iPDC estimator to epileptic EEG signals and constructed G-connectivity networks in three different frequency bands (delta, theta, and alpha) for the left and right brain hemispheres of each patient, and for three seizure phases (pre-ictal, ictal, and post-ictal phase). Subsequently, we constructed directed flag complexes from these networks, computed some chosen simplicial characterization measures (lower variants) and simplicial similarity comparison distances for each hemisphere, frequency, and seizure phase, in five levels of organization $q$ ($q=0,1,2,3,4$) of the complexes, and, finally, we performed a statistical analysis to evaluate the seizure phases and lateralization in each frequency band, in each hemisphere, and at each level $q$. Furthermore, we computed the standard graph measures introduced in Chapter \ref{chap:chap2} and compared their statistical results with the results obtained for their simplicial analogues.

The aim of this analysis is twofold: 1) to understand how epileptic networks and their higher-order counterparts change throughout different seizure phases, in different frequency bands, and for each cerebral hemisphere; 2) to identify novel biomarkers for epileptic brain networks associated with their higher-order structures and higher-order connectivities.

\section{EEG Data Acquisition and Preprocessing}
\label{sec:eeg-data-preprocessing}

The EEG data used in this study was obtained from the Siena Scalp EEG Database (SSED)\footnote{Public available at \url{https://physionet.org/content/siena-scalp-eeg/1.0.0/}} \citep{Detti}, which consists of EEG recordings acquired from 14 patients (9 males and 5 females), diagnosed with epilepsy, at the Unit of Neurology and Neurophysiology at the University of Siena, Italy. The EEG signals were recorded using a Video-EEG with a sampling rate of 512 Hz, and were made available in the European Data Format (EDF).  Three types of seizures were identified and classified according to the criteria of the ILAE, namely:  focal onset impaired awareness, focal onset without impaired awareness, and focal to bilateral tonic–clonic. Moreover, the documentation includes annotations on the start and end times of the ictal phase for each recording.

From the SSED, we selected eight patients with left temporal lobe epilepsy (TLE) based on the quality of the signal. For these patients, the records were performed using a total of 29 electrodes (Fp1, F3, C3, P3, O1, F7, T3, T5, Fc1, Fc5, Cp1, Cp5, F9, Fz, Cz, Pz, Fp2, F4, C4, P4, O2, F8, T4, T6, Fc2, Fc6, Cp2, Cp6, F10), placed according to the international 10–20 system in a referential montage. Table \ref{tab:table-patients} presents detailed information about the selected patients: columns 1 (Pat. id) presents the patient identification; columns 2 (Age) reports their ages; column 3 (Gender) reports their gender; column 4 (Localization) reports the location of the seizure focus; column 5 (Lateralization) reports the lateralization (left or right) of the seizure focus; column 6 (Time) presents the total registration time in minutes. 

\begin{table}[h!]
  \center
    \caption{Patient information.}
    \label{tab:table-patients}
    \begin{tabular}{c c c c c c}
      \toprule 
{\bf Pat. id} &  {\bf Age } & {\bf Gender } & {\bf Localization }  & {\bf Lateralization } & {\bf Time (min)}\\
      \midrule 
PN01 & 46 & Male & Temporal Lobe & Left &  809 \\
PN06 & 36 & Male & Temporal Lobe & Left &  722 \\
PN07 & 20 & Female  & Temporal Lobe & Left & 523  \\
PN09 & 27 & Female & Temporal Lobe & Left &  410 \\
PN12 & 71 & Male & Temporal Lobe & Left &  246 \\
PN13 & 34 & Female & Temporal Lobe & Left & 519  \\
PN14 & 49 & Male & Temporal Lobe & Left &  1408 \\
PN16 & 41 & Female & Temporal Lobe & Left & 303  \\
      \bottomrule 
    \end{tabular}
  
\end{table}

\smallskip

Figure \ref{fig:preprocessing} presents the whole preprocessing workflow performed on the EEG data. Specifically, the preprocessing was performed in EEGLAB and included the following steps: (1) The raw EEG signals were downsampled from 512 Hz to 256 Hz. (2) A 1 Hz high-pass filter was applied to each channel. (3) A notch filter at 50Hz (European standard power line frequency) was applied to each channel using the CleanLine\footnote{\url{https://www.nitrc.org/projects/cleanline}} EEGLAB plugin. (4) The midline electrodes (Fz, Cz, and Pz) were removed as they may be contaminated with the electrical activities of both hemispheres, thus a final configuration with 26 channels was obtained, as is schematically represented in Figure \ref{fig:montage-siena}. (5) Finally, an independent components analysis (ICA) (employing the InfoMax algorithm) was performed to remove artifactual components. Figure \ref{fig:eeg-pre-ic-po} presents a cut of 162 seconds of preprocessed signal (with the midline electrodes) as an example.

\bigskip

\begin{figure}[h!]
    \centering
\includegraphics[scale=1.0]{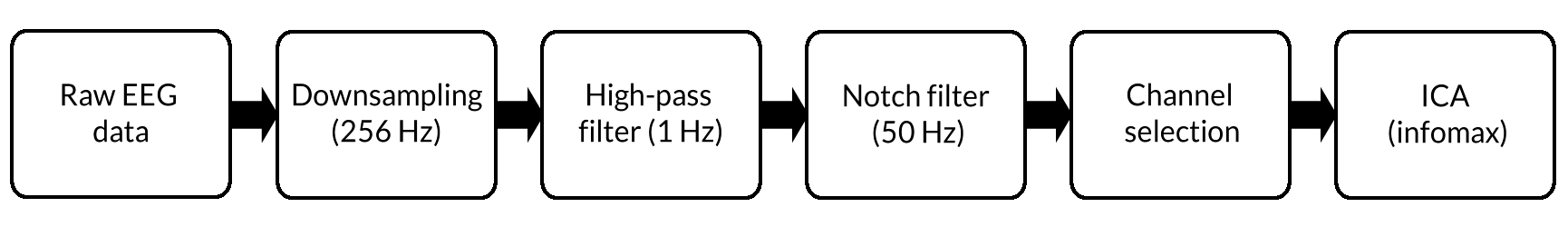}
    \caption{EEG data preprocessing workflow.}
    \label{fig:preprocessing}
\end{figure}

\bigskip

\begin{figure}[H]
    \centering
\includegraphics[scale=0.33]{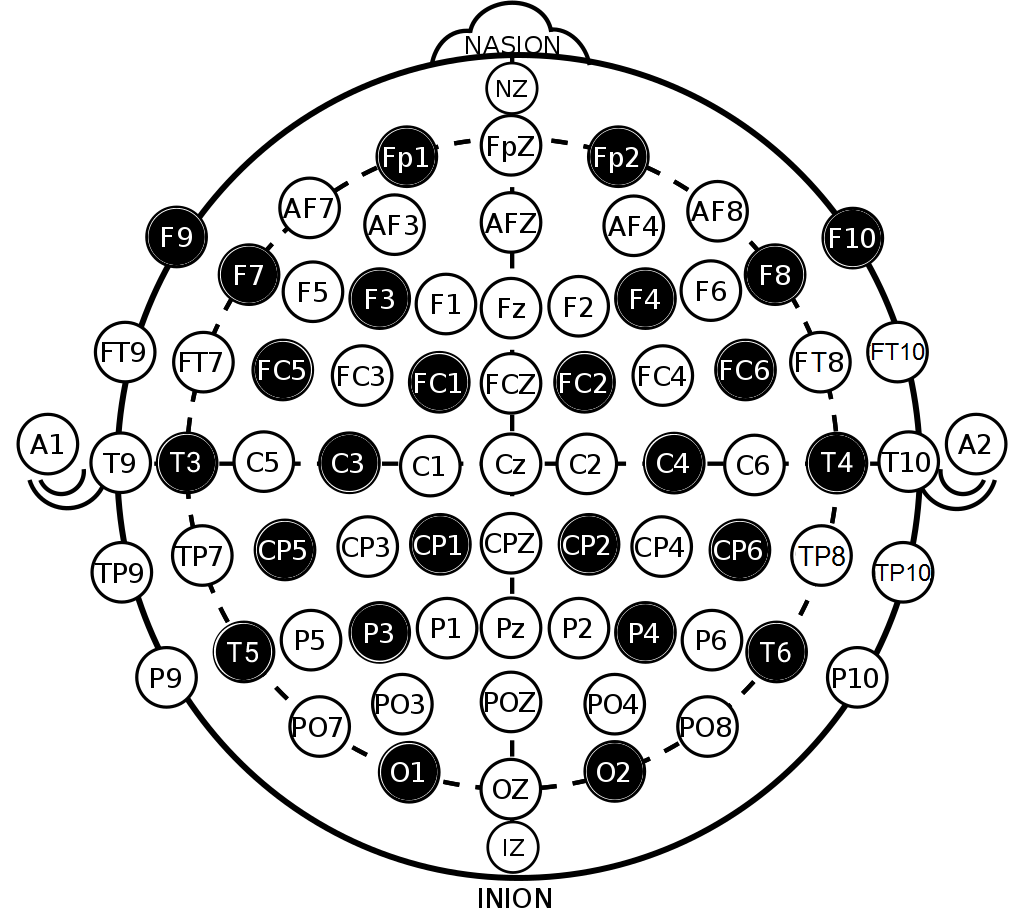}
    \caption{Configuration of the 26 selected electrodes (represented in black), according to the international 10–20 system.}
    \label{fig:montage-siena}
\end{figure}

\begin{figure}[H]
    \centering
\includegraphics[scale=0.4]{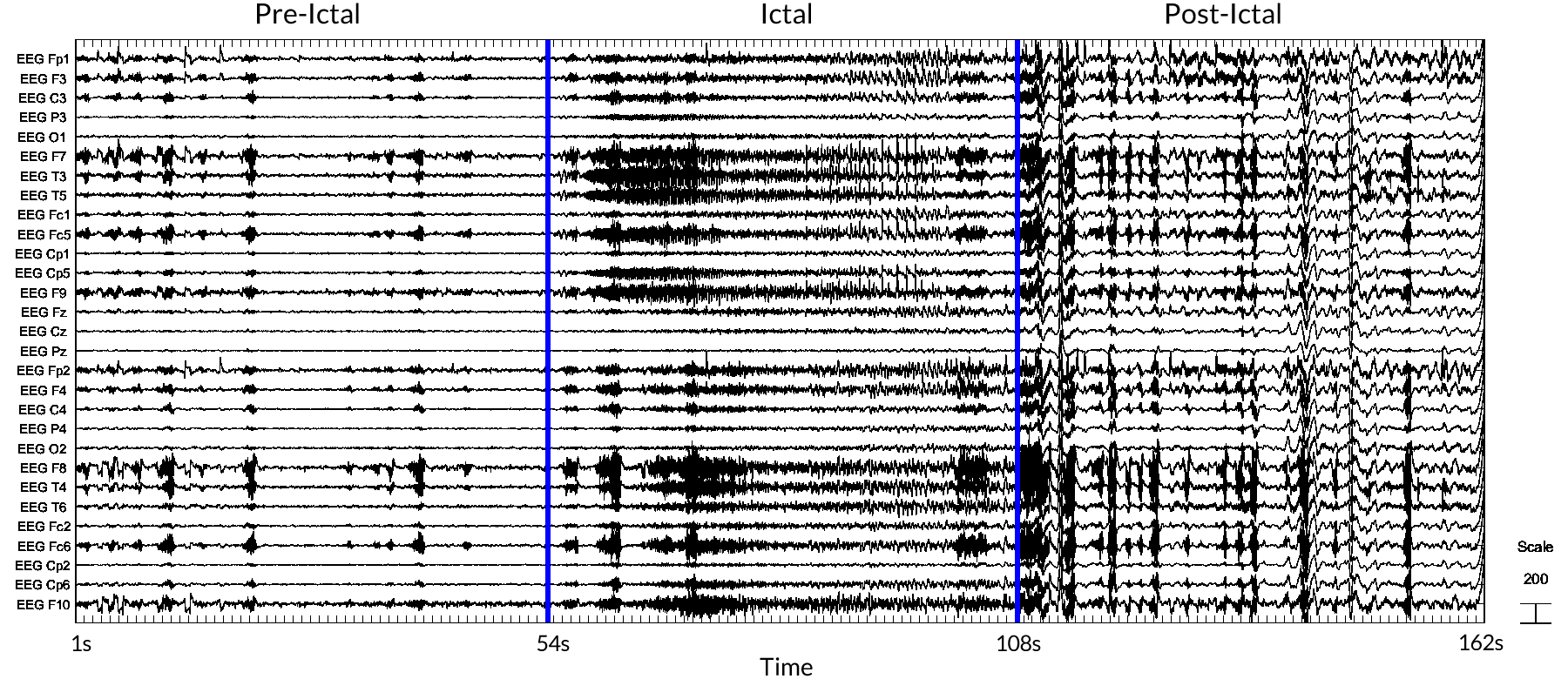}
    \caption{Example of preprocessed signal corresponding to patient PN01 (crisis 1). This signal represents a 162s cut from the original signal, in which the first 54s correspond to the pre-ictal phase, followed by 54s of the ictal phase, and the final 54s correspond to the post-ictal phase (separated by blue lines). The vertical scale is set to 200 $\mu V$.}
    \label{fig:eeg-pre-ic-po}
\end{figure}






\section{Analysis of Seizure Phases and Lateralization in Left Temporal Lobe Epilepsy}

In this section, we present the methodology and the results of the analysis performed on the preprocessed EEG data described in the previous section.

\subsection{Methodology}

\subsubsection{Brain G-Connectivity Networks}

We started our analysis by identifying the pre-ictal, ictal, and post-ictal phases in the preprocessed EEG signals based on the annotations provided together with the data, and then dividing the signals from each phase into 30s epochs: 30s immediately before the seizure, 30s starting on the seizure onset, and 30s immediately after the seizure. 


Afterwards, we applied the iPDC estimator on each 30s epoch using a sliding window technique with fixed-size windows of 10s and $80\%$ overlap, for three different frequency bands, namely: delta band [1 Hz, 4 Hz) (here, we are considering only frequencies $\ge$ 1 Hz), theta band [4 Hz, 8 Hz), and alpha band [8 Hz, 14 Hz). This procedure resulted in 11 weighted digraphs, or G-connectivity networks (corresponding to a discrete-time dynamic digraph), for each seizure phase in each of the frequency bands, where the weights correspond to the strengths of the connections. 

The iPDC networks were computed via the MATLAB package \texttt{asympPDC} version 3.0 (see Appendix \ref{appendix:software}). The VAR autoregression coefficients of these networks were estimated through the Nuttall-Strand's method. The order for the VAR models was estimated by three different information criteria, namely: Akaike's information criterion, Hannan-Quinn's criterion, and  Bayesian–Schwartz's criterion, which yielded an order equal to 2; however, the connectivity patterns differed just slightly from models with a fixed order set to 12, therefore we chose to use a fixed order equal to 12. Only the statistically significant connections were considered, and they were estimated asymptotically (see Subsection \ref{sec:asymp-pdc}) with a significance level of $0.1\%$. Figures \ref{fig:pip-digraphs1}, \ref{fig:pip-digraphs2}, and \ref{fig:pip-digraphs3} present examples of the dynamic iPDC networks (in discrete time) obtained for patient PN01 for the pre-ictal, ictal, and post-ictal phases in the delta band, respectively.

After obtaining the iPDC networks, as we are not taking the weights into account, we converted all of them into binary networks. Also, we removed all double edges in such a way that the clique structures of the networks were preserved. Subsequently, we separated each of these networks into two networks as follows: we separated the nodes corresponding to the left hemisphere from the nodes corresponding to the right hemisphere, preserving the intrahemispheric connections in both hemispheres, disregarding the interhemispheric connections, thus obtaining two networks, one for each hemisphere. We performed this procedure for each of the three seizure phases (pre-ictal, ictal, post-ictal) in each frequency band (delta, theta, alpha).

Furthermore, bearing in mind the density effects \citep{vanWijk}, it is important to clarify that no explicit density normalization was imposed across epileptic phases, since the main objective was to quantify the spontaneous emergence of directed higher-order structures and directed higher-order connectivities in the original iPDC networks. In epileptic phase transitions, increases in connectivity density may themselves reflect neurophysiologically relevant processes; therefore, removing density effects could suppress meaningful topological signatures of seizure dynamics.


\subsubsection{Simplicial Characterization Measures}

In order to analyze the directed higher-order connectivity and higher-order topology of the networks through simplicial characterization measures, we computed the lower $q$-digraphs (we will omit the term ``lower" from now on), for $q=0,1,2,3,4$, for the networks from the right and left hemispheres, for each seizure phase in each frequency band. Figures \ref{fig:fig-q-analysis1} and \ref{fig:fig-q-analysis2} present examples of $q$-digraphs, for $q=0,1$, constructed from networks of the right and left hemispheres of patients PN01 and PN12, respectively.

Due to the limited nature of this study, we chose nine among all simplicial measures introduced in Chapter \ref{chap:chap5}. We chose six global measures and three local measures, and at least one belonging to each of the five categories of simplicial measures (lower variants) (excluding the discrete curvatures):

\begin{itemize}

\item Distance-based simplicial measures:  Global $q$-efficiency (global);

\item Simplicial centralities: In- and out-$q$-degree centralities (local),  $q$-harmonic centrality (local), and global $q$-reaching centrality (global);

\item Simplicial segregation measures: Average  $q$-clustering coefficient (global);

\item Simplicial entropies: In- and out-$q$-degree distribution entropies (global);

\item Spectrum-related simplicial measures: $q$-energy (global).

\end{itemize}

Bearing in mind that the order (N) of the network can impact the simplicial measures that depend on N (N-dependent), we used a fixed N ($\mathrm{N}_{max}$) equal to the largest N among all $q$-digraphs, for $q=0,1,2,3,4$,  for all seizure phases in all frequency bands, for both hemispheres of all patients. For each simplicial measure and for each patient, we computed the mean and the standard deviation between the $q$-digraphs constructed from each of the 11 digraphs, for each seizure phase in each frequency band, and for both hemispheres. For the local simplicial measures, we computed the mean of the maximum values obtained among all nodes.

Furthermore, we applied the usual graph measures presented in Chapter \ref{chap:chap2} corresponding to their simplicial analogues exposed above in the iPDC networks of the right and left hemispheres, following the same procedure described above for the simplicial measures.

Both the construction of the $q$-digraphs and the computation of the simplicial and graph measures were carried out through the Python package \texttt{DigplexQ} (see Appendix \ref{appendix:software}).


 \begin{figure}[h!]
    \centering
\includegraphics[scale=0.5]{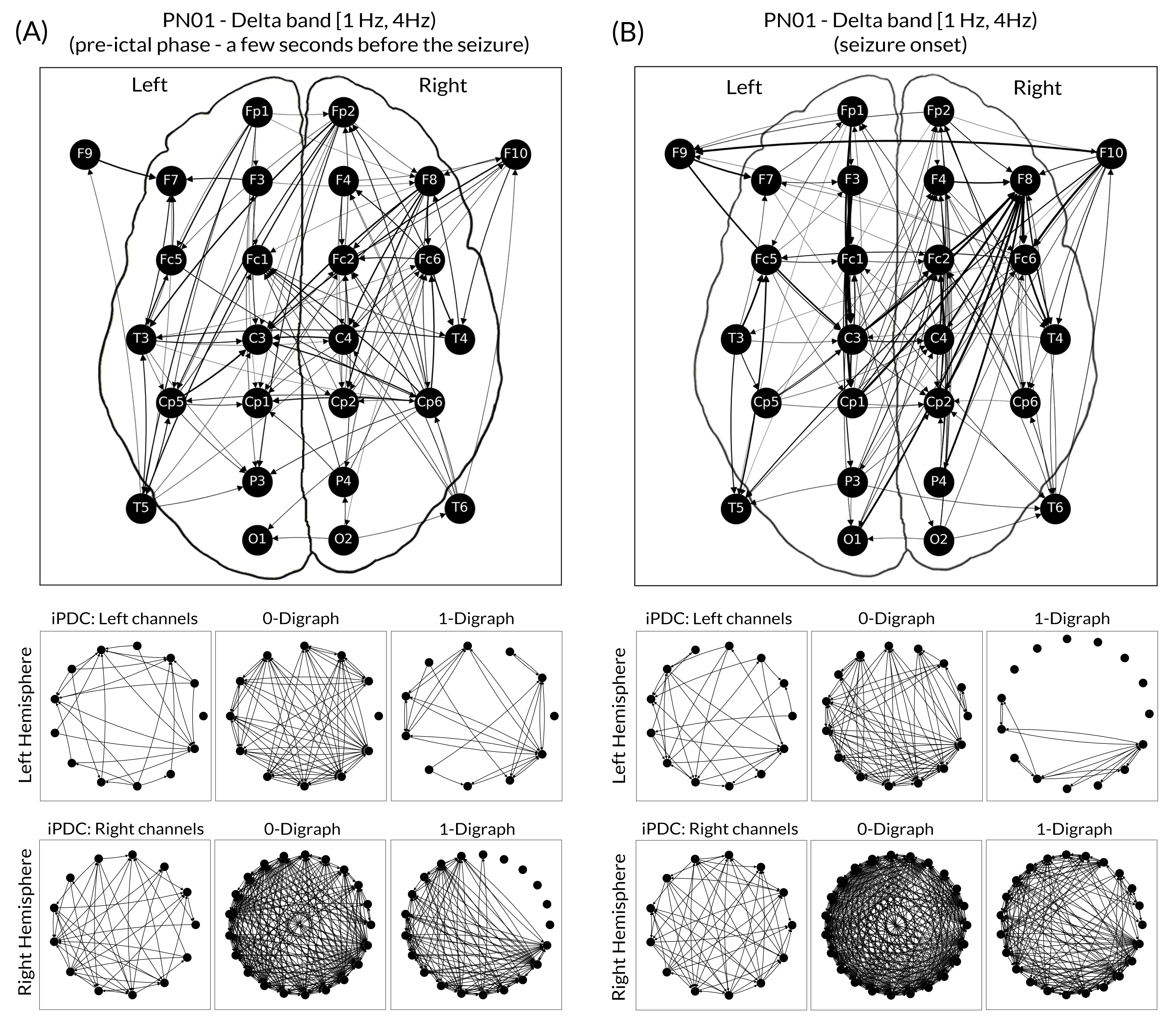}
 \caption{iPDC networks (arcs thicknesses are proportional to the weights) of patient PN01 computed in a 10s interval a few seconds before the seizure (A) and in a 10s interval starting on seizure onset (B) in the delta band, together with the (weightless) $q$-digraphs corresponding to the right and left hemispheres networks (considering only intrahemispheric connections) for $q=0,1$.}
 \label{fig:fig-q-analysis1}
\end{figure}

\begin{figure}[h!]
    \centering
\includegraphics[scale=0.5]{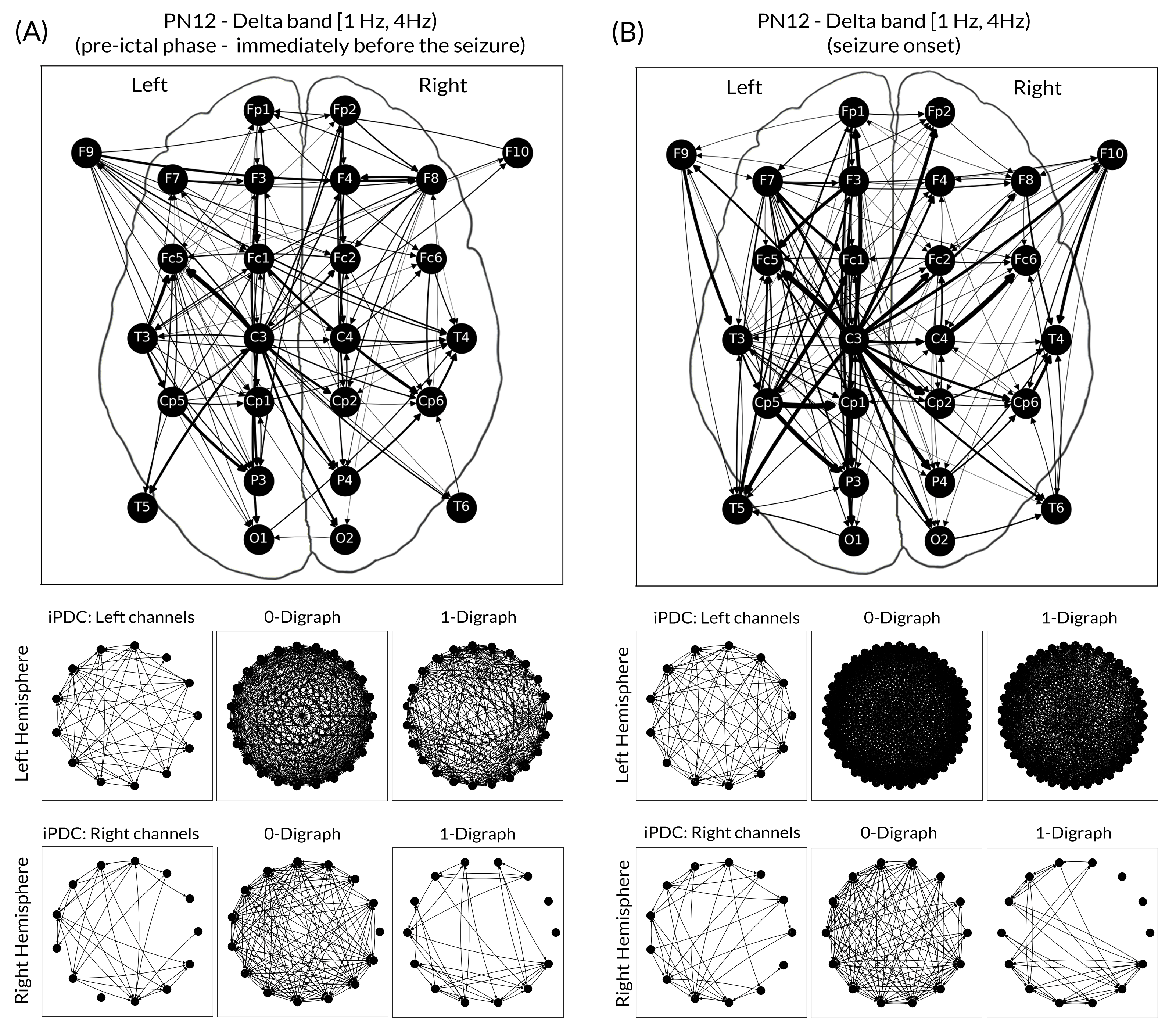}
 \caption{iPDC networks (arcs thicknesses are proportional to the weights) of patient PN12 computed in a 10s interval immediately before the seizure (A) and in a 10s interval starting on seizure onset (B) in the delta band, together with the (weightless) $q$-digraphs corresponding to the right and left hemispheres networks (considering only intrahemispheric connections) for $q=0,1$.}
 \label{fig:fig-q-analysis2}
\end{figure}

\subsubsection{Simplicial Similarity Comparison Distances}

For the simplicial distances, we only considered the iPDC networks of the right and left hemispheres, i.e., we did not take into account their $q$-digraphs. As in the case of the simplicial characterization measures, due to the limited nature of this study, we chose eight simplicial distances, namely: bottleneck distance, Wasserstein distance, Betti distance, 1st, 4th, and 5th topological structure distances, histogram cosine kernel, and Jaccard kernel. Also, for the bottleneck, Wasserstein, Betti, 4th, and 5th topological structure distances, we only considered the 0-th Betti numbers, and for the histogram cosine kernel and Jaccard kernel, we considered the directed flag complexes associated with the networks.

Let $d$ be any of the eight chosen simplicial distances. Given two digraphs, $G_{1}$ and $G_{2}$, let's denote by $d(G_{1}, G_{2})$ the value produced by the distance $d$ in the corresponding structures obtained from these two digraphs, for example, if $d$  is the bottleneck distance, then $d(G_{1}, G_{2}) = d_{W_{\infty}}(P_{1}, P_{2})$, where $P_{1}$ and $P_{2}$ are the persistent diagrams associated with $G_{1}$ and $G_{2}$, respectively, and if $d$ is the Jaccard kernel, then $d(G_{1}, G_{2}) = K_{J}(\mathcal{X}_{1}, \mathcal{X}_{2})$, where $\mathcal{X}_{1}$ and $ \mathcal{X}_{2}$ are the directed flag complexes of $G_{1}$ and $G_{2}$, respectively. For each of the simplicial distances, we produced different distributions, for each patient, according to the cerebral hemisphere and seizure phase, for every frequency band: 

\begin{itemize}
\item Distributions $d(G_{ic}^{R}, G_{ic}^{R})$ and $d(G_{ic}^{L}, G_{ic}^{R})$, where $G_{ic}^{R}$ are randomly chosen digraphs from the right hemisphere and $G^{L}_{ic}$ are randomly chosen digraphs from the left hemisphere corresponding to the ictal phase (without repetition). 

\item Distributions $d(G_{pre}^{L}, G_{pre}^{L})$ and $d(G_{pre}^{L}, G_{ic}^{L})$, where $G_{pre}^{L}$ are randomly chosen digraphs from the left hemisphere corresponding to the pre-ictal phase, and $G_{ic}^{L}$ are randomly chosen digraphs from the left hemisphere corresponding to the ictal phase (without repetition).

\item Distributions $d(G_{pos}^{L}, G_{pos}^{L})$ and $d(G_{ic}^{L}, G_{pos}^{L})$, where $G_{pos}^{L}$ are randomly chosen digraphs from the left hemisphere corresponding to the post-ictal phase, and $G_{ic}^{L}$ are randomly chosen digraphs from the left hemisphere corresponding to the ictal phase (without repetition).

\item Distributions $d(G_{pre}^{R}, G_{pre}^{R})$ and $d(G_{pre}^{R}, G_{ic}^{R})$, where $G_{pre}^{R}$ are randomly chosen digraphs from the right hemisphere corresponding to the pre-ictal phase, and $G_{ic}^{R}$ are randomly chosen digraphs from the right hemisphere corresponding to the ictal phase (without repetition).

\item Distributions $d(G_{pos}^{R}, G_{pos}^{R})$ and $d(G_{ic}^{R}, G_{pos}^{R})$, where $G_{pos}^{R}$ are randomly chosen digraphs from the right hemisphere corresponding to the post-ictal phase, and $G_{ic}^{R}$ are randomly chosen digraphs from the right hemisphere corresponding to the ictal phase (without repetition).

\end{itemize}

Afterwards, we computed the means and standard deviations for each distribution, in each frequency band, for each patient.

The computation of the simplicial distances was carried out through the Python package \texttt{DigplexQ} (see Appendix \ref{appendix:software}).

\subsubsection{Statistical Analysis}

Finally, we performed the following statistical analysis:

\begin{itemize}
\item For each simplicial measure, Wilcoxon paired tests, at a significance level $\alpha = 0.05$, were performed to verify the statistical differences between the simplicial measure of the pre-ictal phase and ictal phase as well as of the ictal phase and post-ictal phase, in each frequency band (delta, theta, alpha) and at each level $q=-1,0,1,2,3,4$ (where $q=-1$ represents the original networks), for the right and left hemispheres separately. In addition, we computed Pearson's correlation between the simplicial measures across the levels $q=-1,0,1,2,3,4$ for each seizure phase (pre-ictal, ictal, and post-ictal), and also across the seizure phases at each level $q=-1,0,1,2,3,4$, both for each frequency band and for each cerebral hemisphere.

\item For each simplicial distance, Wilcoxon paired tests, at a significance level $\alpha = 0.05$, were performed to verify the statistical differences between the means of the following distributions: $d(G_{ic}^{R}, G_{ic}^{R})$ and $d(G_{ic}^{L}, G_{ic}^{R})$ (to verify the differences between the left and the right hemispheres in the ictal phase); $d(G_{pre}^{L}, G_{pre}^{L})$ and $d(G_{pre}^{L}, G_{ic}^{L})$ (to verify the differences between the pre-ictal phase and the ictal phase in the left hemisphere); $d(G_{pos}^{L}, G_{pos}^{L})$ and $d(G_{ic}^{L}, G_{pos}^{L})$ (to verify the differences between the ictal phase and the post-ictal phase in the left hemisphere); $d(G_{pre}^{R}, G_{pre}^{R})$ and $d(G_{pre}^{R}, G_{ic}^{R})$ (to verify the differences between the pre-ictal phase and the ictal phase in the right hemisphere); and $d(G_{pos}^{R}, G_{pos}^{R})$ and $d(G_{ic}^{R}, G_{pos}^{R})$ (to verify the differences between the ictal phase and the post-ictal phase in the right hemisphere), for each frequency band.
\end{itemize}

The statistical analysis was carried out through the Python statistical package \texttt{Pingouin} (see Appendix \ref{appendix:software}).

Figure \ref{fig:two-analysis-pipeline} summarizes the analysis workflow, including the preprocessing step.

\begin{figure}[h!]
    \centering
\includegraphics[scale=1.05]{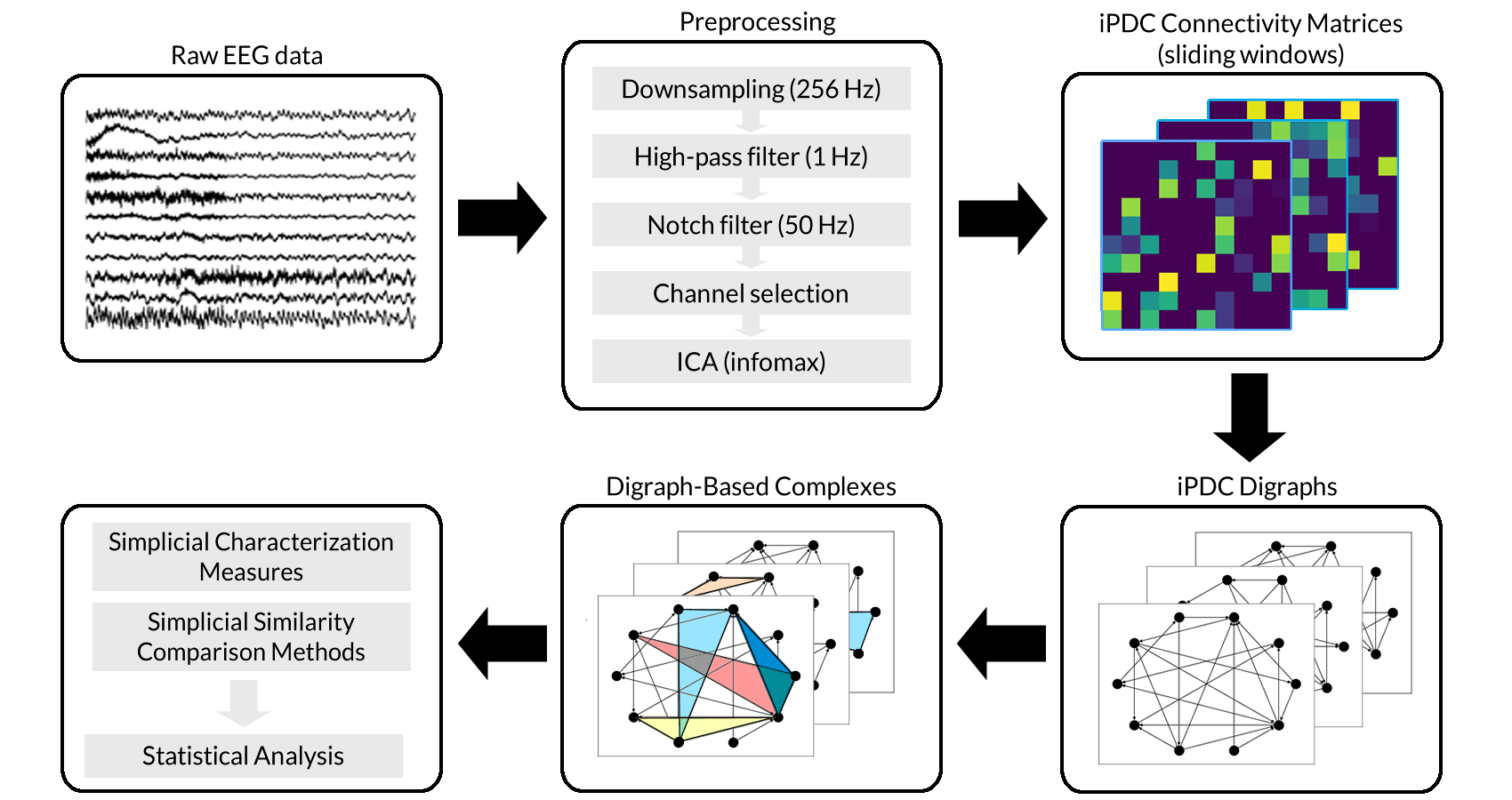}
    \caption{Summary of the analysis workflow.}
    \label{fig:two-analysis-pipeline}
\end{figure}






\subsection{Results and Discussion}

\subsubsection{Simplicial Characterization Measures}

All results discussed in this part can be found in Appendix \ref{appendix:results}. Figure \ref{fig:grand-mean-std_L} through \ref{fig:grand-mean-q4_R} present the grand means (mean of all of the means) and standard deviations of the simplicial measures; Table \ref{tab:w-test-meas-std} through \ref{tab:w-test-meas-q3} present the W-statistics and the p-values of the Wilcoxon paired tests; and Figure C.13 through C.21 present the Pearson's correlation coefficients between the simplicial measures. These tables and figures present the results for the original iPDC networks (henceforth referred to as $(-1)$-digraphs or level $q=-1$) and for the $q$-digraphs, $q=0,1,2,3,4$, for the right and left hemispheres, for each seizure phase (pre-ictal, ictal, post-ictal), and for each frequency band (delta, theta, alpha).

In the following, we present and discuss only statistically significant results ($p<0.05$) and, for the Pearson's correlation coefficient $r$, we consider two measures to be strongly positively correlated if $r \ge 0.7$ and strongly negatively correlated if $r \le -0.7$ ($ p<0.05$). Also, for the sake of simplicity, we refer to directed cliques simply as cliques, we omit the term ``directed simplicial" from the 
nomenclature of the measures, we use the notations $\delta$, $\theta$, and $\alpha$ to indicate the delta, theta, and alpha bands, respectively, and we adopt LH and RH to indicate the left hemisphere and the right hemisphere, respectively.

\bigskip
\noindent \textbf{Global $q$-efficiency:} The global $q$-efficiency increased significantly in the ictal phase compared with the pre-ictal phase in the $\delta$ band at levels $q=-1, 0, 1, 2$ in both hemispheres and at level $q=3$ in the LH, and in $\theta$ band at levels $q=2, 3$ in the LH and at level $q=-1$ in the RH. Also, this measure is strongly positively correlated with most of the measures across all seizure phases at all levels $q$, and across all levels $q$ in all seizure phases, for all frequencies, and in both hemispheres. These results suggest that the efficiency of transmitting information at a global level in each hemisphere increases in the ictal phase, not only at the level of nodes but also at the level of cliques, specifically in the $\delta$ band.

\bigskip
\noindent \textbf{In-$q$-degree centrality:} The in-$q$-degree centrality increased significantly in the ictal phase compared with the pre-ictal phase in the $\delta$ band at levels $q=0, 1, 2,3$ in the LH and at levels $q=1,2$ in the RH, and in $\theta$ band at levels $q=2,3$ in the LH and at levels $q=1,2,3$ in the RH. Also, this measure is strongly negatively correlated with 
most of the measures in relation to all seizure phases at level $q=-1$ in the LH, and it is strongly positively correlated with most of the measures across all seizure phases at all other levels $q$, and across all levels $q$ in all seizure phases, for all frequencies and for both hemispheres. These results suggest that the importance of nodes in network communication, in relation to the inner flow, did not change significantly in the ictal phase compared to the pre-ictal phase, however there was a significant increase at higher levels, that is, the importance of cliques in higher-order communication increased in the ictal phase, in both hemispheres, in both $\delta$ and $\theta$ bands.

\bigskip
\noindent \textbf{Out-$q$-degree centrality:} The out-$q$-degree centrality increased significantly in the ictal phase compared with the pre-ictal phase in the $\delta$ band at levels $q=0, 1, 2,3$ in the LH and at levels $q=1,2$ in the RH, and in $\theta$ band at levels $q=2,3$ in the LH and at levels $q=1,2,3$ in the RH. Also, this measure is strongly negatively correlated with 
most of the measures in relation to all seizure phases at level $q=-1$ in the LH and it is strongly positively correlated with most of the measures across to all seizure phases at all other levels $q$, and across to all levels $q$ in all seizure phases, for all frequencies and for both hemispheres. These results suggest that the importance of nodes in network communication, in relation to the outer flow, did not change significantly in the ictal phase compared to the pre-ictal phase, however, there was a significant increase at higher levels, that is, the importance of cliques in higher-order communication increased in the ictal phase, in both hemispheres, in both $\delta$ and $\theta$ bands.

\bigskip
\noindent \textbf{$q$-Harmonic centrality:} The $q$-harmonic centrality increased significantly in the ictal phase compared with the pre-ictal phase in the $\delta$ band at levels $q=0, 1, 2,3$ in the LH and at levels $q=-1,0,2$ in the RH, and in $\theta$ band at levels $q=2,3$ in the LH and at levels $q=-1,2$ in the RH. Also, this measure is strongly positively correlated with most of the measures across all seizure phases at all levels q, for all frequencies, and for both hemispheres. These results suggest that, in the ictal phase, higher-order cliques transmit information more efficiently, in both hemispheres, especially in the $\delta$ band, on the other hand, nodes transmit information more efficiently in the RH, in both $\delta$ and $\theta$ bands.

\bigskip
\noindent \textbf{Global $q$-reaching centrality:} The global $q$-reaching centrality increased significantly in the ictal phase compared with the pre-ictal phase in the $\delta$ band at levels $q=0, 1, 2,3$ in the LH and at levels $q=0,2,3$ in the RH, and in $\theta$ band at level $q=2$ in the LH and at levels $q=0,1,2$ in the RH. Also, this measure is strongly negatively correlated with most of the measures in relation to all seizure phases at level $q=-1$ in both hemispheres and it is strongly positively correlated with most of the measures across all seizure phases at all other levels $q$, and across all levels $q$ in all seizure phases, for all frequencies and for both hemispheres. These results suggest that, in the ictal phase, higher-order cliques have more influence in the information flow, in both hemispheres, especially in the $\delta$ band.

\bigskip
\noindent \textbf{Average $q$-clustering coefficient:} The average $q$-clustering coefficient increased significantly in the ictal phase compared with the pre-ictal phase in the $\delta$ band at levels $q=-1,0,1,2$ in the LH and at levels $q=1,2$ in the RH, and in $\theta$ band at levels $q=1,2,3$ in the LH and at levels $q=-1,1,2,3$ in the RH. Also, this measure is strongly positively correlated with most of the measures across all seizure phases at all levels $q$, and across all levels $q$ in all seizure phases, for all frequencies and for both hemispheres. These results suggest that, in the ictal phase, cliques tend to form higher-order clusters, in both hemispheres, in both $\delta$ and $\theta$ bands.


\bigskip
\noindent \textbf{In-$q$-degree distribution entropy:} The in-$q$-degree distribution entropy increased significantly in the ictal phase compared with the pre-ictal phase in the $\delta$ band at levels $q=1,2,3$ in the LH and at levels $q=-1,0,2,3$ in the RH, and in $\theta$ band at level $q=2$ in the LH and at level $q=-1,1,2,3$ in the RH. Also, this measure is strongly positively correlated with most of the measures in relation to all seizure phases at all levels $q$ and across all levels $q$ in all seizure phases, for all frequencies and for both hemispheres. These results suggest that, in the ictal phase, the higher-order networks, in relation to their inner higher-order flow, become ``more random" (or ``less regular"), i.e. they get closer to a random model, in both hemispheres, in both $\delta$ and $\theta$ bands.

\bigskip
\noindent \textbf{Out-$q$-degree distribution entropy:} The out-$q$-degree distribution entropy increased significantly in the ictal phase compared with the pre-ictal phase in the $\delta$ band at levels $q=1,2,3$ in the LH and at levels $q=-1,0,2,3$ in the RH, and in $\theta$ band at level $q=2$ in the LH and at level $q=-1,1,2,3$ in the RH. Also, this measure is strongly positively correlated with most of the measures in relation to all seizure phases at all levels $q$ and across all levels $q$ in all seizure phases, for all frequencies and for both hemispheres. These results suggest that, in the ictal phase, the higher-order networks, in relation to their outer higher-order flow, become ``more random" (or ``less regular"), i.e. they get closer to a random model, in both hemispheres, in both $\delta$ and $\theta$ bands.

\bigskip
\noindent \textbf{$q$-Energy:} The $q$-energy increased significantly in the ictal phase compared with the pre-ictal phase in the $\delta$ band at levels $q=0,1,2,3$ in the LH and at levels $q=-1,0,2,3$ in the RH, and in $\theta$ band at levels $q=2,3$ in the LH and at levels $q=-1,3$ in the RH. Also, this measure is strongly positively correlated with most of the measures across all seizure phases at all levels $q$, for all frequency bands and for both hemispheres. These results suggest that, in the ictal phase, the higher-order networks become more connected, in both hemispheres, in both $\delta$ and $\theta$ bands.

\bigskip

In short, all simplicial characterization measures showed an increase from the pre-ictal phase to the ical phase, for various levels of topological organization $q$, in both hemispheres, suggesting an increase in clustering, efficiency, and connectivity, and a shift towards a ``more random" organization (both in relation to inner and outer flow) in higher-order networks, especially in the $\delta$ and $\theta$ bands. However, a caveat is necessary here, because while the directed clustering coefficient takes into account the total degree of the nodes, the in and out entropies only take into account the in and out degrees, respectively; therefore, not necessarily an increase in the directed clustering coefficient will lead to a decrease in some of these entropies. Furthermore, as mentioned earlier in the methods section, density effects were intentionally disregarded; thus, the reported simplicial measures should be interpreted as \textit{raw} higher-order organizational biomarkers rather than density-adjusted network invariants.

\subsubsection{Simplicial Similarity Comparison Distances}

All results discussed in this part can be found in Appendix \ref{appendix:results}. Tables \ref{tab:w-test-dist2} through \ref{tab:w-test-dist5} present the W-statistics and the p-values of the Wilcoxon paired tests. These tables present the results for each frequency band (delta, theta, alpha).

In the following, we present and discuss only statistically significant results (p-values $<0.05$). Also, just as before, we refer to directed cliques simply as cliques, and we use the notations $\delta$, $\theta$, and $\alpha$ to indicate the delta, theta, and alpha bands, respectively, and LH and RH to indicate the left hemisphere and the right hemisphere, respectively.

\bigskip
\noindent \textbf{Bottleneck distance:} The bottleneck distance showed statistically significant difference between the pre-ictal phase and the ictal phase in the RH, in the $\theta$ band.

\bigskip
\noindent \textbf{Wasserstein distance:} The Wasserstein distance showed statistically significant difference between the pre-ictal phase and the ictal phase in both hemispheres, in the $\theta$ band, and between the ictal phase and the post-ictal phase in both hemispheres, but in the $\theta$ band in the LH and in the $\alpha$ band in the RH.

\bigskip
\noindent \textbf{Betti distance:} The Betti distance showed a statistically significant difference between the pre-ictal phase and the ictal phase in both hemispheres, in the $\theta$ band.

\bigskip
\noindent \textbf{First topological distance:} The first topological distance showed a statistically significant difference between the pre-ictal phase and the ictal phase in the LH, in the $\theta$ band, and between the ictal phase and the post-ictal phase in the RH, in the $\delta$, $\theta$, and $\alpha$ bands.

\bigskip
\noindent \textbf{Fifth topological distance:} The fifth topological distance showed a statistically significant difference between the ictal phase and the post-ictal phase in LH, in the $\theta$ band, and in the RH, in the $\delta$ band.

\bigskip
\noindent \textbf{Histogram cosine kernel:} The histogram cosine kernel showed a statistically significant difference between the ictal phase and the post-ictal phase in the RH, in the $\delta$, $\theta$, and $\alpha$ bands.

\bigskip
\noindent \textbf{Jaccard kernel:} The Jaccard kernel showed statistically significant difference between the pre-ictal phase and the ictal phase and between the ictal phase and the post-ictal phase in both hemispheres, in the $\delta$, $\theta$, and $\alpha$ bands.

\bigskip

All these results suggest that there is a change in the clique topology of the networks of both hemispheres, both from the pre-ictal phase to the ictal phase and from the ictal phase to the post-ictal phase, especially in the $\theta$ band. Moreover, we highlight that no statistically significant difference was observed between the left hemispheres and the right hemisphere in the ictal phase for any distance, at any frequency.

\section{Conclusions}

Despite all the efforts of the scientific community, our understanding of the dynamic processes underlying epilepsy still has many gaps. In this study we constructed G-connectivity networks from EEG data recorded before, during and after epileptic seizures in patients diagnosed with left temporal lobe epilepsy, obtaining therefore an estimate of the directionality and dynamics of the information flow between different brain regions in various seizure phases; subsequently, we applied novel simplicial characterization measures on the $q$-digraphs constructed from these networks, in addition to their usual graph counterparts in the original networks, to investigate: 1) How do the topological and functional properties of G-connectivity networks and their respective $q$-digraphs change during the seizure in each hemisphere and in each frequency band, both at the node level and at the various clique topology levels? 2) Does the analysis of higher-order structures of brain connectivity networks provide novel and better biomarkers for seizure dynamics and also for the laterality of the seizure focus than the usual theoretical graph analyses? 

We observed that all simpicial characterization measures showed statistically significant increases in their magnitudes from the pre-ictal phase to the ictal phase, for various levels $q$, for both hemispheres, especially in the $\delta$ and $\theta$ bands but no statistically significant changes were observed from the ictal phase to the post-ictal phase, which may suggest that several topological and functional aspects of the brain networks change from the pre-ictal phase to the ictal phase, at various higher-order levels of topological organization (clique organization), especially in the $\delta$ and $\theta$ bands. However, most of the usual graph measures did not detect significant differences in the left hemisphere, which may suggest that changes in the network topology at the node level ($q=-1$) do not undergo as many changes as in the higher-order network topology. Furthermore, most of the simplicial measures are strongly positively correlated across the seizure phases (at all levels $q$) and across the levels $q$ (in all seizure phases), for all frequency bands and for both hemispheres. Also, most of the simplicial distances revealed changes in the clique topology of the networks of both hemispheres, from the pre-ictal phase to the ictal phase, especially in the $\theta$ band, which reinforces the findings obtained by the simplicial characterization measures. Regarding the laterality of the seizure focus, the analysis through simplicial distances did not find any statistically significant difference between the left and right hemispheres' clique topology in the ictal phase.

In conclusion, despite our several limitations, such as the limited number of patients (eight patients) and all the limitations associated with each method used in the study, we emphasize that we found evidence that the analysis of higher-order structures represented by $q$-digraphs obtained from G-connectivity networks may be a reliable way to find biomarkers associated with epileptic networks, but its establishment as a viable rigorous method will depend on future work, as well as its applicability to other disorders of brain connectivity networks.

\chapter{Final Considerations}
\label{chap:chap8}

\epigraph{Wir müssen wissen. Wir werden wissen. (We must know. We will know.)}{--- Epitaph on the gravestone of David Hilbert}

\bigskip


The initial motivation for writing this thesis was the development of new methods to analyze the topology of directed graphs obtained from brain connectivity estimators based on the concept of Granger causality, especially the iPDC, to contribute to the general investigation of the dynamics of brain activity. As research progressed, we noticed the increasing application and usefulness of methods derived from computational topology in network analysis, mainly based on clique complexes. Thus, we decided to transpose several of these methods and concepts to directed graphs, by considering their directed clique complexes and path complexes, and two main objectives were outlined:

\begin{enumerate}
 \setlength\itemsep{0.1em}

\item To develop rigorously a new quantitative theory for digraph-based complexes (or, as we can consider it, a step towards the formalization of a ``quantitative simplicial theory"), with special emphasis on directed higher-order connectivity between directed cliques;

\item To apply the methods of the new theory to epileptic brain networks obtained through iPDC to quantitatively investigate their higher-order topologies and search for new biomarkers based on their directed higher-order connectivities, thus pointing out potential applications of the theory in network neuroscience.
  
\end{enumerate}

To accomplish these objectives, we divided the thesis into two parts: Part \ref{partI}, which addresses objective 1, and Part \ref{partII}, which addresses objective 2. The developments in these parts are summarized in the following paragraphs.


In Part \ref{partI}, having considered the multidisciplinary nature of this thesis and to make it self-contained, we first presented, in Chapter \ref{chap:chap2}, all the fundamental and necessary concepts of graph theory, starting with a discussion of binary and equivalence relations, followed by the introduction of concepts associated with graphs and digraphs, passing through algebraic and spectral graph theory, graph measures, graph similarities, and at last a brief discussion on random graphs. Subsequently, in Chapter \ref{chap:chap4}, we presented the basic theory of simplicial complexes and directed clique complexes associated with digraphs, including the case of weighted digraphs, passing through simplicial homology, persistent homology, and combinatorial Laplacians associated with these complexes, followed by the presentation of paths complexes and their homologies, and, at last, we introduced a novel theory associated with directed higher-order connectivity between directed cliques, which provided the conception of new concepts such as directed higher-order adjacencies (upper and lower adjacencies) and maximal/lower $q$-digraphs, and new concepts for directed Q-Analysis. Finally, in Chapter \ref{chap:chap5}, we introduced new characterization measures for maximal/lower $q$-digraphs, adapted from graph measures, such as distance-based measures, segregation measures, centrality measures, entropy measures, and spectrum-related measures, and similarity comparison methods for directed clique complexes and path complexes, based on graph similarity comparison methods, such as structure distances and graph kernels, and, at last, we presented some examples with random digraph models.


In Part \ref{partII}, we started by presenting, in Chapter \ref{chap:chap3}, the theory of brain connectivity networks, discussing briefly the biophysical principles of brain signals and the techniques for acquiring such signals, with special attention to EEG, followed by a discussion of connectivity estimators, with an emphasis on PDC and its variants, especially gPDC and iPDC, and, finally, we discussed the different types of brain connectivity networks, especially structural, functional, and effective networks, and also about the applications of graph theory and computational (algebraic) topology in the analysis of these networks to investigate the dynamics of brain activity in different contexts. In Chapter \ref{chap:epilepsy}, we studied the neuropathology of epilepsy in more detail, discussing its main characteristics, etiologies, epidemics, diagnoses, and treatments, and also discussed several studies pointing out alterations in the network properties of epileptic brain networks when compared to brain networks obtained from healthy individuals, as well as changes during the ictal period compared to other periods. Finally, in Chapter \ref{chap:chap6}, we performed an analysis of epileptic brain networks, estimated through iPDC from EEG data from patients diagnosed with left temporal lobe epilepsy, using the novel quantitative methods for directed clique complexes developed in previous chapters, to investigate how certain properties of these networks, and their corresponding higher-order structures ($q$-digraphs), alter according to the phases of seizures, and also according to the cerebral hemispheres, in different frequency bands.

Furthermore, we developed the Python package \texttt{DigplexQ} (see Appendix \ref{appendix:software}), which contains the implementation of algorithms for calculating directed clique complexes and path complexes of some given digraph from its adjacency matrix, and also contains the implementation of various simplicial characterization measures and simplicial similarity comparison methods introduced in Chapter \ref{chap:chap5}.

Due to the limited nature of this thesis, many of the simplicial characterization measures (including the maximal variants) and simplicial similarity comparison methods have not been applied to real-world data, and a comparison between the weighted and unweighted versions of these measures and methods has also been lacking. Moreover, the applicability of concepts developed to path complexes has not been explored. In addition to these gaps, there are several other opportunities to be explored in future work regarding the theory involving directed higher-order connectivity and directed Q-Analysis, not only in relation to its mathematical foundations but also in relation to its applicability in several areas, besides network neuroscience, such as biology, social sciences, transportation planning, etc.

\newpage

\addcontentsline{toc}{chapter}{Bibliography}
\bibliography{bibliography}

\appendix

\chapter{Software Review}
\label{appendix:software}

In this appendix, we present an overview of the software used in the examples of Chapter \ref{chap:chap5} and in the analysis of Chapter \ref{chap:chap6} and also a brief description of the \texttt{DigplexQ} package, a novel Python package that was developed during this thesis.


\section{Software Review}

There are several free/open-source software available for brain connectivity estimation, graph theoretical analysis, topological data analysis, and statistical analysis, written in several different programming languages. In what follows, we present a brief description of the software used in this thesis.

\begin{itemize}
\item \textbf{asympPDC:} asympPDC is a MATLAB/Octave package for the analysis of time series data via partial directed coherence (PDC) and/or directed transfer function (DTF) (and their variants - information/generalized PDC/DTF) to infer directed interactions in the frequency domain between structures.\\
Available at \url{https://github.com/asymppdc}

\item \textbf{DigplexQ:}
DigplexQ is a Python package to perform computations with digraph-based complexes (e.g. directed flag complexes and path complexes).\\
Available at \url{https://github.com/heitorbaldo/DigplexQ}

\item \textbf{Giotto-TDA:}
Giotto-TDA is a Python package built on top of scikit-learn for topological data analysis. It provides functions to compute Betti numbers, persistence diagrams, barcodes, Betti curves, bottleneck distance, etc., from simplicial complexes or graphs/digraphs given as input data. \\
Available at \url{https://giotto-ai.github.io/gtda-docs/0.5.1/library.html}

\item \textbf{Flagser:} 
Flagser is a C++ library built on top of Ripser to compute homologies of directed flag complexes. It is also implemented within the Giotto-TDA package.\\
Available at \url{https://github.com/luetge/flagser}

\item \textbf{HodgeLaplacians:}
HodgeLaplacians is a Python package created to compute the Hodge Laplacian matrices from simplicial complexes given as input data. \\
Available at \url{https://github.com/tsitsvero/hodgelaplacians}

\item \textbf{NetworkX:} NetworkX is a Python package to perform quantitative analysis in undirected and directed graphs. It provides several graph measures, graph distance measures, random models, visualization functions, etc.\\
Available at \url{https://networkx.org}.

\item \textbf{Persim:}
Persim is a Python package for topological data analysis. It provides functions to compute persistence diagrams, persistence landscapes, bottleneck distance, heat kernel, etc., from simplicial complexes given as input data. \\
Available at \url{https://persim.scikit-tda.org/en/latest/}

\item \textbf{Pingouin:}
Pingouin is a Python package based on NumPy and Pandas to perform statistical analysis. It provides standard parametric and non-parametric statistical tests, such as t-test, ANOVAs, Kruskal-Wallis test, Mann-Whitney test, Wilcoxon signed-rank test, etc.\\
Available at \url{https://pingouin-stats.org/build/html/index.html}

\end{itemize}

\section{DigplexQ}

As part of this thesis, the Python package \texttt{DigplexQ} (released as free software under the MIT license\footnote{\url{https://opensource.org/license/mit}}) was developed to perform computations with digraph-based complexes. It is based on other well-known Python packages, such as \texttt{NetworkX} (graph measures), \texttt{Giotto-TDA/Flagser} (persintence diagrams, Betti numbers, and topological distances), \texttt{Persim} (topological distances), and \texttt{HodgeLaplacians} (Hodge Laplacians). 

\texttt{DigplexQ} is an ``adjacency matrix-centered" package since it was designed so that the user can perform all computations just by entering an adjacency matrix as input. It contains the implementation of almost all simplicial characterizations measures and simplicial similarity comparision distances introduced in Chapter \ref{chap:chap5}.

The \texttt{DigplexQ} package has been tested under Windows and Linux-Ubuntu platforms running Python version 3.7 and higher. It is available in the PyPi repository at \url{https://pypi.org/project/digplexq} (current version 0.0.7), and it can be installed through the PIP package manager via the command \texttt{pip3 install digplexq}.

\chapter{Supplementary Information}

\section{Summary of the Simplicial Measures}
\label{appendix:measures-analysis}

\bigskip

\begin{table}[h!]
\centering
\small
\caption{Summary of the directed simplicial measures used in the analysis and/or in the examples (lower variants were used). For simplicity's sake, we omitted the terms ``directed simplicial" from the nomenclature of the measures.}
\begin{tabular}{c l c c}
\toprule 
\textbf{ id } & \textbf{ Measure } & \textbf{Notation} & \textbf{Equation}  \\
\midrule 

1 & Average Shortest $q$-Walk Length & $\vec{L}_{q}(\mathcal{G}_{q})$ & [\ref{eq:average-shortest-simp-q-walk}] \\

2 & Global $q$-Efficiency &  $\vec{E}^{q}_{glob}(\mathcal{G}_{q})$  & [\ref{eq:simp-global-efficiency}] \\


3 & $q$-Returnability &  $K_{r,q}(\mathcal{G}_{q})$  & [\ref{eq:simp-returnability}] \\

4 & In-$q$-Degree Centrality   &  $C_{deg_{q}}^{-}(\sigma)$ &  [\ref{eq:simp-in-degree}] \\

5 & Out-$q$-Degree Centrality & $C_{deg_{q}}^{+}(\sigma)$ & [\ref{eq:simp-out-degree}]  \\

6 & $q$-Harmonic Centrality &  $\vec{HC}_{q}(\sigma)$ &  [\ref{eq:simp-harmonic}]  \\

7 & $q$-Betweenness Centrality & $\vec{B}_{q}(\sigma)$ &  [\ref{eq:simp-betweenness}] \\

8 & Global $q$-Reaching Centrality &  $GRC_{q}(\mathcal{G}_{q})$  & [\ref{eq:simp-global-reaching-centrality}] \\

9 & Average $q$-Clustering Coefficient &  $\vec{\bar{C}}_{q}(\mathcal{G}_{q})$  & [\ref{eq:simp-average-cc}] \\

10 & in-$q$-Rich-Club Coefficient &  $\phi_{q, k}^{-}(\mathcal{G}_{q})$  & [\ref{eq:simp-in-rich-club}]  \\

11 & out-$q$-Rich-Club Coefficient &  $\phi_{q, k}^{+}(\mathcal{G}_{q})$  & [\ref{eq:simp-out-rich-club}]  \\

12 & $q$-Structural Entropy &  $H^{str}_{q}(\mathcal{G}_{q})$  & [\ref{eq:simp-struc-entropy}] \\

13 & in-$q$-Degree Distribution Entropy & $H^{-}_{q}(\mathcal{G}_{q})$ & [\ref{eq:in-dist-entropy}] \\

14 & out-$q$-Degree Distribution Entropy & $H^{+}_{q}(\mathcal{G}_{q})$ & [\ref{eq:out-dist-entropy}]  \\

15 & in-$q$-Forman-Ricci Curvature & $F_{q}^{-}(\sigma)$ &  [\ref{eq:in-forman-ricci}] \\

16 & out-$q$-Forman-Ricci Curvature & $F_{q}^{+}(\sigma)$ &  [\ref{eq:out-forman-ricci}] \\

17 & $q$-Energy & $\varepsilon_{q}(\mathcal{G}_{q})$ & [\ref{eq:simp-graph-energy}] \\

18 & in-$q$-Katz Centrality &  $\vec{K}^{-}_{q}(\sigma)$ &  [\ref{eq:simp-dir-katz-centrality-in}] \\

\bottomrule 
\end{tabular}
\label{tab:measures-analysis}
\end{table}

\newpage
\section{Summary of the Simplicial Distances}
\label{appendix:distances-analysis}

\bigskip

\begin{table}[h!]
\centering
\small
\caption{Summary of the simplicial and persistence distances used in the analysis and/or in the examples.}
\begin{tabular}{c l c c}
\toprule 
\textbf{ id } & \textbf{ Distance } & \textbf{Notation} & \textbf{Equation} \\
\midrule 

1 & Bottleneck Distance & $d_{W_{\infty}}(P, Q)$ & [\ref{eq:bottleneck-dist}] \\

2 & $p$-Wasserstein Distance & $d_{W_{p}}(P, Q)$ & [\ref{eq:wasserstein-dist}] \\

3 & Betti Distance & $d_{\mathcal{B}}(\mathcal{B}_{P}, \mathcal{B}_{Q})$ & [\ref{eq:betti-curve-dist}] \\

4 & First Topological Structure Distance & $\widehat{T}_{tsd}^{1}(\mathcal{X}_{1}, \mathcal{X}_{2})$ &  [\ref{eq:topological-distance}]\\

5 & Fourth Topological Structure Distance & $\widehat{T}_{tsd}^{4}(\mathcal{X}_{1}, \mathcal{X}_{2})$ &  [\ref{eq:topological-distance}]\\

6 & Fifth Topological Structure Distance & $\widehat{T}_{tsd}^{5}(\mathcal{X}_{1}, \mathcal{X}_{2})$ &  [\ref{eq:topological-distance}]\\

7 & Histogram Cosine Kernel & $K_{HC}(\mathcal{X}_{1}, \mathcal{X}_{2})$ &  [\ref{eq:HCK}] \\

8 & Jaccard Kernel & $K_{J}(\mathcal{X}_{1}, \mathcal{X}_{2})$ &  [\ref{eq:jaccard-kernel}] \\

9 & Simplicial $n$-Spectral Distance & $\widehat{D}_{\mathcal{L}_{n}}(\mathcal{X}_{1}, \mathcal{X}_{2})$ &  [\ref{eq:simp-spectral-distance}] \\

\bottomrule 
\end{tabular}
\label{tab:distances-analysis}
\end{table}

\chapter{Results}
\label{appendix:results}

\begin{table}[h!]
  \center
  \footnotesize
    \caption{ {\footnotesize  W-statistics and p-values (in parentheses) for simplicial measures showing significant differences ($p < 0.05$, Wilcoxon paired test) between the pre-ictal phase and ictal phase for the original iPDC networks of the left and right hemispheres in the delta and theta bands. See Table \ref{tab:measures-analysis} for the measure ids.}}
    \label{tab:w-test-meas-std}
    \begin{tabular}{c c c c c}
     \toprule 
      \multicolumn{1}{c}{} & \multicolumn{2}{c}{\textbf{Left}} & \multicolumn{2}{c}{\textbf{Right}} \\
\cmidrule(rl){2-3} \cmidrule(rl){4-5} \textbf{Measure} & \textbf{Delta} & \textbf{Theta} & \textbf{Delta} & \textbf{Theta} \\
      \midrule 
2 & 3.0 (0.039) & --- & 2.0 (0.023) & 0.0 (0.008) \\
4 & ---  & ---  & --- & ---  \\
5 & ---  & ---  & ---  & ---  \\
6 & --- & ---  & 1.0 (0.016) & 3.0 (0.039) \\
8 & --- & --- &  1.0 (0.016) & 3.0 (0.039) \\
9 & 2.0 (0.023) & --- & --- & 0.0 (0.008) \\
13 & --- & ---  & 2.0 (0.023) & 0.0 (0.008) \\
14 & --- & ---  & 3.0 (0.039) & 1.0 (0.016) \\
17 & 3.0 (0.039) & --- & 2.0 (0.023) & 0.0 (0.008) \\

      \bottomrule 
    \end{tabular}

\end{table}

\begin{table}[h!]
    \center
  \footnotesize
    \caption{ {\footnotesize  W-statistics and p-values (in parentheses) for simplicial measures showing significant differences ($p < 0.05$, Wilcoxon paired test) between the pre-ictal phase and ictal phase for the $0$-digraphs of the left and right hemispheres in the delta and theta bands. See Table \ref{tab:measures-analysis} for the measure ids.}}
    \label{tab:w-test-meas-q0}
    \begin{tabular}{c c c c c}
     \toprule 
      \multicolumn{1}{c}{} & \multicolumn{2}{c}{\textbf{Left}} & \multicolumn{2}{c}{\textbf{Right}} \\
\cmidrule(rl){2-3} \cmidrule(rl){4-5} \textbf{Measure} & \textbf{Delta} & \textbf{Theta} & \textbf{Delta} & \textbf{Theta} \\
      \midrule 
2 & 2.0 (0.023) & --- & 3.0 (0.039) & --- \\
4 & 2.0 (0.023) & --- & --- & --- \\
5 & 2.0 (0.023) & --- & --- & --- \\
6 & 2.0 (0.023) & --- & 2.0 (0.023) & --- \\
8 & 2.0 (0.023) & --- & 1.0 (0.016) & 1.0 (0.016) \\
9 & 2.0 (0.023) & --- & --- & --- \\
13 & --- & --- & 3.0 (0.039) & --- \\
14 & --- & --- & 3.0 (0.039) & --- \\
17 & 2.0 (0.023) & --- & 2.0 (0.023) & --- \\

      \bottomrule 
    \end{tabular}
 
\end{table}


\begin{table}[h!]
  \center
  \footnotesize
    \caption{ {\footnotesize W-statistics and p-values (in parentheses) for simplicial measures showing significant differences ($p < 0.05$, Wilcoxon paired test) between the pre-ictal phase and ictal phase for the $1$-digraphs of the left and right hemispheres in the delta and theta bands. See Table \ref{tab:measures-analysis} for the measure ids.}}
    \label{tab:w-test-meas-q1}
    \begin{tabular}{c c c c c}
     \toprule 
      \multicolumn{1}{c}{} & \multicolumn{2}{c}{\textbf{Left}} & \multicolumn{2}{c}{\textbf{Right}} \\
\cmidrule(rl){2-3} \cmidrule(rl){4-5} \textbf{Measure} & \textbf{Delta} & \textbf{Theta} & \textbf{Delta} & \textbf{Theta} \\
      \midrule 
2 & 2.0 (0.023) & --- & 3.0 (0.039) & --- \\
4 & 3.0 (0.039) & --- & 3.0 (0.039) & 3.0 (0.039) \\
5 & 2.0 (0.023) & --- & 0.0 (0.022) &  3.0 (0.039) \\
6 & 2.0 (0.023) & --- & --- & --- \\
8 & 2.0 (0.023) & --- & --- & 0.0 (0.008) \\
9 & 1.0 (0.016) & 2.0 (0.023) & 0.0 (0.008) & 1.0 (0.016) \\
13 & 2.0 (0.023) & --- & --- & 3.0 (0.039) \\
14 & 2.0 (0.023) & --- & --- & 2.0 (0.023) \\
17 & 2.0 (0.023) & --- & --- & --- \\

      \bottomrule 
    \end{tabular}
\end{table}

\begin{table}[h!]
   \center
  \footnotesize
    \caption{ {\footnotesize W-statistics and p-values (in parentheses) for simplicial measures showing significant differences ($p < 0.05$, Wilcoxon paired test) between the pre-ictal phase and ictal phase for the $2$-digraphs of the left and right hemispheres in the delta and theta bands. See Table \ref{tab:measures-analysis} for the measure ids.}}
    \label{tab:w-test-meas-q2}
    \begin{tabular}{c c c c c}
     \toprule 
      \multicolumn{1}{c}{} & \multicolumn{2}{c}{\textbf{Left}} & \multicolumn{2}{c}{\textbf{Right}} \\
\cmidrule(rl){2-3} \cmidrule(rl){4-5} \textbf{Measure} & \textbf{Delta} & \textbf{Theta} & \textbf{Delta} & \textbf{Theta} \\
      \midrule 
2 & 2.0 (0.023) & 2.0 (0.023) & 0.0 (0.008) & --- \\
4 & 1.0 (0.016) & 2.0 (0.023) & 0.0 (0.008) & 1.0 (0.016) \\
5 & 1.0 (0.016) & 2.0 (0.023) & 0.0 (0.008) & 0.0 (0.008) \\
6 & 1.0 (0.016) & 2.0 (0.023) & 0.0 (0.008) & 3.0 (0.039) \\
8 & 1.0 (0.016) & 3.0 (0.039) & 0.0 (0.008) & 1.0 (0.016) \\
9 & 1.0 (0.035) & 1.0 (0.016) & 3.0 (0.039) & 1.0 (0.035) \\
13 & 1.0 (0.016) & 3.0 (0.039) & 0.0 (0.008) & 3.0 (0.039) \\
14 & 1.0 (0.016) & 3.0 (0.039) & 0.0 (0.008) & 2.0 (0.023) \\
17 & 1.0 (0.016) & 2.0 (0.023) & 0.0 (0.008) & --- \\

      \bottomrule 
    \end{tabular}
 
\end{table}

\begin{table}[h!]
  \center
  \footnotesize
    \caption{ {\footnotesize W-statistics and p-values (in parentheses) for simplicial measures showing significant differences ($p < 0.05$, Wilcoxon paired test) between the pre-ictal phase and ictal phase for the $3$-digraphs of the left and right hemispheres in the delta and theta bands. See Table \ref{tab:measures-analysis} for the measure ids.}}
    \label{tab:w-test-meas-q3}
    \begin{tabular}{c c c c c}
     \toprule 
      \multicolumn{1}{c}{} & \multicolumn{2}{c}{\textbf{Left}} & \multicolumn{2}{c}{\textbf{Right}} \\
\cmidrule(rl){2-3} \cmidrule(rl){4-5} \textbf{Measure} & \textbf{Delta} & \textbf{Theta} & \textbf{Delta} & \textbf{Theta} \\
      \midrule 
2 & 0.0 (0.022) & 3.0 (0.039) & --- & --- \\
4 & 0.0 (0.022) & 3.0 (0.039) & --- & 1.0 (0.035) \\
5 & 0.0 (0.022) & 3.0 (0.039) & --- & 1.0 (0.035) \\
6 & 0.0 (0.022) & 3.0 (0.039) & --- & --- \\
8 & 0.0 (0.022) & --- & 2.0 (0.023) & --- \\
9 & --- & 1.0 (0.035) & --- & 1.0 (0.035) \\
13 & 0.0 (0.022) & --- & 2.0 (0.023) & 1.0 (0.035) \\
14 & 0.0 (0.022) & --- & 2.0 (0.023) & 1.0 (0.035) \\
17 & 0.0 (0.022) & 3.0 (0.039) & 3.0 (0.039) & 1.0 (0.035) \\

      \bottomrule 
    \end{tabular}
\end{table}



\begin{table}[h!]
  \center
  \footnotesize
    \caption{ {\footnotesize W-statistics and p-values (in parentheses) for simplicial distances showing significant differences ($p < 0.05$, Wilcoxon paired test) between the means of the distributions $d(G^{L}_{ic}, G^{L}_{pre})$ and $d(G^{L}_{pre}, G^{L}_{pre})$, for each frequency band. See Table \ref{tab:distances-analysis} for the distance ids.}}
    \label{tab:w-test-dist2}
    \begin{tabular}{c c c c c c c c c }
     \toprule 
     \textbf{Distance} & \textbf{Delta} & \textbf{Theta} & \textbf{Alpha} \\
      \midrule 
1 & --- & ---  &  ---\\
2 & --- & 0.0 (0.008) & ---\\
3 & --- & 2.0 (0.023) & ---\\
4 & --- & 3.0 (0.039) & ---\\
5 & --- & --- & ---\\
6 & --- & --- & ---\\
7 & --- & ---  & ---\\
8 & 0.0 (0.008) & 0.0 (0.008) & 0.0 (0.008)\\
      \bottomrule 
    \end{tabular}
  
\end{table}

\begin{table}[h!]

  \center
  \footnotesize
    \caption{ {\footnotesize W-statistics and p-values (in parentheses) for simplicial distances showing significant differences ($p < 0.05$, Wilcoxon paired test) between the means of the distributions $d(G^{L}_{ic}, G^{L}_{pos})$ and $d(G^{L}_{pos}, G^{L}_{pos})$, for each frequency band. See Table \ref{tab:distances-analysis} for the distance ids.}}
    \label{tab:w-test-dist3}
    \begin{tabular}{c c c c c c c c c }
     \toprule 
     \textbf{Distance} & \textbf{Delta} & \textbf{Theta} & \textbf{Alpha} \\
      \midrule 
1 & --- &---  & ---\\
2 &--- & 1.0 (0.016) & ---\\
3 &--- & --- &--- \\
4 &--- & --- & ---\\
5 &--- & --- & ---\\
6 &--- & 2.0 (0.023) & ---\\
7 & ---& --- & ---\\
8 & 0.0 (0.008) & 0.0 (0.008) & 0.0 (0.008)\\
      \bottomrule 
    \end{tabular}
 
\end{table}

\begin{table}[h!]

  \center
  \footnotesize
    \caption{ {\footnotesize W-statistics and p-values (in parentheses) for simplicial distances showing significant differences ($p < 0.05$, Wilcoxon paired test) between the means of the distributions $d(G^{R}_{ic}, G^{R}_{pre})$ and $d(G^{R}_{pre}, G^{R}_{pre})$, for each frequency band. See Table \ref{tab:distances-analysis} for the distance ids.}}
    \label{tab:w-test-dist4}
    \begin{tabular}{c c c c c c c c c }
     \toprule 
     \textbf{Distance} & \textbf{Delta} & \textbf{Theta} & \textbf{Alpha} \\
      \midrule 
1 & --- & 3.0 (0.039) & ---\\
2 & 2.0 (0.023) & 0.0 (0.008) & ---\\
3 &  & 2.0 (0.023 & ---\\
4 & --- & --- & ---\\
5 & --- & --- & ---\\
6 & --- & --- & ---\\
7 & --- & --- & ---\\
8 & 0.0 (0.008) & 0.0 (0.008) & 0.0 (0.008)\\
      \bottomrule 
    \end{tabular}
  
\end{table}

\begin{table}[h!]

  \center
  \footnotesize
    \caption{ {\footnotesize W-statistics and p-values (in parentheses) for simplicial distances showing significant differences ($p < 0.05$, Wilcoxon paired test) between the distributions $d(G^{R}_{ic}, G^{R}_{pos})$ and $d(G^{R}_{pos}, G^{R}_{pos})$, for each frequency band. See Table \ref{tab:distances-analysis} for the distance ids.} }
    \label{tab:w-test-dist5}
    \begin{tabular}{c c c c c c c c c }
     \toprule 
     \textbf{Distance} & \textbf{Delta} & \textbf{Theta} & \textbf{Alpha} \\
      \midrule 
1 & --- & --- & ---\\
2 & --- & --- & 2.0 (0.023) \\
3 & --- & --- & ---\\
4 & 0.0 (0.008)  & 1.0 (0.016) & 1.0 (0.016)\\
5 & --- & --- & ---\\
6 & 2.0 (0.023) & --- & --- \\
7 & 2.0 (0.023) & 0.0 (0.008) & 1.0 (0.016) \\
8 & 0.0 (0.008) & 0.0 (0.008) & 0.0 (0.008)\\
      \bottomrule 
    \end{tabular}
  
\end{table}


\begin{figure}[h!]
    \centering
\includegraphics[scale=1.54]{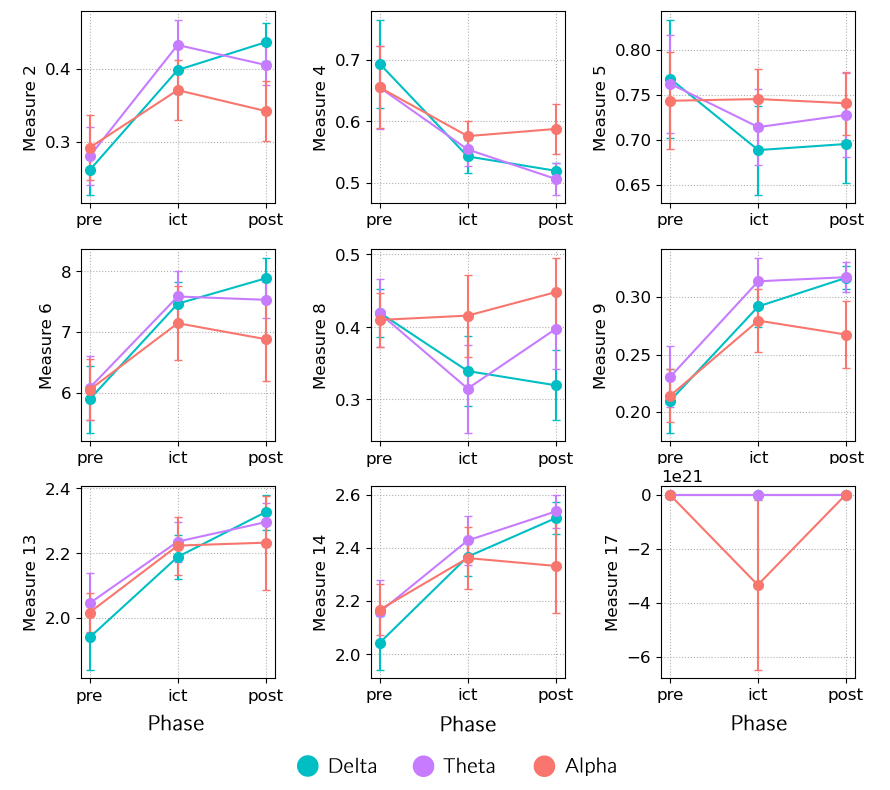}
 \caption{ {\small Grand means and standard deviations of the simplicial measures computed in the original iPDC networks of the left hemisphere for each seizure phase (pre-ictal, ictal, post-ictal) and for each frequency band (delta, theta, alpha).  See Table \ref{tab:measures-analysis} for the measure ids.}}
 \label{fig:grand-mean-std_L}

\medskip

\includegraphics[scale=1.54]{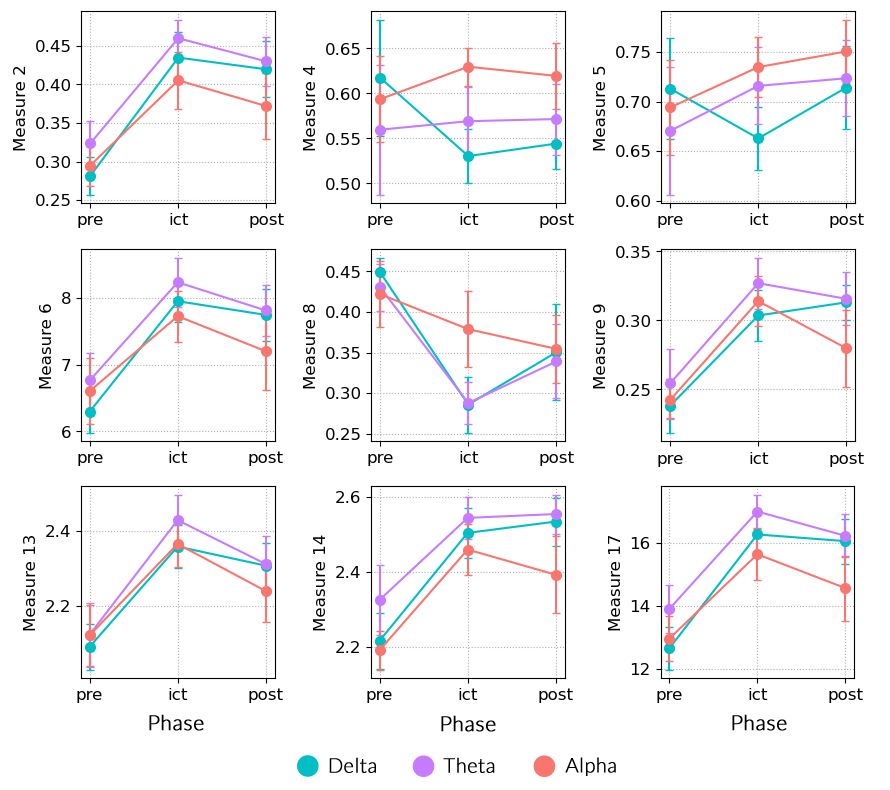}
 \caption{ {\small Grand means and standard deviations of the simplicial measures computed in the original iPDC networks of the right hemisphere for each seizure phase (pre-ictal, ictal, post-ictal) and for each frequency band (delta, theta, alpha).  See Table \ref{tab:measures-analysis} for the measure ids.}}
 \label{fig:grand-mean-std_R}
\end{figure}


\begin{figure}[h!]
    \centering
\includegraphics[scale=1.54]{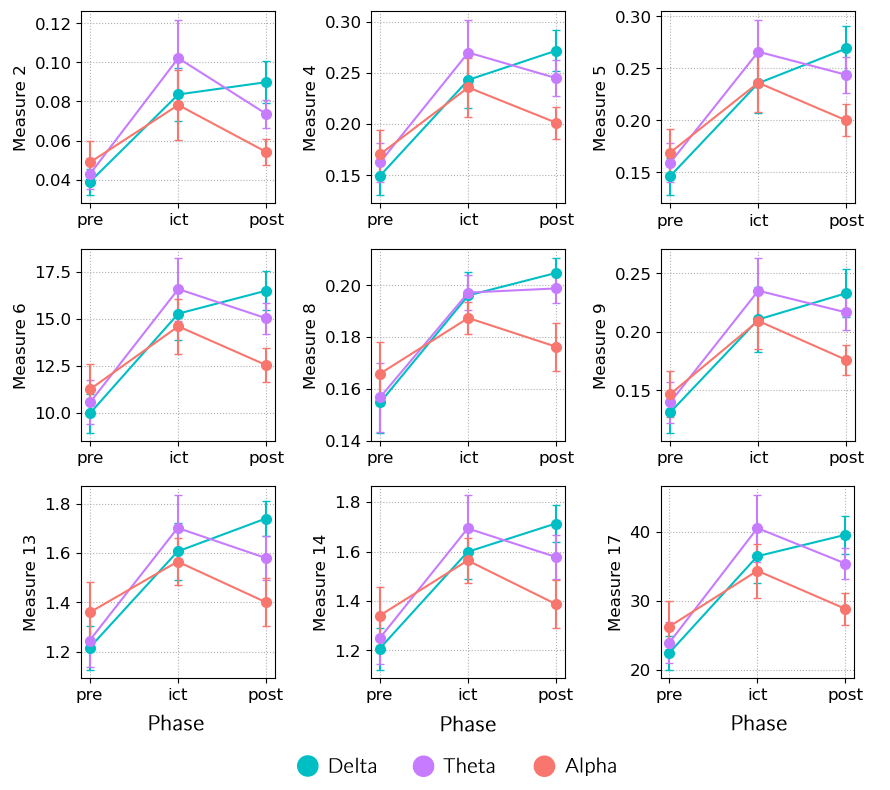}
 \caption{ {\small Grand means and standard deviations of the simplicial measures computed in the $0$-digraphs of the left hemisphere for each seizure phase (pre-ictal, ictal, post-ictal) and for each frequency band (delta, theta, alpha). See Table \ref{tab:measures-analysis} for the measure ids.}}
 \label{fig:grand-mean-q0_L}

\medskip

\includegraphics[scale=1.54]{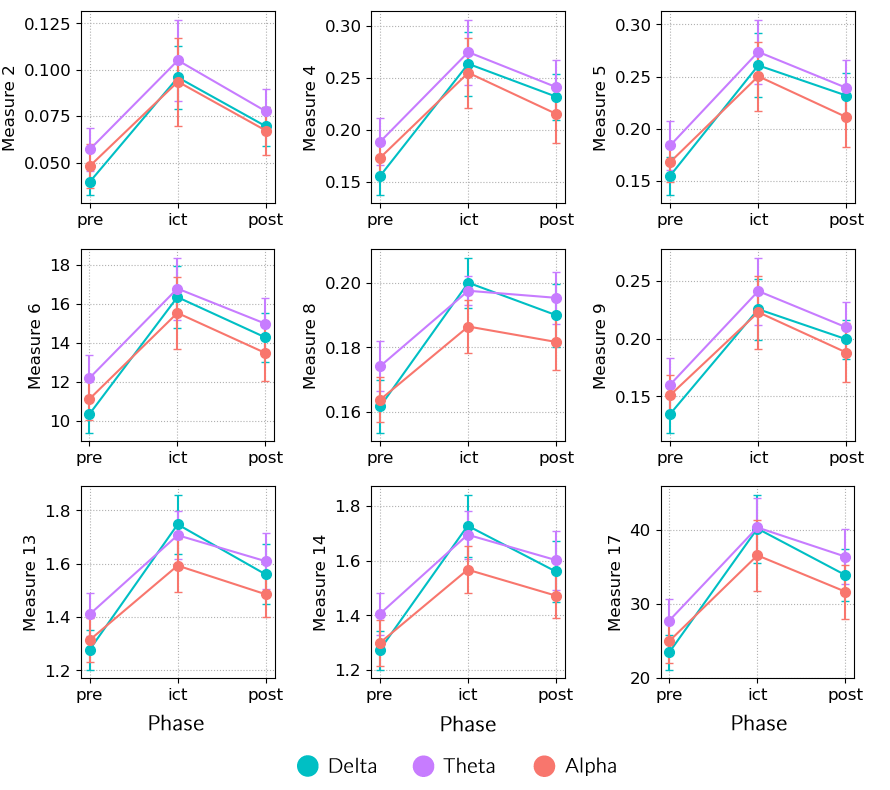}
 \caption{ {\small Grand means and standard deviations of the simplicial measures computed in the $0$-digraphs of the right hemisphere for each seizure phase (pre-ictal, ictal, post-ictal) and for each frequency band (delta, theta, alpha). See Table \ref{tab:measures-analysis} for the measure ids.}}
 \label{fig:grand-mean-q0_R}
\end{figure}


\begin{figure}[h!]
    \centering
\includegraphics[scale=1.54]{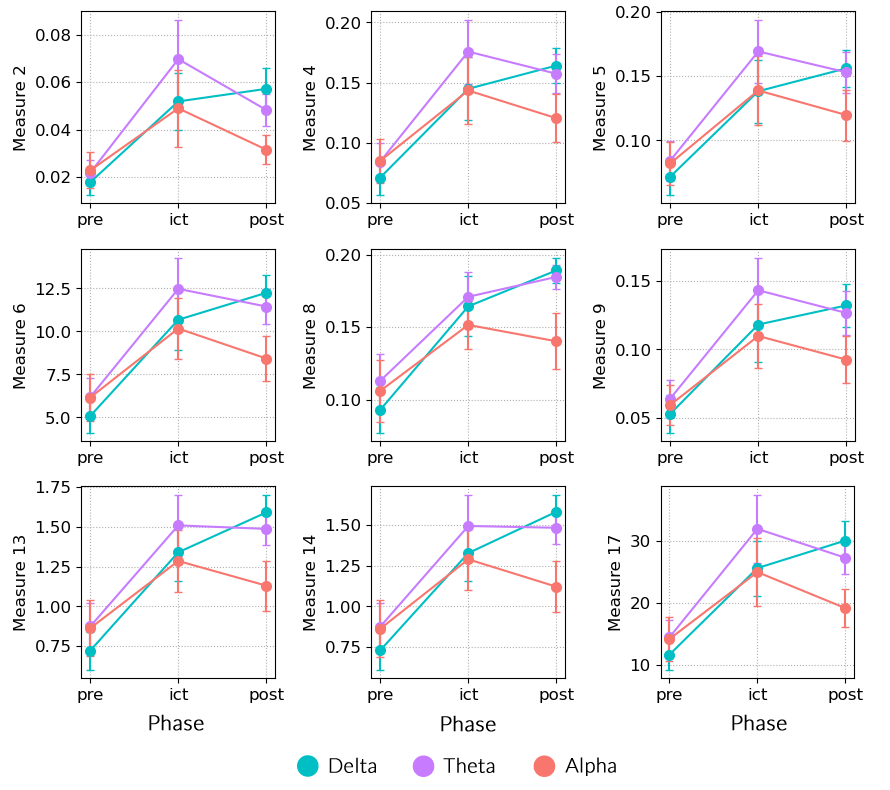}
\caption{ {\small Grand means and standard deviations of the simplicial measures computed in the $1$-digraphs of the left hemisphere for each seizure phase (pre-ictal, ictal, post-ictal) and for each frequency band (delta, theta, alpha). See Table \ref{tab:measures-analysis} for the measure ids.}}
 \label{fig:grand-mean-q1_L}

\medskip

\includegraphics[scale=1.54]{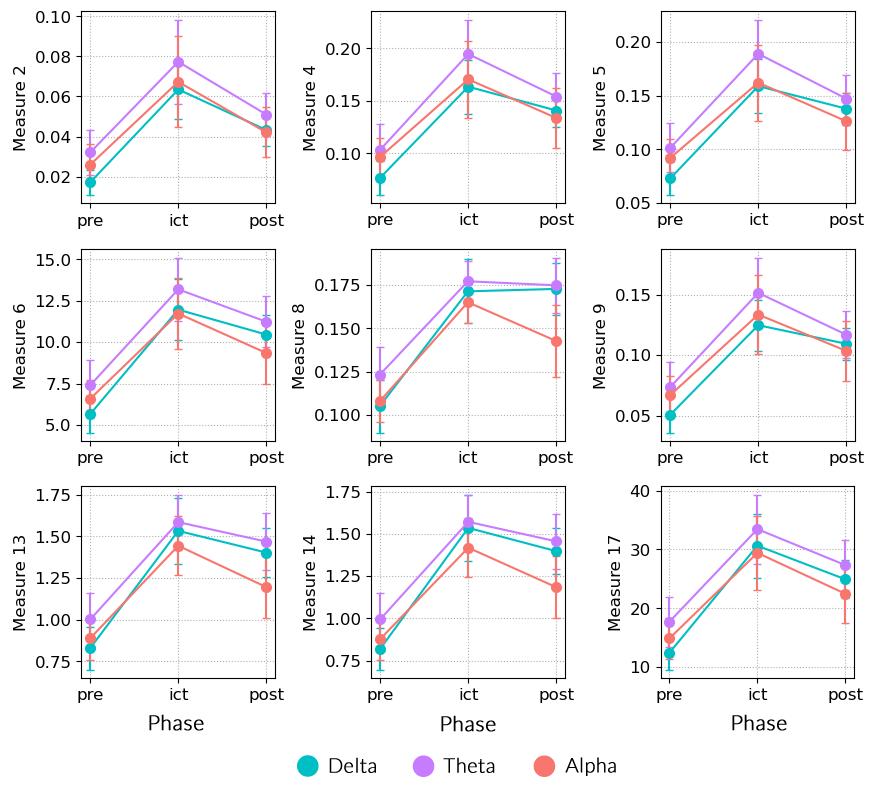}
 \caption{ {\small Grand means and standard deviations of the simplicial measures computed in the $1$-digraphs of the right hemisphere for each seizure phase (pre-ictal, ictal, post-ictal) and for each frequency band (delta, theta, alpha). See Table \ref{tab:measures-analysis} for the measure ids.}}
 \label{fig:grand-mean-q1_R}
\end{figure}


\begin{figure}[h!]
    \centering
\includegraphics[scale=1.54]{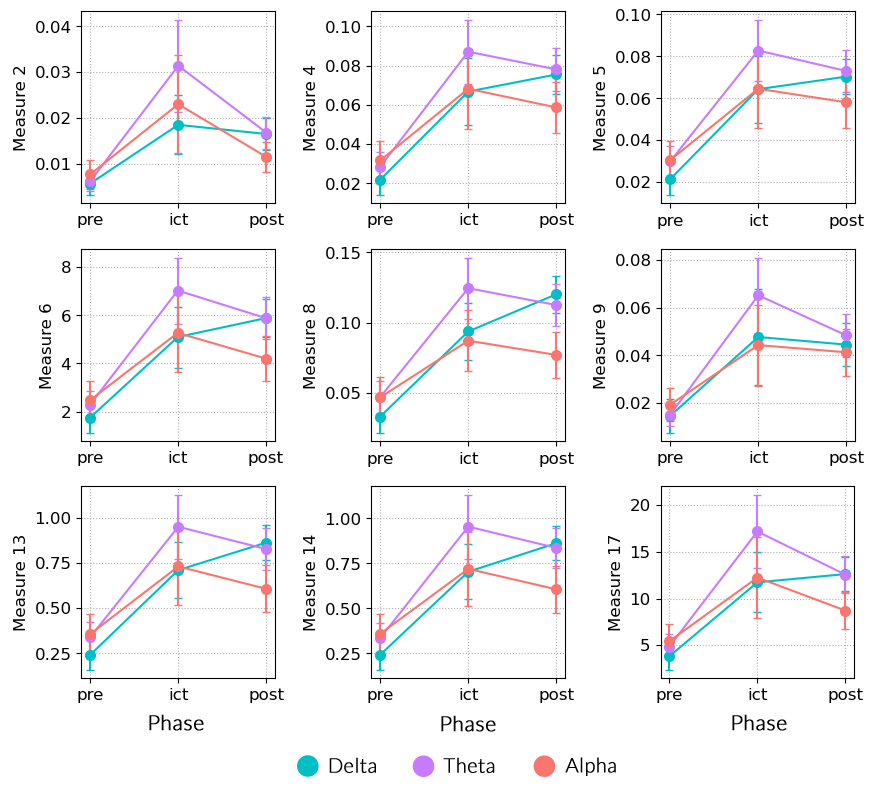}
 \caption{ {\small Grand means and standard deviations of the simplicial measures computed in the $2$-digraphs of the left hemisphere for each seizure phase (pre-ictal, ictal, post-ictal) and for each frequency band (delta, theta, alpha). See Table \ref{tab:measures-analysis} for the measure ids.}}
 \label{fig:grand-mean-q2_L}

\medskip

\includegraphics[scale=1.54]{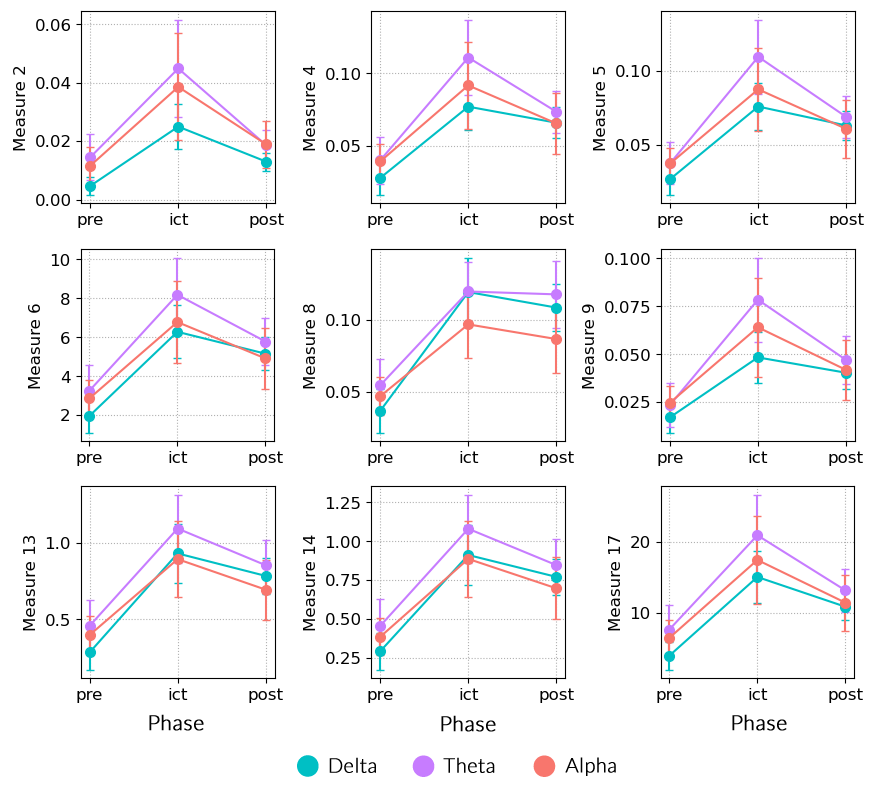}
 \caption{ {\small Grand means and standard deviations of the simplicial measures computed in the $2$-digraphs of the right hemisphere for each seizure phase (pre-ictal, ictal, post-ictal) and for each frequency band (delta, theta, alpha). See Table \ref{tab:measures-analysis} for the measure ids.}}
 \label{fig:grand-mean-q2_R}
\end{figure}


\begin{figure}[h!]
    \centering
\includegraphics[scale=1.54]{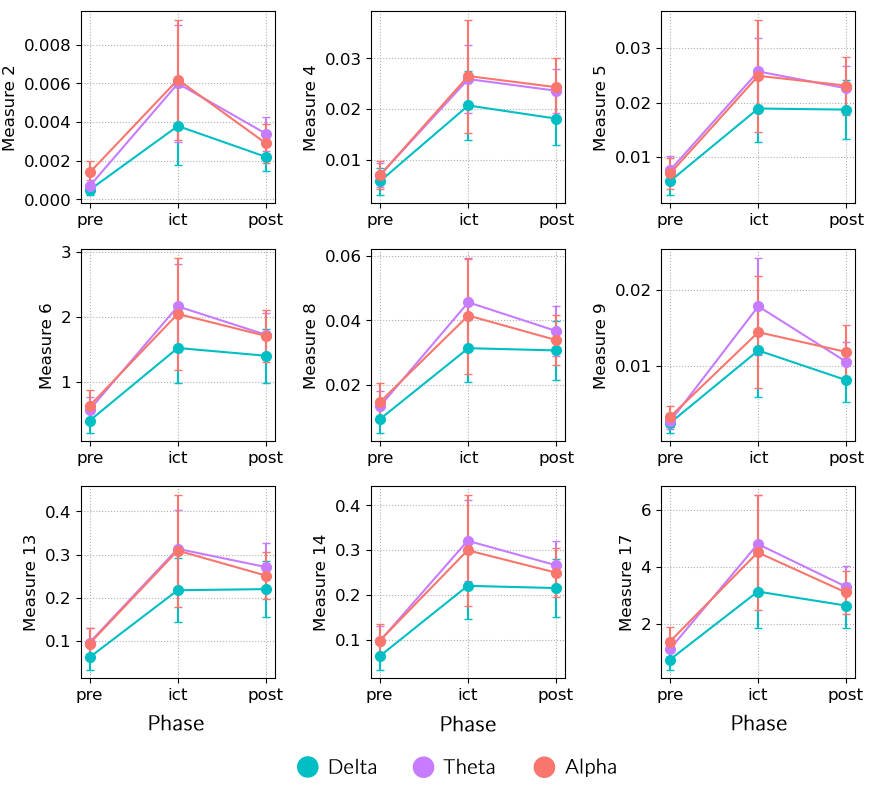}
 \caption{ {\small Grand means and standard deviations of the simplicial measures computed in the $3$-digraphs of the left hemisphere for each seizure phase (pre-ictal, ictal, post-ictal) and for each frequency band (delta, theta, alpha). See Table \ref{tab:measures-analysis} for the measure ids.}}
 \label{fig:grand-mean-q3_L}

\medskip

\includegraphics[scale=1.54]{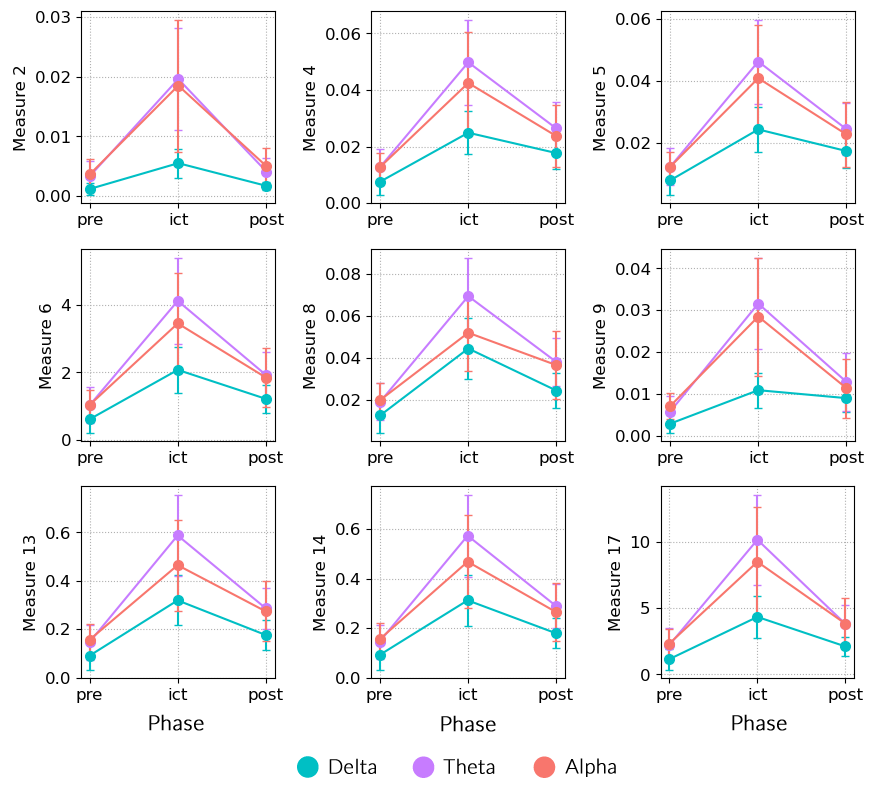}
 \caption{ {\small Grand means and standard deviations of the simplicial measures computed in the $3$-digraphs of the right hemisphere for each seizure phase (pre-ictal, ictal, post-ictal) and for each frequency band (delta, theta, alpha). See Table \ref{tab:measures-analysis} for the measure ids.}}
 \label{fig:grand-mean-q3_R}
\end{figure}


\begin{figure}[h!]
    \centering
\includegraphics[scale=1.54]{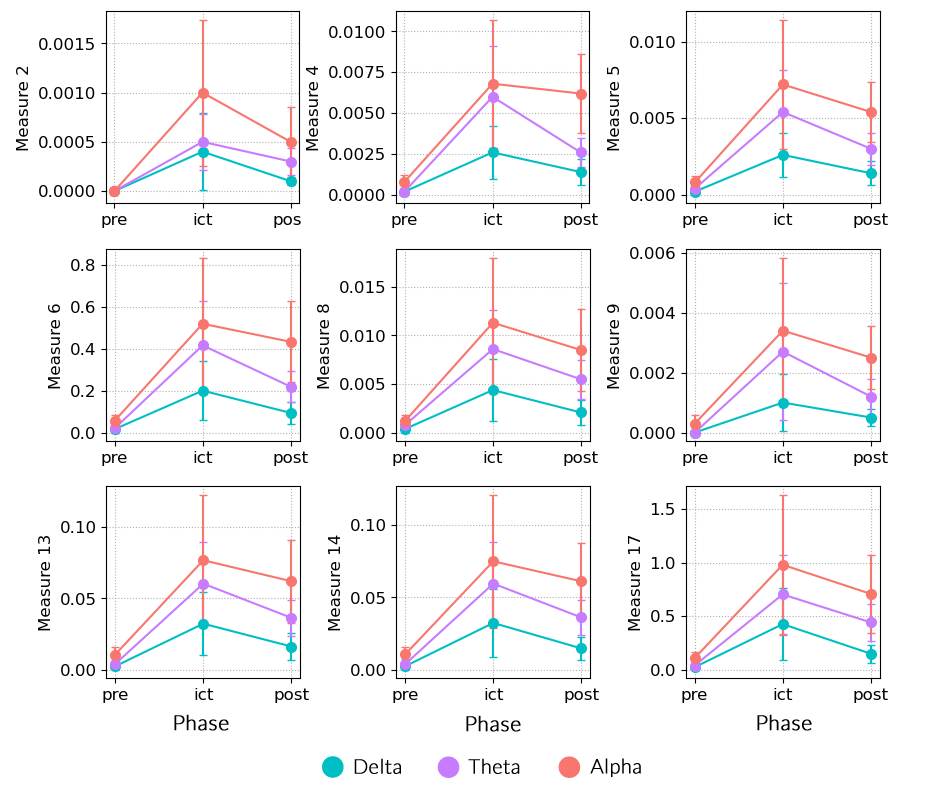}
 \caption{ {\small Grand means and standard deviations of the simplicial measures computed in the $4$-digraphs of the left hemisphere for each seizure phase (pre-ictal, ictal, post-ictal) and for each frequency band (delta, theta, alpha). See Table \ref{tab:measures-analysis} for the measure ids.}}
 \label{fig:grand-mean-q4_L}

\medskip

\includegraphics[scale=1.54]{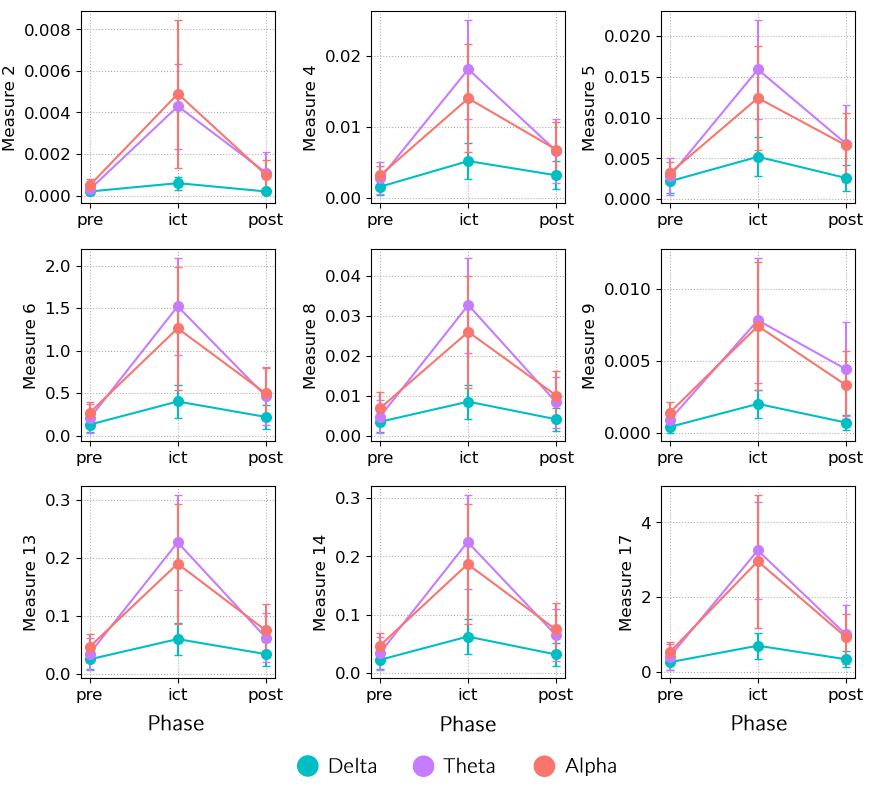}
 \caption{ {\small Grand means and standard deviations of the simplicial measures computed in the $4$-digraphs of the right hemisphere for each seizure phase (pre-ictal, ictal, post-ictal) and for each frequency band (delta, theta, alpha). See Table \ref{tab:measures-analysis} for the measure ids.}}
 \label{fig:grand-mean-q4_R}
\end{figure}


\begin{figure}[h!]
    \centering
\includegraphics[scale=0.75]{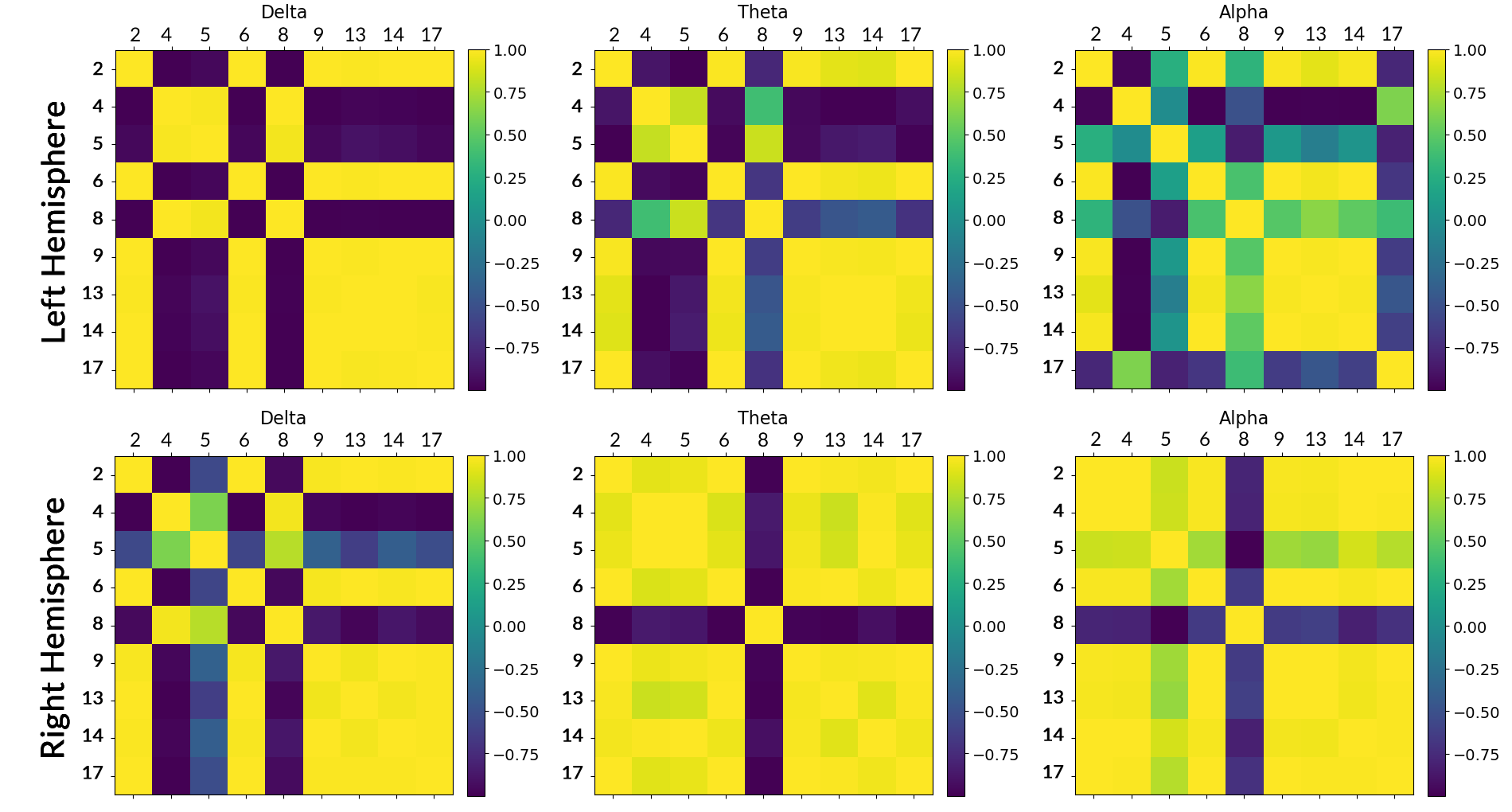}\label{fig:corr-std1}
\caption{ {\small Pearson's correlation coefficients between the simplicial measures across the pre-ictal, ictal, and post-ictal phases for the original iPDC networks (level $q=-1$). See Table \ref{tab:measures-analysis} for the measure ids.}}

\medskip

\includegraphics[scale=0.75]{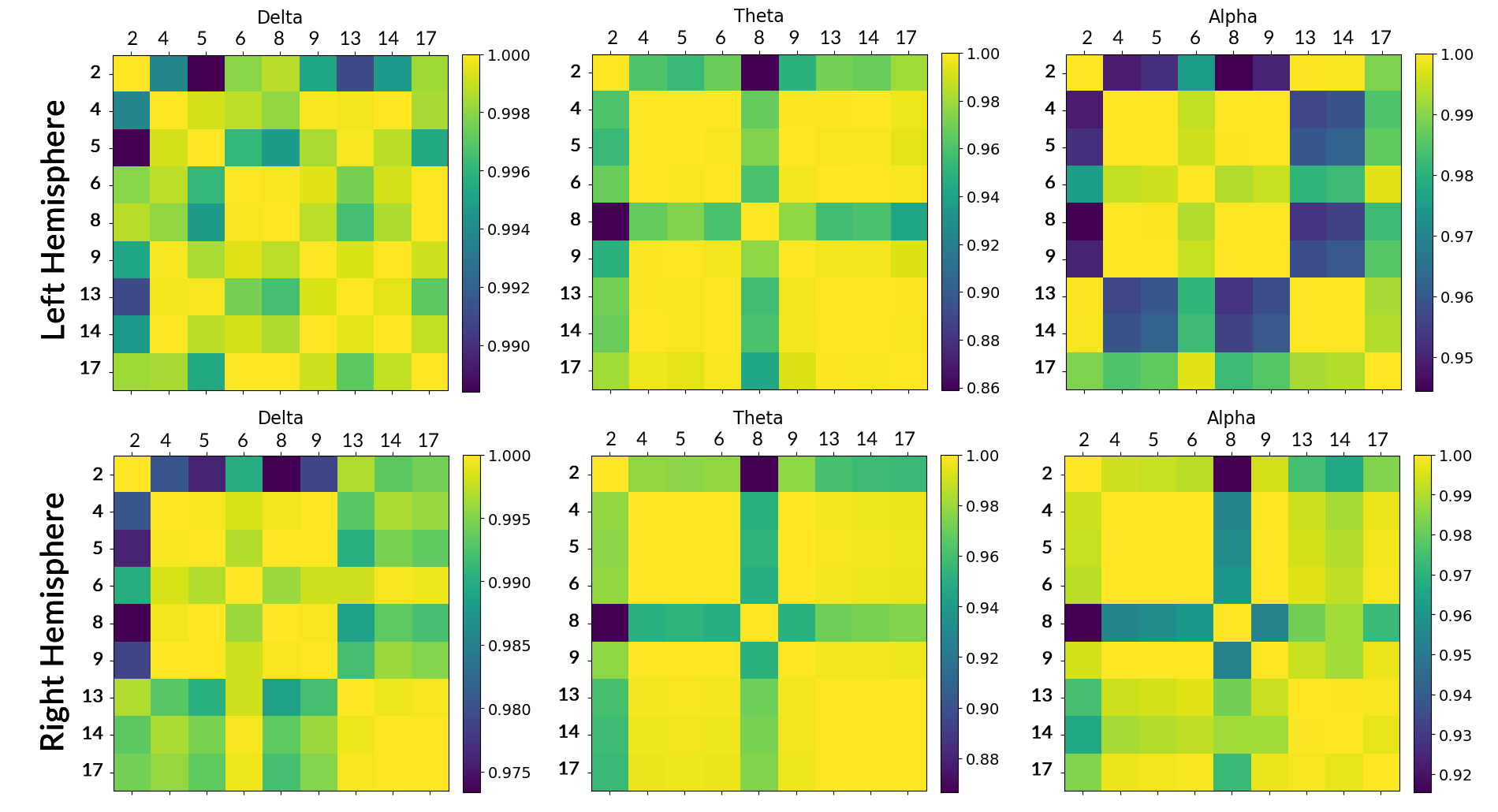}\label{fig:corr-q0}
 \caption{ {\small Pearson's correlation coefficients between the simplicial measures across the pre-ictal, ictal, and post-ictal phases at level $q=0$. See Table \ref{tab:measures-analysis} for the measure ids.}}

\medskip

\includegraphics[scale=0.75]{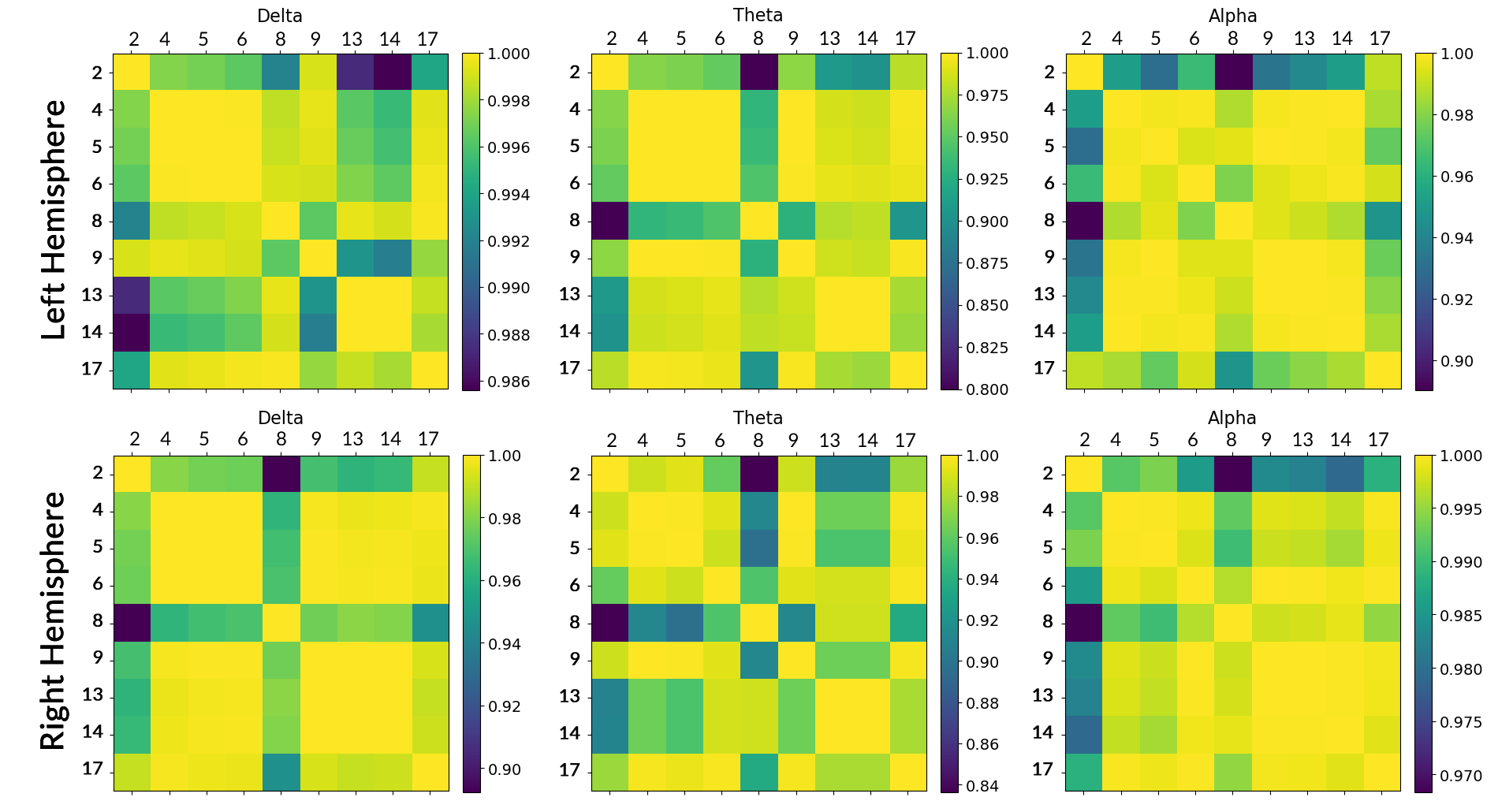}\label{fig:corr-q1}
 \caption{ {\small Pearson's correlation coefficients between the simplicial measures across the pre-ictal, ictal, and post-ictal phases at level $q=1$. See Table \ref{tab:measures-analysis} for the measure ids.}}
 
\end{figure}


\begin{figure}[h!]
    \centering
\includegraphics[scale=0.75]{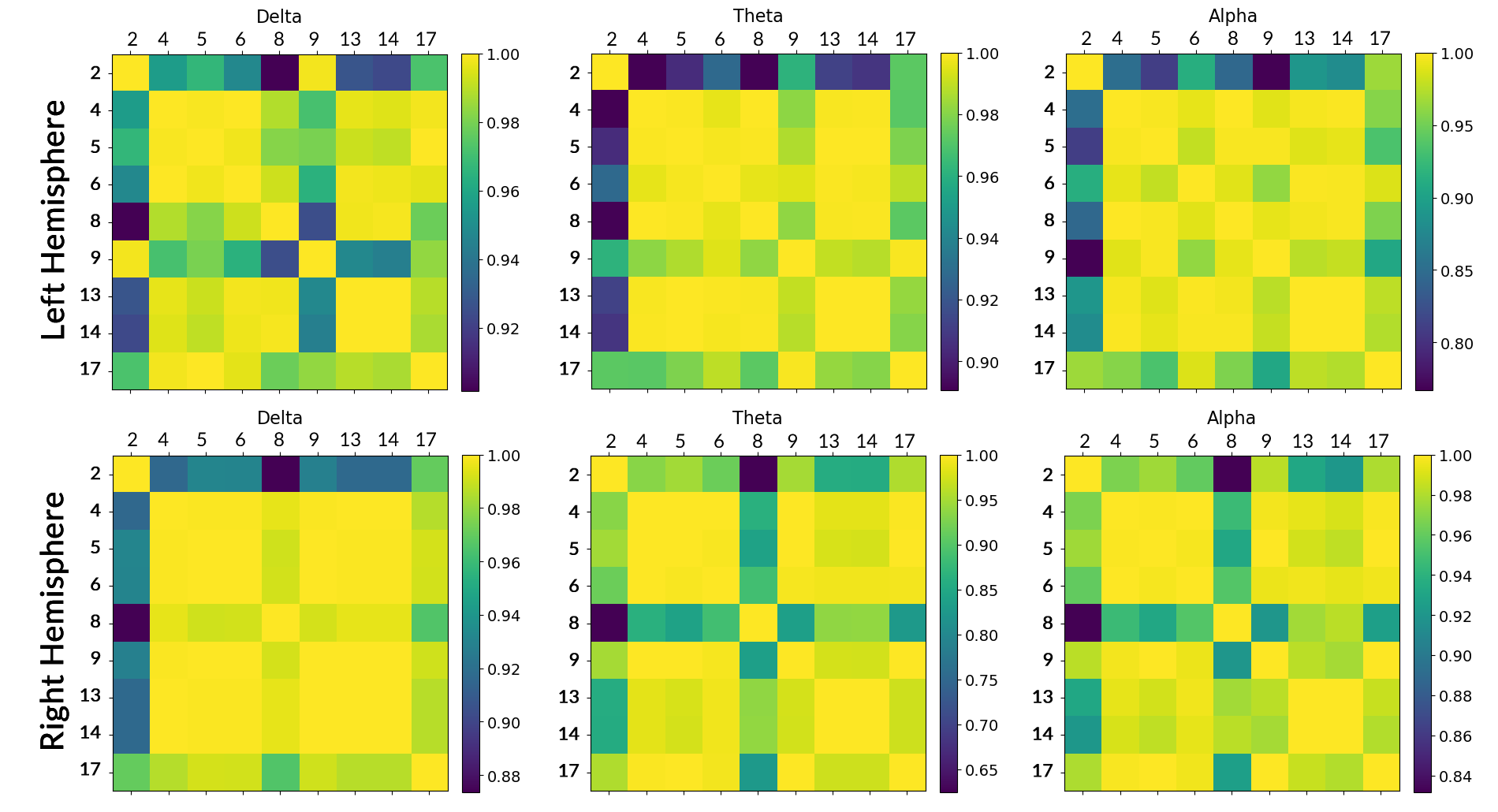}\label{fig:corr-q2}
 \caption{ {\small Pearson's correlation coefficients between the simplicial measures across the pre-ictal, ictal, and post-ictal phases at level $q=2$. See Table \ref{tab:measures-analysis} for the measure ids.}}

\medskip

\includegraphics[scale=0.75]{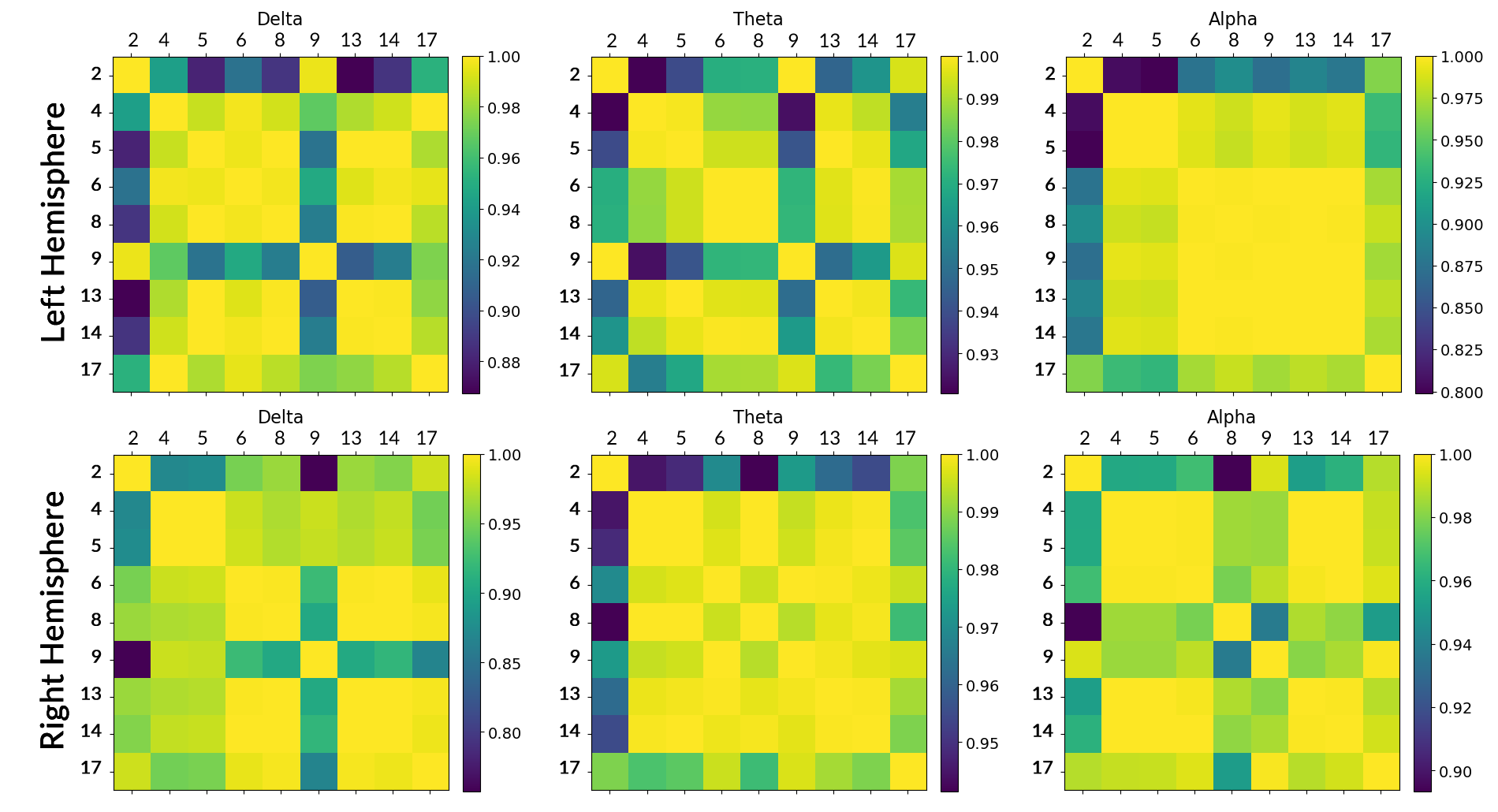}\label{fig:corr-q3}
 \caption{ {\small Pearson's correlation coefficients between the simplicial measures across the pre-ictal, ictal, and post-ictal phases at level $q=3$. See Table \ref{tab:measures-analysis} for the measure ids.}}

\medskip

\includegraphics[scale=0.75]{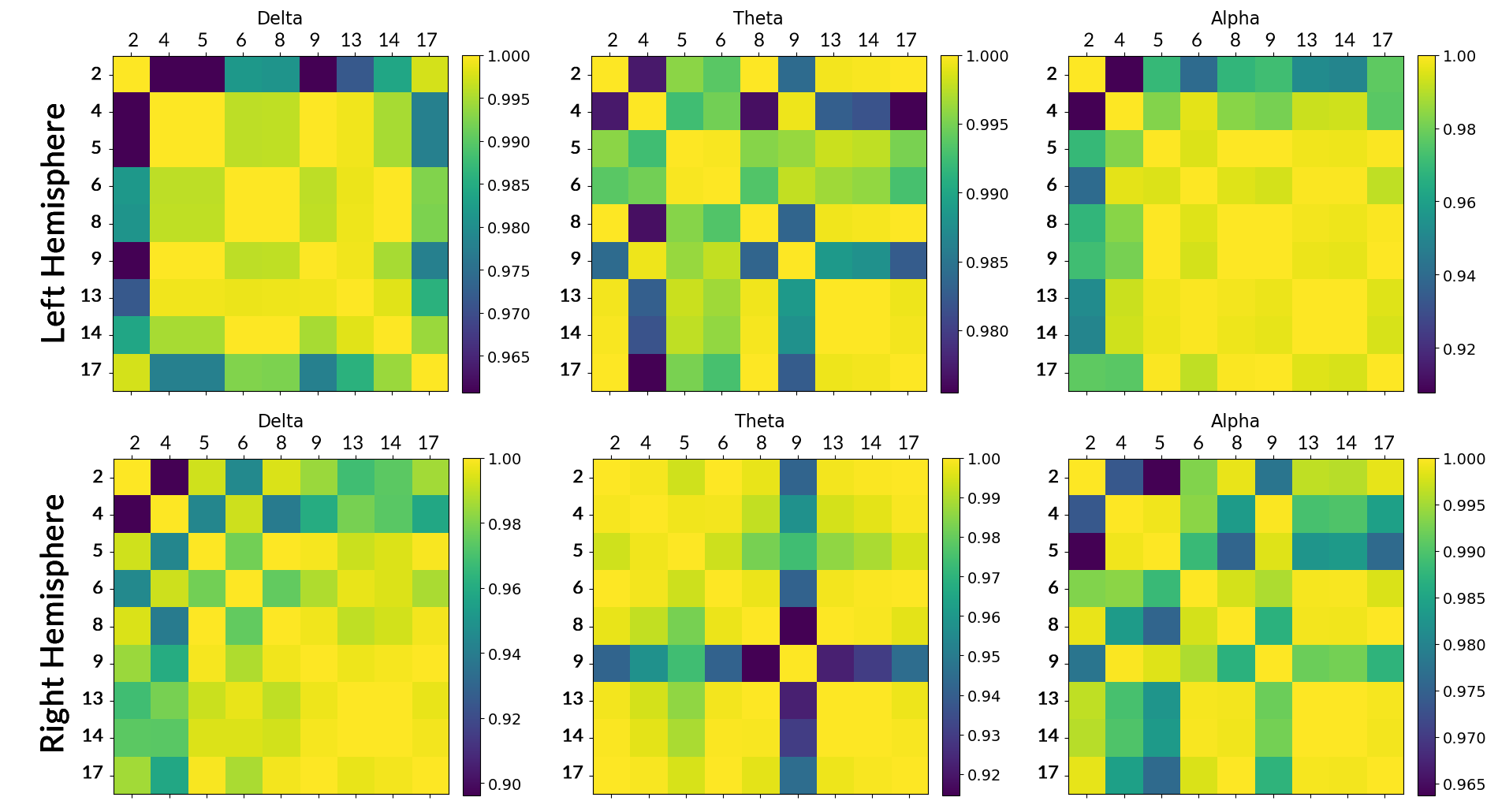}\label{fig:corr-q4}
  \caption{ {\small Pearson's correlation coefficients between the simplicial measures across the pre-ictal, ictal, and post-ictal phases at level $q=4$. See Table \ref{tab:measures-analysis} for the measure ids.}}
 
\end{figure}


\begin{figure}[h!]
    \centering
\includegraphics[scale=0.75]{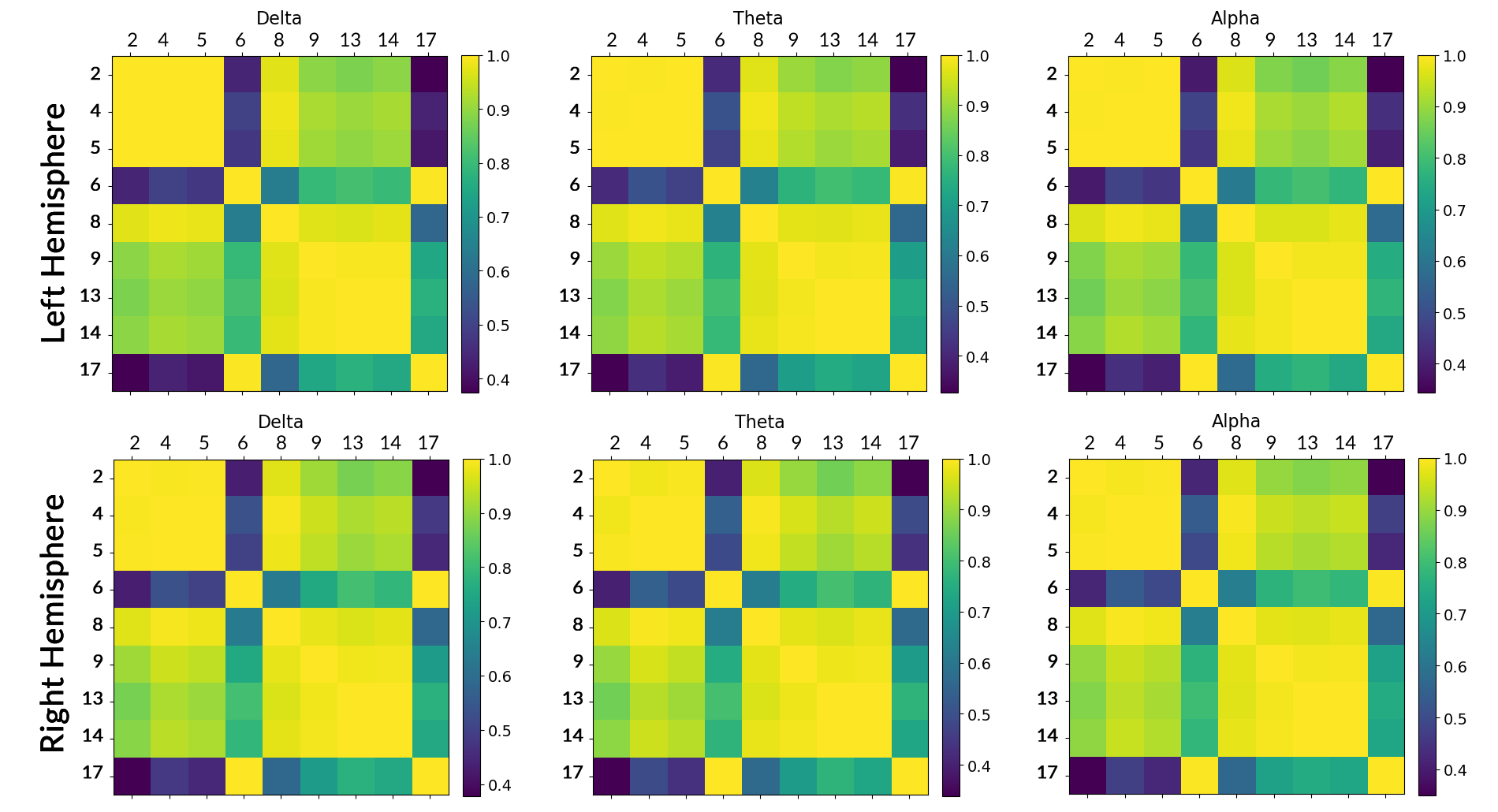}\label{fig:corr-pre}
 \caption{ {\small Pearson's correlation coefficients between the simplicial measures across the levels $q=-1,0,1,2,3,4$ in the pre-ictal phase. See Table \ref{tab:measures-analysis} for the measure ids.}}

\medskip

\includegraphics[scale=0.75]{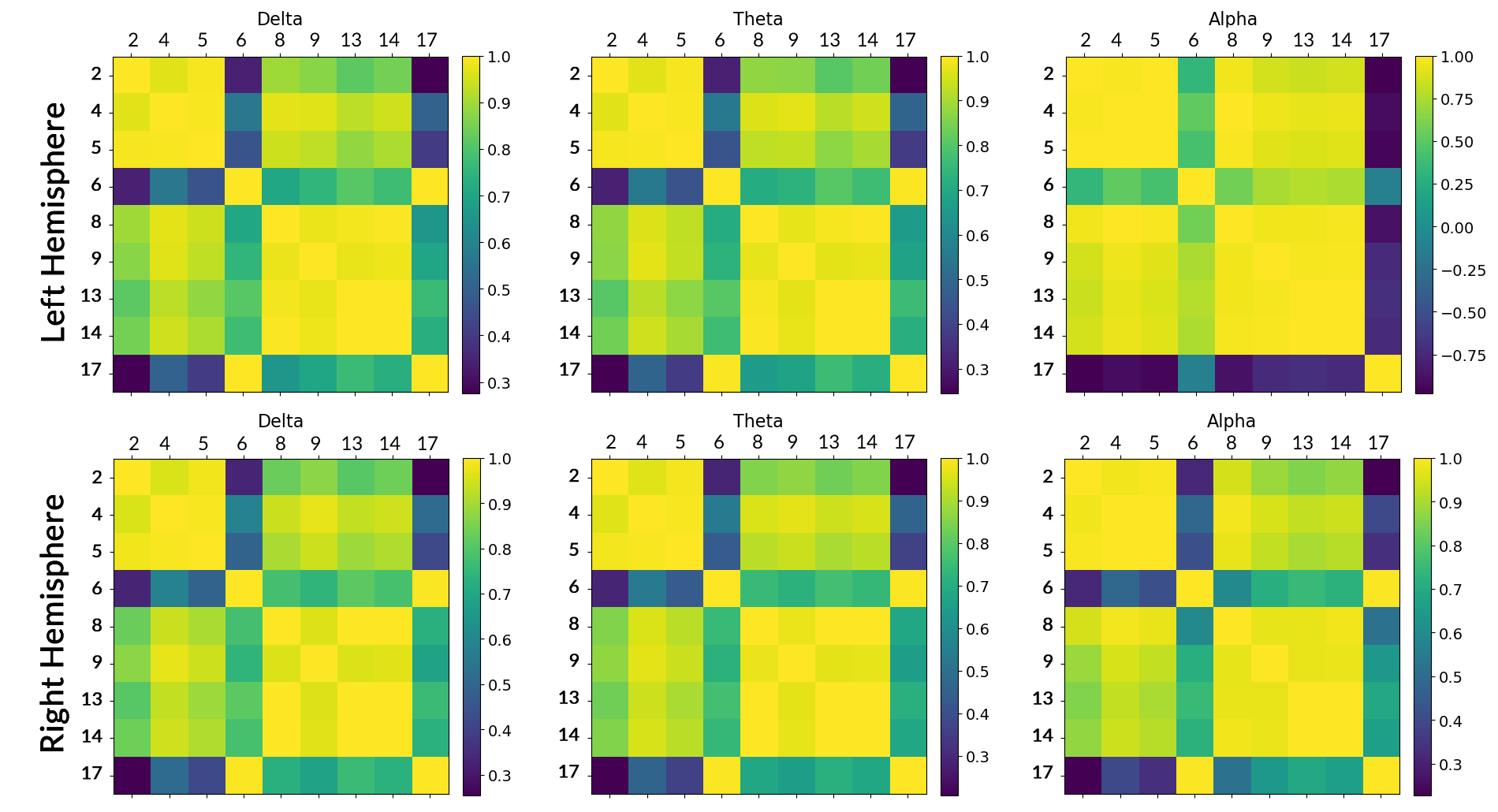}\label{fig:corr-ic}
 \caption{ {\small Pearson's correlation coefficients between the simplicial measures across the levels $q=-1,0,1,2,3,4$ in the ictal phase. See Table \ref{tab:measures-analysis} for the measure ids.}}

\medskip

\includegraphics[scale=0.75]{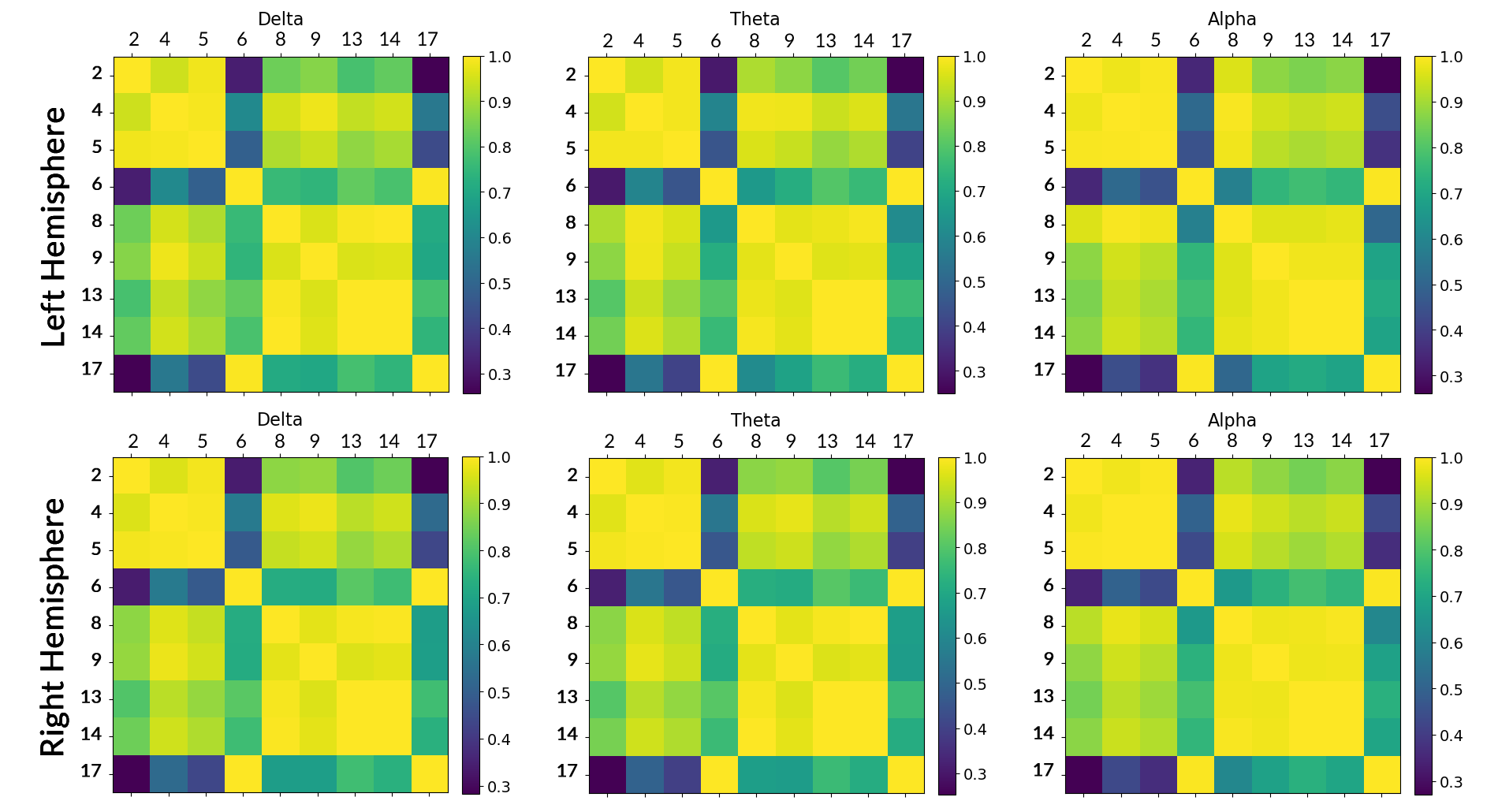}\label{fig:corr-pos}
 \caption{ {\small Pearson's correlation coefficients between the simplicial measures across the levels $q=-1,0,1,2,3,4$ in the post-ictal phase. See Table \ref{tab:measures-analysis} for the measure ids.}}

\end{figure}

\newpage

\begin{figure}[h!]
    \centering
\includegraphics[scale=0.72]{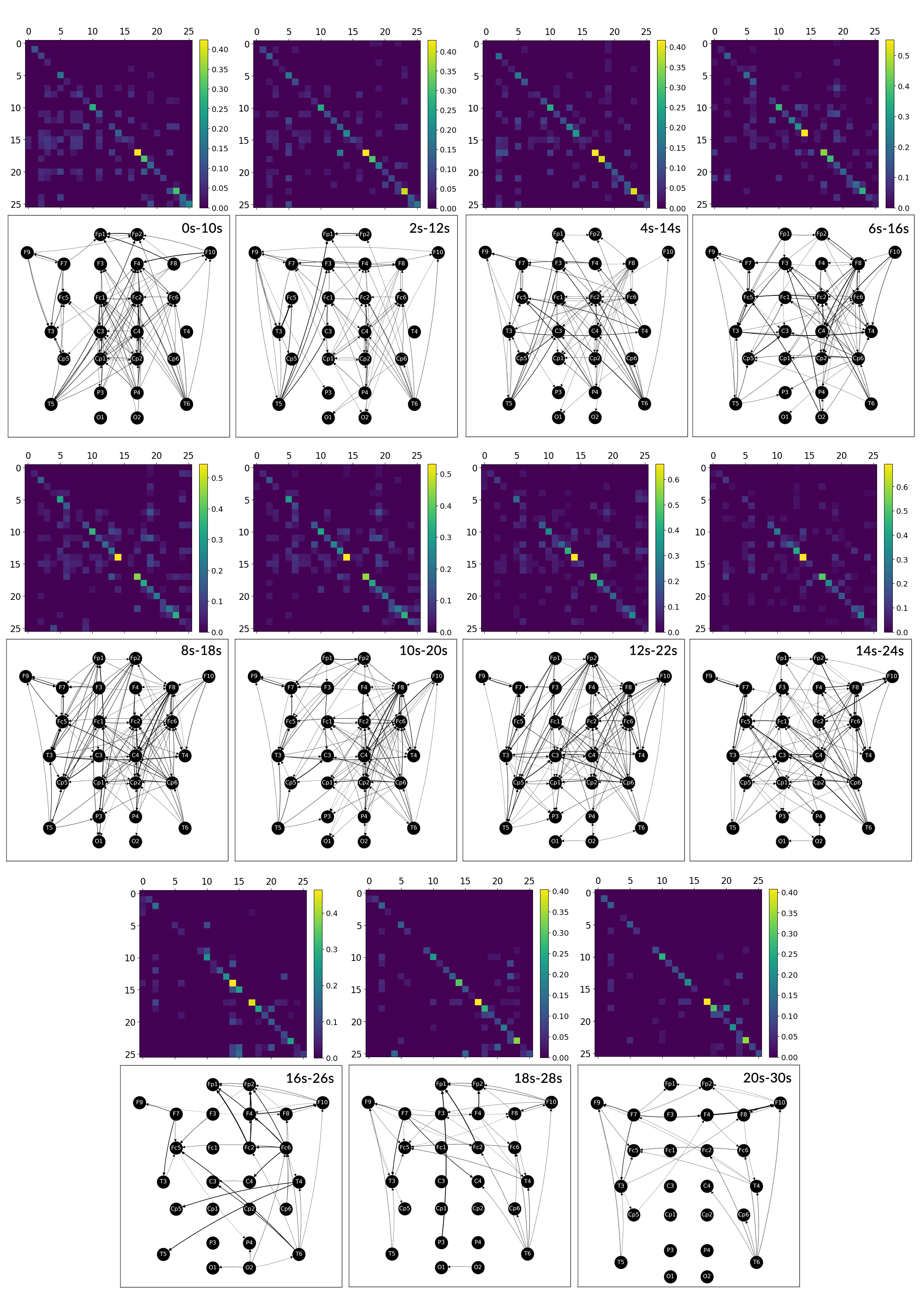}
 \caption{iPDC networks, together with their weighted adjacency matrices (only statistically significant ($p<0.001$) connections were considered), whose entries correspond to $|\iota\pi_{ij}(f)|^{2}$ (squared modulus), computed in the delta band of the pre-ictal phase of patient PN01. The networks were obtained using the sliding window technique with fixed-size windows of 10s and $80\%$ overlap (i.e. 8s) in a 30s interval immediately before the seizure. The nodes in the adjacency matrices are numbered from 0 to 25, and correspond to the following electrodes, starting from bottom to top, right to left: 0 (O2), 1 (O1), 2 (T6), 3 (P4), ..., 25 (Fp1).}
 \label{fig:pip-digraphs1}
\end{figure}

 \begin{figure}[h!]
 \includegraphics[scale=0.72]{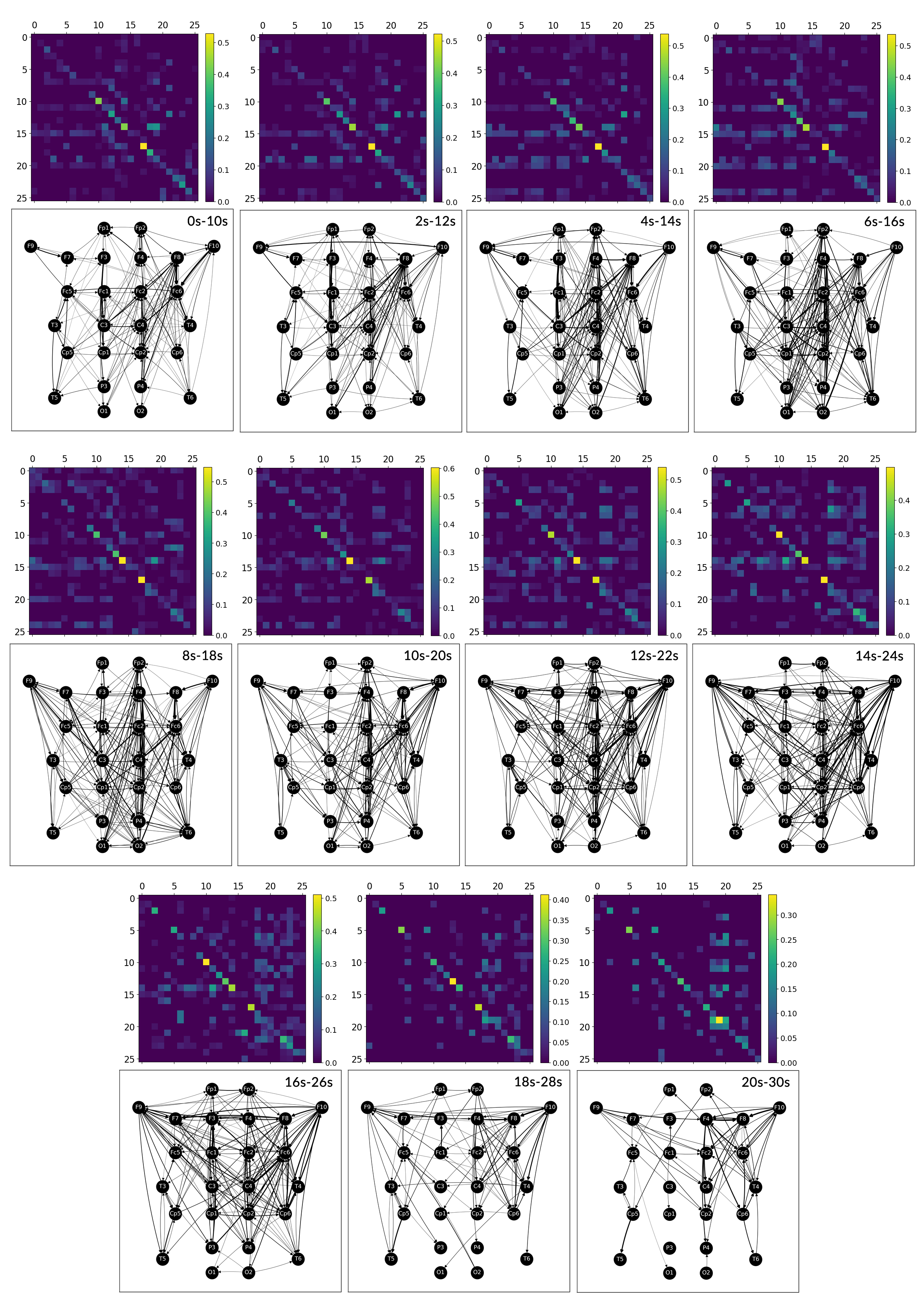}
  \caption{iPDC networks, together with their weighted adjacency matrices (only statistically significant ($p<0.001$) connections were considered), whose entries correspond to $|\iota\pi_{ij}(f)|^{2}$ (squared modulus), computed in the delta band of the ictal phase of patient PN01. The networks were obtained using the sliding window technique with fixed-size windows of 10s and $80\%$ overlap (i.e. 8s) in a 30s interval starting on the seizure onset. The nodes in the adjacency matrices are numbered from 0 to 25, and correspond to the following electrodes, starting from bottom to top, right to left: 0 (O2), 1 (O1), 2 (T6), 3 (P4), ..., 25 (Fp1).}
 \label{fig:pip-digraphs2}
\end{figure}

\begin{figure}[h!]
    \centering
\includegraphics[scale=0.72]{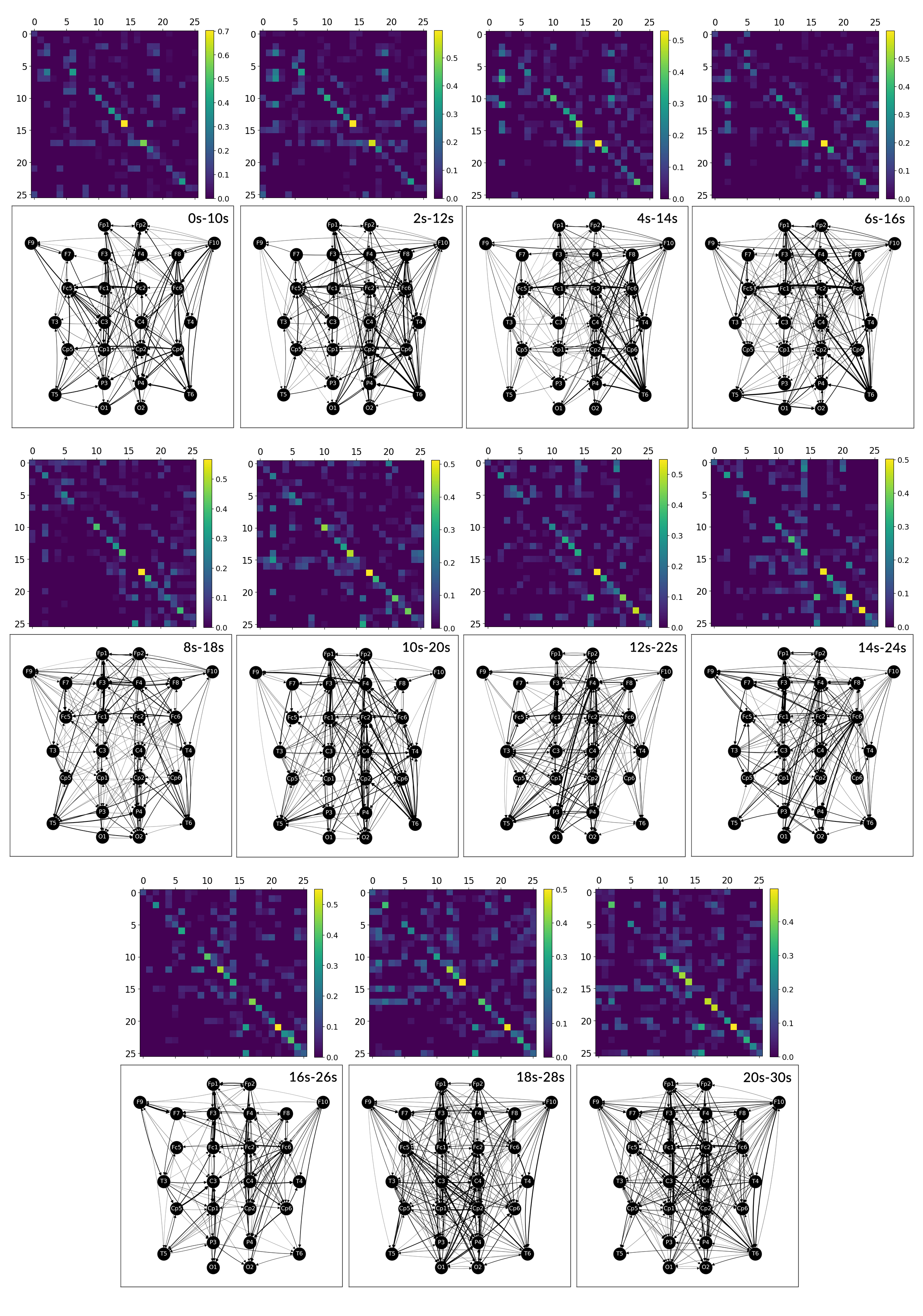}
 \caption{iPDC networks, together with their weighted adjacency matrices (only statistically significant ($p<0.001$) connections were considered), whose entries correspond to $|\iota\pi_{ij}(f)|^{2}$ (squared modulus), computed in the delta band of the post-ictal phase of patient PN01. The networks were obtained using the sliding window technique with fixed-size windows of 10s and $80\%$ overlap (i.e. 8s) in a 30s interval immediately after the seizure. The nodes in the adjacency matrices are numbered from 0 to 25, and correspond to the following electrodes, starting from bottom to top, right to left: 0 (O2), 1 (O1), 2 (T6), 3 (P4), ..., 25 (Fp1).}
 \label{fig:pip-digraphs3}
\end{figure}

\end{document}